\documentclass[acmsmall,10pt]{acmart}

\acmJournal{PACMPL}
\acmVolume{1}
\acmNumber{CONF}
\acmArticle{1}
\acmYear{2018}
\acmMonth{1}
\acmDOI{}
\startPage{1}

\setcopyright{none}

\renewcommand\footnotetextcopyrightpermission[1]{}

\usepackage{booktabs}
\usepackage{subcaption}

\usepackage{multirow}

\newcommand*\circled[1]{\textcircled{\scriptsize{#1}}}

\DeclareMathOperator*{\cart}{\times}
\usepackage{amsmath}
\usepackage{amsthm}
\usepackage{MnSymbol}

\usepackage{xcolor}
\usepackage{soul}
\newcommand{\mathcolorbox}[2]{\colorbox{#1}{$\displaystyle #2$}}
\definecolor{lightgray}{rgb}{0.83, 0.83, 0.83}
\newcommand{\tp}[1]{\mathcolorbox{lightgray}{#1}}

\usepackage{centernot}
\usepackage{mathtools}
\usepackage{stmaryrd}

\makeatletter
\newcommand{\xMapsto}[2][]{\ext@arrow 0599{\Mapstofill@}{#1}{#2}}
\def\Mapstofill@{\arrowfill@{\Mapstochar\Relbar}\Relbar\Rightarrow}
\makeatother

\usepackage{pbox}
\usepackage{arydshln}

\allowdisplaybreaks

\usepackage{listings}
\let\origthelstnumber\thelstnumber
\makeatletter
\newcommand*\Suppressnumber{%
  \lst@AddToHook{OnNewLine}{%
    \let\thelstnumber\relax%
     \advance\c@lstnumber-\@ne\relax%
    }%
}

\newcommand*\Reactivatenumber{%
  \lst@AddToHook{OnNewLine}{%
   \let\thelstnumber\origthelstnumber%
   \advance\c@lstnumber\@ne\relax}%
}
\makeatother

\usepackage{xparse}
\NewDocumentCommand{\overarrow}{O{=} O{\uparrow} m}{%
  \overset{\makebox[0pt]{\begin{tabular}{@{}c@{}}#3\\[0pt]\ensuremath{#2}\end{tabular}}}{#1}
}
\newcommand{\IN}{\mathbb{N}}
\newcommand{\INz}{\mathbb{N}_0}
\newcommand{\IZ}{\mathbb{Z}}
\newcommand{\IQ}{\mathbb{Q}}
\newcommand{\IR}{\mathbb{R}}
\newcommand\dplus{\mbox{+\hspace*{-3px}+}}

\newcommand\dplusn[1]{\mbox{${+\hspace*{-3px}+}_{#1}$}}

\newcommand{\bij}{%
  \hookrightarrow\mathrel{\mspace{-13mu}}\rightarrow
}

\newcommand\mdh{\texttt{md\_hom}}
\newcommand\co[1]{\circledast_{#1}}

\newcommand\concat[2]{\underset{#2}{\dplus_{#1}}}
\newcommand\fsize[4]{{\overset{#1}{\underset{{#2}}{\Rightarrow}}{\hspace{-2.5px}^\texttt{#3}_\texttt{#4}}}}
\newcommand\sfsize[4]{{\overset{#1}{\underset{{#2}}{\Rightarrow}}{\hspace{-0px}^\texttt{#3}_\texttt{#4}}}}
\newcommand\size[3]{{\overset{#1}{\Rightarrow}{\hspace{-2.5px}^\texttt{#2}_\texttt{#3}}}}
\newcommand\bsize[3]{{\overset{#1}{\Rightarrow}{\hspace{-2.5px}^\texttt{#2}_\texttt{#3}}}}
\newcommand\ssize[3]{\overset{#1}{\Rightarrow}{\hspace{-0px}^\texttt{#2}_\texttt{#3}}}
\newcommand\sbsize[3]{\overset{#1}{\Rightarrow}{\hspace{-0px}^\texttt{#2}_\texttt{#3}}}
\newcommand\type{\texttt{TYPE}}

\newcommand\CO{\texttt{CO}}
\newcommand\MDH{\texttt{MDH}}
\newcommand\IDXs{\texttt{MDA-IDX-SETs}}
\newcommand\BUFIDXs{\texttt{BUF-IDX-SETs}}
\newcommand\IDXsxIDXs{\IDXs\,\overset{_{\boldsymbol{\cdot}}}{\times}\,\IDXs}
\newcommand\IDXFCT{\texttt{MDA-IDX-to-BUF-IDX}}

\newcommand\MDA{\mathfrak{a}}
\newcommand\Buf{\mathfrak{b}}
\newcommand\iv{\texttt{inp\_view} }
\newcommand\ov{\texttt{out\_view} }

\newcommand\idx{\mathfrak{idx}}

\DeclareFontFamily{U}{mathb}{}
\DeclareFontShape{U}{mathb}{m}{n}{
  <-5.5> mathb5
  <5.5-6.5> mathb6
  <6.5-7.5> mathb7
  <7.5-8.5> mathb8
  <8.5-9.5> mathb9
  <9.5-11.5> mathb10
  <11.5-> mathbb12
}{}

\newcommand\ASMLVL{\texttt{ASM-LVL}}
\newcommand\MDHLVL{\texttt{MDH-LVL}}

\theoremstyle{definition}

\newtheorem{definition}{Definition}

\newtheorem{lemma}{Lemma}
\newtheorem{proposition}{Proposition}
\newtheorem{notation}{Notation}
\newtheorem{example}{Example}

\newtheorem{note}{Note}

\makeatletter
\newcommand*{\da@rightarrow}{\mathchar"0\hexnumber@\symAMSa 4B }
\newcommand*{\da@leftarrow}{\mathchar"0\hexnumber@\symAMSa 4C }
\newcommand*{\xdashrightarrow}[2][]{%
  \mathrel{%
    \mathpalette{\da@xarrow{#1}{#2}{}\da@rightarrow{\,}{}}{}%
  }%
}
\newcommand{\xdashleftarrow}[2][]{%
  \mathrel{%
    \mathpalette{\da@xarrow{#1}{#2}\da@leftarrow{}{}{\,}}{}%
  }%
}
\newcommand*{\da@xarrow}[7]{%
  \sbox0{$\ifx#7\scriptstyle\scriptscriptstyle\else\scriptstyle\fi#5#1#6\m@th$}%
  \sbox2{$\ifx#7\scriptstyle\scriptscriptstyle\else\scriptstyle\fi#5#2#6\m@th$}%
  \sbox4{$#7\dabar@\m@th$}%
  \dimen@=\wd0
  \ifdim\wd2 >\dimen@
    \dimen@=\wd2
  \fi
  \count@=2
  \def\da@bars{\dabar@\dabar@}%
  \@whiledim\count@\wd4<\dimen@\do{%
    \advance\count@\@ne
    \expandafter\def\expandafter\da@bars\expandafter{%
      \da@bars
      \dabar@
    }%
  }%
  \mathrel{#3}%
  \mathrel{%
    \mathop{\da@bars}\limits
    \ifx\\#1\\%
    \else
      _{\copy0}%
    \fi
    \ifx\\#2\\%
    \else
      ^{\copy2}%
    \fi
  }%
  \mathrel{#4}%
}
\makeatother

\begin{document}

\title[\small{(De/Re)-Composition of Data-Parallel Computations via Multi-Dimensional Homomorphisms}]{(De/Re)-Composition of Data-Parallel Computations \\ via Multi-Dimensional Homomorphisms}
\subtitle{Full Version}
\titlenote{
    A short version of this paper is published at ACM TOPLAS~\cite{10.1145/3665643}.
    The short version relies on a simplified formal foundation, for better illustration and easier understanding of our novel concepts and methodologies introduced in this~paper.
}

\author{Ari Rasch}
\affiliation{
  \institution{University of Muenster}
  \city{Muenster}
  \country{Germany}
}
\email{a.rasch@uni-muenster.de}

\begin{abstract}
Data-parallel computations, such as linear algebra routines (BLAS) and stencil computations, constitute one of the most relevant classes in parallel computing, e.g., due to their importance for deep learning.
Efficiently de-composing such computations for the memory and core hierarchies of modern architectures and re-composing the computed intermediate results back to the final result~--~we say \emph{(de/re)-composition} for short~--~is key to achieve high performance for these computations on, e.g., GPU and CPU.
Current high-level approaches to generating data-parallel code are often restricted to a particular subclass of data-parallel computations and architectures (e.g., only linear algebra routines on only GPU, or only stencil computations), and/or the approaches rely on a user-guided optimization process for a well-performing (de/re)-composition of computations, which is complex and error prone for the user.

We formally introduce a systematic (de/re)-composition approach, based on the algebraic formalism of \emph{Multi-Dimensional Homomorphisms~(MDHs)}\footnote{\url{https://mdh-lang.org}}.
Our approach is designed as general enough to be applicable to a wide range of data-parallel computations and for various kinds of target parallel architectures.
To efficiently target the deep and complex memory and core hierarchies of contemporary architectures, we exploit our introduced (de/re)-composition approach for a correct-by-construction, parametrized cache blocking and parallelization strategy.
We show that our approach is powerful enough to express, in the same formalism, the (de/re)-composition strategies of different classes of state-of-the-art approaches (scheduling-based, polyhedral, etc), and we demonstrate that the parameters of our strategies enable systematically generating code that can be fully automatically optimized (auto-tuned) for the particular target architecture and characteristics of the input and output data (e.g., their sizes and memory layouts).
Particularly, our experiments confirm that via auto-tuning,
we achieve higher performance than state-of-the-art approaches,
including hand-optimized solutions provided by vendors (such as NVIDIA cuBLAS/cuDNN and Intel oneMKL/oneDNN), on real-world data sets and for a variety of data-parallel computations, including:~%
linear algebra routines, stencil and quantum chemistry computations, data mining algorithms, and computations that recently gained high attention due to their relevance for deep learning.
\end{abstract}



\maketitle
\thispagestyle{empty}
\pagestyle{plain}

\section{Introduction}
\label{ch:introduction}

Data-parallel computations constitute one of the most relevant classes in parallel computing.
Important examples of such computations include
linear algebra routines~(BLAS)~\cite{1437325},
various kinds of stencil computations (e.g., Jacobi method and convolutions)~\cite{10.1145/3168824},
quantum chemistry computations~\cite{8661182},
and data mining algorithms~\cite{10.1145/3297280.3297330}.
The success of many application areas critically depends on achieving high performance for their data-parallel building blocks, on a variety of parallel architectures.
For example, highly-optimized BLAS implementations combined with the computational power of modern GPUs currently enable deep learning to significantly outperform other existing machine learning approaches (e.g., for speech recognition and image classification).

Data-parallel computations are characterized by applying the same function (a.k.a \emph{scalar function}) to each point in a multi-dimensional grid of data (a.k.a. \emph{array}), and combining the obtained intermediate results in the grid's different dimensions using so-called \emph{combine operators}.
Figures~\ref{example_intro_1} and~\ref{example_intro_2} illustrate data parallelism using as examples two popular computations:~%
i)~linear algebra \break routine \emph{Matrix-Vector multiplication~(\texttt{MatVec})}, and~%
ii)~stencil computation \emph{Jacobi~(\texttt{Jacobi1D})}.
In the case of \texttt{MatVec}, the grid is 2-dimensional and consists of pairs, each pointing to one element of the input matrix $M_{i,k}$ and the vector $v_{k}$.
To each pair, scalar function $f(M_{i,k},v_k):=M_{i,k}*v_{k}$~(multiplication) is applied, and results in the $i$-dimension are combined using combine operator $\co{1}( \,(x_1,\dotsc,x_n)\, , \,(y_1,\dotsc,y_m)\,):=(x_1,\dotsc,x_n,y_1,\dotsc,y_m)$~(concatenation) and in $k$-dimension using operator
$\co{2}( \,(x_1,\dotsc,x_n)\, , \,(y_1,\dotsc,y_n)\,):=(x_1+y_1,\dotsc,x_n+y_n)$~(point-wise addition).
Similarly, the scalar function of \texttt{Jacobi1D} is $f(v_{i+0},v_{i+1},v_{i+2}):=c*(v_{i+0}+v_{i+1}+v_{i+2})$ which computes the Jacobi-specific function for an arbitrary but fixed constant $c$;~%
\texttt{Jacobi1D}'s combine operator $\co{1}$ is concatenation.
We formally define scalar functions and combine operators
later in this paper.

\begin{figure}[t]
\centering
\includegraphics[width=\textwidth]{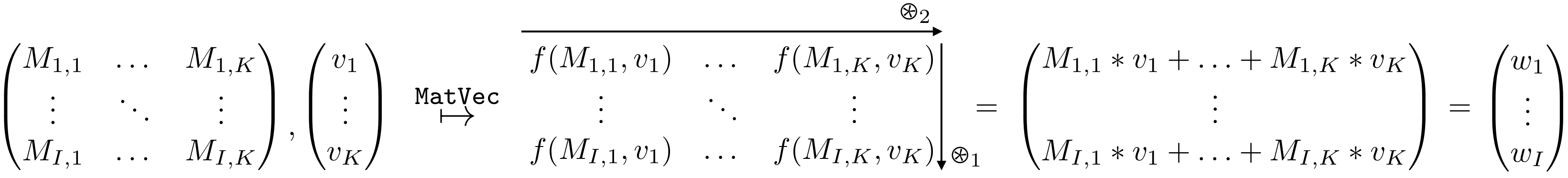}
\caption{Data parallelism illustrated using the example \emph{Matrix-Vector Multiplication} (\texttt{MatVec})}
\label{example_intro_1}
\end{figure}

\begin{figure}[b]
\centering
\includegraphics[scale=0.24]{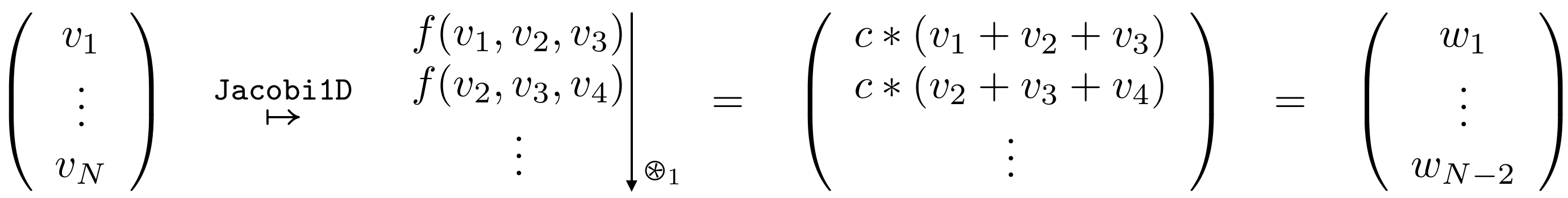}
\caption{Data parallelism illustrated using the example \emph{Jacobi 1D} ($\texttt{Jacobi1D}$)
}
\label{example_intro_2}
\end{figure}

Achieving high performance for data-parallel computations is considered important in both academia and industry, but has proven to be challenging.
In particular, achieving \emph{high performance} that is \emph{portable}
(i.e., the same program code achieves a consistently high level of performance across different architectures and characteristics of the input/output data, e.g., their size and memory layout)
and in a \emph{user-productive} way is identified as an ongoing, major research challenge.
This is because for high performance, an efficient \emph{(de/re)-composition} of computations
(illustrated in Figure~\ref{fig_decomp} and discussed thoroughly in this paper)
is required to efficiently
break down a computation for
the deep and complex memory and core hierarchies of state-of-the-art architectures,
via efficient cache blocking
and parallelization strategies.
Moreover, to achieve performance that is portable across architectures, the programmer has to consider that architectures often differ significantly in their characteristics~\cite{DBLP:journals/corr/abs-1911-11313}~--~depth of memory and core hierarchies, automatically managed caches (as in CPUs) vs manually managed caches (as in GPUs),~etc~--~which poses further challenges on identifying an efficient (de/re)-composition of computations.
Productivity is often also hampered:~state-of-the-art programming models (such as OpenMP~\cite{openmp-specification} for CPU, CUDA~\cite{cuda-specification} for GPU, and OpenCL~\cite{opencl-specification} for multiple kinds of architectures) operate on a low abstraction level;~thereby, the models require from the programmer explicitly implementing a well-performing (de/re)-composition, which involves
complex and error-prone index computations, explicitly managing memory and threads on multiple layers, etc.

\begin{figure}[b!]
    \centering
    \includegraphics[width=\textwidth]{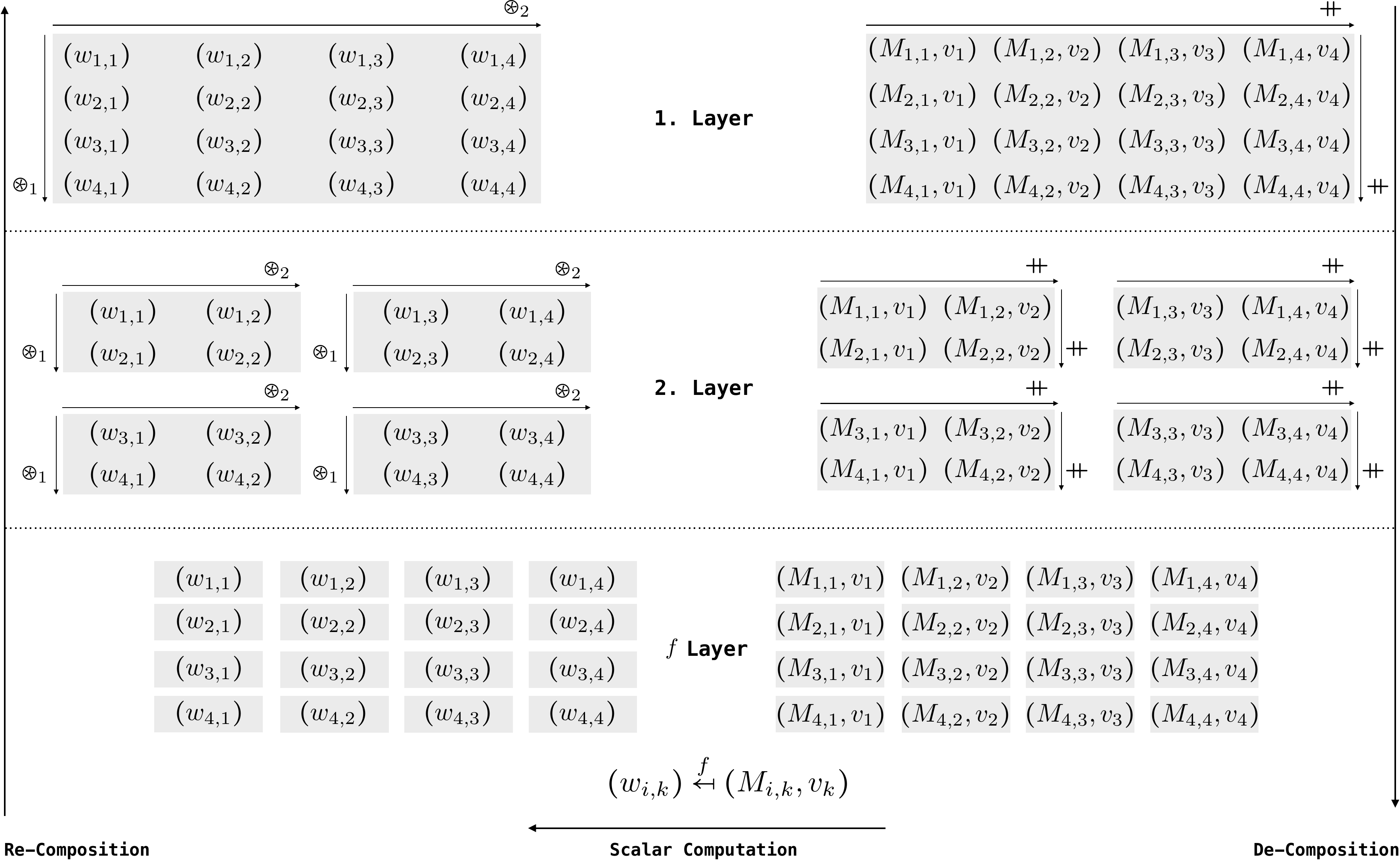}
    \caption{
    Example (de/re)-composition of \texttt{MatVec} (Figure~\ref{example_intro_1}) on a $4\times4$ input matrix~$M$ and a $4$-sized vector~$v$:
    i)~the \emph{de-composition phase} (right part of the figure) partitions the concatenated input data into parts (a.k.a. \emph{tiles} in programming), where $\dplus$ denotes the concatenation operator;
    ii)~to each part, scalar function $f$ is applied in the \emph{scalar phase} (bottom part of figure), which is defined for \texttt{MatVec} as: multiplying matrix element $M_{i,k}$ with vector element $v_k$, resulting in element $w_{i,k}$;
    iii)~the \emph{re-composition phase} (figure's left part) combines the computed parts to the final result, using combine operator $\co{1}$ for the first dimension (defined as \emph{concatenation} in the case of \texttt{MatVec}) and operator $\co{2}$ (\emph{point-wise addition}) for the second dimension.
    All basic building blocks (\emph{scalar function}, \emph{combine operator}, $\dotsc$)
    and concepts (e.g. \emph{partitioning})
    are defined in this paper, based on algebraic concepts.
    For simplicity, this example presents a (de/re)-composition on $2$ layers only, and we partition the input for this example
    into parts that have straightforward, equal sizes.
    Optimized values of
    semantics-preserving
    parameters
    (a.k.a. \emph{tuning parameters}), such as the number of parts and the application order of combine operators, are crucial for achieving high performance, as we discuss in this paper.
    Phases are arranged from right to left, inspired by the application order of function composition, as we also discuss later.
    }
    \label{fig_decomp}
\end{figure}

Current high-level approaches to generating data-parallel code usually struggle with addressing in one combined approach all three challenges:~\emph{performance}, \emph{portability}, and \emph{productivity}.
For example, approaches such as
Halide~\cite{10.1145/2491956.2462176},
Apache TVM~\cite{222575},
Fireiron~\cite{10.1145/3410463.3414632}, and
LoopStack~\cite{https://doi.org/10.48550/arxiv.2205.00618}
achieve high performance, but incorporate the user into the optimization process~--~by requiring from the user explicitly expressing optimizations in a so-called \emph{scheduling language}~--~which is error prone and needs expert knowledge about low-level code optimizations, thus hindering user's productivity.
In contrast, \emph{polyhedral approaches}, such as
Pluto~\cite{bondhugula2008pluto},
PPCG~\cite{10.1145/2400682.2400713}, and
Facebook's TC~\cite{10.1145/3355606},
are often fully automatic and thus productive, but usually specifically designed toward a particular architecture
(e.g., only GPU as TC or only CPU as Pluto)
and thus not portable.
\emph{Functional approaches}, e.g., Lift~\cite{10.1145/2784731.2784754}, are productive for functional programmers (e.g., with experience in \emph{Haskell}~\cite{haskell} programming, which relies on small, functional building blocks for expressing computations), but the approaches often have difficulties in automatically achieving the full performance potential of architectures~\cite{8891668}.
Furthermore, many of the existing approaches are specifically designed toward a particular subclass of data-parallel computations only, e.g., only tensor operations (as LoopStack and TC) or only matrix multiplication (as Fireiron), or they require significant extensions for new subclasses (as Lift for matrix multiplication~\cite{10.1145/2884045.2884046} and stencil computations~\cite{10.1145/3168824}), which further hinders the productivity of the user.

\interfootnotelinepenalty=10000
In this paper, we formally introduce a systematic (de/re)-composition approach for data-parallel computations targeting state-of-the-art parallel architectures.
We express computations via \emph{high-level functional expressions} (specifying \emph{what} to compute), in the form of easy-to-use higher-order functions, based on the algebraic formalism of \emph{Multi-Dimensional Homomorphisms~(MDHs)}~\cite{rasch2018multi}\footnote{
    We thoroughly compare to the existing MDH work in Section~\ref{sec_rw_mdh}.
}.
Our higher-order functions are capable of expressing various kinds of data-parallel computations (linear algebra, stencils, etc), in the same formalism and on a high level of abstraction, independently of hardware and optimization details, thereby contributing to user's productivity\footnote{
We consider as main users of our approach compiler engineers and library designers.
\citet{mdpoly} show that our approach can also take straightforward, sequential code as input,
which makes our approach attractive also to end users.
}.
As target for our high-level expressions, we introduce \emph{functional low-level expressions} (specifying \emph{how} to compute) to formally reason about (de/re)-compositions of data-parallel computations;~our low-level expressions are designed such that they can be straightforwardly transformed to executable program code
(e.g., in OpenMP, CUDA, and OpenCL).
To systematically lower our high-level expressions to low-level expressions, we introduce a formally sound, parameterized \emph{lowering process}.
The parameters of our lowering process enable automatically computing low-level expressions that are optimized (auto-tuned~\cite{8423171}) for the particular target architecture
and characteristics of the input/output data, thereby achieving fully automatically high, portable performance.
For example, we formally introduce parameters for flexibly choosing the target memory regions for de-composed and re-composed computations, and also parameters for flexibly setting an optimized data access pattern.

We show that our high-level representation is capable of expressing various kinds of data-parallel computations, including computations that recently gained high attention due to their relevance for
deep learning~\cite{10.1145/3317550.3321441}.
For our low-level representation, we show that
it can express the cache blocking and parallelization strategies of state-of-the-art parallel implementations%
~--~as generated by scheduling approach TVM and polyhedral compilers PPCG and Pluto~--~in one uniform formalism.
Moreover, we present experimental results to confirm that based on our parameterized lowering process in combination with auto-tuning, we are able to achieve higher performance than the state of the art, including hand-optimized implementations provided by vendors (e.g., NVIDIA cuBLAS and Intel oneMKL for linear algebra routines, and NVIDIA cuDNN and Intel oneDNN for deep learning computations).

\vspace*{0.7em}

Summarized, we make the following three major contributions (illustrated in Figure~\ref{contributions_overall}):

\vspace*{0.7em}

\begin{enumerate}
  \setlength\itemsep{1em}
  \item We introduce a functional \emph{High-Level Representation (HL-REP)}, based on the algebraic formalism of Multi-Dimensional Homomorphisms (MDHs), that enables uniformly expressing data-parallel computations on a high level of abstraction.

  \item We introduce a functional \emph{Low-Level Representation (LL-REP)} that enables formally expressing and reasoning about (de/re)-compositions of data-parallel computations;~our low-level representation is designed such that it can be straightforwardly transformed to executable program code in state-of-practice parallel programming models, including OpenMP, CUDA, and OpenCL.

  \item We introduce a systematic \emph{Lowering} process to fully automatically lower an expression in our high-level representation to a device- and data-optimized expression in our low-level representation, in a formally sound manner, based on auto-tuning.
\end{enumerate}

\begin{figure}[h!]
    \centering
    \includegraphics[width=0.95\textwidth]{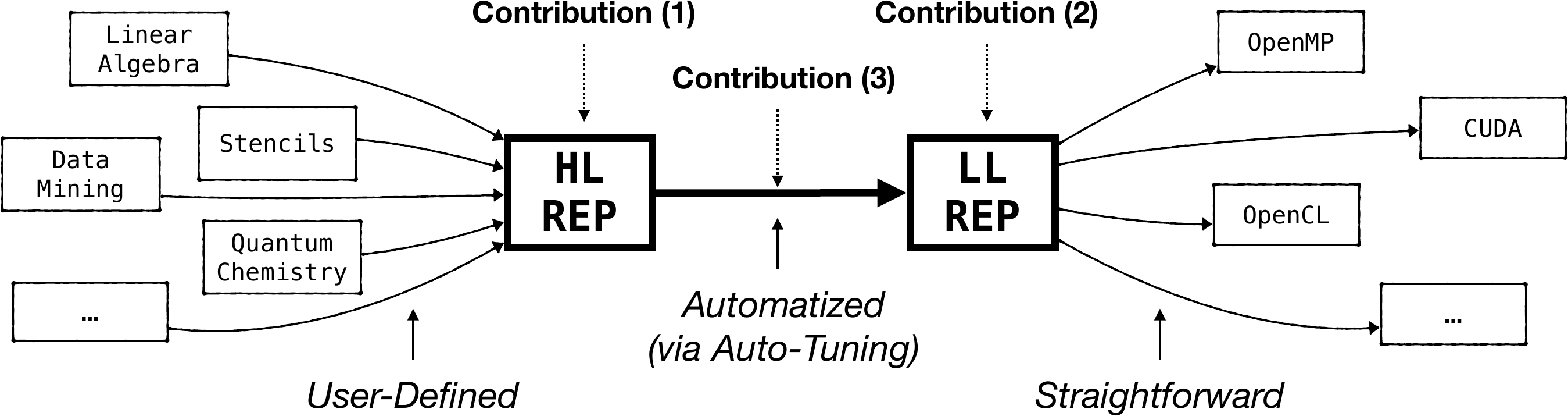}
    \caption{
        Overall structure of our approach (contributions highlighted in bold)
    }
    \label{contributions_overall}
\end{figure}

\newpage

Our three contributions aim to answer the following questions:

\vspace*{0.5em}

\begin{enumerate}
  \setlength\itemsep{0.5em}
  \item \emph{How can data parallelism be formally defined, and how can data-parallel computations be uniformly expressed via higher-order functions that are agonistic from of hardware and optimization details while still capturing all information relevant for generating high-performing, executable program code?}~(Contribution 1);

  \item \emph{How can optimizations for the memory and core hierarchies of state-of-the-art parallel architectures be formally expressed and generalized such that they apply to arbitrary data-parallel computations?}~(Contribution 2);

  \item \emph{How can optimizations for data-parallel computations be expressed and structured so that they can be automatically identified (auto-tuned) for a particular target architecture and characteristics of the input and output data?}~(Contribution 3).
\end{enumerate}

\vspace*{0.7em}

The rest of the paper is structured as follows.
We introduce our high-level functional representation (Contribution 1) in Section~\ref{ch:high_level}, and we show how this representation is used for expressing various kinds of popular data-parallel computations.
In Section~\ref{ch:low_level}, we discuss our low-level functional representation (Contribution 2) which is powerful enough to express
the optimization decisions of
state-of-practice approaches (e.g., scheduling approach TVM and polyhedral compilers PPCG and Pluto) and beyond.
Section~\ref{ch:lowering} shows how we systematically lower a computation expressed in our high-level representation to an expression in our low-level representation, in a formally sound, auto-tunable manner (Contribution 3).
We present experimental results in Section~\ref{ch:eval}, discuss related work in Section~\ref{sec_rw} (including a thorough comparison to previous work on MDHs), conclude in Section~\ref{ch:conclusion}, and we present our ideas for future work in Section~\ref{ch:fw}.
Our Appendix, in Sections~\ref{sec_math_foundation}-\ref{app_sec_code_generation}, provides details for the interested reader that should not be required for understanding the basic ideas and concepts introduced in this paper.

\section{High-Level Representation for Data-Parallel Computations}
\label{ch:high_level}

We introduce functional building blocks, in the form of higher-order functions, that express data-parallel computations on a high abstraction level.
The goal of our high-level abstraction is to express computations agnostic from hardware and optimization details, and thus in a user-productive manner, while still capturing
all information relevant for generating high-performance program code.
The building blocks of our abstraction are based on the algebraic formalism of \emph{Multi-Dimensional Homomorphisms (MDHs)} which is an approach toward formalizing data parallelism~%
(we compare in detail to the existing work on MDHs in Section~\ref{sec_rw_mdh}).

\newpage

Figure~\ref{hl_overview} shows a basic overview of our high-level representation.
We express data-parallel computations using exactly three higher-order functions only (a.k.a. \emph{patterns} or \emph{skeletons}~\cite{gorlatch2011parallel} in programming terminology):~%
1)~\texttt{inp\_view} transforms the domain-specific input data (e.g., a matrix and a vector in the case of matrix-vector multiplication) to a \emph{Multi-Dimensional Array~(MDA)} which is our internal data representation and defined later in this section;
2)~\texttt{md\_hom} expresses the data-parallel computation;
3)~\texttt{out\_view} transforms the computed MDA back to the domain-specific data representation.

In the following, after informally discussing an introductory example in Section~\ref{ch:hl_intro_example}, we formally define and discuss each higher-order function in detail in Section~\ref{sec_mdhom} (function $\mdh$) and Section~\ref{sec_views} (functions $\iv$ and $\ov$).
Note that Section~\ref{sec_mdhom} and Section~\ref{sec_views} introduce and present the internals and formal details of our approach, which are not relevant for the end user of our system~--~the user only needs to operate on the abstraction level discussed in Section~\ref{ch:hl_intro_example}.

\subsection{Introductory Example}
\label{ch:hl_intro_example}

Figure~\ref{fig:intro_example} shows how our high-level representation is used for expressing the example of matrix-vector multiplication \texttt{MatVec}\footnote{
The expression in Figure~\ref{fig:intro_example} can also be extracted from straightforward, annotated sequential code~\cite{mdpoly,mdpoly_src}.
} (Figure~\ref{example_intro_1}).
Computation \texttt{MatVec} takes as input a matrix $M\in T^{I\times K}$ and vector $v\in T^{K}$ of arbitrary scalar type%
\footnote{
We consider as \emph{scalar types}
integers $\IZ$ (a.k.a. \texttt{int} in programming), floating point numbers $\IQ$ (a.k.a. \texttt{float} or \texttt{double}), any fixed collection of types (a.k.a. \emph{record} or \emph{struct}),~etc.~We denote the set of scalar types as $\type$ in the following. Details on scalar types are provided in the Appendix, Section~\ref{app_scalar_type}, for the interested reader.
}
$T$ and sizes $I\times K$ (matrix) and~$K$~(vector), for arbitrary but fixed positive natural numbers $I,K\in\IN$\footnote{
    We denote by $\IN$ the set of positive natural number $\{1,2,\dotsc\}$, and we use $\INz$ for the set of natural numbers including~$0$.
}.
In the figure, based on index function $(i,k)\to(i,k)$ and $(i,k)\to(k)$, high-level function $\iv$ computes a function that takes $M$ and $v$ as input and maps them to a $2$-dimensional
array of size $I\times K$ (referred to as \emph{input MDA} in the following and defined formally in the next subsection).
The MDA contains at each point $(i,k)$ the pair $(M_{i,k},v_k)\in T\times T$ comprising element~$M_{i,k}$ within matrix~$M$ (first component) and element~$v_k$ within vector~$v$~(second component).
The input MDA is then mapped via function $\mdh$
to an output MDA of size $I\times 1$, by applying multiplication $*$ to each pair $(M_{i,k},v_k)$ within the input MDA, and combining the obtained intermediate results within the MDA's first~dimension via~$\dplus$~(concatenation~--~also defined formally in the next subsection) and in second dimension~via~$+$~(point-wise addition).
Finally, function $\ov$ computes a function that straightforwardly maps the output MDA, of size $I\times 1$, to \texttt{MatVec}'s result vector $w\in T^I$, which has scalar type $T$ and is of size~$I$.
For the example of \texttt{MatVec}, the output view is trivial, but it can be used in other computations (such as matrix multiplication) to conveniently express more advanced variants of computations (e.g., computing the result matrix of matrix multiplication as transposed, as we demonstrate later).
\begin{figure}[t!]
    \centering
    \includegraphics[scale=0.25]{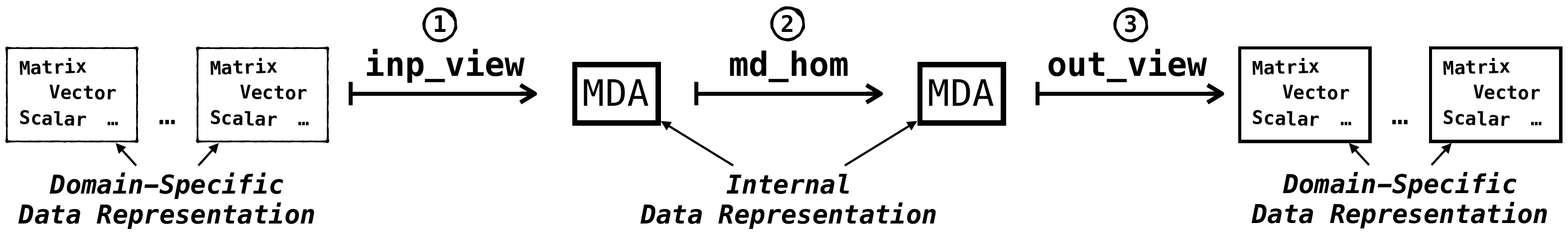}
    \caption{High-level representation (overview)}
    \label{hl_overview}
\end{figure}

\begin{figure}[h!]
    \begin{align*}
    \texttt{MatVec}^{<\texttt{T}\in\type\,|\,\texttt{I,K}\in\IN>} \ :=
    \ \ \
    &\ov\texttt{<T>( w:(i,k)$\mapsto$(i) )} \ \circ \\[2pt]
    &\hspace*{20px}
    \mdh\texttt{<I,K>(\,*,\,(\dplus,+)\,)} \ \circ \\[2pt]
    &\hspace*{40px}
    \iv\texttt{<T,T>( M:(i,k)$\mapsto$(i,k)\,,\,v:(i,k)$\mapsto$(k) )}
    \end{align*}
    \caption{
    High-level expression for Matrix-Vector Multiplication (\texttt{MatVec})\protect\footnotemark
    }
    \label{fig:intro_example}
\end{figure}

\newpage

\subsection{Function \texttt{md\_hom}}
\label{sec_mdhom}

Higher-order function \texttt{md\_hom} is introduced by~\citet{rasch2018multi} to express \emph{Multi-Dimensional Homomorphisms (MDHs)}~--~a formal representation of data-parallel computations~--~in a convenient and structured way.
In the following, we recapitulate the definition of MDHs and function \texttt{md\_hom}, but in a more general and formally more precise setting than done in the original MDH work.

To define MDH functions, we first need to introduce two central building blocks used in the definition of MDHs:~%
i)~\emph{Multi-Dimensional Arrays (MDAs)}~--~the data type on which MDHs operate and which uniformly represent domain-specific input and output data (scalar, vectors, matrices, $\dotsc$), and~%
ii)~\emph{Combine Operators (COs)} which we use to combine elements within a particular dimension of an MDA.

\footnotetext{Our technical implementation takes as input a representation that is equivalent to Figure~\ref{fig:intro_example}, expressed via straightforward program code~(see Appendix, Section~\ref{matvec_in_dsl}).}

\subsection*{Multi-Dimensional Arrays}

\begin{definition}[Multi-Dimensional Array]
\label{def_mda}
Let $\IDXs:=\{ I\subset\INz \ | \ |I|<\infty \}$ be the set of all finite subsets of natural numbers; in the context of MDAs, we refer to the subsets also as \emph{MDA index sets}.
Let further $T\in\type$ be an arbitrary scalar type, $D\in\IN$ a natural number, $I:=(I_1 \, ,\dotsc, \, I_D)
\in\IDXs^D$ a tuple of $D$-many MDA index sets, and $N:=(N_1 \, ,\dotsc, \, N_D):=(|I_1| \, ,\dotsc, \, |I_D|)$ the tuple of index sets' sizes.

A \emph{Multi-Dimensional Array~(MDA)} $\MDA$ that has \emph{dimensionality}~$D$, \emph{size} $N$, \emph{index sets} $I$, and \emph{scalar type} $T$ is a function with the following signature:
\[
\MDA:I_1\,\times\dotsc\times\,I_D \, \to \, T
\]
We refer to $I_1\,\times\dotsc\times\,I_D \, \to \, T$ as the \emph{type} of MDA $\MDA$.
\end{definition}

\begin{notation}
\label{not_mda}
For better readability, we denote MDAs' types and accesses to them using a notation close to programming.
We often write:
\begin{itemize}
\setlength\itemsep{2px}
  \item
  $\MDA\in T[ \, I_1 \, ,\dotsc, \, I_D \, ]$ instead of \ $\MDA:I_1\,\times\dotsc\times\,I_D \, \to \, T$ to denote the type of MDA $\MDA$;

  \item
  $\MDA\in T[ \, N_1 \, ,\dotsc, \, N_D \, ]$ instead of \ $\MDA:[0,N_1)_{\INz}\times\cdots\times[0,N_D)_{\INz}\to T$;\,\footnote{
    We denote by $[L,U)_{\INz}:=\{ \ n\in\INz \ | \ L\le n<U \ \}$ the half-open interval of natural numbers (including $0$) between $L$~(incl.) and $U$~(excl.).
      }

  \item
  $\MDA[ \, i_1 \, ,\dotsc, \, i_D \, ]$ instead of \ $a(i_1 \, ,\dotsc, \, i_D )$ to access MDA $\MDA$ at position $(i_1 \, ,\dotsc, \,i_D)$.
\end{itemize}
\end{notation}

\begin{figure}[t!]
    \centering
    \includegraphics[width=\textwidth]{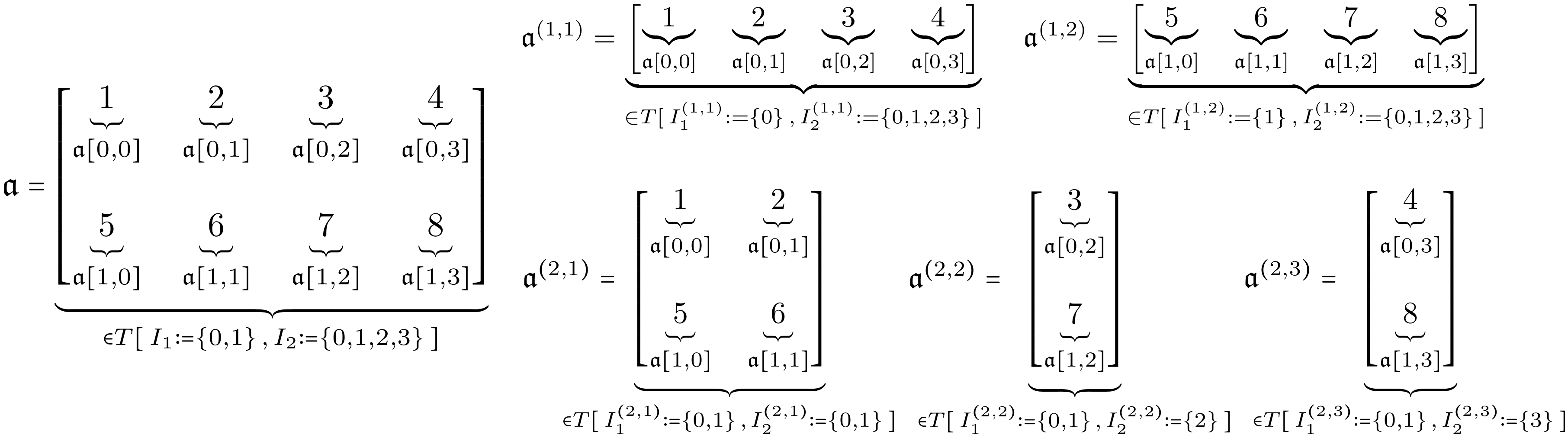}
    \caption{
    MDA examples}
    \label{mda_examples}
\end{figure}

Figure~\ref{mda_examples} shows six MDAs for illustration.
The left part of the figure shows MDA $\MDA$ which is of type $\MDA:I_1\,\times I_2 \, \to \, T$, for $I_1=\{0,1\}$, $I_2=\{0,1,2,3\}$, and $T=\IZ$ (integer numbers).
On the right side, five MDAs are shown, named $\MDA^{(1,1)}$, $\MDA^{(1,2)}$, $\MDA^{(2,1)}$, $\MDA^{(2,2)}$, $\MDA^{(2,3)}$~--~the superscripts are part of the names and represent a two-dimensional numbering of the five MDAs.
The MDAs $\MDA^{(1,1)}$ and $\MDA^{(1,2)}$ are of types
$\MDA^{(1,1)}:I^{(1,1)}_1\,\times I^{(1,1)}_2 \, \to \, T$ and
$\MDA^{(1,2)}:I^{(1,2)}_1\,\times I^{(1,2)}_2 \, \to \, T$,
for $I^{(1,1)}_1=\{0\}$ and $I^{(1,2)}_1=\{1\}$, and coinciding
second dimensions $I^{(1,1)}_2=I^{(1,2)}_2=\{0,1,2,3\}$.
The MDAs $\MDA^{(2,1)}$, $\MDA^{(2,2)}$, and $\MDA^{(2,3)}$ are of types
$\MDA^{(2,1)}:I^{(2,1)}_1\,\times I^{(2,1)}_2 \, \to \, T$,
$\MDA^{(2,2)}:I^{(2,2)}_1\,\times I^{(2,2)}_2 \, \to \, T$, and
$\MDA^{(2,3)}:I^{(2,3)}_1\,\times I^{(2,3)}_2 \, \to \, T$, and they
coincide in their first dimensions $I^{(2,1)}_1=I^{(2,2)}_1=I^{(2,3)}_1=\{0,1\}$;~their second dimensions are $I^{(2,1)}_2=\{0,1\}$, $I^{(2,2)}_2=\{2\}$, and $I^{(2,2)}_2=\{3\}$.
Note that MDAs $\MDA^{(1,1)},\MDA^{(1,2)},\MDA^{(2,1)},\MDA^{(2,2)},\MDA^{(2,3)}$ can be considered as \emph{parts} (a.k.a. \emph{tiles} in programming) of MDA $\MDA$.
We formally define and use \emph{partitionings} of MDAs in Section~\ref{ch:low_level}.

\subsection*{Combine Operators}

A central building block in our definition of MDHs is a \emph{Combine Operator (CO)}.
Intuitively, we use a combine operator to combine all elements within a particular dimension of an MDA.
For example, in Figure~\ref{example_intro_1} (matrix-vector multiplication), we combine elements of the $2$-dimensional MDA via combine operator \emph{concatenation} in MDA's first dimension and via operator \emph{point-wise addition} in the second dimension.

Technically, combine operators are functions that take as input two MDAs and yield a single MDA as their output (formal definition follows soon).
By definition, we require that the index sets of the two input MDAs coincide in all dimensions except in the dimension to combine;~thereby, we catch undefined cases already at the type level, e.g., trying to concatenate improperly sized MDAs:

\vspace*{10px}

\begin{center}
\resizebox{.85\hsize}{!}{
$
\underbrace{
\underbrace{
\begin{bmatrix}
1 & 2 & 3 \\
4 & 5 & 6 \\
7 & 8 & 9 \\
\end{bmatrix}
}_{\substack{3\times3\text{-many} \\ \text{elements}}}
\dplus
\underbrace{
\begin{bmatrix}
11 & 12 \\
13 & 14 \\
15 & 16 \\
\end{bmatrix}
}_{\substack{3\times2\text{-many} \\ \text{elements}}}
=
\underbrace{
\begin{bmatrix}
1 & 2 & 3 & 11 & 12 \\
4 & 5 & 6 & 13 & 14 \\
7 & 8 & 9 & 15 & 16 \\
\end{bmatrix}
}_{\substack{3\times5\text{-many} \\ \text{elements}}}
}_{\text{well defined}}
\hspace{40px}
\underbrace{
\underbrace{
\begin{bmatrix}
1 & 2 & 3 \\
4 & 5 & 6 \\
7 & 8 & 9 \\
\end{bmatrix}
}_{\substack{3\times3\text{-many} \\ \text{elements}}}
\dplus
\underbrace{
\begin{bmatrix}
11 & 12 \\
13 & 14 \\
\end{bmatrix}
}_{\substack{2\times2\text{-many} \\ \text{elements}}}
= \ ?
}_{\lightning\text{undefined}}
$
}
\end{center}

\vspace*{5px}

\noindent
Here, on the left, we can reasonably define the concatenation of MDAs that contain $3\times3$-many elements and $3\times2$ elements.
However, as indicated in the right part of the figure,
it is not possible to intuitively concatenate MDAs of sizes $3\times3$ and $2\times2$,
as the MDAs do not match in their number of elements in any of the two dimensions.

\begin{figure}[b!]
    \centering
    \includegraphics[width=\textwidth]{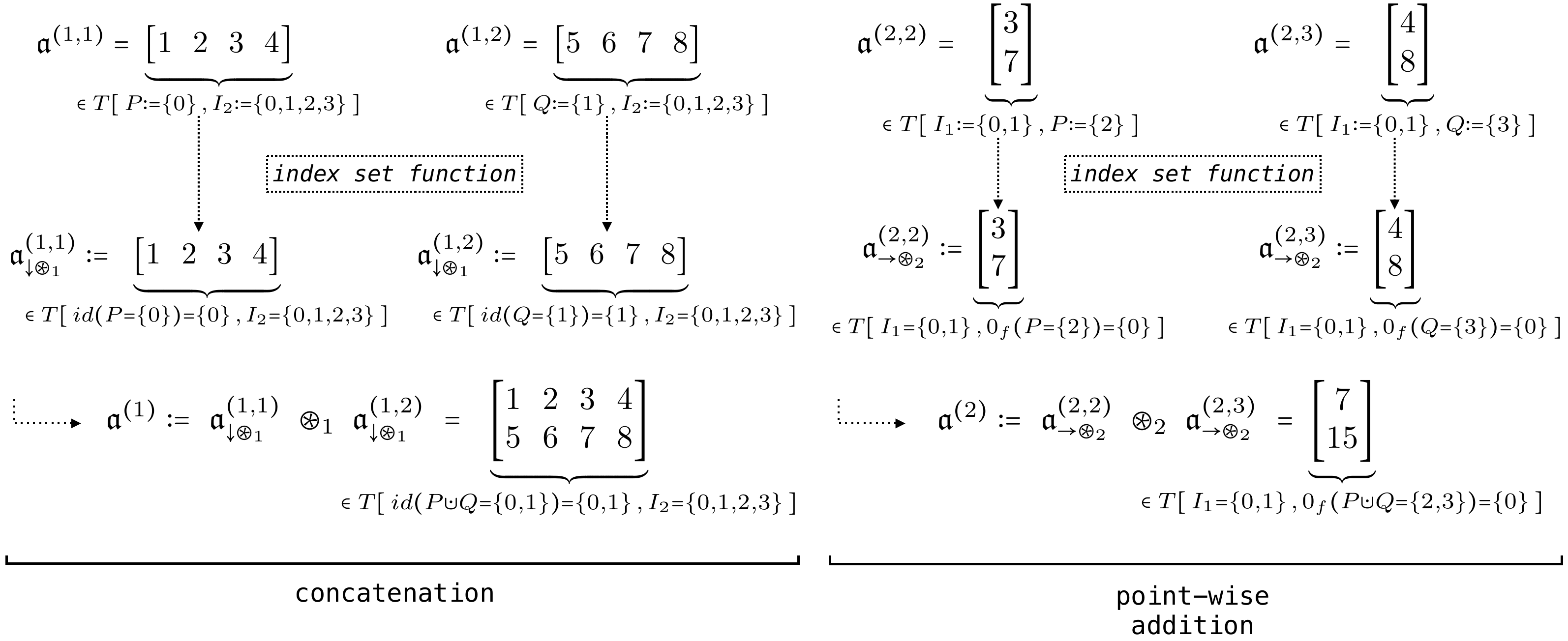}
    \caption{
    Illustration of \emph{combine operators} using the examples \emph{concatenation} (left) and \emph{point-wise addition} (right)}
    \label{co_illustration}
\end{figure}

Figure~\ref{co_illustration} illustrates combine operators informally using the example operators \emph{concatenation} (left part of the figure) and \emph{point-wise addition} (right part).
We illustrate concatenation using the example MDAs $\MDA^{(1,1)}$ and $\MDA^{(1,2)}$ from Figure~\ref{mda_examples};~for point-wise addition, we use MDAs $\MDA^{(2,2)}$ and $\MDA^{(2,3)}$ from Figure~\ref{mda_examples} (all MDAs are chosen arbitrarily, and the example works the same for other MDAs).

\newpage

In the case of concatenation (left part of Figure~\ref{co_illustration}), MDAs $\MDA^{(1,1)}$ and $\MDA^{(1,2)}$ coincide in their second dimension $I_2:=\{0,1,2,3\}$, which is important, because we concatenate in the first dimension, thus requiring coinciding index sets in all other dimensions (as motivated above).
In the case of the point-wise addition example (right part of Figure~\ref{co_illustration}), the example MDAs $\MDA^{(2,2)}$ and $\MDA^{(2,3)}$ coincide in their first dimension $I_1:=\{0,1\}$, as required for combining the MDAs in the second dimension.
The varying index sets of
the four MDAs are denoted as $P$ and $Q$ in the figure, which are in the case of the concatenation example, index sets in the first dimension of MDAs
$\MDA^{(1,1)}$
and
$\MDA^{(1,2)}$;~%
in the case of the point-wise addition example, the varying index sets of MDAs
$\MDA^{(2,1)}$
and
$\MDA^{(2,2)}$
belong to the second dimensions.

In the following, we assume w.l.o.g. that the varying index sets $P$ and $Q$ of MDAs to combine are disjoint.
Our assumption will not be a limitation for our approach:~we will apply combine operators always to parts of MDAs that belong to the same MDA, causing that the index sets of the parts are disjoint by construction.
For example, in the case of the concatenation example in Figure~\ref{co_illustration}, the parts $\MDA^{(1,1)}$ and $\MDA^{(1,2)}$ of MDA $\MDA$ correspond to the first and second row of the same MDA $\MDA$ in Figure~\ref{mda_examples} and thus have different index sets in their first dimension, and in the case of the point-wise addition example in Figure~\ref{co_illustration}, the parts $\MDA^{(2,2)}$ and $\MDA^{(2,3)}$ represent the third and fourth column of MDA $\MDA$ and thus have different index sets in their second dimension.

We define combine operators based on \emph{index set functions} (also defined formally soon).
Index set functions precisely describe, on the type level, the index set of the combined output MDA and thus how an MDA's index set evolves during combination.
For this, an index set function takes as input the input MDA's index set in the dimension to combine, and the function yields as its output the index set of the output MDA which is combined in this dimension.
In the case of the concatenation example in Figure~\ref{co_illustration}, the index set function is identity $id$ and thus trivial.
However, in the case of point-wise addition, the corresponding index set function is the constant function $0_f$ which maps any index set to the singleton set $\{0\}$ containing index $0$ only.
This is because when combining via point-wise addition, the MDA shrinks in the combined dimension to only one element which we aim to uniformly access via MDA index $0$.
In Figure~\ref{co_illustration}, we denote MDAs
$\MDA^{(1,1)}$,
$\MDA^{(1,2)}$,
$\MDA^{(2,2)}$,
$\MDA^{(2,3)}$
after applying the corresponding index set function as:
$\MDA^{(1,1)}_{\downarrow{\co{1}}}$,
$\MDA^{(1,2)}_{\downarrow{\co{1}}}$,
$\MDA^{(2,2)}_{\rightarrow{\co{2}}}$,
$\MDA^{(2,3)}_{\rightarrow{\co{2}}}$;~the combined MDAs are denoted as $\MDA^{(1)}$ and $\MDA^{(2)}$ in the figure.
The concatenation operator is denoted in the figure generically as $\co{1}$, and point-wise addition is denoted as $\co{2}$, correspondingly.

\vspace{7px}

We now define \emph{combine operators} formally, and we illustrate this formal definition afterward using the example operators \emph{concatenation} and \emph{point-wise combination}.
For the interested reader, details on some technical design decisions of combine operators are outlined in the Appendix, Section~\ref{app_design_dec_co}.

\vspace*{5px}

\begin{definition}[Combine Operator]
  \label{def_combine_op}
  Let $\IDXsxIDXs:=\{ \, (P,Q)\in\IDXs$ $\times\IDXs \ | \ P\cap Q=\emptyset \, \}$ denote the set of all pairs of MDA index sets that are disjoint.
  Let further $\size{}{MDA}{MDA}:\IDXs\to\IDXs$ be a function on MDA index sets, $T\in\type$ a scalar type, $D\in\IN$ an MDA dimensionality, and $d\in[1,D]_\IN$ an MDA dimension.

  We refer to any binary
  function $\circledast$ of type (parameters in angle brackets are type parameters)
  \begin{align*}
  &
  \circledast^{<(I_1,\dotsc,I_{d-1},I_{d+1},\dotsc,I_D)\in\IDXs^{D-1} \,,\, (P,Q)\in\IDXsxIDXs>}: \\[0pt]
  &\hspace*{20px}
  T[ \, I_1 \, ,\dotsc, \, \underset{\underset{d}{\uparrow}}{\underbrace{\size{}{MDA}{MDA}(P)}} \, ,\dotsc, \, I_D \, ] \ \times  \
  T[ \, I_1 \, ,\dotsc, \, \underset{\underset{d}{\uparrow}}{\underbrace{\size{}{MDA}{MDA}(Q)}} \, ,\dotsc, \, I_D \, ]
  \\[-25pt]
  &\hspace*{250px}
  \to  T[ \, I_1 \, ,\dotsc, \, \underset{\underset{d}{\uparrow}}{\underbrace{\size{}{MDA}{MDA}(P\cupdot Q)}} \, ,\dotsc, \, I_D \, ]
  \end{align*}
  as \emph{Combine Operator (CO)} that has \emph{index set function} $\size{}{MDA}{MDA}$, \emph{scalar type} $T$, \emph{dimensionality} $D$, and \emph{operating dimension} $d$.
  We denote CO's type concisely as ${\CO}^{<\ssize{}{MDA}{MDA}\,|\,T\,|\,D\,|\,d>}$.
\end{definition}

\vspace*{5px}

Since function $\circledast$'s ordinary function type $T[\dotsc]\times T[\dotsc]\to T[\dotsc]$ is generic in parameters $(I_1,\dotsc,I_{d-1},$ $I_{d+1},\dotsc,I_D)$ and~$(P,Q)$ (these type parameters are denoted in angle brackets in Definition~\ref{def_combine_op}), we refer to function $\circledast$ as \emph{meta-function}, to the type parameters
in angle brackets
as \emph{meta-parameters}, and we say \emph{meta-types} to
$T[ \, I_1 \, ,\dotsc, \, \size{}{MDA}{MDA}(P) \, ,\dotsc, \, I_D \, ]$ (first input MDA),
$T[ \, I_1 \, ,\dotsc, \, \size{}{MDA}{MDA}(Q) \, ,$ $\dotsc, \, I_D \, ]$ (second input MDA), and
$T[ \, I_1 \, ,\dotsc, \, \size{}{MDA}{MDA}(P\cupdot Q) \, ,\dotsc, \, I_D \, ]$ (output MDA) as these types are generic in meta-parameters.
Formal definitions and details about our meta-parameter concept
are provided in Section~\ref{sec_math_foundation} of our Appendix for the interested reader.

We use meta-functions as an analogous concept to \emph{metaprogramming} in programming language theory to achieve high generality.
For example, by defining combine operators as meta-functions, we
can use the operators on input MDAs
that operate on arbitrary index sets while still guaranteeing correctness,
e.g., that index sets of the two input MDAs match in all dimensions except in the dimension to combine (as discussed above).
For simplicity, we often refrain from explicitly stating meta-parameters when they are clear from the context;~for example, when they can be deduced from the types of their particular inputs (a.k.a. \emph{type deduction} in programming).

\vspace*{10px}

We now formally discuss the example operators \emph{concatenation} and \emph{point-wise combination}.
For high flexibility, we define both operators generically in the scalar type $T$ of their input and output MDAs, the MDAs' dimensionality $D$, as well as in the dimension $d$ to combine.

\begin{example}[Concatenation]
\label{def:mda_concat}

We define \emph{concatenation} as function $\dplus$ of type
\begin{align*}
&
\dplus^{<T\in\type\,|\,D\in\IN\,|\,d\in[1,D]_\IN\,|\,(I_1,\dotsc,I_{d-1},I_{d+1},\dotsc,I_D)\in\IDXs^{D-1},(P,Q)\in\IDXsxIDXs>}: \\
&\hspace*{20px}
T[ \, I_1 \, ,\dotsc, \, \underset{\underset{d}{\uparrow}}{\underbrace{id(P)}} \, ,\dotsc, \, I_D \, ]\ \times \ T[ \,
I_1 \, ,\dotsc, \, \underset{\underset{d}{\uparrow}}{\underbrace{id(Q)}} \, ,\dotsc, \, I_D \, ]
\ \to \
T[ \, I_1 \, ,\dotsc, \, \underset{\underset{d}{\uparrow}}{\underbrace{id(P\cupdot Q)}} \, ,\dotsc, \, I_D \, ]
\end{align*}
where $id:\IDXs\to\IDXs$ is the identity function on MDA index sets.

\newpage
\noindent
The function is computed as:

\begin{align*}
&
\dplus^{<T\,|\,D\,|\,d\,|\,(I_1,\dotsc,I_{d-1},I_{d+1},\dotsc,I_D),(P,Q)>}( \, \MDA_1, \MDA_2 \ )[ \, i_1 \, ,\dotsc, \ \ i_d \ \ ,\dotsc, \, i_D \, ] \\[5pt]
&\hspace*{200px}
:= \
\begin{cases}
\ \MDA_1[ \, i_1 \, ,\dotsc, \ \ i_d \ \ ,\dotsc, \, i_D \, ] &, \ \ i_d\in P \\
\ \MDA_2[ \, i_1 \, ,\dotsc, \ \ i_d \ \ ,\dotsc, \, i_D \, ] &, \ \ i_d\in Q \\
\end{cases}
\end{align*}
The function is well defined, because $P$ and $Q$ are disjoint.
We usually use an infix notation for $\dplus{}^{<\dotsc>}$ (meta-parameters omitted via ellipsis), i.e., we write $\MDA_1\,\dplus{}^{<\dotsc>} \,\MDA_2$ instead of $\dplus{}^{<\dotsc>}(\MDA_1,\MDA_2)$.

\end{example}

\vspace*{10px}

The vertical bar in the superscript of $\dplus$ denotes that function $\dplus$ can be partially evaluated (a.k.a. \emph{Currying}~\cite{CURRY198085} in math and \emph{multi staging}~\cite{10.1145/258993.259019} in programming) for particular values of meta-parameters:~$T\in\type$ (first stage), $D\in\IN$ (second stage), etc.
Partial evaluation (formally defined in the Appendix, Definition~\ref{app_multi_stage_meta_function}) enables both:~%
1)~expressive typing and thus better error elimination:~for example, parameter $(I_1,\dotsc,I_{d-1},I_{d+1},\dotsc,I_D)\in\IDXs^{D-1}$ can depend on meta-parameter $D\in\IN$, because $D$ is defined in an earlier stage, which allows precisely limiting the length of the tuple $(I_1,\dotsc,I_{d-1},I_{d+1},\dotsc,I_D)$ to ${D-1}$ index sets;~%
2)~generality:~for example, we can instantiate $\dplus$ to $\dplus^{<T>}$ which is specific for a particular scalar type $T\in\type$, but still generic in meta-parameters $D\in\IN$, $d\in[1,D]_\IN$, $\dotsc$, as these meta-parameters are defined in later stages.
We specify stages and their order according to the recommendations in~\citet{haskell_parameter_order}, e.g., using earlier stages for meta-parameters that are expected to change less frequently than other meta-parameters.

\vspace*{5px}

It is easy to see that $\dplus^{<T\,|\,D\,|\,d>}$ is a combine operator of type ${\CO}^{<id\,|\,T\,|\,D\,|\,d>}$ for any particular choice of meta-parameters $T\in\type$, $D\in\IN$, and $d\in[1,D]_\IN$.

\begin{example}[Point-Wise Combination]
\label{def:mda_pw}
We define \emph{point-wise combination},
according to a binary function $\oplus:T\times T\to T$ (e.g. addition),
as function $\overrightarrow{\bullet}$ of type
\begin{align*}
&
\overrightarrow{\bullet}^{<T\in\type\,|\,D\in\IN\,|\,d\in[1,D]_\IN\,|\,(I_1,\dotsc,I_{d-1},I_{d+1},\dotsc,I_D)\in\IDXs^{D-1},(P,Q)\in\IDXsxIDXs>}: \\[5pt]
&\hspace*{0px}
\underbrace{T\times T\to T
}_{\oplus}
\ \to \
\underbrace{
T[I_1,\dotsc,\underset{\underset{d}{\uparrow}}{\underbrace{0_f(P)}},\dotsc,I_D]
\times
T[I_1,\dotsc,\underset{\underset{d}{\uparrow}}{\underbrace{0_f(Q)}},\dotsc,I_D]
\to T[I_1,\dotsc,\underset{\underset{d}{\uparrow}}{\underbrace{0_f(P\cupdot Q)}},\dotsc,I_D]
}_{\substack{\texttt{point-wise combination (according to $\oplus$)} }}
\end{align*}
where $0_f:\IDXs\to\IDXs$, $I\mapsto\{0\}$ is the constant MDA index set function that maps any index set $I$ to the index set containing MDA index $0$ only.
The function is computed as:
\begin{align*}
&
\overrightarrow{\bullet}^{<T\,|\,D\,|\,d\,|\,(I_1,\dotsc,I_{d-1},I_{d+1},\dotsc,I_D),(P,Q)>}(\,\oplus\,)(\, \MDA_1,\MDA_2 \,)[i_1,\dotsc,\underset{\underset{d}{\uparrow}}{0},\dotsc,i_D] := \\
&\hspace*{224px}
\MDA_1[i_1,\dotsc,\underset{\underset{d}{\uparrow}}{0},\dotsc,i_D] \oplus \MDA_2[i_1,\dotsc,\underset{\underset{d}{\uparrow}}{0},\dotsc,i_D]
\end{align*}
We often write $\oplus$ only, instead of $\overrightarrow{\bullet}(\oplus)$, and we usually use an infix notation for $\oplus$.

\end{example}

Function $\overrightarrow{\bullet}^{<T\,|\,D\,|\,d>}(\oplus)$ (meaning:~$\overrightarrow{\bullet}$ is partially applied to ordinary function parameter $\oplus$ and thus still generic in parameters $(I_1,\dotsc,I_{d-1},I_{d+1},\dotsc,I_D)$ and $(P,Q)$~--~formal details provided in the Appendix, Definition~\ref{app_partial_meta_func_app}) is a combine operator of type ${\CO}^{<0_f\,|\,T\,|\,D\,|\,d>}$ for any binary operator $\oplus:T\times T\to T$.

\subsection*{Multi-Dimensional Homomorphisms}

Now that we have defined MDAs (Definition~\ref{def_mda}) and COs (Definition~\ref{def_combine_op}), we can define \emph{Multi-Dimensional Homomorphisms (MDHs)}.
Intuitively, a function $h$ operating on MDAs is
an MDH
iff we can apply the function independently to parts of its input MDA and combine the obtained intermediate results to the final result using combine operators;~this can be imagined as a typical divide-and-conquer pattern.
Compared to classical approaches, e.g., \emph{list homomorphisms}~\cite{10.1007/978-3-642-74884-4_5,doi:10.1142/S0129626495000175,GORLATCH19991}, a major characteristic of MDH functions is that they allow
(de/re)-composing computations in multiple dimensions (e.g., in Figure~\ref{example_intro_1}, in both the concatenation dimension as well as in the point-wise addition dimensions), rather than being limited to a particular dimension only (e.g., only the concatenation dimension or only point-wise addition dimension, respectively).
We will see later in this paper that a multi-dimensional (de/re)-composition approach is essential to efficiently exploit the hardware of modern architectures which require fine-grained cache blocking and parallelization strategies to achieve their full performance potential.

Figure~\ref{mdh_examples} illustrates the MDH property informally on a simple, two-dimensional input MDA.
In the left part of the figure, we split the input MDA in dimension $1$ (i.e., horizontally) into two parts $\MDA_1$ and $\MDA_2$, apply the MDH function $h$ independently to each part, and combine the obtained intermediate results to the final result using the MDH function $h$'s combine operator $\co{1}$.
Similarly, in the right part of Figure~\ref{mdh_examples}, we split the input MDA in dimension $2$ (i.e., vertically) into parts and combine the results via MDH function $h$'s second combine operator $\co{2}$.

\vspace*{-6px}

\begin{figure}[h!]
    \centering
    \includegraphics[scale=0.25]{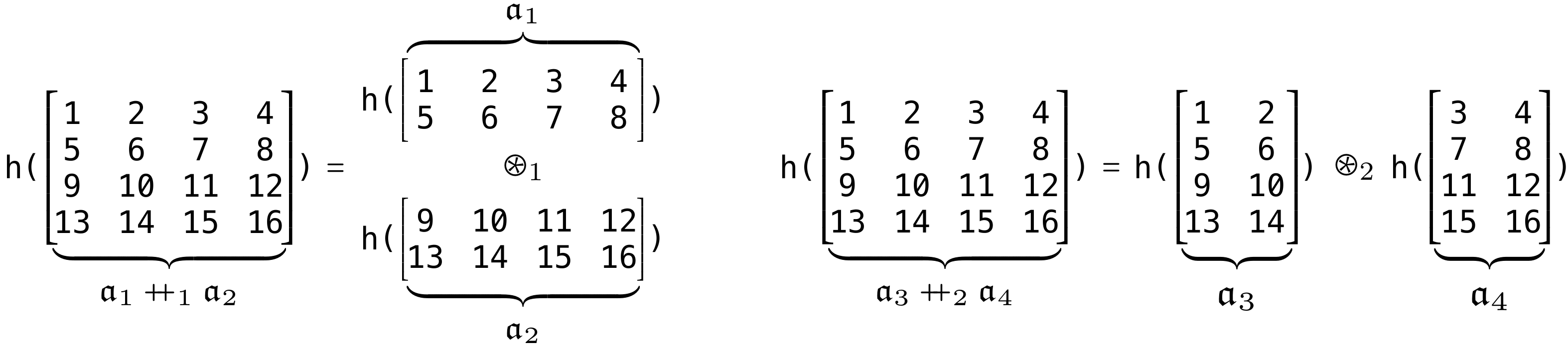}
    \caption{
    MDH property illustrated on a two-dimensional example computation
    }
    \label{mdh_examples}
\end{figure}

\vspace*{-4px}

\begin{figure}[p]
    \centering
    \includegraphics[width=0.9\textwidth]{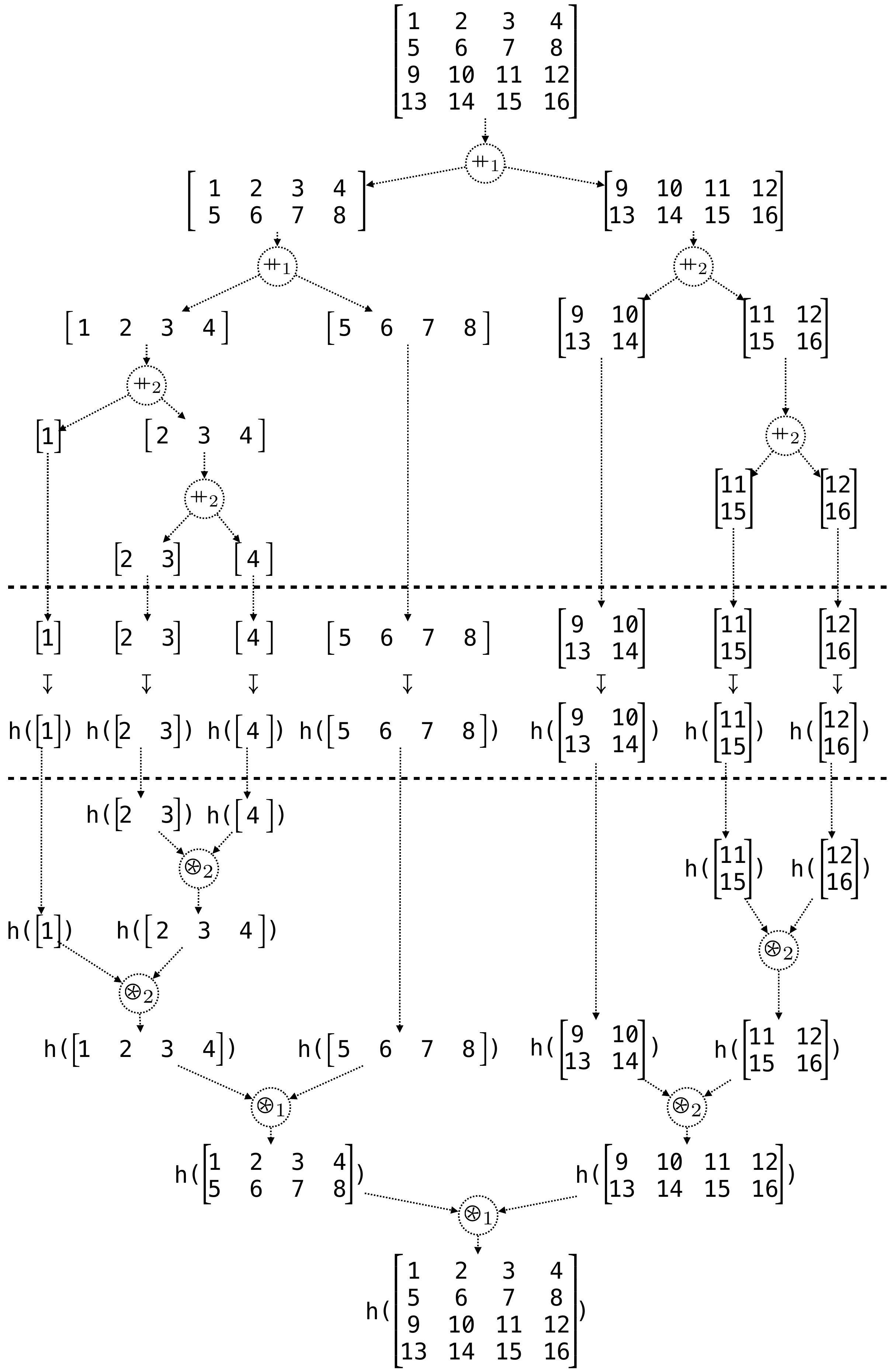}
    \caption{
    MDH property recursively applied to a two-dimensional example computation
    }
    \label{mdh_example_large}
\end{figure}

Figure~\ref{mdh_example_large} shows an artificial example in which we apply the MDH property
(illustrated in Figure~\ref{mdh_examples})
recursively.
We refer in Figure~\ref{mdh_example_large} to the part above the horizontal dashed lines as
\emph{de-composition phase}
and to the part below dashed lines as \emph{re-composition} phase.

\vspace{5px}

\begin{definition}[Multi-Dimensional Homomorphism]
\label{def_mdh}
    Let $T^\texttt{INP},T^\texttt{OUT}\in\type$ be two arbitrary scalar types, $D\in\IN$ a natural number, and $\size{1}{MDA}{MDA},\dotsc,\size{D}{MDA}{MDA}:\IDXs\to\IDXs$ functions on MDA index sets.
    Let further $\dplus_d:={\dplus}^{<T^\texttt{INP}\,|\,D\,|\,d>}\in{\CO}^{<id\,|\,T^\texttt{INP}\,|\,D\,|\,d>}$ denote concatenation (Definition~\ref{def:mda_concat}) in dimension $d\in[1,D]_\IN$ on $D$-dimensional MDAs that have scalar type $T^\texttt{INP}$.

A function
\begin{align*}
&
h^{<I_1,\dotsc,I_D\in\IDXs>} : \
    T^\texttt{INP}[ \ I_1  \ , \, \dotsc \, , \ I_D \ ] \ \to \ T^\texttt{OUT}[ \ \size{1}{MDA}{MDA}(I_1) \ ,\dotsc, \ \size{D}{MDA}{MDA}(I_D) \ ]
\end{align*}
is a \emph{Multi-Dimensional Homomorphism (MDH)} that has \emph{input scalar type} $T^\texttt{INP}$, \emph{output scalar type} $T^\texttt{OUT}$, \emph{dimensionality} $D$, and \emph{index set functions} $\size{1}{MDA}{MDA},\dotsc,\size{D}{MDA}{MDA}$,  iff for each $d\in[1,D]_\IN$, there exists a combine operator $\co{d}\in{\CO}^{<\ssize{d}{MDA}{MDA}\,|\,T^\texttt{OUT}\,|\,D\,|\,d>}$ (Definition~\ref{def_combine_op}),
such that for any
concatenated input MDA $\MDA_1 \dplusn{d} \ \MDA_2$ in dimension $d$,
the \emph{homomorphic property} is satisfied:
\[h( \ \MDA_1 \dplusn{d} \ \MDA_2 \ ) \ = \ h(\MDA_1)  \, \circledast_d \, h(\MDA_2)\]
We denote the type of MDHs concisely as~%
$\MDH^{<T^\texttt{INP},T^\texttt{OUT}\,|\,D\,|\,(\ssize{d}{MDA}{MDA})_{d\in[1,D]_\IN}>}$.

\end{definition}

MDHs are defined such that applying them to a concatenated MDA in dimension $d$ can be computed by applying the MDH $h$ independently to the MDA's parts $\MDA_1$ and $\MDA_2$ and combining the intermediate results afterward by using its combine operator $\circledast_d$, as also informally discussed above.
Note that by definition of MDHs, their combine operators are associative and commutative (which follows from the associativity and commutativity of $\dplus_d$).
Note further that for simplicity, Definition~\ref{def_mdh} is specialized to MDHs whose input algebraic structure relies on concatenation, as such kinds of MDHs already cover the currently practice-relevant data-parallel computations (as we will see later).
We provide a generalized definition of MDHs in Section~\ref{app_gen_notion_mdh} of our Appendix, for the algebraically interested reader.

\begin{example}[Function Mapping]
\label{mdh_map}
    A simple example MDH is \emph{function mapping}~\cite{https://doi.org/10.1002/spe.1026}, expressed by higher-order function $\texttt{map}(f)(\MDA)$, which applies a user-defined scalar function $f:T^\texttt{INP}\to T^\texttt{OUT}$ to each element within a $D$-dimensional MDA $\MDA$.
    Function $\texttt{map}(f)$ is an MDH of type $\MDH^{<T^\texttt{INP},T^\texttt{OUT}\,|\,D\,|\,id,\dotsc,id>}$ whose combine operators are concatenation $\dplus$ in all of its $D$ dimensions (Example~\ref{def:mda_concat}).
    Function $id$ is the index set function of $\dplus$ (see Example~\ref{def:mda_concat}) and consequently also of MDH $\texttt{map}(f)$.
    Formal details and definitions for \emph{function mapping} can be found in the Appendix, Section~\ref{sec_mdh_examples}.
\end{example}
\begin{example}[Reduction]
\label{mdh_reduce}
    A further MDH function is \emph{reduction}~\cite{https://doi.org/10.1002/spe.1026}, expressed by higher-order function $\texttt{red}(\oplus)(\MDA)$, which combines all elements within a $D$-dimensional MDA $\MDA$ using a user-defined binary function $\oplus:T\times T\to T$.
    Reduction is an MDH of type $\MDH^{<T,T\,|\,D\,|\,0_f,\dotsc,0_f>}$, and its combine operators are point-wise combination $\overrightarrow{\bullet}(\,\oplus\,)$ in all dimensions (Example~\ref{def:mda_pw}), which have $0_f$ as index set function.
    Formal details and definitions for \emph{reduction} can be found in the Appendix, Section~\ref{sec_mdh_examples}.
\end{example}

We show how Examples~\ref{mdh_map} and~\ref{mdh_reduce}
(and particularly also more advanced examples)
are expressed in our high-level representation in Section~\ref{ch_hl_examples}, based on higher-order functions $\mdh$, $\iv$, and $\ov$ (Figure~\ref{hl_overview}) which we introduce in the following.

\subsection*{Higher-Order Function $\mdh$}

We define higher-order function
$\mdh$ which conveniently expresses MDH functions in a uniform and structured manner.
For this, we exploit that any MDH function is uniquely determined by its combine operators and its behavior on singleton MDAs, as informally illustrated in the following figure:

\begin{center}
    \includegraphics[width=0.85\textwidth]{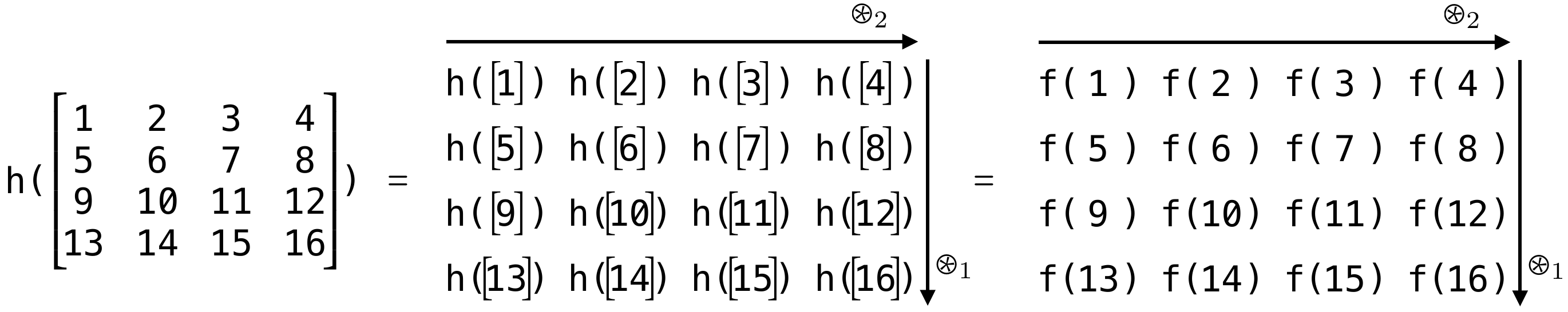}
\end{center}

\vspace*{6px}
\noindent
Here,
$f$
is the function on scalar values that behaves the same as $h$ when restricted to singleton MDAs: $f(\MDA[i_1,\dotsc,i_D]) := h(\MDA)$, for any MDA $\MDA\in T[\{i_1\},\dotsc,\{i_D\}]$ consisting of only one element that is accessed by (arbitrary) indices $i_1,\dotsc,i_D\in\INz$.
For singleton MDAs, we usually use~$f$ instead of~$h$, because $f$ can be defined more conveniently by the user as $h$ (which needs to handle MDAs of arbitrary sizes, and not only singleton MDAs as $f$).
Also, since $f$ takes as input a scalar value~(rather than a singleton MDA, as $h$), the type of $f$ also becomes simpler, which further contributes to simplicity.

\newpage

We now formally introduce function $\mdh$ which uniformly expresses any MDH function, by using only the MDH's behavior $f$ on scalar values~--~which we call \emph{Scalar Function (SF)} and define also in the following~--~and the MDH's combine operators.

\begin{definition}[Scalar Function]
\label{def_sf}
Let $T^\texttt{INP}$ and $T^\texttt{OUT}$ be two arbitrary scalar types.

We refer to any function $f$ of type
\[
f:T^\texttt{INP}\to T^\texttt{OUT}
\]
as \emph{Scalar Function (SF)} from $T^\texttt{INP}$ to $T^\texttt{OUT}$.
We denote scalar function's type concisely as $\mathrm{SF}^{<T^\texttt{INP},T^\texttt{OUT}>}$.
\end{definition}

\begin{definition}[Higher-Order Function $\mdh$]
\label{def_md_hom}
The higher-order function $\mdh$ is of type
\begin{align*}
&
\texttt{md\_hom}^{<T^\texttt{INP},T^\texttt{OUT}\in\type \ | \ D\in\IN \ | \ (\ssize{d}{MDA}{MDA}:\IDXs\to\IDXs)_{d\in[1,D]_\IN}>}: \\[2pt]
&\hspace*{20px}
\underbrace{\mathrm{SF}^{<T^\texttt{INP},T^\texttt{OUT}>}}_{f} \ \ \times \ \
\underbrace{( \ \CO^{<\ssize{1}{MDA}{MDA}\,|\,T^\texttt{OUT}\,|\,D\,|\,1>}\times\dotsc\times\CO^{<\ssize{D}{MDA}{MDA}\,|\,T^\texttt{OUT}\,|\,D\,|\,D>} \ )}_{\co{1} \ ,\,\dotsc\,, \ \co{D}}
\\[-10px]
&
\hspace*{240px}
\to_p \ \
\underbrace{\MDH^{<T^\texttt{INP},T^\texttt{OUT}\,|\,D\,|\,(\ssize{d}{MDA}{MDA})_{d\in[1,D]_\IN}>}}_{
\texttt{md\_hom}( \ f \ , \ (\circledast_1,\dotsc, \circledast_D) \ )
}
\end{align*}
%where $\mathrm{SF}^{<T^\texttt{INP},T^\texttt{OUT}>}$ denotes the set of scalar functions of type $T^\texttt{INP}\to T^\texttt{OUT}$.
Function \texttt{md\_hom} is partial (indicated by $\to_p$ instead of $\rightarrow$), which we motivate after this definition.
The function takes as input a scalar function $f$ and a tuple of $D$-many combine operators $(\co{1},\dotsc,\co{D})$,
and it yields a function $\texttt{md\_hom}( \ f \ , \ (\circledast_1,\dotsc, \circledast_D) \ )$ which is defined as
\begin{align*}
&
\texttt{md\_hom}( \ f \ , \ (\circledast_1,\dotsc, \circledast_D) \ )( \, \MDA \, )
\ := \
\underset{i_1\in I_1}{\circledast_1}\dotsc \underset{i_D\in I_D}{\circledast_D} \
\vec{f}( \ \MDA|_{\{i_1\}\times\dotsc\times\{i_D\}} \ )
\end{align*}
The combine operators' underset notation denotes straightforward iteration (explained formally in the Appendix, Notation~\ref{app_not_iterative_func_app}), and the MDA $\MDA|_{\{i_1\}\times\dotsc\times\{i_D\}}$ is the restriction of $\MDA$ to the MDA containing the single element accessed via MDA indices $(i_1,\dotsc,i_D)$.
Function $\vec{f}$ behaves like scalar function~$f$, but~$\vec{f}$ operates on singleton MDAs (rather than scalars).
Function $\vec{f}$ is of type
\[
\vec{f}^{<i_1,\dotsc,i_D\in\INz>}:T^\texttt{INP}[ \ \{i_1\},\dotsc,\{i_D\} \ ]\to T^\texttt{OUT}[ \ \size{1}{MDA}{MDA}(\,\{i_1\}\,) \, ,\dotsc, \, \size{D}{MDA}{MDA}\,(\{i_D\}\,) \ ]
\]
and defined as~%
\[
\vec{f}(x)[\,j_1,\dotsc,j_D\,]:=f(\,x[\,i_1,\dotsc,i_D\,]\,)\footnote{
    We assume that MDH functions, when applied to singleton MDAs, return a singleton MDA, as such MDHs already cover all real-world cases we currently are aware of.
    }
\]
\vspace*{2px}
For $\texttt{md\_hom}( \, f , \, (\circledast_1,\dotsc, \circledast_D) \, )$, we require by definition the homomorphic property (Definition~\ref{def_mdh}),
i.e., for each $d\in[1,D]_\IN$, it must hold:
\begin{align*}
&
\texttt{md\_hom}( \ f \ , \ (\circledast_1,\dotsc, \circledast_D) \ ) (\ \MDA_1 \dplus_d \ \MDA_2 \ ) \ = \\[5pt]
&\hspace*{80px}
\texttt{md\_hom}( \ f \ , \ (\circledast_1,\dotsc, \circledast_D) \ ) ( \, \MDA_1 \, )
\ \ \circledast_d \ \
\texttt{md\_hom}( \ f \ , \ (\circledast_1,\dotsc, \circledast_D) \ ) ( \, \MDA_2 \, )
\end{align*}
\end{definition}

Using Definition~\ref{def_md_hom}, we express any MDH function uniformly via higher-order function $\mdh$ using only the MDH's behavior $f$ on scalar values\footnote{
     For simplicity, we ignore that the scalar functions of some MDHs (such as Mandelbrot) also take as input MDA indices, which requires slight, straightforward extension of function $\mdh$, as outlined in the Appendix, Section~\ref{app_design_dec_mdh}.
}
and its combine operators $\circledast_1,\dotsc,\circledast_D$.
The other direction also holds: each function expressed via $\mdh$ is an MDH function, because we require the homomorphic property for $\mdh$.

Note that we can potentially allow in Definition~\ref{def_md_hom} the case $D=0$ in which we would define the $\mdh$ instance equal to the scalar function $f$:
\[
    \mdh( \ f \ , \ () \ \, ) \  := \ f
\]

Note further that function $\mdh$ is defined as partial function,
because the homomorphic property is not met
for all potential combinations of combine operators, e.g., $\co{1}=+$~(point-wise addition) \mbox{and $\co{2}=*$~(point-wise multiplication)}.
However, in many real-world examples, an MDH's combine operators are a mix of concatenations and point-wise combinations according to the same binary function.
The following lemma proves that any instance of the $\mdh$ higher-order function for such a mix of combine operators is a well-defined MDH function.

\begin{lemma}
\label{lemma_mdhom}
    Let $\oplus:T\to T$ be an arbitrary but fixed associative and commutative binary function on scalar type $T\in\type$.
    Let further $\circledast_1,\dotsc,\circledast_D$ be combine operators of which any is either concatenation (Example~\ref{def:mda_concat}) or point-wise combination according to binary function $\oplus$ (Example~\ref{def:mda_pw}).

    It holds that $\texttt{md\_hom}( \, f , \, (\circledast_1,\dotsc, \circledast_D) \, )$ is well defined.
\end{lemma}
\begin{proof}
    See Appendix, Section~\ref{app_proof_mdhhom_lemma}.
\end{proof}

\vspace{5px}

MDH functions are defined (Definition~\ref{def_mdh}) such that they uniformly operate on MDAs (Figure~\ref{hl_overview}).
We introduce higher-order function $\iv$ to prepare domain-specific inputs (e.g., a matrix and a vector in the case of matrix-vector multiplication) as an MDA, and we use function $\ov$ to transform the output MDA back to the domain-specific data requirements (like storing it as a transposed matrix in the case of matrix multiplication, or splitting it into multiple outputs as we will see later with examples).
We introduce both higher-order functions in the following.

\vspace*{20px}

\subsection{View Functions}
\label{sec_views}

We start, in Section~\ref{sec_views_prep}, by formally introducing \emph{Buffers (BUF)} and \emph{Index Functions}~--~both concepts are central building blocks in our definition of higher-order functions $\iv$ and $\ov$.
In our approach, we use BUFs to represent domain-specific input and output data (scalars, vectors, matrices, etc), and \emph{index functions} are used by the user to conveniently instantiate higher-order functions $\iv$ and $\ov$, e.g., index function $(i,k)\mapsto(i,k)$ and $(i,k)\mapsto(k)$ used in Figure~\ref{fig:intro_example} to instantiate function $\iv$, and $(i,k)\mapsto(i)$ is used for $\ov$.

Sections~\ref{sec_views_iv} and~\ref{sec_views_ov} introduce \emph{input views} and \emph{output views} which are central concepts in our approach.
We define \emph{input views} as arbitrary functions that map a collection of BUFs to an MDA~(Figure~\ref{hl_overview});~higher-order function $\iv$ is then defined to conveniently compute an important class of input view functions that are relevant for expressing real-world computations.
Correspondingly, Section~\ref{sec_views_ov} defines \emph{output views} as functions that transform an MDA to a collection of BUFs, and higher-order function $\ov$ is defined to conveniently compute important output views.

Finally, we discuss in Section~\ref{sec_views_iv_and_ov} the relationship between higher-order function $\iv$ and $\ov$:~we prove that both functions are inversely related to each other, allowing arbitrarily switching between our internal MDA representation and our domain-specific BUF representation~(as required for our code generation process discussed later).

\newpage

\subsubsection{Preparation}
\label{sec_views_prep}

We formally introduce \emph{Buffers (BUF)} and \emph{Index Functions} in the following.

\begin{definition}[Buffer]
\label{def_buffer}
Let $T\in\type$ be an arbitrary scalar type, $D\in\INz$ a natural number\footnote{
    We use the case $D=0$ to represent scalar values (formal details provided in the Appendix, Section~\ref{app_scalar_values}).
}, and $N:=(N_1 \, ,\dotsc, \, N_D)\in\IN^D$ a tuple of natural numbers.

A \emph{Buffer (BUF)} $\Buf$ that has \emph{dimensionality}~$D$, \emph{size}~$N$, and \emph{scalar type}~$T$ is a function with the following signature:
\[
  \Buf:[0,N_1)_{\INz}\,\times\dotsc\times\,[0,N_D)_{\INz} \, \to \, T\cup\{\bot\}
\]
Here, we use $\bot$ to denote the \emph{undefined value}.
We refer to $[0,N_1)_{\INz}\,\times\dotsc\times\,[0,N_D)_{\INz} \, \to \, T\cup\{\bot\}$ as the \emph{type} of BUF $\Buf$, which we also denote as $T^{N_1\times\dotsc\times N_D}$, and we refer to the set $\BUFIDXs:=\{ \, [0,N)_{\INz} \ | \ N\in\IN \}$
as \emph{BUF index sets}.
Analogously to Notation~\ref{not_mda}, we write $\Buf[\,i_1,\dotsc,i_D\,]$ instead of $\Buf(i_1,\dotsc,i_D)$ to avoid a too heavy usage of parentheses.
\end{definition}

In contrast to MDAs (Definition~\ref{def_mda}), a BUF always operates on a contiguous range of natural numbers starting from~$0$, and a BUF may contain undefined values.
These two differences allow straightforwardly transforming BUFs to data structures provided by low-level programming languages (e.g., \emph{C arrays}, as used in OpenMP, CUDA, and OpenCL).

Note that in our generated program code~(discussed later in Section~\ref{ch:low_level}), we implement MDAs on top of BUFs, as straightforward aliases that access BUFs, so that we do not need to transform MDAs to low-level data structures and/or store them otherwise physically in memory.

\begin{definition}[Index Function]
    \label{def_idx_func}
    Let $D\in\IN$ be a natural number (representing an MDA's dimensionality) and $D_b\in\INz$ (representing a BUF's dimensionality).

    An \emph{index function} $\idx$ from $D$-dimensional MDA indices to $D_b$-dimensional BUF indices is any meta-function of type
    \[
    \idx^{<I^\texttt{MDA}_1,\dotsc,I^\texttt{MDA}_D\in\IDXs^D>}:
    I^\texttt{MDA}_1\times\dotsc\times I^\texttt{MDA}_D
    \ \to \
    I^\texttt{BUF}_1\times\dotsc\times I^\texttt{BUF}_{D_b}
    \]
    for
    {$
    (I^\texttt{BUF}_1,\dotsc,I^\texttt{BUF}_{D_b}) := \bsize{}{MDA}{BUF}(I^\texttt{MDA}_1,\dotsc,I^\texttt{MDA}_D)
    $} where
    $\bsize{}{MDA}{BUF}:{\IDXs}^D\to{\BUFIDXs}^{D_b}$
    is an arbitrary but fixed function
    that maps $D$-many MDA index sets to $D_b$-many BUF index sets.
    We denote the type of index functions as $\IDXFCT^{<D,D_b \ | \ \sbsize{}{MDA}{BUF}>}$.

    In words:
    Index functions have to be capable of operating on any potential MDA index set.
    This generality will be required later for using index functions also on parts of MDAs whose index sets are subsets of the original MDA's index sets.
\end{definition}

We will use index functions to access BUFs.
For example, in the case of \texttt{MatVec} (Figure~\ref{example_intro_1}), we access its input matrix using index function
$(i,k)\mapsto(i,k)$
which is of type
\[
\IDXFCT^{<D:=2,D_b:=2 \ | \
\sbsize{}{MDA}{BUF}(I^\texttt{MDA}_1,I^\texttt{MDA}_2)\ := \ [0,\max(I^\texttt{MDA}_1)]_{\INz},[0,\max(I^\texttt{MDA}_2)]_{\INz}
>}
\]
and we use index function
$(i,k)\mapsto(k)$
to access \texttt{MatVec}'s input vector, which is of type
\[
\IDXFCT^{<D:=2,D_b:=1 \ | \
\sbsize{}{MDA}{BUF}(I^\texttt{MDA}_1,I^\texttt{MDA}_2) \ := \
[0,\max(I^\texttt{MDA}_2)]_{\INz}
>}
\]
Further examples of index function, e.g., for stencil computation \texttt{Jacobi1D}, are presented in the Appendix, Section~\ref{app_examples_idx_funcs}, for the interested reader.

\subsubsection{Input Views}
\label{sec_views_iv}

We define \emph{input views} as any function that compute an MDA from a collection of (user-defined) BUFs.
For example, in the case of \texttt{MatVec}, its input view takes as input two BUFs~--~a matrix and a vector~--~and it yields a two-dimensional MDA containing pairs of matrix and vector elements (illustrated in Figure~\ref{example_intro_1}).
In contrast, the input view of \texttt{Jacobi1D} takes as input a single BUF (representing a vector) only, and it computes an MDA containing triples of BUF elements~(Figure~\ref{example_intro_2}).

\begin{definition}[Input View]
    \label{def_input_view}

    An \emph{Input View (IV)} from $B$-many BUFs, $B\in\IN$, of arbitrary but fixed types $T^{N^b_1\times\dotsc\times N^b_{D_b}}_b$, $b\in[1,B]_\IN$,
    to an MDA of arbitrary but fixed type $T[I_1,\dotsc,I_D]$
    is any function $\mathfrak{iv}$ of type:
    \vspace*{-3px}
    \begin{align*}
    &
    {\mathfrak{iv}}: \
        \underbrace{
        \cart_{b=1}^{B}
            T^{N^b_1\times\dotsc\times N^b_{D_b}}_b
        }_\text{BUFs}
        \hspace{3px}
        \to_p
        \hspace{3px}
        \underbrace{
        T[ \ I_1 \, ,\dotsc, \, I_D \ ]
        }_\text{MDA}
    \end{align*}
    We denote the type of $\mathfrak{iv}$ as
    $
    \texttt{IV}^{<
    \overbrace{\scriptstyle
    B \ | \
    (D_b\,)_{b\in[1,B]_\IN} \ | \
    (\,N^b_1,\dotsc,N^b_{D_b}\,)_{b\in[1,B]_\IN} \ | \
    (T_b\,)_{b\in[1,B]_\IN}
    }^\text{BUFs' Meta-Parameters}
    |
    \overbrace{\scriptstyle
    D
    \ | \ I_1,\dotsc,I_D
    \ | \ T
    }^\text{MDA's Meta-Parameters}
    \: >}
    $.
\end{definition}

\begin{example}[Input View~--~\texttt{MatVec}]
    The input view
    of \texttt{MatVec} on a $1024\times512$ matrix and $512$-sized vector (sizes chosen arbitrarily), both of integers $\IZ$, is of type
    \[
        \texttt{IV}^{<
        \overbrace{\scriptstyle
        B=2\,|\,D_1=2,D_2=1\,|\,(N^1_1=1024,N^1_2=512),(N^2_1=512)\,|\,T_1=\IZ,T_2=\IZ
        }^\text{BUFs' meta-parameters}
        |
        \overbrace{\scriptstyle
        D=2\,|\,I_1=[0,1024)_{\IN_0},I_2=[0,512)_{\IN_0}\,|\,T=\IZ\times\IZ
        }^\text{MDA's meta-parameters}
        >}
        \]
        and defined as
        \[
        \underbrace{[\,M(i,k)\,]_{i\in [0,1024)_{\IN_0}, k\in [0,512)_{\IN_0}}}_{\text{BUF (Matrix)}},
        \underbrace{[\,v(k)\,]_{k\in [0,512)_{\IN_0}}}_{\text{BUF (Vector)}}
        \mapsto
        \underbrace{[\,M(i,k), v(k)\,]_{i\in [0,1024)_{\IN_0}, k\in [0,512)_{\IN_0}}}_{\text{MDA}}
        \]
        Here, the BUFs' meta-parameters are as follows:
        $B=2$ is the number of BUFs (matrix and vector);
        $D_1=2$ is dimensionality of the matrix and $D_2=1$ the vector's dimensionality;
        $(N^1_1,N^1_2)=(1024,512)$ is the matrix size and
        $N^2_1=512$ the vector's size;
        $T_1,T_2=\IZ$ are the scalar types of the matrix and vector.
        The MDA's meta-parameters are:
        $D=2$ is the computed MDA's dimensionality;
        $I_1,I_2$ are the MDA's index sets;
        parameter $T=\IZ\times\IZ$ is MDA's scalar type (pairs of matrix/vector elements~--~see Figure~\ref{example_intro_1}).
\end{example}
\begin{example}[Input View~--~\texttt{Jacobi1D}]
        The input view of \texttt{Jacobi1D} on a $512$-sized vector of integers is of type
        \[
            \texttt{IV}^{<
            \overbrace{\scriptstyle
            B=1\,|\,D_1=1\,|\,(N^1_1=512)\,|\,T_1=\IZ
            }^\text{BUFs' meta-parameters}
            |
            \overbrace{\scriptstyle
            D=1\,|\,I_1=[0,512-2)_{\IN_0}\,|\,T=\IZ\times\IZ\times\IZ
            }^\text{MDA's meta-parameters}
            >}
            \]
            and defined as
            \[
            \underbrace{[\,v(i)\,]_{i\in [0,512)_{\IN_0}}}_{\text{BUF (Vector)}}
            \mapsto
            \underbrace{[\,v(i+0),v(i+1),v(i+2)\,]_{i\in [0,512-2)_{\IN_0}}}_{\text{MDA}}
            \]
\end{example}

\vspace*{7px}

We introduce higher-order function $\iv$:~it computes~--~in a structured way~--~important input views from user-defined index functions~%
$(\idx_{b,a})_{b\in[1,B]_\IN,a\in[1,A_b]_\IN}$~(Definition~\ref{def_idx_func}).
Here, $B\in\IN$ represents the number of BUFs that the computed input view will take as input, and $A_b$ represents the number of accesses to the $b$-th BUF  required for computing an individual MDA element.

In the case of \texttt{MatVec} (Figure~\ref{example_intro_1}), we use
$B:=2$ because \texttt{MatVec} has two input BUFs:~a matrix~$M$ (the first input of \texttt{MatVec} and thus identified by $b=1$)
and
a
vector~$v$ (identified by $b=2$).
For the number of accesses, we use for the matrix $A_1:=1$, as one element is accessed within matrix~$M$ to compute
an individual MDA element~--~matrix element $M[i,k]$ for computing MDA element at position $(i,k)$.
For the vector, we use $A_2:=1$, as the single element $v[k]$ is accessed within the vector.
The index functions of \texttt{MatVec} are:~$\idx_{1,1}(i,k):=(i,k)$ (used to access the matrix) and $\idx_{2,1}(i,k):=(k)$ (used for the vector).

In contrast, for \texttt{Jacobi1D}~(Figure~\ref{example_intro_2}), we use $B:=1$ because \texttt{Jacobi1D} has vector $v$ as its only input, and we use $A_1:=3$ because the vector is accessed three times to compute an individual MDA element at arbitrary position~$i$:~first access $v[i+0]$, second access $v[i+1]$, and third access $v[i+2]$.
The index functions of \texttt{Jacobi1D} are:~%
$\idx_{1,1}(i):=(i+0)$, $\idx_{1,2}(i):=(i+1)$, and $\idx_{1,3}(i):=(i+2)$.

More generally, higher-order function $\iv$ uses the index functions
$\idx_{b,a}$
to compute an input view that maps BUFs $\Buf_1,\dotsc,\Buf_B$ to an MDA $\MDA$ that contains at position $i_1,\dotsc,i_D$ the following element:
{
\begin{align*}
\MDA[i_1,\dotsc,i_D] \ := \ \ \
&
( \,
\underbrace{( \
\underbrace{\Buf_1[ \, \idx_{1,1}(i_1,\dotsc,i_D) \, ]\in T_1}_{a=1}  \ , \, \dotsc \, , \  \underbrace{\Buf_1[ \, \idx_{1,A_1}(i_1,\dotsc,i_D) \, ]\in T_1}_{a=A_1} \ )}_{b=1}
, \\[2pt]
&\hspace*{132px} \vdots  \\[2pt]
&\hspace*{-3px}
\
\underbrace{( \
\underbrace{\Buf_B[ \, \idx_{B,1}(i_1,\dotsc,i_D) \, ]\in T_B}_{a=1} \ , \, \dotsc \, , \ \underbrace{\Buf_B[ \, \idx_{B,A_B}(i_1,\dotsc,i_D) \, ]\in T_B}_{a=A_B} \ )}_{b=B}
\,)
\end{align*}
}
The element consists of $B$-many tuples~--~one per BUF~--~and each such tuple contains $A_b$-many elements~--~one element per access to the $b$-th BUF.
For \texttt{MatVec}, the element is of the form
{
\begin{align*}
\MDA[i_1,i_2] \ := \ \ \
&
( \,
\underbrace{( \
\underbrace{\Buf_1[ \, \idx_{1,1}(i_1,i_2):=(i_1,i_2) \, ]\in T_1}_{a=1} \ )}_{b=1}
\ ,
\
\underbrace{( \
\underbrace{\Buf_2[ \, \idx_{2,1}(i_1,i_2):=(i_2) \, ]\in T_2}_{a=1} \ )}_{b=2}
\,)
\end{align*}
}
and for \texttt{Jacobi1D}, the element is
{
\begin{align*}
&
\MDA[i_1] \ :=  \\[3px]
& \ \ \
( \,
\underbrace{( \
\underbrace{\Buf_1[ \, \idx_{1,1}(i_1):=(i_1+0) \, ]\in T_1}_{a=1},
\underbrace{\Buf_1[ \, \idx_{1,2}(i_1):=(i_1+1) \, ]\in T_1}_{a=2},
\underbrace{\Buf_1[ \, \idx_{1,3}(i_1):=(i_1+2) \, ]\in T_1}_{a=3}
\ )}_{b=1}
\,)
\end{align*}
}

\begin{figure}[p!]
    \centering

    \vspace{60px}

    \includegraphics[scale=0.32]{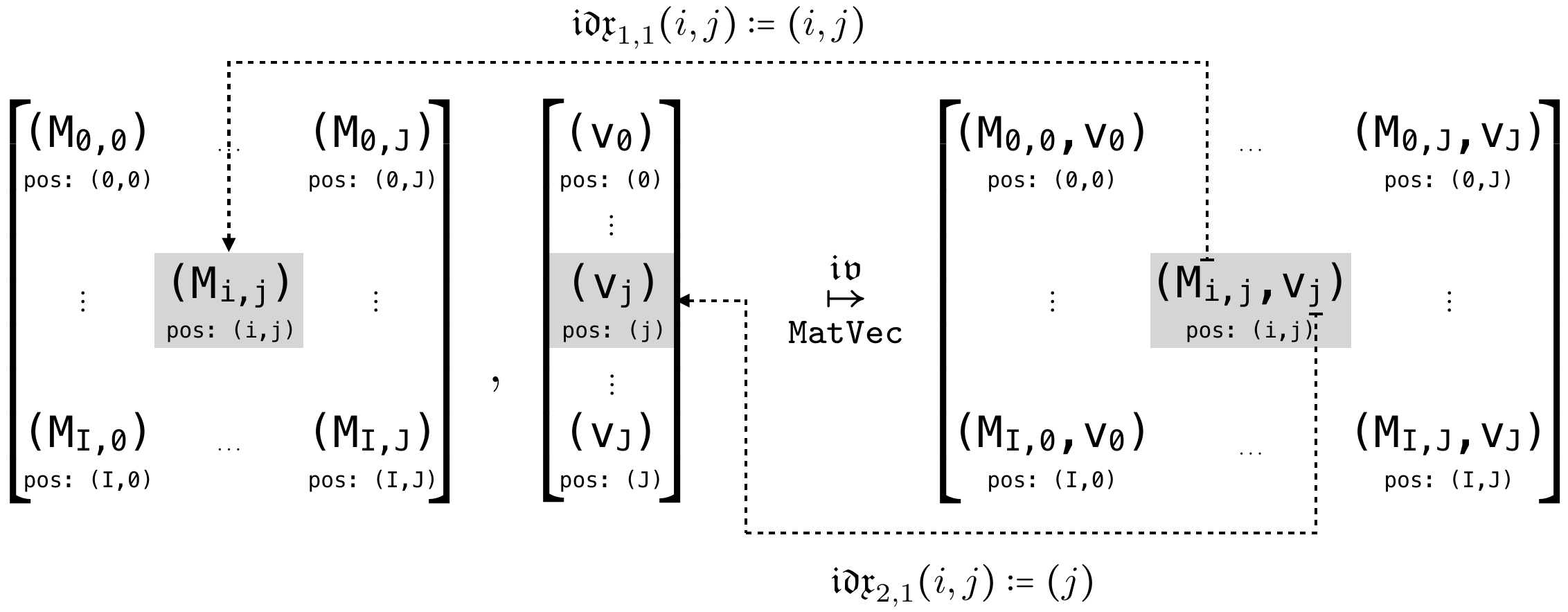}
    \caption{
    Input view illustrated using the example \texttt{MatVec}
    }
    \label{img_example_inp_view_matvec}

    \vspace{65px}

    \includegraphics[scale=0.32]{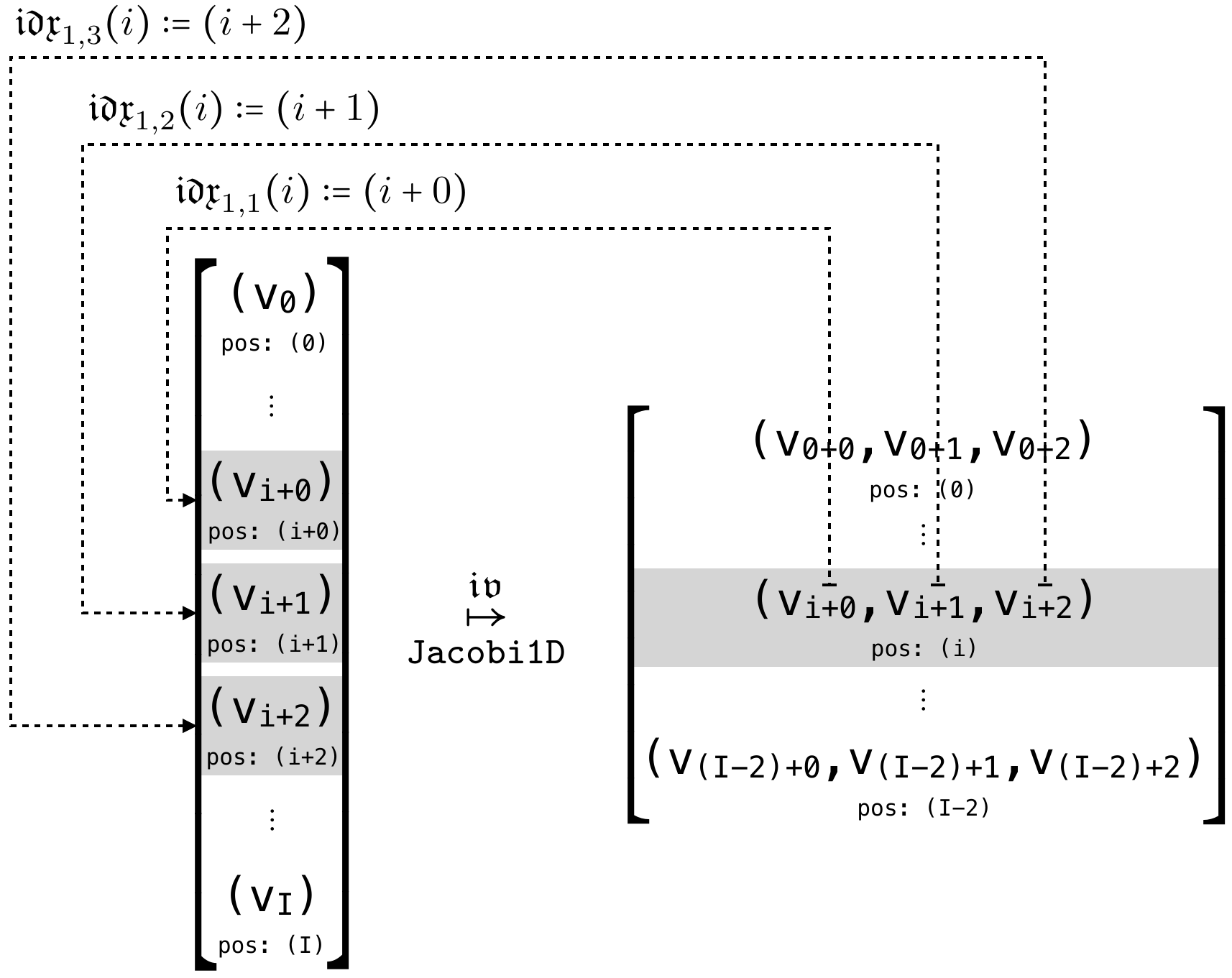}
    \caption{
    Input view illustrated using the example \texttt{Jacobi1D}
    }
    \label{img_example_inp_view_jacobi1d}

    \vspace{70px}

\end{figure}

In the following, we introduce higher-order function $\iv$ which computes important input views conveniently and in a uniform, structured manner.
Function $\iv$ takes as input a collection of index functions (Definition~\ref{def_idx_func}), and it uses these index functions to compute a corresponding input view (Definition~\ref{def_input_view}), as outlined above and described in detail in the following.

Figures~\ref{img_example_inp_view_matvec} and~\ref{img_example_inp_view_jacobi1d} use the examples \texttt{MatVec} and \texttt{Jacobi1D} to informally illustrate how function $\iv$ uses index functions to compute input views.
In the two figures,
we use domain-specific identifiers for better clarity:~in the case of \texttt{MatVec}, we use for its two input BUFs the identifiers $M$ and $v$ instead of $\Buf_1$ and $\Buf_2$, as well as identifiers $i$ and $j$ instead of $i_1$ and $i_2$ for index variables;~for \texttt{Jacobi1D}, we use identifier $v$ instead of $\Buf_1$ and $i$ instead of $i_1$.

We now formally define higher-order function $\iv$.
For high flexibility and formal correctness, function $\iv$ relies on a type that involves many meta-parameters.
The high number of meta-parameters, and the resulting complex type of function $\iv$, might appear daunting to the user.
However, Notation~\ref{notation_iv} confirms that despite the complex type of function $\iv$, the function can be conveniently expressed by the user (as also illustrated in Figure~\ref{fig:intro_example}), because meta-parameters can be automatically deduced from $\iv$'s input parameters.

\newpage

\begin{definition}[Higher-Order Function $\iv$]
\label{def_iv}
Function $\iv$ is of type

\vspace*{10px}

\hspace*{6px}
\texttt{inp\_view}$^\texttt{<}\,$%
  $\underbrace{^{B\in\IN}}_\texttt{Number BUFs}\,^|\,$
  $\underbrace{^{A_1,\dotsc,A_B\in\IN}}_\texttt{Number BUFs' Accesses}\,^|\,$
  $\underbrace{^{D_1,\dotsc,D_B\in\INz}}_\texttt{BUFs' Dimensionalities}\,^|\,$
  $\underbrace{^{D\in\IN}}_\texttt{MDA Dimensionality}\,^|\,$
\\[7pt]
\hspace*{63px}
    $\underbrace{^{(\sbsize{}{MDA}{BUF}^{b,a}:{\IDXs}^D\to{\BUFIDXs}^{D_b})_{b\in[1,B]_\IN, a\in[1,A_{b}]_\IN}}}_\texttt{Index Set Functions (MDA indices to BUF indices)}$
$^\texttt{>}:$\\[0pt]

\noindent
$
\underbrace{
\underbrace{\cart_{b=1}^{B}}_\text{Buffer}
\underbrace{\cart_{a=1}^{A_b}}_\text{Access}
\ \underbrace{
\IDXFCT^{<D,D_b\,|\,\sbsize{}{MDA}{BUF}^{b,a}>}%
}_\text{Index Function: $\idx_{b,a}$}
}_\text{Index Functions: $\idx_{1,1},\dotsc,\idx_{B,A_B}$}
\ \to \
\underbrace{
\texttt{IV}{
^{\tiny
<
\overbrace{
B \,|\,
D_1,\dotsc,D_B \,|\,
\rightarrow \,|\,
T_1,\dotsc,T_B\in\type
}^{\text{BUFs' Meta-Parameters}}
\,|\,
}
_{\hspace{4px}\tiny
\underbrace{
D
\,|\, I_1,\dotsc,I_D\in\IDXs
\,|\,
\rightarrow \,
}_{\text{MDA's Meta-Parameters}}
><
\underbrace{
 N^1_1,\dotsc,N^B_{D_B}
\, | \,
T
}_{\text{Postponed Parameters}}
>}}
}_{\text{Input View: $\mathfrak{iv}$}}
$

\vspace*{10px}
\noindent
for
$
    N^b_d:=1+\max( \ \bigcup_{a\in[1,A_b]_\IN}\bsize{d}{MDA}{BUF}{^{b,a}}(I_1,\dotsc,I_D) \ )
$ and $
    T:=\cart_{b=1}^{B}\cart_{a=1}^{A_b}T_b
$,
and it is defined as:

\begin{align*}
&
\underbrace{ \,(\idx_{b,a})_{b\in[1,B]_\IN,a\in[1,A_b]_\IN}\, }_{\text{Index Functions}}
\ \mapsto \
\underbrace{
(\,\underbrace{\Buf_1,\dotsc,\Buf_B}_{\text{BUFs}}\,)
\overset{\mathfrak{iv}}{\mapsto}
\underbrace{
\MDA
}_\text{MDA}
}_\text{Input View}
\end{align*}
for
\begin{align*}
\MDA[i_1,\dotsc,i_D] := ( \, \MDA_{b,a}[i_1,\dotsc,i_D] \, )_{b\in[1,B]_\IN,a\in[1,A_b]_\IN}
\end{align*}
and
\begin{align*}
    \MDA_{b,a}[i_1,\dotsc,i_D] \ \ := \ \ \Buf_b[ \ \idx_{b,a}(i_1,\dotsc,i_D) \ ]
\end{align*}
\end{definition}

\noindent
Higher-order function $\iv$ takes as input a collection of index functions that are of types $\IDXFCT$ (Definition~\ref{def_idx_func}), and it computes an input view of type \texttt{IV} (Definition~\ref{def_input_view}) based on the index functions, as illustrated in Figures~\ref{img_example_inp_view_matvec} and~\ref{img_example_inp_view_jacobi1d}.

As concrete meta-parameter values of type $\IDXFCT$ (listed in angle brackets), we use straightforwardly the values of meta-parameters passed to function $\iv$.
Similarly, we use the particular meta-parameter values of function $\iv$ also for type \texttt{IV}'s meta-parameters~$B$,~$D_1,\dotsc,D_B$, and~$D$

To be able using the computed input view on arbitrarily typed input buffers and letting the input view compute MDAs that have arbitrary index sets, we keep \texttt{IV}'s meta-parameters $T_1,\dotsc,T_B$ (scalar types of the computed view's input buffers) and $I_1,\dotsc,I_D$ (index sets of the view's returned MDA) flexible.
Being flexible in the BUFs' scalar types and MDA's index sets is important for convenience:~for example, in the case of \texttt{MatVec}, this flexibility allows using the computed input view generically for matrices and vectors that have arbitrary scalar types (e.g., either \texttt{int} or \texttt{float}) and sizes $(I,J)$ (matrix) and $J$ (vector), for arbitrary $I,J\in\IN$, without needing to re-compute a new input view every time again when BUFs' scalar types and/or sizes change.

We automatically
compute
the sizes $N^b_d$ of BUFs
in \texttt{IV}'s meta-parameter list
(e.g., in the case of \texttt{MatVec}, the size of the input matrix~$(I,J)$ and vector size~$J$), according to the formula in Definition~\ref{def_iv}, based on the flexible MDA's index sets (e.g., sets~$[0,I)_{\INz}$ and~$[0,J)_{\INz}$ for \texttt{MatVec}).
By computing BUF sizes from MDA index sets (rather than requesting the sizes explicitly from the user), we achieve strong error checking:~for example, for \texttt{MatVec}, we can ensure~--~already on the type level~--~that the number of columns of its input matrix and the size of its input vector match.
To compute the BUF sizes, we \emph{postpone} via $\rightarrow$ (defined formally in the Appendix, Definition~\ref{app_del_meta_func}) the sizes in \texttt{IV}'s meta-parameter list to later meta-parameter stages;~this is because the sizes are defined in early stages and thus have no access to the MDA's index sets which are defined in later stages.
Our formula in Definition~\ref{def_iv} then
works
as follows:~%
for each BUF~$b$, its size~$N^b_d$ has to be well-defined in each of its dimensions $d$,
for all accesses $a$, which is checked by using the BUF's index functions on all indices within the MDA index sets $I_1,\dotsc,I_D$.
Here, in the computation of~$N^b_d$, function
$\bsize{d}{MDA}{BUF}{^{b,a}}$ computes the~$d$-th component of the $D_b$-sized output tuple of~$\bsize{}{MDA}{BUF}{^{b,a}}$ (the computed component is the index set of BUF~$b$ in dimension~$d$ for the~$a$-th index function used to access the BUF).

We automatically compute also MDA's scalar type~$T$ using the formula presented in Definition~\ref{def_iv}.
The formula computes $T$ as a tuple that consists of the BUFs' scalar types, as each MDA element consist of BUF elements (illustrated in Figures~\ref{img_example_inp_view_matvec} and~\ref{img_example_inp_view_jacobi1d}).
Postponing $T$ in \texttt{IV}'s meta-parameter list is done (analogously as for $N^b_d$), but is actually not required, because the BUFs' scalar types $T_1,\dotsc,T_B$ are already defined in earlier meta-parameter stages than $T$.
However, we will see that postponing $T$ is required later in the Definition~\ref{def_ov} of higher-order function $\ov$;~therefore, we postpone $T$ also in our definition of $\iv$ to increase consistency between our
definitions of $\iv$ (Definition~\ref{def_iv}) and $\ov$ (Definition~\ref{def_ov}).

\vspace{5px}

Note that function $\iv$ is not capable of computing every kind of input view function (Definition~\ref{def_input_view}).
For example, $\iv$ cannot be used for computing MDAs that are required for expressing computations on sparse data formats~\cite{9407251}, because such MDAs need dynamically accessing BUFs.
This limitation of $\iv$ can be relaxed by generalizing our index functions toward taking additional, dynamic input arguments, which we consider as future work (as outlined in Section~\ref{ch:fw}).

\begin{notation}[Input Views]
\label{notation_iv}
Let
$
\iv^{<\dotsc>}( \
(\,\idx_{b,a}\,)_{b\in[1,B]_\IN,a\in[1,A_b]_\IN}
\ )
$
be a particular instance of higher-order function $\iv$ (meta-parameters omitted via ellipsis for simplicity) for an arbitrary but fixed choice of index functions.
Let further $\texttt{ID}_1,\dotsc,\texttt{ID}_B\in\Sigma^*$ be arbitrary, user-defined BUF identifiers (e.g., $\texttt{ID}_1=\texttt{"M"}$ and $\texttt{ID}_2=\texttt{"v"}$ in the case of \texttt{MatVec}), for an arbitrary, fixed collection of letters $\Sigma=\{A,B,C,\dotsc,a,b,c,\dotsc,1,2,3,\dotsc\}$.

For better readability, we use the following notation for the 2-dimensional structure of index functions taken as input by function $\iv$, inspired by \citet{9370308}:
\begin{align*}
\iv
( \
\texttt{ID}_1:\,\idx_{1,1}\,,\dotsc,\,\idx_{1,A_1}
\ \ \ , \,\dotsc\, , \ \ \
\texttt{ID}_B:\,\idx_{B,1}\,,\dotsc,\,\idx_{B,A_B}
\ )
\end{align*}
We refrain from stating $\iv$'s meta-parameters in our notation, as the parameters can be automatically deduced from the number and types of index functions.
\end{notation}

\begin{example}[Notation Input Views]
Function $\iv$ is used for $\texttt{MatVec}$ and $\texttt{Jacobi1D}$ (in Notation~\ref{notation_iv}) as follows:

\begin{align*}
&
\texttt{\underline{MatVec:}} \ \ \ \
&&\iv( \
\underbrace{\texttt{M:}\underbrace{(i,k)\mapsto(i,k)}_\texttt{a=1}}_\texttt{b=1}
\,,\,
\underbrace{\texttt{v:}\underbrace{(i,k)\mapsto(k)}_\texttt{a=1}}_\texttt{b=2}
\ )
\\
&
\texttt{\underline{Jacobi1D:}} \ \ \ \
&&\iv
( \
\underbrace{
\texttt{v:}
  \underbrace{(i)\mapsto (i+0)}_{a=1}
  \,,\,
  \underbrace{(i)\mapsto (i+1)}_{a=2}
  \,,\,
  \underbrace{(i)\mapsto (i+2)}_{a=3}
}_{b=1}
\ )
\end{align*}
\end{example}

\vspace*{20px}

\subsubsection{Output Views}
\label{sec_views_ov}

An \emph{output view}
is the counterpart
of an input view:~in contrast to an input view which maps BUFs to an MDA, an output view maps an MDA to a collection of BUFs.
In the following, we define output views, and we introduce higher-order function $\ov$ which computes output views in a structured manner (analogously to function $\iv$ for input views).

\begin{figure}[h!]
    \centering
    \includegraphics[scale=0.32]{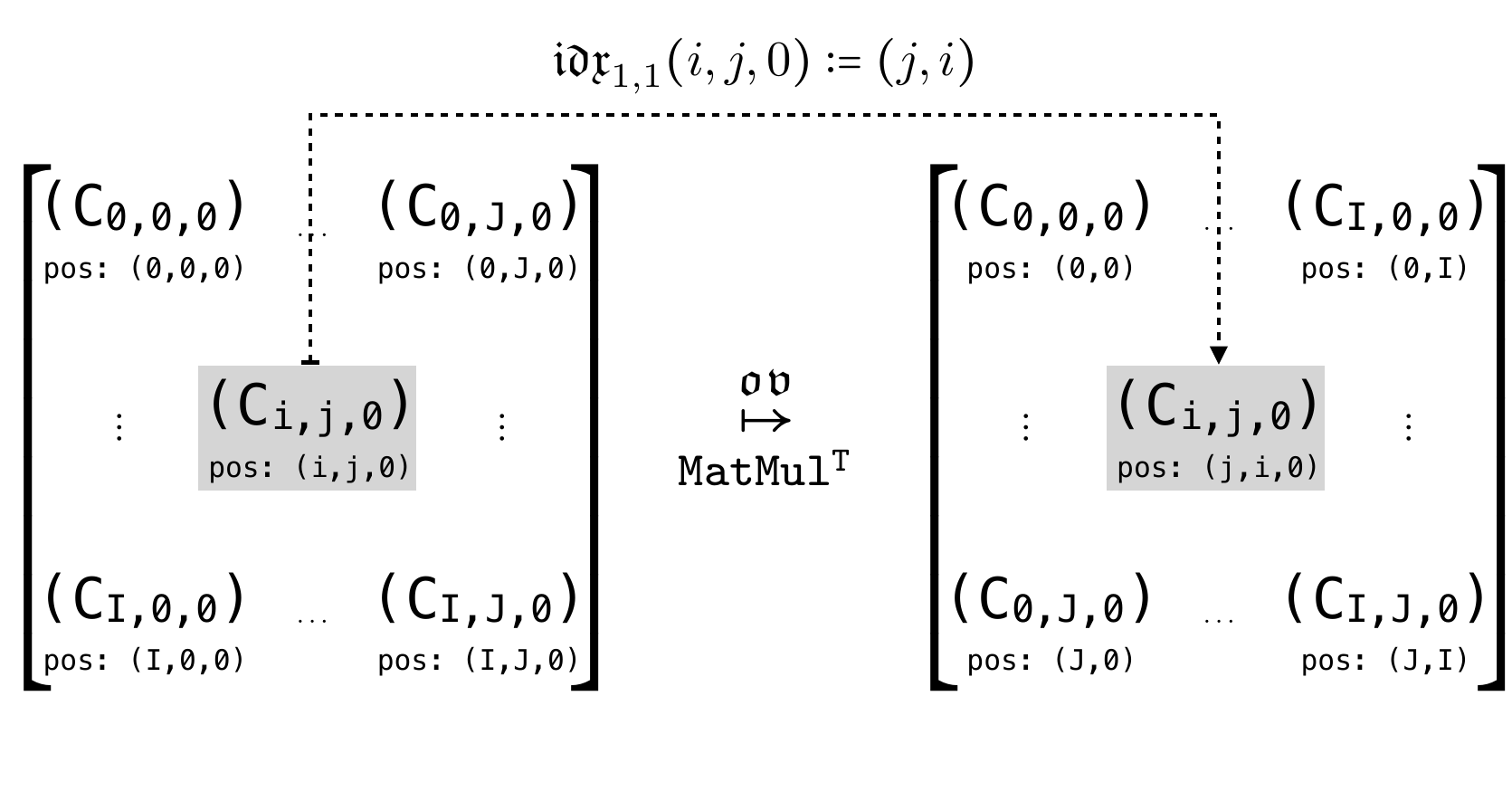}
    \caption{
    Output view illustrated using the example \emph{transposed Matrix Multiplication}
    }
    \label{fig_ov_examples_trans_matmul}
\end{figure}
\begin{figure}[h!]
    \centering
    \includegraphics[scale=0.32]{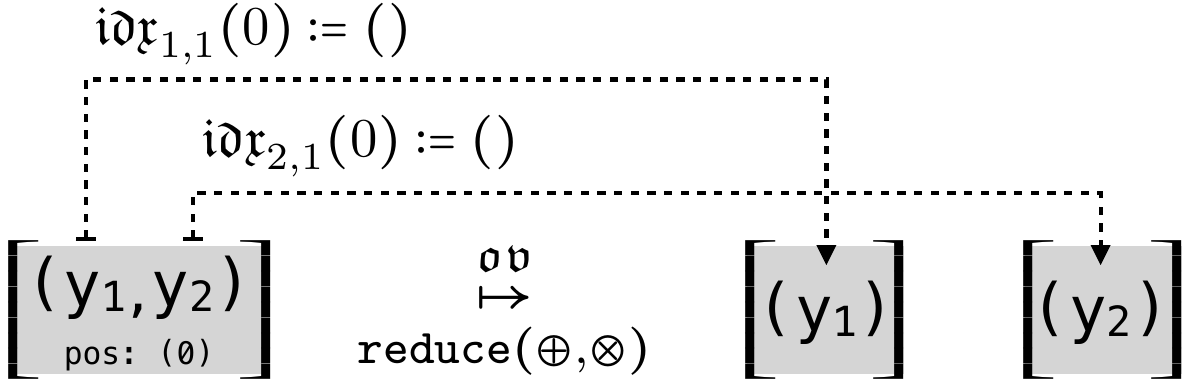}
    \caption{
    Output view illustrated using the example \emph{Double Reduction}
    }
    \label{fig_ov_examples_double_reduce}
\end{figure}

Figures~\ref{fig_ov_examples_trans_matmul} and~\ref{fig_ov_examples_double_reduce} illustrate output views informally using the examples \emph{transposed Matrix Multiplication} and \emph{Double Reduction}.

In the case of transposed matrix multiplication (Figure~\ref{fig_ov_examples_trans_matmul}), the computed output MDA (the computation of matrix multiplication is presented later and not relevant for our following considerations) is stored via an output view as a matrix in a transposed fashion, using
index function $(i,j,0)\mapsto(j,i)$.
Here, the MDA's third dimension (accessed via index $0$) represents the so-called reduction dimension
of matrix multiplication, and it contains only one element after the computation, as all elements in this dimension are combined via addition.

For double reduction (Figures~\ref{fig_ov_examples_double_reduce}), we combine the elements within the vector twice~--~once using operator~$\oplus$ (e.g.,~$\oplus=+$~addition) and once using operator~$\otimes$ (e.g,~$\otimes=*$~multiplication).
The final outcome of double reduction is a singleton MDA containing a pair of two elements that represent the combined vector elements (e.g., the elements' sum and product).
We store this MDA via an output view as two individual scalar values, using index functions $(0)\mapsto()$\footnote{
    The empty braces denote accessing a scalar value (formal details provided in the Appendix, Section~\ref{app_scalar_values}).
} for both pair elements.

\begin{definition}[Output View]
    \label{def_output_view}

    An \emph{Output View (OV)} from an MDA of arbitrary but fixed type $T[I_1,\dotsc,I_D]$ to $B$-many BUFs, $B\in\IN$, of arbitrary but fixed types $T^{N^b_1\times\dotsc\times N^b_{D_b}}_b$, $b\in[1,B]_\IN$, is any function $\mathfrak{ov}$ of type:
    \begin{align*}
    &
    {\mathfrak{ov}}%
    : \
    \underbrace{
        T[ \ I_1 \, ,\dotsc, \, I_D \ ]
        }_\text{MDA}
        \hspace{3px}
        \to_p
        \hspace{3px}
        \underbrace{
        \cart_{b=1}^{B}
            T^{N^b_1\times\dotsc\times N^b_{D_b}}_b
        }_\text{BUFs}
    \end{align*}
    We denote the type of $\mathfrak{ov}$ as
    $
    \texttt{OV}^{<
    \overbrace{\scriptstyle
    D
    \ | \ I_1,\dotsc,I_D
    \ | \ T
    }^\text{MDA's Meta-Parameters}
    |
    \overbrace{\scriptstyle
    B \ | \
    (D_b\,)_{b\in[1,B]_\IN} \ | \
    (\,N^b_1,\dotsc,N^b_{D_b}\,)_{b\in[1,B]_\IN} \ | \
    (T_b\,)_{b\in[1,B]_\IN}
    }^\text{BUFs' Meta-Parameters}
    \: >}
    $.
\end{definition}

\begin{example}[Output View~--~\texttt{MatVec}]
    The output view of \texttt{MatVec} computing a $1024$-sized vector (size is chosen arbitrarily), of integers $\IZ$, is of type
    \[
        \texttt{OV}^{<
        \overbrace{\scriptstyle
        D=2\,|\,I_1=[0,1024)_{\IN_0},I_2=\{0\}\,|\,T=\IZ
        }^\text{MDA's meta-parameters}
        |
        \overbrace{\scriptstyle
        B=1\,|\,D_1=1\,|\,(N^1_1=1024)\,|\,T_1=\IZ
        }^\text{BUFs' meta-parameters}
        >}
        \]
        and defined as
        \[
        \underbrace{[\,w(i)\,]_{i\in [0,1024)_{\IN_0}, k\in\{0\}}}_{\text{MDA}}
        \mapsto
        \underbrace{[\,w(i)\,]_{i\in [0,1024)_{\IN_0}}}_{\text{BUF (Vector)}}
        \]
\end{example}
\begin{example}[Output View~--~\texttt{Jacobi1D}]
        The output view of \texttt{Jacobi1D} computing a $(512-2)$-sized vector of integers is of type
        \[
        \texttt{OV}^{<
        \overbrace{\scriptstyle
        D=1\,|\,I_1=[0,512-2)_{\IN_0}\,|\,T=\IZ
        }^\text{MDA's meta-parameters}
        |
        \overbrace{\scriptstyle
        B=1\,|\,D_1=1\,|\,(N^1_1=(512-2))\,|\,T_1=\IZ
        }^\text{BUFs' meta-parameters}
        >}
        \]
        and defined as
        \[
        \underbrace{[\,w(i)\,]_{i\in [0,512-2)_{\IN_0}, k\in\{0\}}}_{\text{MDA}}
        \mapsto
        \underbrace{[\,w(i)\,]_{i\in [0,512-2)_{\IN_0}}}_{\text{BUF (Vector)}}
        \]
\end{example}

We define higher-order function $\ov$ formally as follows.

\begin{definition}[Higher-Order Function $\ov$]
\label{def_ov}
Function $\ov$ is of type

\vspace*{5px}

\hspace*{6px}
\texttt{out\_view}$^\texttt{<}\,$%
  $\underbrace{^{B\in\IN}}_\texttt{Number of BUFs}\,^|\,$
  $\underbrace{^{A_1,\dotsc,A_B\in\IN}}_\texttt{Number BUFs' Accesses}\,^|\,$
  $\underbrace{^{D_1,\dotsc,D_B\in\INz}}_\texttt{BUFs' Dimensionalities}\,^|\,$
  $\underbrace{^{D\in\IN}}_\texttt{MDA Dimensionality}\,^|\,$
\\[7pt]
\hspace*{63px}
    $\underbrace{^{(\sbsize{}{MDA}{BUF}^{b,a}:{\IDXs}^D\to{\BUFIDXs}^{D_b})_{b\in[1,B]_\IN, a\in[1,A_b]_\IN}}}_\texttt{Index Set Functions (MDA indices to BUF indices)}$
$^\texttt{>}:$\\[0pt]

\noindent
$
\underbrace{
\underbrace{\cart_{b=1}^{B}}_\text{Buffer}
\underbrace{\cart_{a=1}^{A_b}}_\text{Access}
\ \underbrace{
\IDXFCT^{<D,D_b\,|\,\sbsize{}{MDA}{BUF}^{b,a}>}%
}_\text{Index Function: $\idx_{b,a}$}
}_\text{Index Functions: $\idx_{1,1},\dotsc,\idx_{B,A_B}$}
\, \to \,
\underbrace{
\texttt{OV}{
^{\tiny
<
\overbrace{
D
\,|\, I_1,\dotsc,I_D\in\IDXs
\,|\, \rightarrow
}^{\text{MDA's Meta-Parameters}}
\,|\,
}
_{\hspace{4px}\tiny
\underbrace{
B \,|\,
D_1,\dotsc,D_B \,|\,
\rightarrow \,|\,
T_1,\dotsc,T_B\in\type
}_{\text{BUFs' Meta-Parameters}}
><
\underbrace{
N^1_1,\dotsc,N^B_{D_B}
\, | \,
T
}_{\text{Postponed Parameters}}
>}}
}_{\text{Output View: $\mathfrak{ov}$}}
$

\vspace*{10px}
\noindent
which differs from $\iv$'s type only in mapping index functions to \texttt{OV} (Definition~\ref{def_output_view}), rather than \texttt{IV} (Definition~\ref{def_input_view}).
Function $\ov$ is defined as:
\begin{align*}
    &
    \underbrace{ \,(\idx_{b,a})_{b\in[1,B]_\IN,a\in[1,A_b]_\IN}\, }_{\text{Index Functions}}
    \ \mapsto \
    \underbrace{
    \underbrace{\MDA}_\text{MDA}
    \overset{\mathfrak{ov}}{\mapsto}
    (\,\underbrace{\Buf_1,\dotsc,\Buf_B}_{\text{BUFs}}\,)
    }_\text{Output View}
\end{align*}
for
\begin{align*}
    \Buf_b[ \ \idx_{b,a}(i_1,\dotsc,i_D) \ ] \ \ := \ \ \MDA_{b,a}[i_1,\dotsc,i_D]
\end{align*}
and
\begin{align*}
( \, \MDA_{b,a}[i_1,\dotsc,i_D] \, )_{b\in[1,B]_\IN,a\in[1,A_b]_\IN} := \MDA[i_1,\dotsc,i_D]
\end{align*}
i.e., $\MDA_{b,a}[i_1,\dotsc,i_D]$ is the element at point $i_1,\dotsc,i_D$ within MDA $\MDA$ that belongs to the $a$-th access of the $b$-th BUF.
We set $\Buf_b[ \, j_1,\dotsc,j_{D_b} \, ] := \bot$ (symbol $\bot$ denotes the undefined value) for all BUF indices
$(j_1,\dotsc,j_{D_b})\in[0,N^b_1)_{\IN_0}\times\dotsc\times[0,N^b_D)_{\IN_0}\setminus
\bigcup_{a\in[1,A_b]_\IN}\bsize{d}{MDA}{BUF}{^{b,a}}(I_1,\dotsc,I_D)
$
which are not in the function range
of the index functions.

Note that the computed output view $\mathfrak{ov}$ is partial (indicated by $\rightarrow_p$ in Definition~\ref{def_output_view}), because for non-injective index functions, it must hold
$\idx_{b,a}(i_1,\dotsc,i_D)=\idx_{b,a'}(i'_1,\dotsc,i'_D) \Rightarrow \MDA_{b,a}[i_1,\dotsc,i_D]=\MDA_{b,a'}[i'_1,\dotsc,i'_D]$
which may not be satisfied for each potential input MDA of the computed view.
\end{definition}

\begin{notation}[Output Views]
\label{notation_ov}
Analogously to Notation~\ref{notation_iv}, we denote $\ov$ for a particular choice of index functions as:
\begin{align*}
\ov
( \
\texttt{ID}_1:\,\idx_{1,1}\,,\dotsc,\,\idx_{1,A_1}
\ \ \ , \,\dotsc\, , \ \ \
\texttt{ID}_B:\,\idx_{B,1}\,,\dotsc,\,\idx_{B,A_B}
\ )
\end{align*}
\end{notation}

\begin{example}[Notation Output Views]
Function $\ov$ is used for $\texttt{MatVec}$ and $\texttt{Jacobi1D}$ (in Notation~\ref{notation_ov}) as follows:
\[
\texttt{\underline{MatVec:}} \ \ \ \
\ov( \
\underbrace{
\texttt{w:}
  \underbrace{(i,k)\mapsto(i)}_\texttt{a=1}
}_\texttt{b=1}
\ )
\hspace*{30px}
\texttt{\underline{Jacobi1D:}} \ \ \ \
\ov
( \
\underbrace{
\texttt{w:}
  \underbrace{(i)\mapsto (i)}_{a=1}
}_{b=1}
\ )
\]
\end{example}

\subsubsection{Relation between View Functions}
\label{sec_views_iv_and_ov}

We use view functions to transform data from their domain-specific representation (represented in our formalism as BUFs, Definition~\ref{def_buffer}) to our internal, MDA-based representation (via input views) and back (via output views), as also illustrated in Figure~\ref{hl_overview}.
In our implementation presented later, we aim to access data uniformly in the form of MDAs, thereby being independent of domain-specific data representations.
However, we aim to store the data physically in the domain-specific format, as such format is usually the more efficient representation.
For example, we aim to store the input data of \texttt{MatVec} in the domain-specific matrix and vector format, rather than as an MDA, because the input MDA of \texttt{MatVec} contains many redundancies~--~each vector element once per row of the input matrix (as illustrated in Figure~\ref{img_example_inp_view_matvec}).

\newpage

The following lemma proves that functions $\iv$ and $\ov$ are invertible and that they are each others inverses.
Consequently, the lemma shows how we can arbitrarily switch between the domain-specific and our MDA-based representation, and consequently also that
we can implicitly identify MDAs with the domain-specific data representation.
For example, for computing \texttt{MatVec}, we will specify the computations via pattern $\mdh$ which operates on MDAs (see Figure~\ref{hl_overview}), but we use the view functions in our implementation to implicitly forward the MDA accesses to the physically stored BUF representation.

\begin{lemma}
\label{theorem_views}
Let
\begin{align*}
    \iv
    ( \
    \texttt{ID}_1:\,\idx_{1,1}\,,\dotsc,\,\idx_{1,A_1}
    \ \ \ , \,\dotsc\, , \ \ \
    \texttt{ID}_B:\,\idx_{B,1}\,,\dotsc,\,\idx_{B,A_B}
    \ )
\end{align*}
and
\begin{align*}
    \ov
    ( \
    \texttt{ID}_1:\,\idx_{1,1}\,,\dotsc,\,\idx_{1,A_1}
    \ \ \ , \,\dotsc\, , \ \ \
    \texttt{ID}_B:\,\idx_{B,1}\,,\dotsc,\,\idx_{B,A_B}
    \ )
\end{align*}
be two arbitrary instances of functions $\iv$ and $\ov$ (in Notations~\ref{notation_iv} and~\ref{notation_ov}), both using the same index functions $\idx_{1,1},\dotsc,\idx_{B,A_B}$.

It holds (index functions omitted via ellipsis for brevity):
\begin{align*}
    &\iv
    ( \
    \dotsc
    \ )
    \ \circ \
    \ov
    ( \
    \dotsc
    \ )
    \ \ = \ \
    \ov
    ( \
    \dotsc
    \ )
    \ \circ \
    \iv
    ( \
    \dotsc
    \ )
    \ \ = \ \
    id
\end{align*}
\end{lemma}
\begin{proof}
    Follows immediately from Definitions~\ref{def_iv} and~\ref{def_ov}.
\end{proof}

\vspace{5px}

\noindent
The following figure illustrates the lemma using as example the inverse of \texttt{MatVec}'s input view (shown in Figure~\ref{img_example_inp_view_matvec}):

\vspace{5px}
\begin{center}
\includegraphics[scale=0.32]{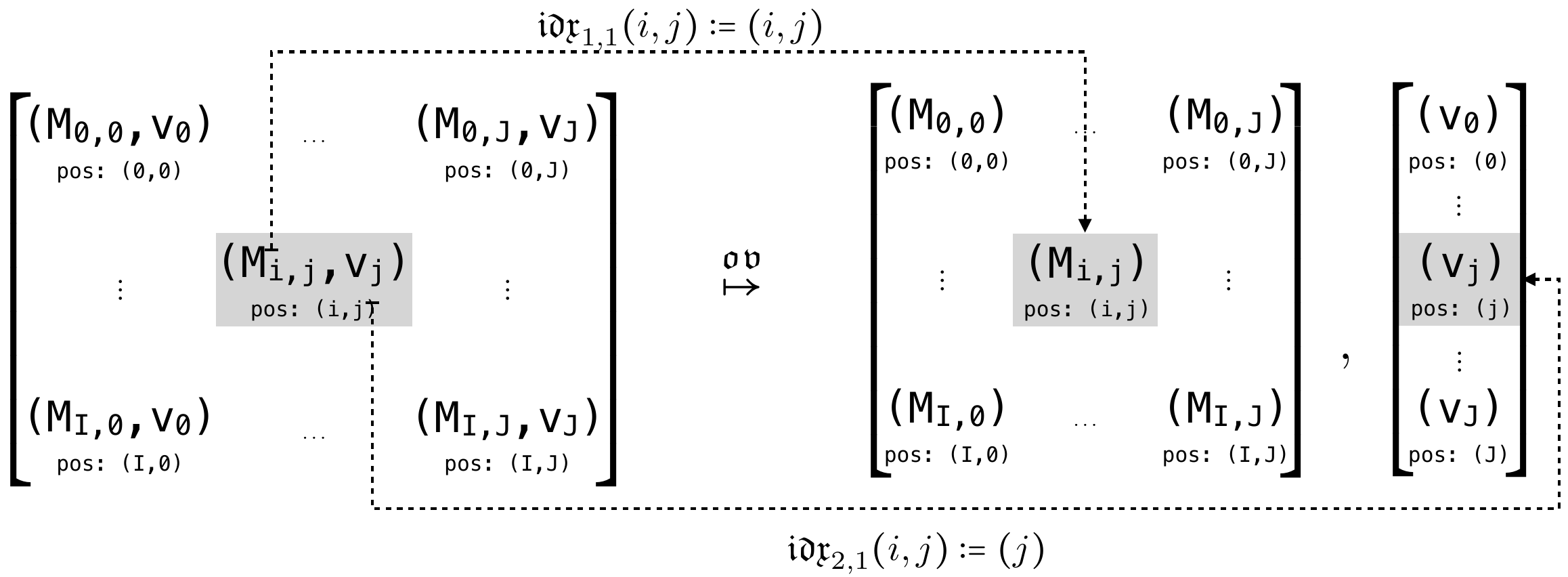}
\end{center}
\vspace{5px}

\subsection{Generic High-Level Expression}

Figure~\ref{fig_generic_hl} shows an expression in our high-level representation~--~consisting of higher-order functions $\iv$, $\mdh$, and $\ov$ (Figure~\ref{hl_overview})~--~that is generic in an arbitrary but fixed choice of index functions, scalar function, and combine operators.
We express data-parallel computations using a particular instance of this generic expression in Figure~\ref{fig_generic_hl}.

Note that meta-parameters of higher-order function $\iv$, $\ov$, and $\mdh$ are omitted in Figure~\ref{fig_generic_hl}, because all parameters can be automatically deduced from the particular numbers and types of their inputs
(index functions in the case of $\iv$ and $\ov$, and scalar function and combine operators
for $\mdh$).

\newpage

The concrete instance of $\mdh(\dotsc)$ (i.e., the MDH function returned by $\mdh$ for the particular input scalar function and combine operators) has as meta-parameter the MDH function's index sets (see Definition~\ref{def_mdh}) for high flexibility.
We use as index sets straightforwardly the input size $(N_1,\dotsc,N_D)\in\IN^D$ (which abbreviates $([0,N_1)_{\INz},\dotsc,[0,N_D)_{\INz})$~--~see Notation~\ref{not_mda})\footnote{
    Our formalism allows dynamic shapes, by using symbol $*$ instead of a particular natural number for $N_i$ (formal details provided in the Appendix, Definition~\ref{app_gen_meta_func}),
    which we aim to discuss thoroughly in future work.
}.
Instances of $\iv(\dotsc)$ and $\ov(\dotsc)$ (i.e., the input and output view
returned by $\iv$ and $\ov$ for concrete index functions) have as meta-parameters the MDA's index sets and the scalar types of BUFs.
We explicitly state only the meta-parameter for the BUFs' scalar types in our generic high-level expression (Figure~\ref{fig_generic_hl}), and we avoid explicitly stating the MDA's index sets for simplicity and to avoid redundancies, because the sets can be taken from the $\mdh$'s meta-parameter list.

Note that for better readability of our high-level expressions, we list meta-parameters before parentheses, i.e., instead of writing
$\iv(\dotsc)^{\texttt{<}\dotsc\texttt{>}}$, $\ov(\dotsc)^{\texttt{<}\dotsc\texttt{>}}$, and $\mdh(\dotsc)^{\texttt{<}\dotsc\texttt{>}}$ for the particular instances of higher-order functions, where meta-parameters are listed at the end, we write
$\iv\texttt{<}\dotsc\texttt{>}(\dotsc)$,
$\ov\texttt{<}\dotsc\texttt{>}(\dotsc)$,
and $\mdh\texttt{<}\dotsc\texttt{>}(\dotsc)$.

\begin{figure}[h!]
\begin{align*}
&
\ov\texttt{<}
T^\texttt{OB}_1,\dotsc,T^\texttt{OB}_{B^\texttt{OB}}\texttt{>}
( \
\texttt{OB}_1\texttt{:}\,
\idx^\texttt{OUT}_{1,1} ,\dotsc, \idx^\texttt{OUT}_{1,A^\texttt{OB}_1}
\ \: ,\dotsc, \ \:
\texttt{OB}_{B^\texttt{OB}}\texttt{:}\,
\idx_{B^\texttt{OB},1}^\texttt{OUT} ,\dotsc, \idx_{{B^\texttt{OB}},A^\texttt{OB}_{B^\texttt{OB}}}^\texttt{OUT}
\ ) \ \ \circ \\[5px]
&\hspace*{10px}
\mdh\texttt{<}N_1,\dotsc,N_D\texttt{>}(\,f,\,(\co{1},\dotsc,\co{D})\,) \ \ \circ \\[5px]
&\hspace*{20px}
\iv\texttt{<}
T^\texttt{IB}_1,\dotsc,T^\texttt{IB}_{B^\texttt{IB}}\texttt{>}
( \
\texttt{IB}_1\texttt{:}\,
\idx^\texttt{INP}_{1,1} ,\dotsc, \idx^\texttt{INP}_{1,A^\texttt{IB}_1}
\ \: ,\dotsc, \ \:
\texttt{IB}_{B^\texttt{IB}}\texttt{:}\,
\idx_{B^\texttt{IB},1}^\texttt{INP} ,\dotsc, \idx_{{B^\texttt{IB}},A^\texttt{IB}_{B^\texttt{IB}}}^\texttt{INP}
\ )
\end{align*}
\caption{
Generic high-level expression for data-parallel computations
}
\label{fig_generic_hl}
\end{figure}

\subsection{Examples}
\label{ch_hl_examples}

\begin{figure}[p]
\centering
\includegraphics[width=0.95\textwidth]{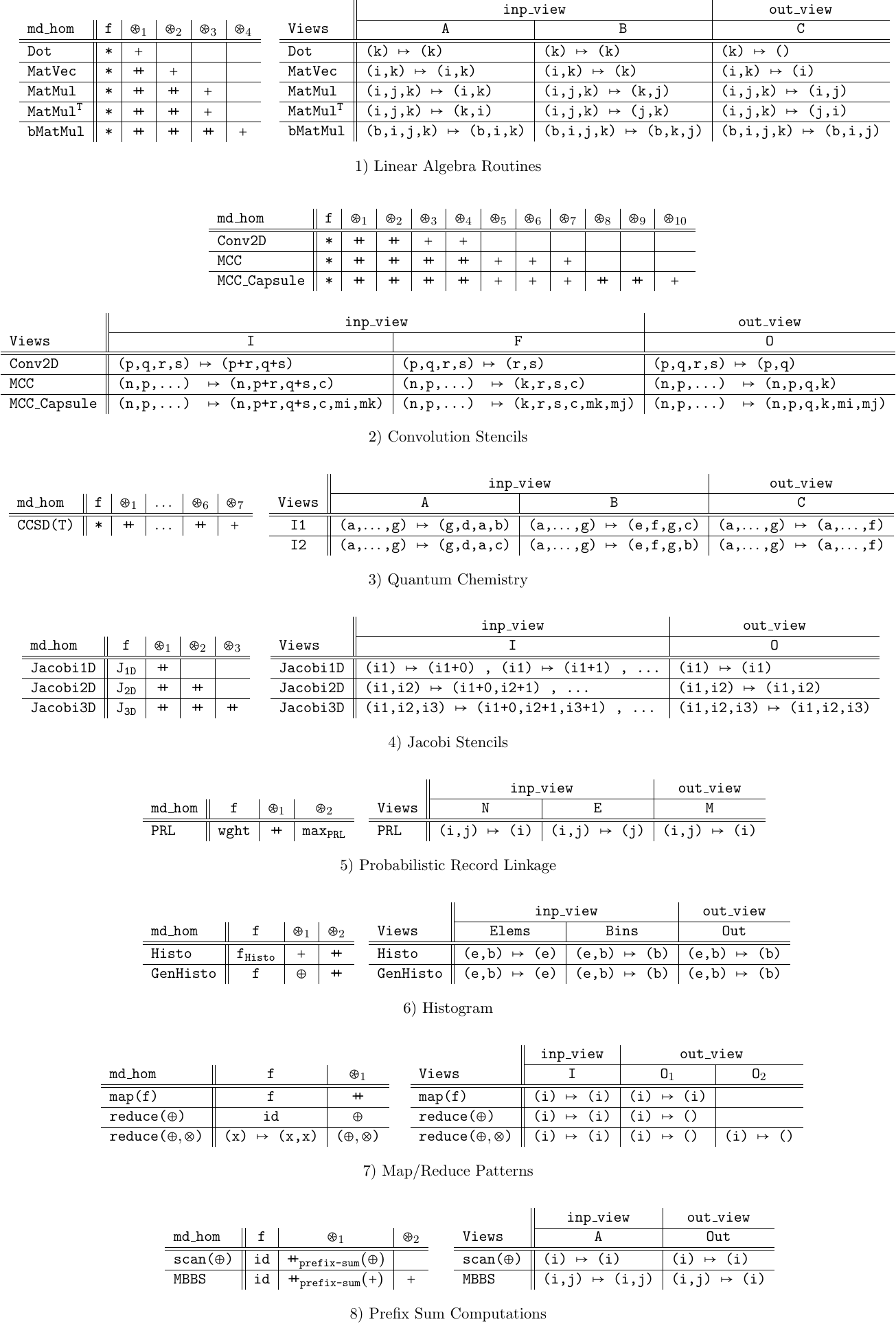}
\caption{
Data-parallel computations expressed in our high-level representation
}
\label{fig_hl_examples}
\end{figure}

Figure~\ref{fig_hl_examples} shows how our high-level representation is used for expressing different kinds of popular data-parallel computations.
For brevity, we state only the index functions, scalar function, and combine operators of the higher-order functions%
;~the expression in Figure~\ref{fig_generic_hl} is then obtained by straightforwardly inserting these building blocks into the higher-order functions.

\paragraph*{Subfigure~1}
We show how our high-level representation is used for expressing linear algebra routines:
1)~\texttt{Dot} (\emph{Dot Product});
2)~\texttt{MatVec} (\emph{Matrix-Vector Multiplication});
3)~\texttt{MatMul} (\emph{Matrix Multiplication});
4)~$\texttt{MatMul}^\texttt{T}$ (\emph{Transposed Matrix Multiplication}) which computes matrix multiplication on transposed input and output matrices;
5)~\texttt{bMatMul} (\emph{batched Matrix Multiplication}) where multiple matrix multiplications are computed using matrices of the same sizes.

We can observe from the subfigure that our high-level expressions for the routines naturally evolve from each other.
For example, the $\mdh$ instance for \texttt{MatVec} differs from the $\mdh$ instance for \texttt{Dot} by only containing a further concatenation dimension $\dplus$ for its $i$ dimension.
We consider this close relation between the high-level expressions of \texttt{MatVec} and \texttt{Dot} in our approach as natural and favorable, as \texttt{MatVec} can be considered as computing multiple times \texttt{Dot}~--~one computation of \texttt{Dot} for each value of \texttt{MatVec's} $i$ dimension.
Similarly, the $\mdh$ instance for \texttt{MatMul} is very similar to the expression of \texttt{MatVec}, by containing the further concatenation dimension $j$ for \texttt{MatMul}'s $j$ dimension.
The same applies to \texttt{bMatMul}:~its $\mdh$ instance is the expression of \texttt{MatMul} augmented with one further concatenation dimension.

Regarding $\texttt{MatMul}^\texttt{T}$, the basic computation part of $\texttt{MatMul}^\texttt{T}$ and \texttt{MatMul} are the same, which is exactly reflected in our formalisms:~both $\texttt{MatMul}^\texttt{T}$ and $\texttt{MatMul}$ are expressed using exactly the same $\mdh$ instances.
The differences between $\texttt{MatMul}^\texttt{T}$ and \texttt{MatMul} lies only in the data accesses~--~transposed accesses in the case of $\texttt{MatMul}^\texttt{T}$ and non-transposed accesses in the case of \texttt{MatMul}.
Data accesses are expressed in our formalism, in a structured way, via view functions (as discussed in Section~\ref{sec_views}):~for example, for $\texttt{MatMul}^\texttt{T}$, we use for its first input matrix $A$ the index function $(i,j,k)\mapsto(k,i)$ for transposed access, instead of using index function $(i,j,k)\mapsto(i,k)$ as for \texttt{MatMul}'s non-transposed accesses.

Note that all $\mdh$ instances in the subfigure are well defined according to Lemma~\ref{lemma_mdhom}.

\paragraph*{Subfigure~2}
We show how convolution-style stencil computations are expressed in our high-level representation:~%
1)~\texttt{Conv2D} expresses a standard convolution that uses a $2$D sliding window~\cite{podlozhnyuk2007image};~%
2)~\texttt{MCC} expresses a so-called \emph{Multi-Channel Convolution}~\cite{dumoulin2018guide}~--~a generalization of \texttt{Conv2D} that is heavily used in the area of deep learning;~%
3)~\texttt{MCC\_Capsule} is a recent generalization of \texttt{MCC}~\cite{e2018matrix} which attracted high attention due to its relevance for advanced deep learning neural networks~\cite{10.1145/3317550.3321441}.

While our $\mdh$ instances for convolutions are quite similar to those of linear algebra routines (they all use multiplication $*$ as scalar function, and a mix of concatenations $\dplus$ and point-wise additions $+$ as combine operators), the index functions used for the view functions of convolutions are notably different from those used for linear algebra routines: the index functions of convolutions contain arithmetic expressions (e.g., \texttt{p+r} and \texttt{q+s}) and thus access neighboring elements in their input~--~%
a typical access pattern in stencil computations that requires special optimizations~\cite{10.1145/3168824}.
Moreover, convolution-style computations are often high-dimensional (e.g., $10$ dimensions in the case of \texttt{MCC\_Capsule}), whereas linear algebra routines usually rely on less dimensions.
Our experiments in Section~\ref{ch:eval} confirm that respecting the data access patterns and the high dimensionality of convolutions in the optimization process (as in our approach, which we discuss later) often achieves significantly higher performance than using optimizations chosen toward linear algebra routines, as in vendor libraries provided by NVIDIA and Intel for convolutions~\cite{7573804}.

\paragraph*{Subfigure~3}
We show how quantum chemistry computation \emph{Coupled Cluster} (\texttt{CCSD(T)})~\cite{8661182} is expressed in our high-level representation.
The computation of \texttt{CCSD(T)} notably differs from those of linear algebra routines and convolution-style stencils, by accessing its high-dimensional input data in sophisticated transposed fashions:~for example, the view function of \texttt{CCSD(T)}'s \emph{instance one} (denoted as \texttt{I1} in the subfigure)
uses indices \texttt{a} and \texttt{b} to access the last two dimensions of its $A$ input tensor (rather than the first two dimensions of the tensor, as would be the case for non-transposed accesses).

For brevity, the subfigure presents only two \texttt{CCSD(T)} instances~--~in our experiments in Section~\ref{ch:eval}, we present experimental results for nine different real-world \texttt{CCSD(T)} instances.

\paragraph*{Subfigures~4-6}
The subfigures present computations whose scalar functions and combine operators are different from those used in Subfigures~1-3 (which are in Subfigures~1-3 straightforward multiplications $*$, concatenation $\dplus$, and point-wise additions $+$ only).
For example, Jacobi stencils (Subfigure~4) use as scalar function the Jacobi-specific computation $\texttt{J}_\texttt{nD}$~\cite{10.1007/978-3-642-28151-8_17},
and \emph{Probabilistic Record Linkage~(PRL)}~\cite{10.5555/2344108}, which is heavily used in data mining to identify duplicate entries in a data base, uses a PRL-specific both scalar function \texttt{wght} and combine operator $\texttt{max}_\texttt{PRL}$ (point-wise combination via the PRL-specific binary operator $\texttt{max}_\texttt{PRL}$)~\cite{10.1145/3297280.3297330}.
Histograms, in their generalized version~\cite{9355244} (denoted as \texttt{GenHisto} in Subfigure~6), use an arbitrary, user-defined scalar function $f$ and a user-defined associative and commutative combine operator $\oplus$;~the standard histogram variant \texttt{Histo} is then a particular instance of \texttt{GenHist}, for $\oplus=+$ (point-wise addition) and $f=f_\texttt{Histo}$, where $f_\texttt{Histo}(e,b)=1$ iff $e=b$ and $f_\texttt{Histo}(e,b)=0$ otherwise.
Histogram's are often analyzed regarding their runtime complexity~\cite{9355244}; we provide such a discussion for our MDH-based Histogram implementation in the Appendix, Section~\ref{histo_complexity}, for the interested reader.

\paragraph*{Subfigure~7}
We show how typical \texttt{map} and \texttt{reduce} patterns~\cite{https://doi.org/10.1002/spe.1026}
are implemented in our high-level representation.
Examples \texttt{map(f)} and \texttt{reduce($\oplus$)} (discussed in Examples~\ref{mdh_map} and~\ref{mdh_reduce}) are simple and thus straightforwardly expressed in our representation.
In contrast, example \texttt{reduce($\oplus,\otimes$)} is more complex and shows how \texttt{reduce($\oplus$)} is extended to combine the input vector simultaneously twice~--~once combining vector elements via operator $\oplus$ and once using operator $\otimes$.
The outcome of \texttt{reduce($\oplus,\otimes$)} are two scalars~--~one representing the result of combination via $\oplus$ and the other of combination via $\otimes$~--~which we map via the output view to output elements $\texttt{O}_\texttt{1}$ (result of~$\oplus$) and $\texttt{O}_\texttt{2}$ (result of~$\otimes$), correspondingly;~this is also illustrated in Figure~\ref{fig_ov_examples_double_reduce}.

\paragraph*{Subfigure~8} We present \emph{prefix-sum computations}~\cite{BlellochTR90} which differ from the computations in Subfigures~1-7 in terms of their combine operators:~the operator used for expressing computations in Subfigure~8 is different from concatenation (Example~\ref{def:mda_concat}) and point-wise combinations (Example~\ref{def:mda_pw}).
Computation \texttt{scan($\oplus$)} uses as combine operator
$\dplus_\texttt{prefix-sum}(\oplus)$ which computes prefix-sum~\cite{715958} (formally defined in the Appendix, Section~\ref{app_co_scan}) according to binary operator $\oplus$, and \texttt{MBBS} (Maximum Bottom Box Sum)~\cite{10.1145/3314221.3314612} uses a
particular instance of prefix-sum
for $\oplus=+$ (addition).

\section{Low-Level Representation for Data-Parallel Computations}
\label{ch:low_level}

We introduce our low-level representation for expressing data-parallel computations.
In contrast to our high-level representation, our low-level representation explicitly expresses the de-composition and re-composition of computations (informally illustrated in Figure~\ref{fig_decomp}).
Moreover, our low-level representation is designed such that it can be straightforwardly transformed to executable program code, because it explicitly captures and expresses the optimizations for the memory and core hierarchy of the target architecture.

In the following, after briefly discussing an introductory example in Section~\ref{ch:ll_intro_example}, we introduce in Section~\ref{sec_amm} our formal representation of
computer systems, to which we refer to as \emph{Abstract System Model~(ASM)}.
Based on this model, we define \emph{low-level MDAs}, \emph{low-level BUFs}, and \emph{low-level COs} in Section~\ref{sec_bbb}, which are basic building blocks of our low-level representation.

\vspace*{3px}

Note that all details and concepts discussed in this section are not exposed to the end users of our system and therefore transparent for them:~expressions in our low-level representation are generated fully automatically for the user, from expressions in our high-level representation (Figure~\ref{contributions_overall}), according to the methodologies presented later in Section~\ref{ch:lowering} and auto-tuning~\cite{10.1145/3427093}.

\subsection{Introductory Example}
\label{ch:ll_intro_example}

\begin{figure}[hbtp]
\centering
\includegraphics[width=0.95\textwidth]{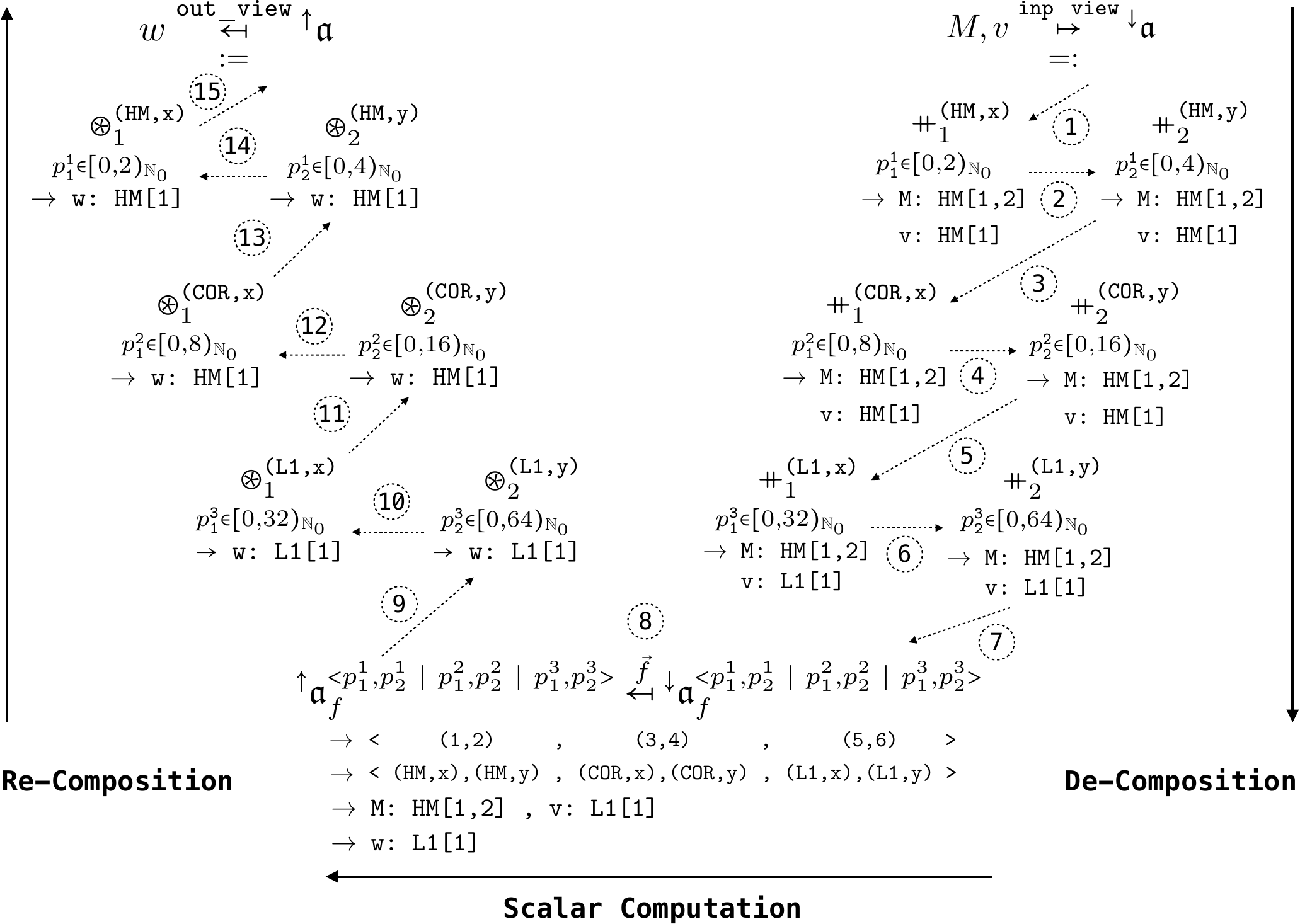}
\caption{
Low-level expression for straightforwardly computing Matrix-Vector Multiplication (\texttt{MatVec}) on a simple, artificial architecture with two memory layers (\texttt{HM} and \texttt{L1}) and one core layer (\texttt{COR}).
Dotted lines
indicate data flow.
}
\label{fig_ll_example}
\end{figure}

\begin{figure}[hbtp]
    \centering
    \includegraphics[width=\textwidth]{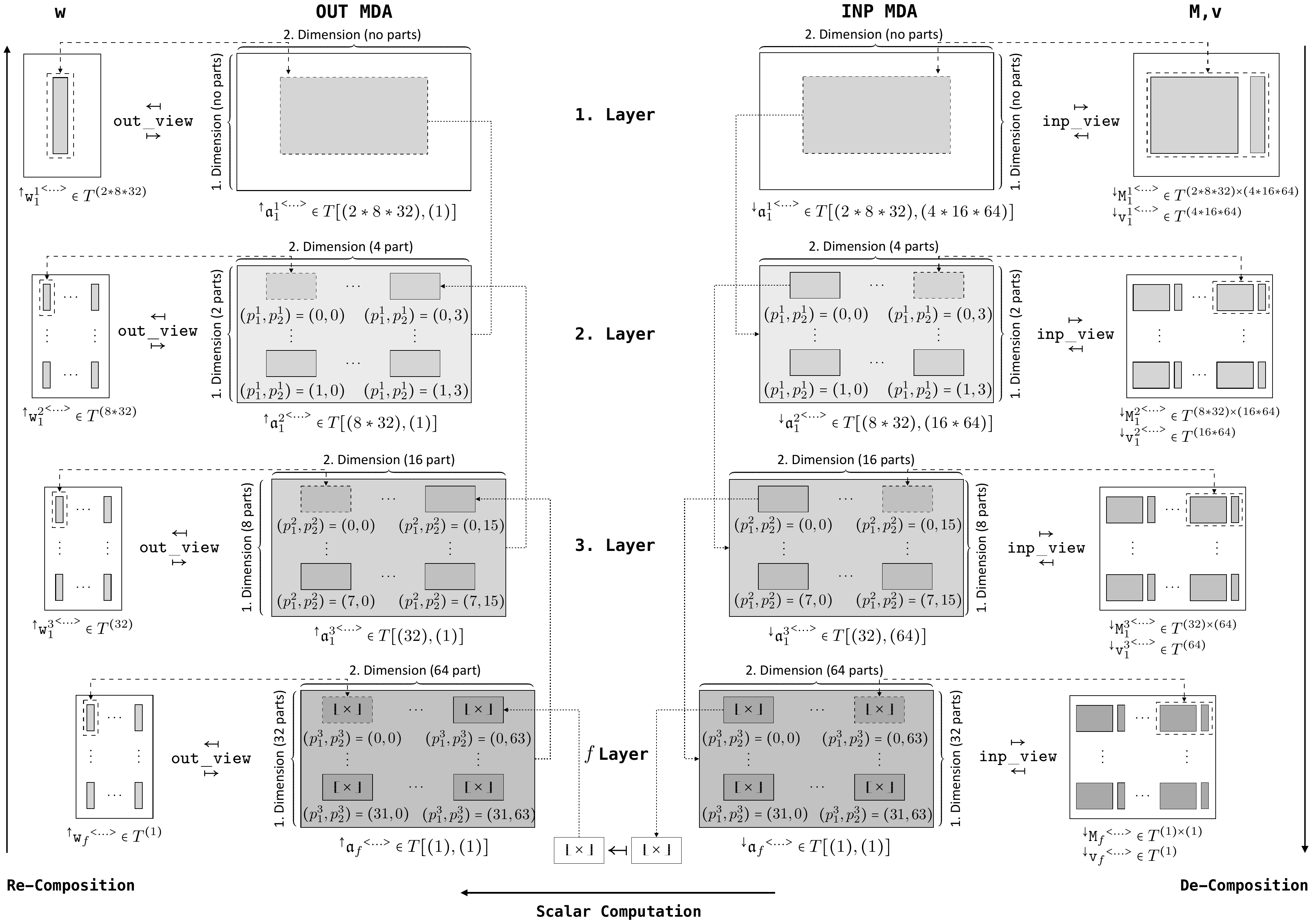}
    \caption{
    Illustration of multi-layered, multi-dimensional MDA partitioning using the example MDA from Figure~\ref{fig_ll_example}. In this example, we use three layers
    and two dimensions, according to Figure~\ref{fig_ll_example}.
    }
    \label{mda_part_example_informal}
\end{figure}

Figure~\ref{fig_ll_example} illustrates our low-level representation by showing how
\texttt{MatVec} (Matrix-Vector Multiplication) is expressed in our representation.
In our example, we use an input matrix
$M\in T^{512\times4096}$
of size $512\times4096$ (size chosen arbitrarily)
that has an arbitrary but fixed scalar type $T\in\type$;~the input vector
$v\in T^{4096}$ is of size $4096$, correspondingly.

For better illustration, we consider for this introductory example a straightforward, artificial target architecture that has only two memory layers~--~\emph{Host Memory~(\texttt{HM})} and \emph{Cache Memory~(\texttt{L1})}~--~and one \emph{Core Layer}~(\texttt{COR}) only;~our examples presented and discussed later in this section target real-world architectures (e.g., CUDA-capable NVIDIA GPUs).
The particular values of tuning parameters (discussed in detail later in this section), such as the number of threads and the order of combine operators, are chosen by hand for this example and as straightforward
for simplicity.

\vspace*{5px}

Our low-level representation works in three phases:~%
1)~\emph{de-composition} (steps 1-7, in the right part of Figure~\ref{fig_ll_example}),~%
2)~\emph{scalar} (step 8, bottom part of the figure),~%
3)~\emph{re-composition} (steps 9-15, left part).
Steps are arranged from right to left,
inspired by the application order of function composition.

\paragraph{1. De-Composition Phase:}
The de-composition phase (steps $1$-$7$ in Figure~\ref{fig_ll_example}) partitions input MDA ${^\downarrow}{\MDA}$ (in the top right of Figure~\ref{fig_ll_example}) to the structure
${^\downarrow}\MDA_f^{<\,\dotsc\,>}$
(bottom right) to which we refer to as \emph{low-level MDA} and define formally in the next subsection.
The low-level MDA
represents a partitioning of MDA ${^\downarrow}{\MDA}$ (a.k.a \emph{hierarchical, multi-dimensional tiling} in programming), where each particular choice of
indices
$p^1_1\in[0,2)_{\IN_0}$, $p^1_2\in[0,4)_{\IN_0}$, $p^2_1\in[0,8)_{\IN_0}$, $p^2_2\in[0,16)_{\IN_0}$, $p^3_1\in[0,32)_{\IN_0}$, $p^3_2\in[0,64)_{\IN_0}$
refers to an MDA that represents an individual part of MDA ${^\downarrow}{\MDA}$ (a.k.a. \emph{tile} in programming~--~informally illustrated in Figure~\ref{mda_examples}).
The partitions are arranged on multiple layers (indicated by the $p$'s superscripts) and in multiple dimensions (indicated by subscripts)~--~as illustrated in Figure~\ref{mda_part_example_informal}~--~according to the memory/core layers of the target architecture and dimensions of the MDH computation:~we partition for each of the target architecture's three layers (\texttt{HM}, \texttt{L1}, \texttt{COR}) and in each of the two dimensions of the MDH (dimensions $1$ and~$2$, as we use example \texttt{MatVec} in Figure~\ref{fig_ll_example}, which represents a two-dimensional MDH computation).
Consequently, our partitioning approach allows efficiently exploiting each particular layer of the target architecture (both memory and core layers), and also optimizing for both dimensions of the target computation (in the case of \texttt{MatVec}, the $i$-dimension and also the $k$-dimension~--~see Figure~\ref{example_intro_1}), allowing fine-grained optimizations.

We compute the partitionings of MDAs by applying the concatenation operator (Example~\ref{def:mda_concat}) inversely\footnote{It is easy to see that operator \emph{concatenation} (Example~\ref{def:mda_concat}) is invertible for any particular choice of meta-parameters~(formally proved in the Appendix, Section~\ref{app_inverse_concat}).}~(indicated by using $=:$ instead of $:=$ in the top right part of Figure~\ref{fig_ll_example}).
For example, we partition in Figure~\ref{fig_ll_example} MDA ${^\downarrow}{\MDA}$ first via the inverse of $\dplus^\texttt{(HM,x)}_1$ in dimension $1$ (indicated by the subscript~$1$ of $\dplus^\texttt{(HM,x)}_1$;~the superscript \texttt{(HM,x)} is explained later)
into $2$ parts, as $p^1_1$ iterates over interval $[0,2)_{\INz}=\{0,1\}$ which consists of two elements ($0$ and $1$)~--~the interval is chosen arbitrarily for this example.
Afterward, each of the obtained parts is further partitioned, in the second dimension, via $\dplus^\texttt{(HM,y)}_2$ into $4$ parts ($p^1_2$ iterates over $[0,4)_{\INz}=\{0,1,2,3\}$ which consists of four elements).
The ${(2*4)}$-many \texttt{HM} parts are then each further partitioned in both dimensions for the \texttt{COR} layer into $(8*16)$ parts, and each individual \texttt{COR} part is again partitioned for the \texttt{L1} layer into $(32*64)$ parts, resulting in $(2*8*32)*(4*16*64)=512*4096$ parts in total.

We always use a \emph{full partitioning} in our low-level expressions\footnote{
Our future work (outlined in Section~\ref{ch:fw}) aims to additionally allow coarser-grained partitioning schemas, e.g., to target domain-specific hardware extensions (such as \emph{NVIDIA Tensor Cores}~\cite{cuda_tensor_cores} which compute $4\times 4$ matrices immediately in hardware, rather than $1\times1$ matrices as obtained in the case of a full partitioning).
}, i.e., each
particular choice of indices
$p^1_1$, $p^1_2$, $p^2_1$, $p^2_2$, $p^3_1$, $p^3_2$
points to an MDA that contains a single element only
(in Figure~\ref{mda_part_example_informal}, the individual elements are denoted via symbol $\times$, in the bottom part of the figure).
By relying on a full partitioning, we can apply scalar function $f$ to the fully partitioned MDAs later in the scalar phase (described in the next paragraph).
This is because function $f$ is defined on scalar values (Definition~\ref{def_md_hom}) to make defining scalar functions more convenient for the user (as discussed in Section~\ref{sec_mdhom}).

The superscript of combine operators, e.g., \texttt{(COR,x)} of operator $\dplus^\texttt{(COR,x)}_1$, is a so-called \emph{operator tag}
(formal definition given in the next subsection).
Such a tag indicates to our code generator whether its combine operator is assigned to a memory layer (and thus computed sequentially in our generated code) or to a core layer (and thus computed in parallel).
For example,
tag \texttt{(COR,x)} indicates that parts processed by operator $\dplus^\texttt{(COR,x)}_1$ should be computed by cores \texttt{COR}, and thus in parallel;~the dimension tag \texttt{x} indicates that \texttt{COR} layer's \texttt{x} dimension should be used for computing the operator (we use dimension \texttt{x} for our example architecture as an analogous concept to CUDA's thread/block dimensions \texttt{x},\texttt{y},\texttt{z} for GPU architectures~\cite{cuda-specification}), as we also discuss in the next subsection.
In contrast, tag $\texttt{(HM,x)}$ refers to a memory layer (host memory \texttt{HM}) and thus, operator $\dplus^\texttt{(HM,x)}_1$ is computed sequentially.
Since the current state-of-practice programming approaches (OpenMP, CUDA, OpenCL, $\dotsc$) have no explicit notion of memory tiles (e.g., by offering the potential variables \texttt{tileIdx.x}/\texttt{tileIdx.y}/\texttt{tileIdx.z}, as analogous concepts to CUDA variables \texttt{threadIdx.x}/\texttt{threadIdx.y}/\texttt{threadIdx.z}), the dimensions tag \texttt{x} in \texttt{(HM,x)} is currently ignored by our code generator,
because \texttt{HM} refers to a memory layer.

Note that the number of parts (e.g., $2$~parts on layer~$1$ in dimension~$1$, and~$4$~parts on layer~$1$ in dimension~$2$), the combine operators' tags, and our partition order~(e.g. first partitioning in MDA's dimension $1$ and afterward in dimension $2$) are chosen arbitrarily for this example.
These choices are critical for performance and should be optimized\footnote{
    We currently rely on auto-tuning~\cite{10.1145/3427093} for choosing optimized values of performance-critical parameters, as we discuss in Section~\ref{ch:eval}.
} for a particular target architecture and characteristics of the input and output data (size, memory layouts, etc.), as we discuss in detail later in this section.

\paragraph{2. Scalar Phase:}
In the scalar phase (step $8$ in Figure~\ref{fig_ll_example}), we apply MDH's scalar function $f$ to the individual MDA elements
\[
{^\downarrow}\MDA_f^{
<\, p^1_1,p^1_2 \ | \ p^2_1,p^2_2 \ | \ p^3_1,p^3_2 \,>
}
\]
for each particular choice of indices $p^1_1$, $p^1_2$, $p^2_1$, $p^2_2$, $p^3_1$, $p^3_2$,
which results in
\[
{^\uparrow}\MDA_f^{
<\, p^1_1,p^1_2 \ | \ p^2_1,p^3_2 \ | \ p^3_1,p^3_2 \,>
}
\]
In the figure, $\vec{f}$ (introduced in Definition~\ref{def_md_hom}) is the slight adaption of function $f$ that operates on a singleton MDA, rather than a scalar.

Annotation \texttt{$\rightarrow$  < (1,2) , $\dotsc$ >} indicates the application order of applying scalar function (in this example, first iterating over $p^1_1$, then over $p^1_2$, etc), and we use annotation \mbox{\texttt{$\rightarrow$  < (HM,x) , $\dotsc$ >}} to indicate how the scalar computation is assigned to the target architecture (this is described in detail later in this section).
Annotations \texttt{$\rightarrow$ M:\,HM , v:\,L1} and \texttt{$\rightarrow$ w:\,L1} \ (in the bottom part of Figure~\ref{fig_ll_example}) indicate the memory regions to be used for reading and writing the input scalar of function $f$ (also described later in detail).

\paragraph{3. Re-Composition Phase:}
Finally, the re-composition phase (steps $9$-$15$ in Figure~\ref{fig_ll_example}) combines the computed parts ${^\uparrow}\MDA_f^{<\, p^1_1,p^1_2 \ | \ p^2_1,p^2_2 \ | \ p^3_1,p^3_2 \,>}$ (bottom left in the figure) to the final result ${^\uparrow}{\MDA}$~(top left) via MDH's combine operators, which are in the case of matrix-vector multiplication $\co{1}:=\dplus$~(concatenation) and $\co{2}:=+$~(point-wise addition).
In this example, we first combine the \texttt{L1} parts in dimension $2$ and then in dimension $1$;~afterward, we combine the \texttt{COR} parts in both dimensions, and finally the \texttt{HM} parts.
Analogously to before, this order of combine operators and their tags are chosen arbitrarily for this example and should be auto-tuned for high performance.

\vspace*{5px}

In the de- and re-composition phases, the arrow notation below combine operators allow efficiently exploiting architecture's memory hierarchy, by indicating
the memory region to read from (de-composition phase) or to write to (re-composition phase);~the annotations also indicate the memory layouts to use.
We exploit these memory and layout information in both:
i) our code generation process to assign combine operators' input and output data to memory regions and to chose memory layouts for the data (row major, column major, etc);~%
ii) our formalism to specify constraints of programming models, e.g., that in CUDA, results of GPU cores can only be combined in designated memory regions~\cite{cuda-programming-guide}.
For example, annotation \texttt{\small$\rightarrow$ M:\:HM[1,2]\:,\:v:\:L1[1]} \ below an operator in the de-composition phase indicates to our code generator that the parts (a.k.a tiles) of matrix $M$
used
for this computation step should be read from host memory \texttt{HM} and that parts of vector $v$ should be copied to and accessed from fast \texttt{L1} memory.
The annotation also indicates that \texttt{M} should be stored using a row-major memory layout (as we use \texttt{[1,2]} and not \texttt{[2,1]}).
The memory regions and layouts are chosen arbitrarily for this example and should be chosen as optimized (auto-tuned) for the particular target architecture and characteristics of the input and output data.
Formally, the arrow notation of combination operators is a concise notation to hide MDAs and BUFs for intermediate results (discussed
in the Appendix, Section~\ref{app_ll_example}, for the interested reader).

\subsection*{Excursion: Code Generation\protect\footnote{Our implementation of MDH is open source: \url{https://mdh-lang.org}}}
Our low-level expressions can be straightforwardly transformed to executable program code in imperative-style programming languages (such as OpenMP, CUDA, and OpenCL).
As code generation is not the focus of this work, we outline our code generation process briefly using the example of Figure~\ref{fig_ll_example}.
Details about our code generation
are provided in Section~\ref{app_sec_code_generation} of our Appendix, and will be presented and illustrated in detail in our future work.

We implement combine operators as sequential or parallel loops.
For example, the operator $\dplus^\texttt{(HM,x)}_1$
is assigned to memory layer \texttt{HM} and thus implemented as a sequential loop
(loop range indicated by $[0,2)_{\INz}$), and operator $\dplus^\texttt{(COR,x)}_1$
is assigned to core layer \texttt{COR} and thus implemented as a parallel loop (e.g., a loop
annotated with \texttt{\small\#pragma omp parallel for} in OpenMP~\cite{openmp-specification}, or
variable \texttt{\small threadIdx.x}
in CUDA~\cite{cuda-specification}%
).
Correspondingly, our three phases
(de-composition, scalar, and re-composition)
each correspond to an individual loop nest;~we generate the nests as fused when the tags of combine operators have the same order in phases, as in Figure~\ref{fig_ll_example}.
Note that our currently targeted programming models (OpenMP, CUDA, and OpenCL) have no explicit notion of \emph{tiles}, e.g., by offering the potential variable \texttt{tileIdx.x} for managing tiles automatically in the programming model (similarly as variable \texttt{threadIdx.x} automatically manages threads in CUDA).
Consequently, when the operator tag refers to a memory layer, the dimension information within tags are currently ignored by our code generator (such as dimension~\texttt{x} in tag~\texttt{(HM,x)} which refers to memory layer \texttt{HM}).

Operators' memory regions correspond to straightforward allocations (e.g., in CUDA's \texttt{device}, \texttt{shared}, or \texttt{register} memory~\cite{cuda-specification}, according to the arrow annotations in our low-level expression).
Memory layouts are implemented as aliases, e.g., \emph{preprocessor directives} such as \texttt{\small\#define M(i,k) M[k][i]} for storing \texttt{MatVec}'s input matrix $M$ as transposed.

We implement MDAs also as aliases (according to Definition~\ref{def_iv}), e.g., \texttt{\small\#define inp\_mda(i,k) M[i][k],v[k]} for \texttt{MatVec}'s input MDA.

Code optimizations that
are
applied on
a lower abstraction level than proposed by our representation in Example~\ref{fig_ll_example} are beyond the scope of this work
and
outlined
in Section~\ref{ch_ll_opts} of our Appendix
e.g., loop fusion and loop unrolling which are applied on the loop-based abstraction level.

We provide an open source \emph{MDH compiler} for code generation~\cite{mdh_in_python}.
Our compiler takes as input a high-level MDH expression (as in Figure~\ref{fig:intro_example}), in the form of a Python program (see Appendix, Section~\ref{matvec_in_dsl}), and it fully automatically generates auto-tuned program code from it.

\vspace{10px}

In the following, we introduce in Section~\ref{sec_amm} our formal representation of a computer system~(which can be a single device, but also a multi-device or a multi-node system, as we discuss soon), and we illustrate our formal system representation using the example architectures targeted by programming models OpenMP, CUDA, and OpenCL.
Afterward, in Section~\ref{sec_bbb}, we formally define the basic building blocks of our low-level representation~--~\emph{low-level MDAs}, \emph{low-level BUFs}, and \emph{low-level COs}~--~based on our formal system representation.

\subsection{Abstract System Model (ASM)}
\label{sec_amm}

\begin{definition}[Abstract System Model]
\label{def_asm}
An \emph{$L$-Layered Abstract System Model (ASM)}, $L\in\IN$, is any pair of two positive natural numbers
\[
(\, \texttt{NUM\_MEM\_LYRs} \, , \, \texttt{NUM\_COR\_LYRs} \,)\in\IN\times\IN
\]
for which $\texttt{NUM\_MEM\_LYRs}+\texttt{NUM\_COR\_LYRs} \ = \ L$.
\end{definition}
Our ASM representation is capable of modeling architectures with arbitrarily deep memory and core hierarchies\footnote{
    We deliberately do not model into our ASM representation an architecture's particular number of cores and/or sizes of memory regions, because our
    optimization process
    is designed to be generic in these numbers and sizes, for high flexibility.
}:~\texttt{NUM\_MEM\_LYRs}
denotes the target architecture's number of memory layers and \texttt{NUM\_COR\_LYRs} the architecture's number of core layers, correspondingly.
For example, the artificial architecture we use in Figure~\ref{fig_ll_example} is represented as an ASM instance as follows (bar symbols denote set cardinality):
\[
    \texttt{ASM}_\texttt{artif.}:=( \ \ |\{\texttt{HM},\texttt{L1}\}| \ \ , \ |\{\texttt{COR}\}| \ \ ) = (2,1)
\]
The instance is a pair consisting of the numbers $2$ and $1$ which represent the artificial architecture's two memory layers (\texttt{HM} and \texttt{L1}) and its single core layer (\texttt{COR}).

\begin{example}[ASM Instances]
\label{example_asm}
We show particular ASM instances that represent the device models of the state-of-practice approaches OpenMP, CUDA, and OpenCL:
\begin{align*}
&
\texttt{ASM}_\texttt{OpenMP} &:=& &(& \ &|\{\texttt{MM},\texttt{L2},\texttt{L1}\}|& \, &,& \, &|\{\texttt{COR}\}|& \ &)& &=& &(3,1)\\
&
\texttt{ASM}_\texttt{OpenMP+L3} &:=& &(& \ &|\{\texttt{MM},\texttt{L3},\texttt{L2},\texttt{L1}\}|& \, &,& \, &|\{\texttt{COR}\}|& \ &)& &=& &(4,1)\\
&
\texttt{ASM}_\texttt{OpenMP+L3+SIMD} &:=& &(& \ &|\{\texttt{MM},\texttt{L3},\texttt{L2},\texttt{L1}\}|& \, &,& \, &|\{\texttt{COR},\texttt{SIMD}\}|& \ &)& &=& &(4,2)\\[5pt]
&
\texttt{ASM}_\texttt{CUDA} &:=& &(& \ &|\{\texttt{DM},\texttt{SM},\texttt{RM}\}|& \, &,& \, &|\{\texttt{SMX},\texttt{CC}\}|& \ &)& &=& &(3,2)\\
&
\texttt{ASM}_\texttt{CUDA+WRP} &:=& &(& \ &|\{\texttt{DM},\texttt{SM},\texttt{RM}\}|& \, &,& \, &|\{\texttt{SMX},\texttt{WRP},\texttt{CC}\}|& \ &)& &=& &(3,3)\\[5pt]
&
\texttt{ASM}_\texttt{OpenCL} &:=& &(& \ &|\{\texttt{GM},\texttt{LM},\texttt{PM}\}|& \, &,& \, &|\{\texttt{CU},\texttt{PE}\}|& \ &)& &=& &(3,2)\\
\end{align*}
\end{example}

OpenMP is often used to target $(3+1)$-layered architectures which rely on $3$~memory regions~(main memory \texttt{MM}, and caches \texttt{L2} and \texttt{L1}) and $1$~core layer~(\texttt{COR}).
OpenMP-compatible architectures sometimes also contain the \texttt{L3} memory region, and they may allow exploiting SIMD parallelization (a.k.a. \emph{vectorization}~\cite{10.1007/978-3-642-30961-8_5}), which are expressed in our ASM representation as a further memory or core layer, respectively.

CUDA's target architectures are $(3+2)$-layered:~they consist of \emph{Device Memory~(\texttt{DM})}, \emph{Shared Memory~(\texttt{SM})}, and \emph{Register Memory~(\texttt{RM})}, and they offer as cores so-called \emph{Streaming Multiprocessors~(\texttt{SMX})} which themselves consist of
\emph{Cuda Cores~(\texttt{CC})}.
CUDA also has an implicit notion of so-called \emph{Warps~(\texttt{WRP})} which are not explicitly represented in the CUDA programming model~\cite{cuda-specification}, but often exploited by programmers~--~via special intrinsics (e.g., \emph{shuffle} and \emph{tensor core intrinsics}~\cite{cuda_shuffles,cuda_tensor_cores})~--~to achieve highest performance.

OpenCL-compatible architectures are designed analogously to those targeted by the CUDA programming model;~consequently, both OpenCL- and CUDA-compatible architectures are represented by the same ASM instance in our formalism.
Apart from straightforward syntactical differences between OpenCL and CUDA~\cite{opencl-vs-cuda-syntax}, we see as the main differences between the two programming models (from our ASM-based abstraction level) that OpenCL has no notion of warps, and it uses a different terminology~--~%
\emph{Global/Local/Private Memory (\texttt{GM}/\texttt{LM}/\texttt{PM})} instead of device/shared/register memory, and \emph{Compute Unit~(\texttt{CU})} and \emph{Processing Element~(\texttt{PE})}, rather than~\texttt{SMX} and~\texttt{CC}.

In the following, we consider memory regions and cores of ASM-represented architectures as arrangeable in an arbitrary number of dimensions.
Programming models for such architectures often have native support for such arrangements.
For example, in the CUDA model, memory is accessed via arrays which can be arbitrary-dimensional (a.k.a \emph{multi-dimensional C arrays}), and cores are programmed in CUDA via threads which are arranged in CUDA's so-called dimensions \texttt{x}, \texttt{y}, \texttt{z};~further thread dimensions can be explicitly programmed in CUDA, e.g., by embedding them in the last dimension \texttt{z}.
Details on our arrangements of memory and cores are provided in the Appendix, Section~\ref{app_md_arrangement_mem_cor}.

We express constraints of programming models~--~for example, that in CUDA, \texttt{SMX} can combine their results in \texttt{DM} only~\cite{cuda-programming-guide}~--~via so-called \emph{tuning
parameter constraints}, which we discuss later in this section.

Note that we call our abstraction \emph{Abstract System Model} (rather than \emph{Abstract Architecture Model}, or the like), because it can also represent systems consisting of multiple devices and/or nodes, etc.
For example, our ASM representation of a multi-GPU system is:
\[
\texttt{ASM}_\texttt{Multi-GPU} \ := \ ( \ |\{\texttt{HM},\texttt{DM},\texttt{SM},\texttt{RM}\}| \ , \, |\{\texttt{GPU},\texttt{SMX},\texttt{CC}\}| \ ) \ = \  (4,3)
\]
It extends our ASM-based representation of CUDA devices (Example~\ref{example_asm}) by \emph{Host Memory~(\texttt{HM})} which represents the memory region of the system containing the GPUs (and in which the intermediate results of different GPUs are combined), and it introduces the further core layer~\texttt{GPU} representing the system's GPUs.
Analogously, our ASM representation of a multi-node, multi-GPU system is:
\[
\texttt{ASM}_\texttt{Multi-Node-Multi-GPU} \ := \ ( \ |\{\texttt{NM},\texttt{HM},\texttt{DM},\texttt{SM},\texttt{RM}\}| \ , \, |\{\texttt{NOD},\texttt{GPU},\texttt{SMX},\texttt{CC}\}| \ ) \ = \  (5,4)
\]
It adds to $\texttt{ASM}_\texttt{Multi-GPU}$ the memory layer \emph{Node Memory~ (\texttt{NM})} which represents the memory region of the host node, and it adds core layer \emph{Node}~(\texttt{NOD}) which represents the compute nodes.
Our approach
is currently designed for \emph{homogeneous systems}, i.e., all devices/nodes/$\dots$ are assumed to be identical.
We aim to extend our approach to \emph{heterogeneous systems} (which may consist of different devices/nodes/$\dotsc$) as future work, inspired by dynamic load balancing approaches~\cite{5470413}.

\subsection{Basic Building Blocks}
\label{sec_bbb}

We introduce the three main basic building blocks of our low-level representation:
1)~\emph{low-level MDAs} which we use as partitioning of MDAs and that
represent multi-layered, multi-dimensionally arranged collection of ordinary MDAs (Definition~\ref{def_mda})~--~one ordinary MDA per memory/core layer of their target~ASM and for each dimension of the MDH computation (as illustrated in Figure~\ref{mda_part_example_informal});~%
2)~\emph{low-level BUFs} which
are a collection of ordinary BUFs (Definition~\ref{def_buffer}) and that are augmented with a \emph{memory region}
and a
\emph{memory layout};~%
3)~\emph{low-level COs} which represent combine operators (Definition~\ref{def_combine_op}) to which the layer and dimension of their target ASM is assigned to be used in our generated code to compute the operator (e.g., a core layer to compute the operator in~parallel).

\begin{definition}[Low-Level MDA]
\label{def_ll_mda}
Let be $L\in\IN$
(representing an ASM's number of layers)
and $D\in\IN$ (representing an MDH's number of dimensions).
Let further be
$
P=
(\,
(P^1_1,\dotsc,P^1_D)
\ ,\,\dotsc\,,$ $
(P^L_1,\dotsc,P^L_D)
\,)\in\IN^{L\times D}
$
an arbitrary tuple of $L$-many $D$-tuples of positive natural numbers,
$T\in\type$ a scalar type, and
$
I
:=
( \ (I_d^{^{<p^1_1,\dotsc,p^1_D\,|\,\dotsc\,|\,p^L_1,\dotsc,p^L_D>}}\in\IDXs)_{d\in[1,D]_\IN} \ )^{<\, (p^1_1,\dotsc,p^1_D)\in P^1_1\times\dotsc\times P^1_D \ | \
}$ $^{
\dotsc \ | \ (p^L_1,\dotsc,p^L_D)\in P^L_1\times\dotsc\times P^L_D \,>}
$
an arbitrary collection of $D$-many MDA index sets (Definition~\ref{def_mda}) for each particular choice of indices
$(p^1_1,\dotsc,p^1_D)\in P^1_1\times\dotsc\times P^1_D \ \ , \ \dotsc \ , \ \ (p^L_D,\dotsc,p^L_D)\in P^L_1\times\dotsc\times P^L_D$\footnote{
    Analogously to Notation~\ref{not_mda}, we identify each $P^l_d\in\IN$ implicitly also with the interval $[0,P^l_d)_{\IN_0}$ (inspired by set theory).
} (illustrated in Figure~\ref{mda_part_example_informal}).

An $L$-layered, $D$-dimensional, $P$-partitioned
\emph{low-level Multi-Dimensional Array (low-level MDA)} that has scalar type~$T$ and index sets~$I$ is any function $\MDA_{ll}$ of type:
\vspace*{-2px}
\begin{align*}
  &
  \MDA_{ll}^{<\,
  \overbrace{\scriptstyle(p^1_1,\dotsc,p^1_D)\in P^1_1\times\dotsc\times P^1_D}^\text{\tiny Partitioning: \ Layer $1$}
  \ | \ \dotsc \ | \
  \overbrace{\scriptstyle(p^L_1,\dotsc,p^L_D)\in P^L_1\times\dotsc\times P^L_D}^\text{\tiny Partitioning: \ Layer $L$}
  \,>}: \\[-2pt]
  &\hspace*{145px}
  I^{<p^1_1,\dotsc,p^1_D\,|\,\dotsc\,|\,p^L_1,\dotsc,p^L_D>}_1
  \,\times\dotsc\times\,
  I^{<p^1_1,\dotsc,p^1_D\,|\,\dotsc\,|\,p^L_1,\dotsc,p^L_D>}_D
  \, \to \, T
\end{align*}
\end{definition}

We use low-level MDAs in the following to represent partitionings of MDAs (as illustrated soon and formally discussed the Appendix, Section~\ref{app_sec_mda_part}).

\vspace*{5px}

Next, we introduce \emph{low-level BUFs} which work similarly as BUFs (Definition~\ref{def_buffer}), but are tagged with a memory region
and a
memory layout.
While these tags have no effect on the operators' semantics, they indicate later to our code generator in which memory region the BUF should be stored and accessed, and which memory layout to chose for storing the BUF.
Moreover, we use these tags to formally define constraints of programming models, e.g., that according to the CUDA specification~\cite{cuda-programming-guide}, \texttt{SMX} cores can combine their results in memory region \texttt{DM} only.

\begin{definition}[Low-Level BUF]
\label{ll_buffers}
Let be $L\in\IN$
(representing an ASM's number of layers)
and $D\in\IN$ (representing an MDH's number of dimensions).
Let further
$
P=(\,
(P^1_1,\dotsc,P^1_D)
\ ,\,\dotsc\,,$ $
(P^L_1,\dotsc,P^L_D)\,)\in\IN^{L\times D}$
be an arbitrary tuple of $L$-many $D$-tuples of positive natural numbers,
$T\in\type$ a scalar type, and
$
N
:=
( \ (N_d^{^{<p^1_1,\dotsc,p^1_D\,|\,\dotsc\,|\,p^L_1,\dotsc,p^L_D>}}\in\IN)_{d\in[1,D]_\IN} \ )^{<\, (p^1_1,\dotsc,p^1_D)\in P^1_1\times\dotsc\times P^1_D \ | \ \dotsc \ | \ (p^L_1,
}$ $^{
\dotsc,p^L_D)\in P^L_1\times\dotsc\times P^L_D \,>}
$
be a BUF's size (Definition~\ref{def_buffer}) for each particular choice of $p^1_1,\dotsc,p^L_D$.

An $L$-layered, $D$-dimensional, $P$-partitioned
\emph{low-level Buffer (low-level BUF)} that has scalar type $T$ and size $N$ is any function $\Buf_{ll}$ of type ($\bij$ denotes bijection):
\vspace*{-2px}
\begin{align*}
&
\Buf_{ll}^{<
\overbrace{\scriptstyle \texttt{MEM}\in[1,\texttt{NUM\_MEM\_LYRs}]_\IN}^\text{Memory Region}
\, | \,
\overbrace{\scriptstyle \sigma:[1,D]_\IN\bij [1,D]_\IN}^\text{Memory Layout}
><
\overbrace{\scriptstyle(p^1_1,\dotsc,p^1_D)\in P^1_1\times\dotsc\times P^1_D}^\text{\tiny Partitioning: \ Layer $1$}
\ | \ \dotsc \ | \
\overbrace{\scriptstyle(p^L_1,\dotsc,p^L_D)\in P^L_1\times\dotsc\times P^L_D}^\text{\tiny Partitioning: \ Layer $L$}
>}: \\[-1pt]
&\hspace*{80px}
[0,N_1^{<\, p^1_1,\dotsc,p^1_D \ | \ \dotsc \ | \ p^L_1,\dotsc,p^L_D \,>})_{\IN_0}
\,\times\dotsc\times\,
[0,N_D^{<\, p^1_1,\dotsc,p^1_D \ | \ \dotsc \ | \ p^L_1,\dotsc,p^L_D \,>})_{\IN_0}
 \, \to \, T
\end{align*}
We refer to $\texttt{MEM}$ as low-level BUF's \emph{memory region} and to $\sigma$ as its \emph{memory layout}, and we refer to the function
\begin{align*}
    &
    {\Buf^\texttt{trans}_{ll}}^{<
    \overbrace{\scriptstyle \texttt{MEM}\in[1,\texttt{NUM\_MEM\_LYRs}]_\IN}^\text{Memory Region}
    \, | \,
    \overbrace{\scriptstyle \sigma:[1,D]_\IN\bij [1,D]_\IN}^\text{Memory Layout}
    ><
    \overbrace{\scriptstyle(p^1_1,\dotsc,p^1_D)\in P^1_1\times\dotsc\times P^1_D}^\text{\tiny Partitioning: \ Layer $1$}
    \ | \ \dotsc \ | \
    \overbrace{\scriptstyle(p^L_1,\dotsc,p^L_D)\in P^L_1\times\dotsc\times P^L_D}^\text{\tiny Partitioning: \ Layer $L$}
    >}: \\
    &\hspace*{80px}
    [0,N^{<\, p^1_1,\dotsc,p^1_D \ | \ \dotsc \ | \ p^L_1,\dotsc,p^L_D \,>}_{\sigma(1)})_{\IN_0}
    \,\times\dotsc\times\,
    [0,N^{<\, p^1_1,\dotsc,p^1_D \ | \ \dotsc \ | \ p^L_1,\dotsc,p^L_D \,>}_{\sigma(D)})_{\IN_0}
     \, \to \, T
    \end{align*}
    that is defined as
    \[
        {\Buf^\texttt{trans}_{ll}}^{<\, \texttt{MEM} \, | \, \sigma><p^1_1,\dotsc,p^1_D\,|\,\dotsc\,|\,p^L_1,\dotsc,p^L_D\,>}(i_{\sigma(1)},\dotsc,i_{\sigma(D)})
        \ := \
        \Buf_{ll}^{<\, \texttt{MEM} \, | \, \sigma><p^1_1,\dotsc,p^1_D\,|\,\dotsc\,|\,p^L_1,\dotsc,p^L_D\,>}(i_1,\dotsc,i_D)
    \]
as $\Buf_{ll}$'s \emph{transposed function representation} (which we use to store the buffer in our generated code).
\end{definition}

Finally, we introduce \emph{low-level COs}.
We define such operators to behave the same as ordinary combine operators (Definition~\ref{def_combine_op}), but we additionally tag them with a layer of their target ASM.
Similarly as for low-level BUFs, the tag has no effect on semantics, but it is used in our code generation process to assign the computation to the hardware (e.g., indicating that the operator is computed
by either an \texttt{SMX}, \texttt{WRP}, or \texttt{CC} when targeting CUDA~--~see Example~\ref{example_asm}).
Also, we use the tags to define model-specific constraints in our formalism (as also discussed for low-level BUFs).
We also tag the combine operator with a dimension of the ASM layer, enabling later in our optimization process to express advanced data access patterns (a.k.a. \emph{swizzles}~\cite{10.1145/3297858.3304059}).
For example, when targeting CUDA, flexibly mapping ASM dimensions on \texttt{CC} layer (in CUDA terminology, the dimensions are called \texttt{threadIdx.x}, \texttt{threadIdx.y}, and \texttt{threadIdx.z}) to array dimensions enables the well-performing \emph{coalesced global memory accesses}~\cite{cuda-programming-guide} for both transposed and non-transposed data layouts, by only using different dimension tags.

\begin{definition}[ASM Level]
\label{def_asm_lvl}
We refer to pairs $(l_\texttt{ASM},d_\texttt{ASM})$~--~consisting of an ASM layer $l_\texttt{ASM}\in[1,L]_\IN$ and an ASM dimension $d_\texttt{ASM}\in[1,D]_\IN$~--~as \emph{ASM Levels}  (\texttt{ASM-LVL})\footnote{
For simplicity, we refrain from annotating identifier \texttt{ASM-LVL} with values $L$ and $D$ (e.g., $\texttt{ASM-LVL}^{<L,D>}$ ), because both values will usually be clear from the context.
} (terminology motivated in the Appendix, Section~\ref{app_asm_levels}):
\[
    \ASMLVL:=\{\ (l_\texttt{ASM},d_\texttt{ASM}) \ | \ l_\texttt{ASM}\in[1,L]_\IN, d_\texttt{ASM}\in[1,D]_\IN\}
\]
\end{definition}

\begin{definition}[Low-Level CO]
\label{def_ll_comb}

Let be $L\in\IN$
(representing an ASM's number of layers)
and $D\in\IN$ (representing an MDH's number of dimensions).

A \emph{low-level Combine Operator (low-level CO)}
\[\co{}^{<(l_\texttt{ASM},d_\texttt{ASM})\in\ASMLVL=\{ \ (l,d) \ | \ l\in[1,L]_\IN \, , \ d\in[1,D]_\IN \, \}>}\]
is a function for which $\co{}^{<(l_\texttt{ASM},d_\texttt{ASM})>}$ is an ordinary combine operator (Definition~\ref{def_combine_op}), for each $(l_\texttt{ASM},d_\texttt{ASM})\in\ASMLVL$.
\end{definition}

\vspace*{5px}

Note that in Figure~\ref{fig_ll_example}, for better readability, we use domain-specific identifiers for ASM layers:~$\texttt{HM:=1}$ as an alias for the ASM layer that has id $1$, $\texttt{L1:=2}$ for the layer with id $2$, and $\texttt{COR:=3}$ for the layer with id $3$.
For dimensions, we use aliases $x:=1$ for ASM dimension $1$ and $y:=2$ for ASM dimension $2$, correspondingly.

\subsection{Generic Low-Level Expression}
\label{ch_gen_ll_exp}

\begin{figure}[p]
\begin{subfigure}{\textwidth}
\center
\includegraphics[width=0.9\textwidth]{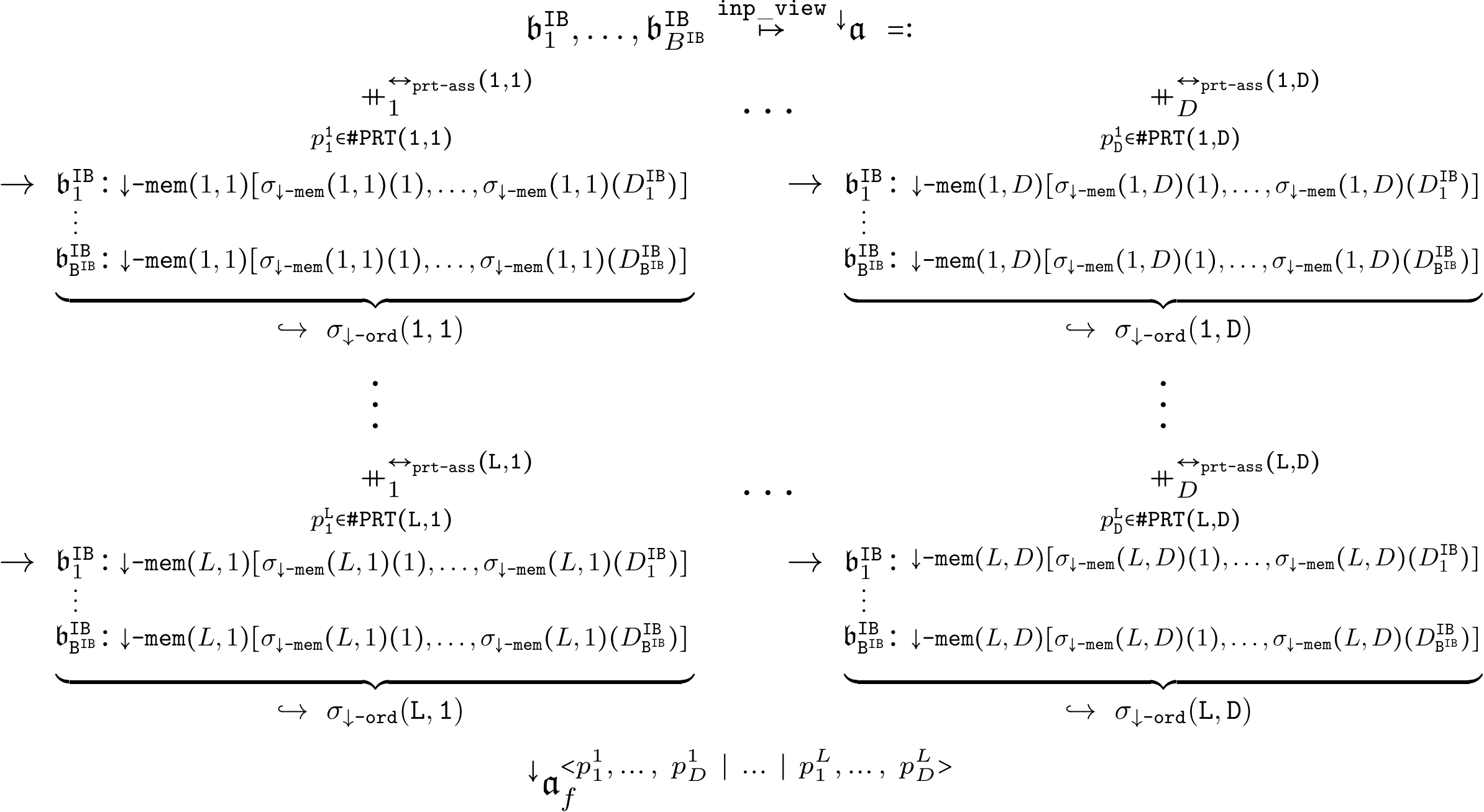}
\caption{
De-Composition Phase
}
\end{subfigure}

\vspace*{15px}

\begin{subfigure}{\textwidth}
\center
\includegraphics[width=0.65\textwidth]{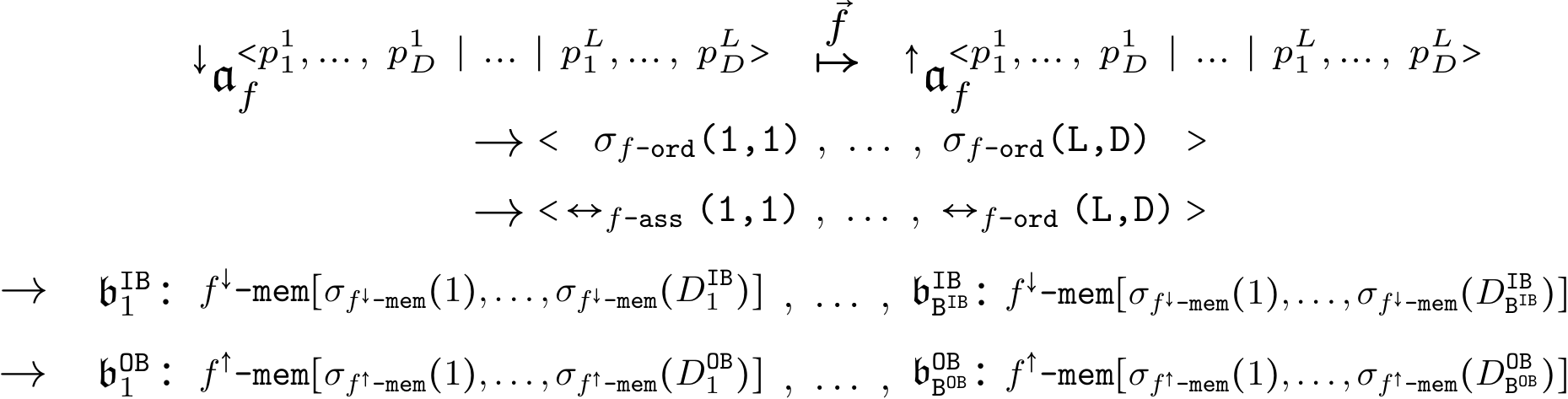}
\caption{
Scalar Phase
}
\end{subfigure}

\vspace*{15px}

\begin{subfigure}{\textwidth}
\center
\includegraphics[width=0.9\textwidth]{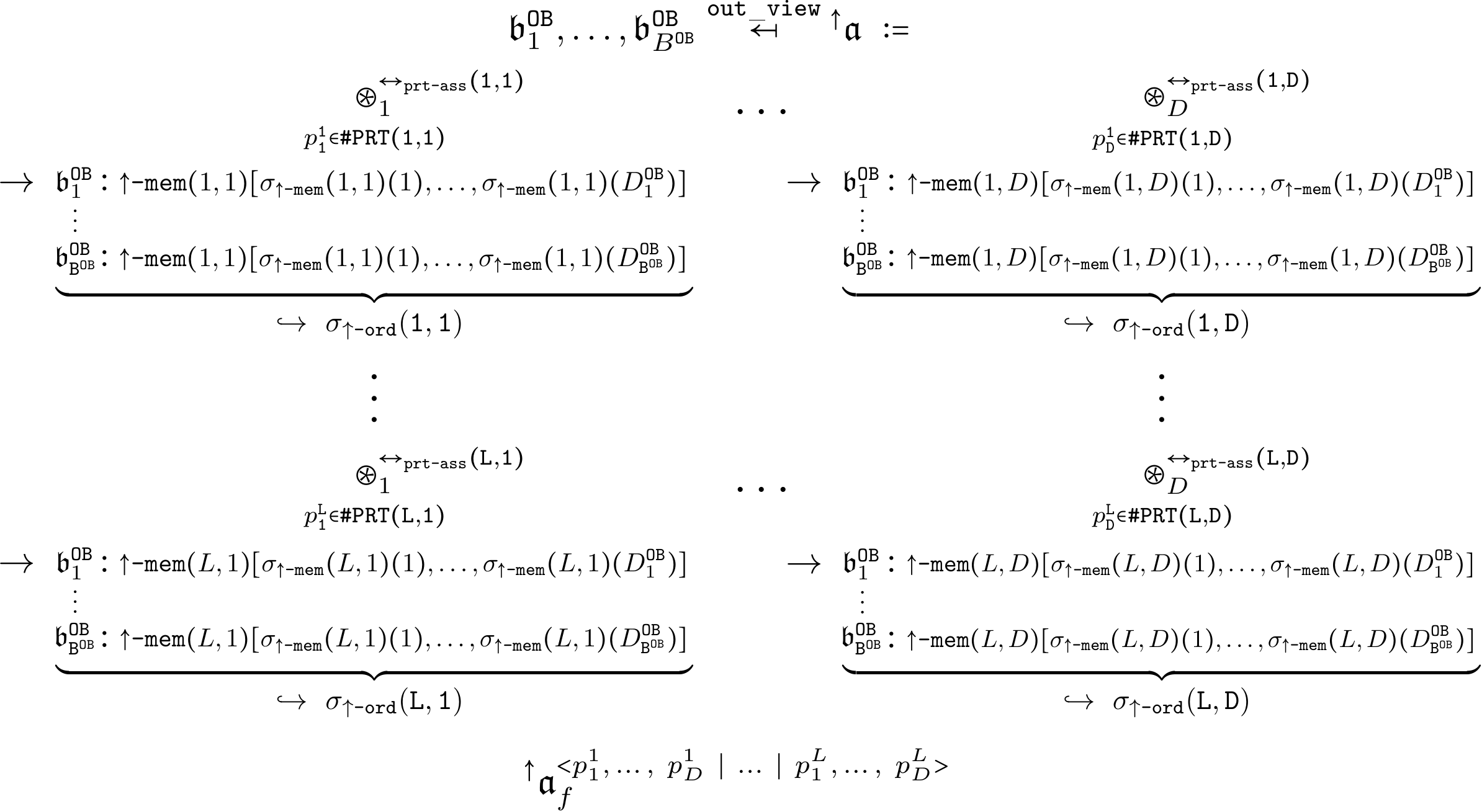}
\caption{
Re-Composition Phase
}
\end{subfigure}

\vspace*{-5px}

\caption{
Generic low-level expression for data-parallel computations
}
\label{fig_gen_ll}
\end{figure}

\setlength\dashlinedash{0.5pt}
\setlength\dashlinegap{3pt}
\setlength\arrayrulewidth{0.3pt}
\begin{table}[t!]
\center
\begin{tabular}{c|l|l|l}
No. & Name & Range & Description\\
\hline &&& \\
\texttt{0} & $\#\texttt{PRT}$ \hspace*{25px} & $\MDHLVL\to\IN$ & \texttt{number of parts} \\
&&& \\
\hline &&& \\
\texttt{D1} & $\sigma_\texttt{$\downarrow$-ord}$ &
$\MDHLVL\bij\MDHLVL$
& \texttt{de-composition order} \\[4pt]
\texttt{D2} & $\leftrightarrow_\texttt{$\downarrow$-ass}$ &
$\MDHLVL\bij\ASMLVL$
& \texttt{ASM assignment (de-composition)}     \\
&&& \\
\texttt{D3} & $\texttt{$\downarrow$-mem}^\texttt{<ib>}$ & $\MDHLVL\to\texttt{MR}$ & \texttt{memory regions of input BUFs (\texttt{ib})} \\[4pt]
\texttt{D4} & $\sigma_{\texttt{$\downarrow$-mem}}^\texttt{<ib>}$ &
$\MDHLVL\to[1,\dotsc,D^\texttt{IB}_\texttt{ib}]_\mathcal{S}$
& \texttt{memory layouts of input BUFs (\texttt{ib})} \\
&&& \\
\hdashline &&& \\
\texttt{S1} & $\sigma_\texttt{$f$-ord}$ &
$\MDHLVL\bij\MDHLVL$
& \texttt{scalar function order} \\[4pt]
\texttt{S2} & $\leftrightarrow_\texttt{$f$-ass}$ &
$\MDHLVL\bij\ASMLVL$ & \texttt{ASM assignment (scalar function)} \\
&&& \\
\texttt{S3} & $f^\downarrow\texttt{-mem}^\texttt{<ib>}$ & $\texttt{MR}$ & \texttt{memory region of input BUF (\texttt{ib})} \\[4pt]
\texttt{S4} & $\sigma_{\texttt{$f^\downarrow$-mem}}^\texttt{<ib>}$ &
$[1,\dotsc,D^\texttt{IB}_\texttt{ib}]_\mathcal{S}$
& \texttt{memory layout of input BUF (\texttt{ib})} \\
&&& \\
\texttt{S5} & $f^\uparrow\texttt{-mem}^\texttt{<ob>}$ & $\texttt{MR}$ & \texttt{memory region of output BUF (\texttt{ob})} \\[4pt]
\texttt{S6} & $\sigma_{\texttt{$f^\uparrow$-mem}}^\texttt{<ob>}$ &
$[1,\dotsc,D^\texttt{OB}_\texttt{ob}]_\mathcal{S}$
& \texttt{memory layout of output BUF (\texttt{ob})} \\
&&& \\
\hdashline &&& \\
\texttt{R1} & $\sigma_\texttt{$\uparrow$-ord}$ &
$\MDHLVL\bij\MDHLVL$
& \texttt{re-composition order} \\[4pt]
\texttt{R2} & $\leftrightarrow_\texttt{$\uparrow$-ass}$ &
$\MDHLVL\bij\ASMLVL$
& \texttt{ASM assignment (re-composition)}     \\
&&& \\
\texttt{R3} & $\texttt{$\uparrow$-mem}^\texttt{<ob>}$ & $\MDHLVL\to\texttt{MR}$ & \texttt{memory regions of output BUFs (\texttt{\texttt{ob}})} \\[4pt]
\texttt{R4} & $\sigma_{\texttt{$\uparrow$-mem}}^\texttt{<ob>}$ &
$\MDHLVL\to[1,\dotsc,D^\texttt{OB}_\texttt{ob}]_\mathcal{S}$
& \texttt{memory layouts of output BUFs (\texttt{ob})} \\
&&& \\
\end{tabular}
\caption{
Tuning parameters of our low-level expressions
}
\label{tab_tps}
\end{table}

Figure~\ref{fig_gen_ll} shows a generic expression in our low-level representation:~%
it targets an arbitrary but fixed $L$-layered ASM instance,
and it implements~--~on low level~--~the generic instance of our high-level expression in Figure~\ref{fig_generic_hl}.
Inserting into the low-level expression a particular value for ASM's numbers of layer $L$, as well as particular values for the generic parameters of the high-level expression in Figure~\ref{fig_generic_hl} (dimensionality $D$, combine operators $\co{1},\dotsc,\co{d}$, and input/output views) results in an instance of the expression in Figure~\ref{fig_gen_ll} that remains generic in tuning parameters only;~this auto-tunable instance
will be the focus of our discussion in the remainder of this section.

In Section~\ref{ch:lowering}, we show that we fully automatically compute the auto-tunable low-level expression for a concrete ASM instance and high-level expression, and we automatically optimize this tunable expression for a particular target architecture and characteristics of the input and output data using auto-tuning~\cite{10.1145/3427093}.
Our final outcome is a concrete (non-generic) low-level expression (as in Figure~\ref{fig_ll_example}) that is auto-tuned for the particular target architecture (represented via an ASM instance, e.g., ASM instance $\texttt{ASM}_\texttt{CUDA}$ when targeting an NVIDIA Ampere GPU) and high-level MDH expression.
From this auto-tuned low-level expression, we can straightforwardly generate executable program code, because all the major optimization decisions have already been made in the previous auto-tuning step.
Our overall approach is illustrated in Figure~\ref{contributions_overall}.

\subsubsection*{Auto-Tunable Parameters}

Table~\ref{tab_tps} lists the tuning parameters of our auto-tunable low-level expressions~--~different values of tuning parameters lead to semantically equal variants of the auto-tunable low-level expression (which we prove formally in Section~\ref{ch:lowering}), but the variants will be translated to differently optimized code variants.

In the following, we explain the $15$ tuning parameters in Table~\ref{tab_tps}.
We give our explanations in a general, formal setting that is independent of a particular computation and programming model;~the parameters are discussed afterward for the concrete example computation \emph{matrix multiplication} in the models OpenMP, CUDA, and OpenCL.

Our tuning parameters in Table~\ref{tab_tps} have constraints:~%
1)~\emph{algorithmic constraints} which have to be satisfied by all target programming models, and
2)~\emph{model constraints} which are specific for particular programming models only (CUDA-specific constraints, OpenCL-specific constraints, etc), e.g., that the results of CUDA's thread blocks can be combined in designated memory regions only~\cite{cuda-programming-guide}.
We discuss algorithmic constraints in the following, together with our tuning parameters;~model constraints are discussed
in our Appendix, Section~\ref{app_constraints}, for the interested reader.

In the following, we present our $15$ tuning parameters in Table~\ref{tab_tps}.
Dotted lines separate parameters for different phases:~parameters \texttt{D1}-\texttt{D4} customize the de-composition phase, parameters \texttt{S1}-\texttt{S6} the scalar phase, and parameters \texttt{R1}-\texttt{R4} the re-composition phase, correspondingly;~the parameter \texttt{0} impacts all three phases (separated by a straight line in the table).

Note that our parameters do not aim to introduce novel optimization techniques, but to unify, generalize, and combine together well-proven optimizations, based on a formal foundation, toward an efficient, overall optimization process that applies to various combinations of data-parallel computations, architectures, and characteristics of input and output data (e.g., their size and memory layout).

In Table~\ref{tab_tps}, we point to combine operators in Figure~\ref{fig_ll_example} using pairs $(l,d)$ to which we refer as \emph{MDH Levels} (terminology motivated in the Appendix, Section~\ref{app_mdh_levels}).
We use the pairs as enumeration for operators in the de-composition and re-composition phases.

\begin{definition}[MDH Level]
\label{def_mdh_lvl}
We refer to pairs $(l_\texttt{MDH},d_\texttt{MDH})$~--~consisting of a layer $l_\texttt{MDH}\in[1,L]_\IN$ and dimension $d_\texttt{MDH}\in[1,D]_\IN$~--~as~\emph{MDH Levels} ($\MDHLVL$):
\[
    \MDHLVL:=\{\ (l_\texttt{MDH},d_\texttt{MDH}) \ | \ l_\texttt{MDH}\in[1,L]_\IN, d_\texttt{MDH}\in[1,D]_\IN\}\footnote{
        The same as for identifier $\ASMLVL$ (Definition~\ref{def_asm_lvl}), we refrain from annotating identifier $\MDHLVL$ with values $L$ and $D$.
        Note that $\MDHLVL$ and $\ASMLVL$ both refer to the same set of pairs, but we use identifier $\MDHLVL$ when referring to MDH levels and identifier $\ASMLVL$ when referring to ASM levels, correspondingly, for better clarity.
    }
\]
\end{definition}
For example, in the de-composition phase of Figure~\ref{fig_ll_example} (right part of the figure), pair $(1,1)\in\MDHLVL$ points to the first combine operator, as the operator operates on the first layer $l=1$ and in the first dimension $d=1$ (discussed in Section~\ref{ch:ll_intro_example}).
Analogously, pairs $(1,2),(2,1)\in\MDHLVL$ point to the second and third operator, etc.
An operator's enumeration can be easily deduced from its corresponding $p$ variable:~the variable's superscript indicates the operator's corresponding layer~$l$ and the variable's subscript indicates its dimension~$d$.

\paragraph{Parameter \texttt{0}:}

Parameter $\#\texttt{PRT}$ is a function that maps pairs in
$\MDHLVL$
to natural numbers;~the parameter determines \emph{how much} data are grouped together into parts in our low-level expression in Figure~\ref{fig_gen_ll} (and consequently also in our generated code later), by
setting the particular number of parts (a.k.a. \emph{tiles}) used in our expression.
For example, in Figure~\ref{fig_ll_example}, we use
$\#\texttt{PRT}(1,1):=2$
which causes combine operators
$\dplus^\texttt{(HM,x)}_1$ and $\circledast^\texttt{(HM,x)}_1$
to iterate over interval $[0,2)_{\IN_0}$ (and thus partitioning the MDH computation on level $(1,1)$ into two parts), and we use
$\#\texttt{PRT}(1,2):=4$
to let operators
$\dplus^\texttt{(HM,y)}_2$ and $\circledast^\texttt{(HM,x)}_2$
iterate over interval $[0,4)_{\IN_0}$ (partitioning on level $(1,2)$ into four parts), etc.

To ensure a full partitioning (so that we obtain singleton MDAs to which scalar function $f$ can be applied in the scalar phase, as discussed above), we require the following algorithmic constraint for the parameter ($N_d$ denotes the input size in dimension $d$, see Figure~\ref{fig_generic_hl}):

\[
\prod_{l\in[1,L]_\IN} \#\texttt{PRT}(l,d)\,=\,N_d \text{, \ for all } d\in[1,D]_\IN
\]

In our generated code, the number of parts directly translates to the number of \emph{tiles} which are computed either sequentially (a.k.a. \emph{cache blocking}~\cite{10.1145/106972.106981}) or in parallel,
depending on the combine operators' tags (which are chosen via Parameters~\texttt{D2,S2,R2}, as discussed soon).
In our example from Figure~\ref{fig_ll_example}, we process parts belonging to combine operators tagged with \texttt{HM} and \texttt{L1} sequentially, via \texttt{for}-loops, because \texttt{HM} and \texttt{L1} correspond to ASM's memory layers (note that Parameter~\texttt{0} only chooses the number of tiles;~the parameter has no effect on explicitly copying data into fast memory resources, which is the purpose of Parameters~\texttt{D3,R3,S1,S2}).
The \texttt{COR} parts are computed in parallel in our generated code, because \texttt{COR} corresponds to ASM's core layer, and thus, the number of \texttt{COR} parts determines the number of threads used in our code.

An optimized number of tiles is essential for achieving high performance~\cite{10.1145/197405.197406}, e.g., due to its impact for locality-aware data accesses (number of sequentially computed tiles) and efficiently exploiting parallelism (number of tiles computed in parallel, which corresponds to the number of threads in our generated code).

\paragraph{Parameters \texttt{D1,S1,R1}:}
These three parameters are permutations on \MDHLVL{} (indicated by symbol~$\bij$ in Table~\ref{tab_tps}), determining \emph{when} data are accessed and combined.
The parameters specify the order (indicated by symbol $\hookrightarrow$ in Figure~\ref{fig_gen_ll}) of combine operators in the de-composition and re-composition phases (parameters \texttt{D1} and \texttt{R1}), and the order of applying scalar function $f$ to parts (parameter \texttt{S1}).
Thereby, the parameters specify when parts are processed during the computation.

In our generated code, combine operators are implemented as sequential/parallel loops such that the parameters enable optimizing loop orders (a.k.a. \emph{loop permutation}~\cite{10.1145/233561.233564}).
For combine operators assigned (via parameter \texttt{R2}) to ASM's core layer and thus computed in parallel, parameter \texttt{R1} particularly determines when the computed results of threads are combined:~%
if we used in the re-composition phase of Figure~\ref{fig_ll_example} combine operators tagged with \texttt{(COR,x)} and \texttt{(COR,y)} immediately after applying scalar function $f$ (i.e., in steps \circled{10} and \circled{11}, rather than steps \circled{12} and \circled{13}), we would combine the computed intermediate results of threads multiple times, repeatedly after each individual computation step of threads, and using the two operators at the end of the re-composition phase (in steps \circled{14} and \circled{15}) would combine the result of threads only once, at the end of the re-composition phase.
Combining the results of threads early in the computation usually has the advantages of reduced memory footprint, because memory needs to be allocated for one thread only, but at the cost of more computations, because the results of threads need to be combined multiple times.
In contrast, combing the results of threads late in the computation reduces the amount of computations, but at the cost of higher memory footprint.
Our parameters make this trade-off decision generic in our approach such that the decision can be left to an auto-tuning system, for example.

Note that each phase corresponds to an individual loop nest which we fuse together when parameters \texttt{D1,S1,R1} (as well as parameters \texttt{D2,S2,R2}) coincide (as also outlined in our Appendix, Section~\ref{ch_ll_opts}).

\paragraph{Parameters \texttt{D2,S2,R2}:}
These parameters (symbol $\bij$ in the table denotes bijection) assign MDH levels to ASM levels, by setting the tags of low-level combine operators (Definition~\ref{def_ll_comb}).
Thereby, the parameters determine \emph{by whom} data are processed (e.g., threads or \texttt{for}-loops),
similar to the concept of \texttt{bind} in scheduling languages~\cite{tvm_bind}.
Consequently, the parameters determine which parts should be computed sequentially in our generated code and which parts in parallel.
For example, in Figure~\ref{fig_ll_example}, we use $\leftrightarrow_\texttt{$\downarrow$-ass}(2,1):=(\texttt{COR},\texttt{x})$ and $\leftrightarrow_\texttt{$\downarrow$-ass}(2,2):=(\texttt{COR},\texttt{y})$, thereby assigning the computation of MDA parts on layer $2$ in both dimensions
to ASM's \texttt{COR} layer in the de-composition phase,
which causes processing the parts in parallel in our generated code.
For multi-layered core architectures, the parameters particularly determine the thread layer to be used for the parallel computation (e.g., \texttt{block} or \texttt{thread} in CUDA).

Using these parameters, we are able to flexibly set data access patterns in our generated code.
In Figure~\ref{fig_ll_example}, we assign parts on layer $2$ to \texttt{COR} layers, which results in a so-called \emph{block access} pattern of cores:~we start $8*16$ threads, according to the $8*16$ core parts, and each thread processes a part of the input MDA representing a block of $32\times64$ MDA elements within the input data.
If we had assigned in the figure the first computation layer to ASM's \texttt{COR} layer (in the figure, this layer is assigned to ASM's \texttt{HM} layer), we would start $2*4$ threads and each thread would process MDA parts of size $(8*32)\times(16*64)$;~assigning the last MDH layer to \texttt{COR}s would result in $(2*8*32)\times(4*16*64)$ threads, each processing a singleton MDA~%
(a.k.a. \emph{strided access}).

The parameters also enable expressing so-called \emph{swizzle} access patterns~\cite{10.1145/3297858.3304059}.
For example, in CUDA, processing consecutive data elements in data dimension $1$ by threads that are consecutive in thread dimension $2$ (a.k.a \texttt{threadIdx.y} dimension in CUDA) can achieve higher performance due to the hardware design of fast memory resources in NVIDIA GPUs.
Such swizzle patterns can be easily expressed and auto-tuned in our approach;~for example, by interchanging in Figure~\ref{fig_ll_example} tags \texttt{(COR,x)} and \texttt{(COR,y)}.
For memory layers (such as \texttt{HM} and \texttt{L1}), the dimension tags \texttt{x} and \texttt{y} currently have no effect on our generated code, as the programming models we target at the moment (OpenMP, CUDA, and OpenCL) have no explicit notion of tiles.
However, this might change in the future when
targeting new kinds of programming models, e.g., for upcoming architectures.

\paragraph{Parameters \texttt{D3,R3} and \texttt{S3,S5}:}
Parameters \texttt{D3} and \texttt{R3} set for each BUF the memory region to be used, thereby determining \emph{where} data are read from or written to, respectively.
In the table, we use $\texttt{ib}\in[1,B^\texttt{IB}]_\IN$~to refer to a particular input BUF (e.g., \texttt{ib=1} to refer to the input matrix of matrix-vector multiplication, and \texttt{ib=2} to refer to the input vector), and $\texttt{ob}\in[1,B^\texttt{OB}]_\IN$ refers to an output BUF, correspondingly.
Parameter \texttt{D3} specifies the memory region to read from, and parameter \texttt{R3} the region to write to.
The set $\texttt{MR}:=[1,\texttt{NUM\_MEM\_LYRs}]_\IN$ denotes the ASM's memory regions.

Similarly to parameters \texttt{D3} and \texttt{R3}, parameters \texttt{S3} and \texttt{S5} set the memory regions for the input and output of scalar function $f$.

Exploiting fast memory resources of architectures is a fundamental optimization~\cite{10.1145/509593.509638,10.1007/978-3-662-44917-2_13,10.1145/3579990.3580024,bondhugula2020high},
particularly due to the performance gap between processors' cores and their memory systems~\cite{wilkes2001memory,9530719}.

\paragraph{Parameters \texttt{D4,R4} and \texttt{S4,S6}:}
These parameters set the memory layouts of BUFs, thereby determining \emph{how} data are accessed in memory;~for brevity in Table~\ref{tab_tps}, we denote the set of all BUF permutations $[1,D]_\IN\bij [1,D]_\IN$ (Definition~\ref{ll_buffers})
as $[1,\dotsc,D]_\mathcal{S}$ (symbol~$\mathcal{S}$ is taken from the notation of \emph{symmetric groups}~\cite{sagan2001symmetric}).
In the case of our matrix-vector multiplication example in Figure~\ref{fig_ll_example}, we use a standard memory layout for all matrices, which we express via the parameters by setting them to the identity function, e.g., $\sigma_{\texttt{$\downarrow$-mem}}^\texttt{<M>}(1,1):=id$ (Parameter~\texttt{D4}) for the matrix read by operator $\dplus^\texttt{(HM,x)}_1$.

An optimized memory layout is important to access data in a locality-aware and thus efficient manner.

\subsection{Examples}
\label{ll_examples_matmul_resnet}

Figures~\ref{tp_ll_tab_tvm_gpu}-\ref{tp_ll_tab_pluto_cpu} show how our low-level representation is used for expressing the (de/re)-compositions of concrete, state-of-the-art implementations.
For this, we use the popular example of matrix multiplication~(\texttt{MatMul}), on a real-world input size taken from the \texttt{ResNet-50}~\cite{DBLP:journals/corr/HeZRS15} deep learning neural network (training phase).

To challenge our formalism:
i) we express implementations generated and optimized according to notably different approaches:~scheduling approach TVM using its recent Ansor~\cite{258858} optimization engine which is specifically designed and optimized toward optimizing deep learning computations (e.g., \texttt{MatMul});~polyhedral compilers PPCG and Pluto with auto-tuned tile sizes;~%
ii) we consider optimizations for two fundamentally different kinds of architectures: NVIDIA Ampere GPU and Intel Skylake CPU.
We consider our study as challenging for our formalism, because it needs to express~--~in the same formal framework~--~the (de/re)-compositions of implementations generated and optimized according to notably different approaches (scheduling-based and polyhedral-based) and for significantly different kinds of architectures (GPU and CPU).
Experimental results for TVM, PPCG, and Pluto (including the \texttt{MatMul} study used in this section)
are presented and discussed in Section~\ref{ch:eval}, as the focus of this section is on analyzing and discussing the expressivity of our low-level representation, rather than on its performance potential~(which is often higher than that of TVM, PPCG, and Pluto, as we will see in Section~\ref{ch:eval}).

In Figures~\ref{tp_ll_tab_tvm_gpu}-\ref{tp_ll_tab_pluto_cpu}, we list our low-level representation's particular tuning parameter values for expressing the TVM- and PPCG/Pluto-generated implementations.
The parameters concisely describe the concrete (de/re)-composition strategies used by TVM, PPCG and Pluto for \texttt{MatMul} on GPU or CPU using the ResNet-50's input size. ~%
Inserting these tuning parameter values into our generic low-level expression in Figure~\ref{fig_gen_ll} results in the concrete formal representation of the (de/re)-composition strategies used by TVM, PPCG and Pluto (similarly as in~Figure~\ref{fig_ll_example}).

In the following, we describe the columns of the tables in Figures~\ref{tp_ll_tab_tvm_gpu}-\ref{tp_ll_tab_pluto_cpu}, each of which listing particular values of tuning parameters in Table~\ref{tab_tps}:~column~\texttt{0} lists values of tuning parameter~\texttt{0} in Table~\ref{tab_tps}, column~\texttt{D1} of tuning parameter~\texttt{D1}, etc.
As all four tables follow the same structure, we focus on describing the particular example table in Figure~\ref{tp_ll_tab_tvm_gpu} (example chosen arbitrarily), which shows the (de/re)-composition used in TVM's generated CUDA code for \texttt{MatMul} on NVIDIA Ampere GPU using input matrices of sizes $16\times2048$ and $2048\times1000$ taken from ResNet-50\footnote{
  For the interested reader, TVM's corresponding, Ansor-generated scheduling program is presented in our Appendix, Section~\ref{app_tvm_schedule_matmul_resnet}.
}.
Note that for clarity, we use in the figures domain-specific aliases, instead of numerical values, to better differentiate between different ASM layers and memory regions.
For example, we use in Figure~\ref{tp_ll_tab_tvm_gpu} as aliases
$\texttt{DEV}:=1$, $\texttt{SHR}:=2$, and $\texttt{REG}:=3$ to refer to CUDA's three memory layers (device memory layer \texttt{DEV}, shared memory layer \texttt{SHR}, and register memory layer \texttt{REG}), and we use $\texttt{DM}:=1$, $\texttt{SM}:=2$, and $\texttt{RM}:=3$ to refer to CUDA's memory regions device \texttt{DM}, shared \texttt{SM}, and register \texttt{RM};~aliases $\texttt{BLK}:=4$ and $\texttt{THR}:=5$ refer to CUDA's two core layers which are programmed via blocks
and threads
in CUDA.

We differentiate between memory layers and memory regions for the following reason:~for example, using tuning parameter~\texttt{0} in Table~\ref{tab_tps}, we partition input data hierarchically for each particular memory layer of the target architecture (sometime possibly into one part only, which is equivalent to not partitioning).
However, depending on the value of tuning parameter~\texttt{D3}, we do not necessarily copy the input's parts always into the corresponding memory regions (e.g., a part on \texttt{SHR} layer is not necessarily copied into shared memory \texttt{SM}), for example, when the architecture provides automatically managed memory regions (as caches in CPUs) or when only some parts of the input are accessed multiple times (e.g., the input vector in the case of matrix-vector multiplication, but not the input matrix), etc.

\paragraph*{Column~\texttt{0}}
The column lists the particular number of parts (a.k.a. \emph{tiles}) used in TVM's multi-layered, multi-dimensional partitioning strategy for \texttt{MatMul} on the ResNet-50's input matrices which have the sizes $(I,K)=16\times2048$ and $(K,J)=2048\times1000$.
We can observe from this column that the input MDA, which is initially of size $(I,J,K)=(16,1000,2048)$ for the ResNet-50's input matrices,
is partitioned into $(2*50*1)$-many parts (indicated by the first three rows in column $1$)~--~$2$ parts in the first dimension, $50$ parts in the second dimension, and $1$ part in the third dimension.
Each of these parts is then further partitioned into $(2*1*8)$-many parts (rows 4,5,6), and these parts are again partitioned into $(4*20*1)$-many further parts (rows 7,8,9), etc.

\paragraph*{Columns~\texttt{D1}, \texttt{S1}, \texttt{R1}}
These 3 columns describe the order in which parts are processed in the different phases:~de-composition (column~\texttt{D1}), scalar phase (column~\texttt{S1}), and re-composition (column~\texttt{R1}).
For example, we can observe from column~\texttt{R1} that TVM's generated CUDA code first starts to combine parts on layer $1$ in dimensions $1$,~$2$,~$3$ (indicated by $(1,1)$,~$(1,2)$,~$(1,3)$ in rows~$1$,~$2$,~$3$ of column~\texttt{R1});~afterward, the code combines parts on layer $3$ in dimensions $1$,~$2$,~$3$ (indicated by~$(2,1)$,~$(2,2)$,~$(2,3)$ in rows $7$,~$8$,~$9$ of column~\texttt{R1}), etc.
~%

Note that TVM uses the same order in the three phases (i.e., columns~\texttt{D1},~\texttt{S1},~\texttt{R1} coincide).
Most likely this is because in CUDA, iteration over memory tiles are programmed via \texttt{for}-loops such that columns~\texttt{D1},~\texttt{S1},~\texttt{R1} represent loop orders;~using the same order of loops in columns~\texttt{D1},~\texttt{S1},~\texttt{R1} thus allows TVM to generate the loop nests as fused for  the three different phases (rather than generating three individual nests), which usually achieves high performance in~CUDA.

Note further that the order of parts that are processed in parallel (columns~\texttt{D2}, \texttt{S2}, \texttt{R2} determine if parts are processed in parallel or not, as described in the next paragraph)
effects when results of blocks and threads are combined (a.k.a. \emph{parallel reduction}~\cite{harris2007optimizing}), e.g., early in the computation and thus often (but thereby often requiring less memory) or late and thus only once~(but at the cost of higher more memory consumption), etc.

\paragraph*{Columns~\texttt{D2}, \texttt{S2}, \texttt{R2}}
The columns determine how computations are assigned to the target architecture.
In our example in Figure~\ref{tp_ll_tab_tvm_gpu}, we have $(2*50*1)$-many parts in the MDA's first partitioning layer, and each of these parts is assigned to be computed by an individual CUDA block (\texttt{BLK}) in the de-composition phase (rows $1$-$3$ in column \texttt{D2}), i.e., TVM uses a so-called \emph{grid size} of \texttt{2,50,1} in its generated CUDA code for \texttt{MatMul}.
The $(4*20*1)$-many parts in the third partitioning layer (rows $7$-$9$) are processed by CUDA threads (\texttt{THR}), i.e., the CUDA \emph{block size} in the TVM-generated code is \texttt{4,20,1}.
All other parts, e.g, those belonging to the $(2*1*8)$-many parts in the second partitioning layer (rows $4$-$6$), are assigned to CUDA's memory layers (denoted as \texttt{DEV}, \texttt{SHR}, \texttt{REG} in Figure~\ref{tp_ll_tab_tvm_gpu}) and thus
processed sequentially, via \texttt{for}-loops.

\paragraph*{Columns~\texttt{D3},~\texttt{S3},~\texttt{S5},~\texttt{R3}}
While column~\texttt{0} shows the multi-layered, multi-dimensional partitioning strategy used in TVM's CUDA code (according to the CUDA model's multi-layered memory and core hierarchies, shown in Example~\ref{example_asm}), column~\texttt{0} does not indicate how CUDA's fast memory regions are exploited in the TVM-generated CUDA code for \texttt{MatMul}~--~column~\texttt{0} only describes TVM's partitioning of the input/output computations such that
parts of the input and output data can potentially fit into fast memory resources.

The actual mapping of parts to memory regions is done via columns~\texttt{D3} (memory regions to be used for input data), columns~\texttt{S3} and~\texttt{S5} (memory regions to be used for the scalar computations), and column~\texttt{R3} (memory regions to be used for storing the computed output data).
For example, column~\texttt{D3} indicates that in TVM's CUDA code for \texttt{MatMul}, parts of the $A$ and $B$ input matrices are stored in CUDA's fast shared memory \texttt{SM} (column \texttt{D3}, rows $4$-$5$ and $10$-$15$), and column~\texttt{R3} indicates that each thread computes its results within CUDA's faster register memory \texttt{RM} (column \texttt{R3}, rows $4$-$6$ and $10$-$15$).

Our flexibility of separating the tiling strategy (Parameter~\texttt{0}) from the actual usage of memory regions (columns~\texttt{D3},~\texttt{S3},~\texttt{S5},~\texttt{R3}) allows us, for example, to store parts belonging to one input buffer into fast memory resources (e.g., the input vector of matrix-vector multiplication, whose values are accessed multiple times), but not parts of other buffers (e.g., the input matrix of matrix-vector multiplication, whose values are accessed only once) or only subparts of buffers, etc.

Note that in the case of Figure~\ref{tp_ll_tab_pluto_cpu} which shows Pluto's (de/re)-composition for OpenMP code, the memory tags in columns~\texttt{D3},~\texttt{S3},~\texttt{S5},~\texttt{R3} have no effect on the generated code:~OpenMP relies on its target architecture's implicit memory system (e.g., CPU caches), rather than exposing the memory hierarchy explicitly to the programmer.
Consequently, the memory tags are ignored by our OpenMP code generator and only emphasize the implementer's intention, e.g., that in Figure~\ref{tp_ll_tab_pluto_cpu}, each of the $(2*962*218)$-many tiles in the MDA's third partitioning layer are intended by the implementer to be processed in \texttt{L2} memory (rows $7$-$9$ in columns \texttt{D3} and~\texttt{R3}), even though this decision is eventually made by the automatic cache engine of the~CPU.

\paragraph*{Columns~\texttt{D4},~\texttt{S4},~\texttt{S6},~\texttt{R4}}
These columns set the memory layout to be used for memory allocations in the CUDA code.
TVM choses in all cases CUDA's standard transpositions layout (indicated by \texttt{[1,2]} which is called \emph{row-major} layout, instead of \texttt{[2,1]} which is known as \emph{column-major} layout).
Since the same layout and memory region is used on consecutive layers, the same memory allocation is re-used in the CUDA code.
For example, parameters \texttt{D3} and \texttt{D4} contain the same values in rows $4$-$5$ and $10$-$15$, and thus only one memory buffer is allocated in shared memory for input buffer \texttt{A};~the buffer is accessed in the computations of all \texttt{SHR} and \texttt{REG} tiles, as well as \texttt{DEV} tiles in dimensions $x$ and $z$.
Similarly, only one buffer is allocated in register memory for computing the results of \texttt{SHR}, \texttt{REG}, and \texttt{DEV} tiles, because the rows $4$-$6$ and $10$-$15$ in columns \texttt{R3} and \texttt{R4} coincide.

\vspace*{20px}

\begin{figure}[hbpt]
    \centering
    \includegraphics[scale=0.7]{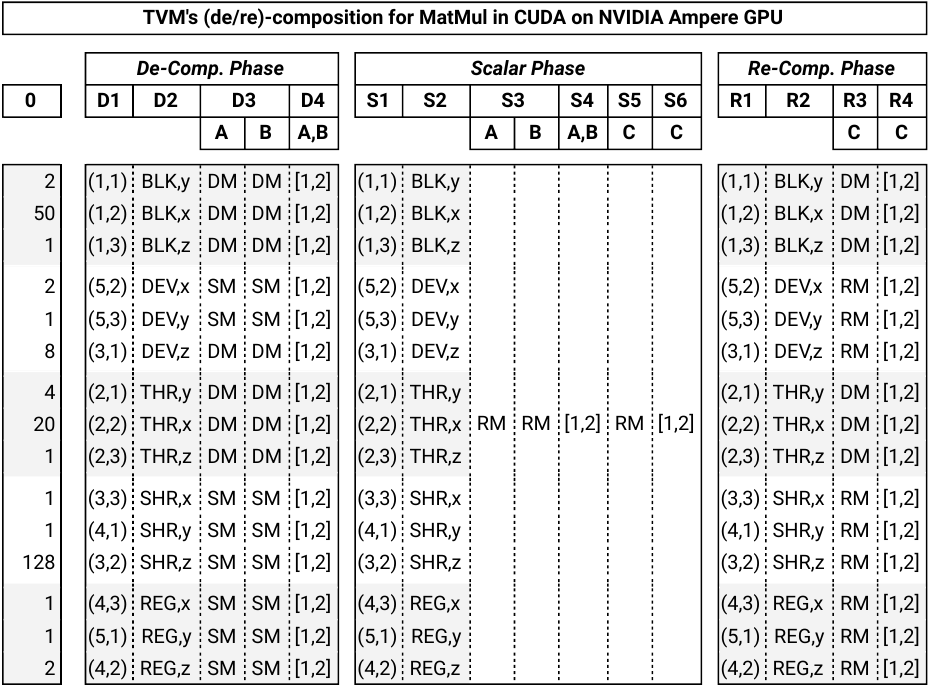}
    \caption{TVM's (de/re)-composition for \texttt{MatMul} in CUDA on GPU expressed in our low-level representation
    }
    \label{tp_ll_tab_tvm_gpu}
\end{figure}
\begin{figure}[hbpt]
    \centering
    \includegraphics[scale=0.7]{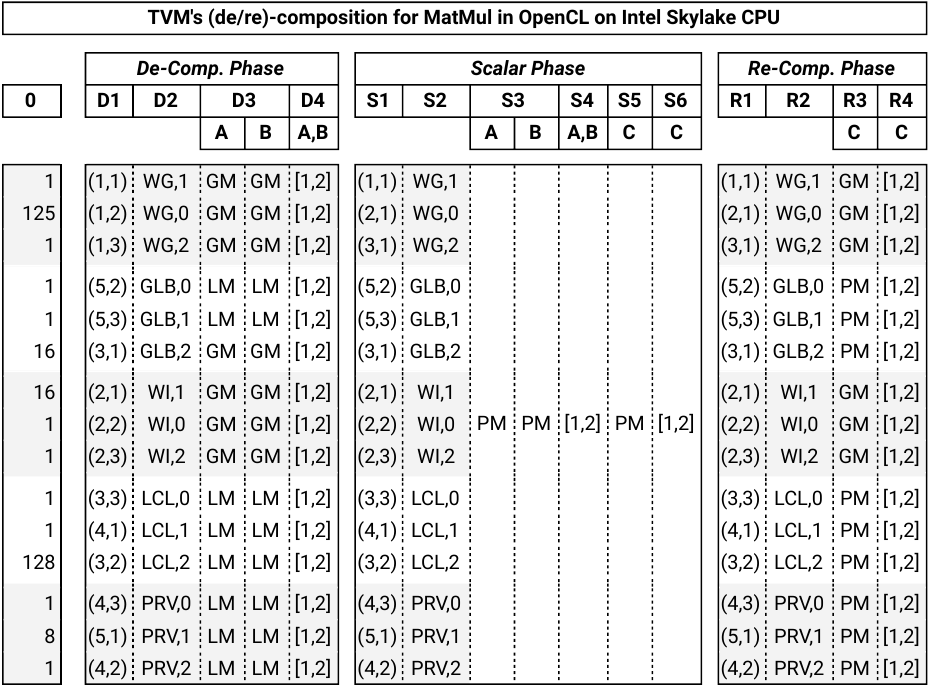}
    \caption{TVM's (de/re)-composition for \texttt{MatMul} in OpenCL on CPU expressed in our low-level representation
    }
    \label{tp_ll_tab_tvm_cpu}
\end{figure}

\begin{figure}[hbpt]
    \centering
    \includegraphics[scale=0.7]{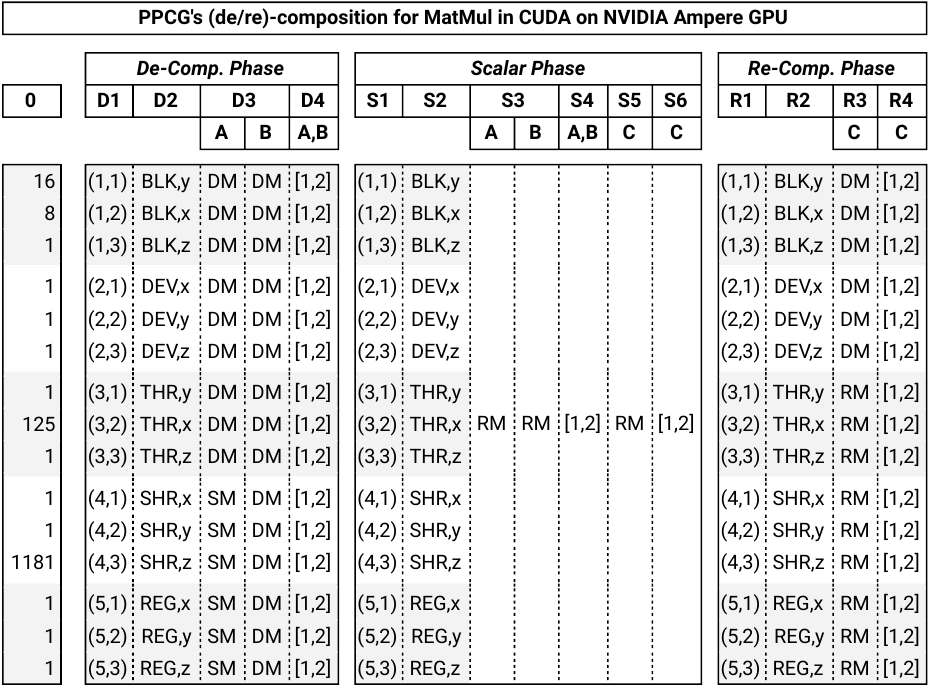}
    \caption{PPCG's (de/re)-composition for \texttt{MatMul} in CUDA on GPU expressed in our low-level representation
    }
    \label{tp_ll_tab_ppcg_gpu}
\end{figure}
\begin{figure}[hbpt]
    \centering
    \includegraphics[scale=0.7]{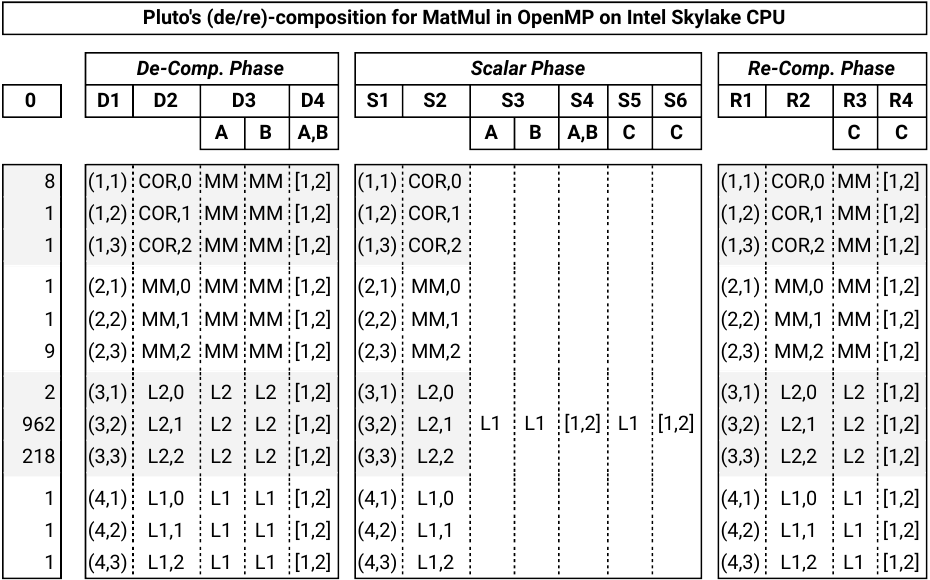}
    \caption{Pluto's (de/re)-composition for \texttt{MatMul} in OpenMP on CPU expressed in our low-level representation
    }
    \label{tp_ll_tab_pluto_cpu}
\end{figure}

\newpage
\phantom{xxx}
\newpage

\section{Lowering: From High Level to Low Level}
\label{ch:lowering}

We have designed our formalism such that an expression in our high-level representation (such as in Figure~\ref{fig:intro_example})
can be \emph{systematically lowered} to an expression in our low-level representation (as in Figure~\ref{fig_ll_example}).
We confirm this by parameterizing the generic high-level expression in Figure~\ref{fig_generic_hl}~--~step-by-step~--~in the tuning parameters listed in Table~\ref{tab_tps}, in a formally sound manner, which results exactly in the generic low-level expression in Figure~\ref{fig_gen_ll}.

Note that the tuning parameters in Table~\ref{tab_tps} can also be interpreted as parameters of the lowering process (instead of the low-level representation).
This is because in practice, our lowering process takes as input a particular configuration of the tuning parameter in Table~\ref{tab_tps} (automatically chosen via auto-tuning), such that it lowers a particular instance in our high-level representation (i.e., for a concrete choice of: scalar function, combine operator, etc) straight to a particular instance in our low-level representation (instead of lowering first to the generic low-level instance in Figure~\ref{fig_gen_ll} and then inserting tuning parameters in this generic instance).

\begin{proof}[Parameter~\texttt{0}]
Let ${^\downarrow}\MDA$ be the input MDA.
Let further be ${^\downarrow}\MDA_f$ the $L$-layered, $D$-dimensional, $P$-partitioned low-level MDA (according to Definition~\ref{def_ll_mda}), for
\[
  P :=
( \
  \underbrace{(\underbrace{\#\texttt{PRT}(1,1)}_\text{Dimension 1},\dotsc,\underbrace{\#\texttt{PRT}(1,D)}_\text{Dimension $D$})}_\text{Layer $1$}
    \ \ ,\,\dotsc\,, \ \
  \underbrace{(\underbrace{\#\texttt{PRT}(L,1)}_\text{Dimension 1},\dotsc,\underbrace{\#\texttt{PRT}(L,D)}_\text{Dimension $D$})}_\text{Layer $L$}
\ )
\]
where \#\texttt{PRT} denotes the number of partitions (Parameter~\texttt{0} in Table~\ref{tab_tps}), which is defined as:
\[
{^\downarrow}\MDA \ \ =: \ \ \ \
\underbrace{\underbrace{\concat{1}{p^1_1\in P^1_1}}_\text{Dimension $1$}\dotsc\underbrace{\concat{D}{p^1_D\in P^1_D}}_\text{Dimension $D$}}_\text{Layer $1$}
\ \ \ \ \dotsc \ \ \ \
\underbrace{\underbrace{\concat{1}{p^L_1\in P^L_1}}_\text{Dimension $1$}\dotsc\underbrace{\concat{D}{p^L_D\in P^L_D}}_\text{Dimension $D$}}_\text{Layer $L$}
\ \ \ \
{^\downarrow}\MDA_f^{<
\, p^1_1,\dotsc,p^1_D \ | \ \dotsc \ | \ p^L_1,\dotsc,p^L_D \,
>}
\]

Applying $L*D$ times the homomorphic property (Definition~\ref{def_md_hom}), we get:
\begin{align*}
&
{^\uparrow}\MDA \ \ := \ \
\texttt{md\_hom}( \ f \ , \ (\circledast_1,\dotsc, \circledast_D) \ )( \, {^\downarrow}\MDA \, )
\ = \\[7px]
& \hspace{60px}
\underbrace{\underbrace{\underset{p^1_1\in\#\texttt{PRT}(1,1)}{\circledast_1}}_\text{Dimension $1$}\dotsc\underbrace{\underset{p^1_D\in\#\texttt{PRT}(1,D)}{\circledast_D}}_\text{Dimension $D$}}_\text{Layer $1$}
\ \ \ \dotsc \ \ \
\underbrace{\underbrace{\underset{p^L_1\in\#\texttt{PRT}(L,1)}{\circledast_1}}_\text{Dimension $1$}\dotsc\underbrace{\underset{p^L_D\in\#\texttt{PRT}(L,D)}{\circledast_D}}_\text{Dimension $D$}}_\text{Layer $L$} \\[0px]
& \hspace*{150px}
\ \
\texttt{md\_hom}( \ f \ , \ (\circledast_1,\dotsc, \circledast_D) \ )( \
{^\downarrow}\MDA_f^{
<\, p^1_1,\dotsc,p^1_D \ | \ \dotsc \ | \ p^L_1,\dotsc,p^L_D \,>
}
\ )
\end{align*}
Since each part ${^\downarrow}\MDA_f^{<\, p^1_1,\dotsc,p^1_D \ | \ \dotsc \ | \ p^L_1,\dotsc,p^L_D \,>}$ contains a single scalar value only (according to the algorithmic constraint of Parameter~1, discussed in Section~\ref{ch_gen_ll_exp}), it holds
\[
  \texttt{md\_hom}( \ f \ , \ (\circledast_1,\dotsc, \circledast_D) \ )( \ {^\downarrow}\MDA_f^{<\, p^1_1,\dotsc,p^1_D \ | \ \dotsc \ | \ p^L_1,\dotsc,p^L_D \,>}\ ) = {^\uparrow}\MDA_f^{<\, p^1_1,\dotsc,p^1_D \ | \ \dotsc \ | \ p^L_1,\dotsc,p^L_D \,>}
\]
for
\[
    {^\uparrow}\MDA_f^{<\, p^1_1,\dotsc,p^1_D \ | \ \dotsc \ | \ p^L_1,\dotsc,p^L_D \,>} := \vec{f}( \ {^\downarrow}\MDA_f^{<\, p^1_1,\dotsc,p^1_D \ | \ \dotsc \ | \ p^L_1,\dotsc,p^L_D \,>} \ )
\]
and $\vec{f}$ defined as in Definition~\ref{def_md_hom}.
\end{proof}

\begin{proof}[Parameter~\texttt{D1}]
Parameter~\texttt{D1} reorders concatenation operators $\dplus_1,\dotsc,\dplus_D$ (Example~\ref{def:mda_concat}).
We prove our assumption w.l.o.g. for the case $D=2$; the general case $D\in\IN$ follows analogously.

Let $\dplus_{d_1}\in{\CO}^{<id\,|\,T\,|\,D\,|\,d_1>}$ and $\dplus_{d_2}\in{\CO}^{<id\,|\,T\,|\,D\,|\,d_2>}$ be two arbitrary concatenation operators that coincide in meta-parameters $T$ and $D$, but may differ in their operating dimensions $d_1$ and $d_2$.
We have to show
\[
    (a_1 \ \dplus_{d_1} \ a_2) \ \dplus_{d_2} \ (a_3 \ \dplus_{d_1} \ a_4)
    \ \ = \ \
    (a_1 \ \dplus_{d_2} \ a_3) \ \dplus_{d_1} \ (a_2 \ \dplus_{d_2} \ a_4)
\]
which follows from the definition of the concatenation operator $\dplus$ in Example~\ref{def:mda_concat}.
\end{proof}

\begin{proof}[Parameters~\texttt{D2},~\texttt{S2},~\texttt{R2}]
    These parameters replaces combine operators (Definition~\ref{def_combine_op}) by low-level combine operators (Definition~\ref{def_ll_comb}), which has no effect on semantics.
\end{proof}

\begin{proof}[Parameters~\texttt{D3},~\texttt{S3},~\texttt{S5},~\texttt{R3}]
    These parameters set the memory tags of low-level BUFs (Definition~\ref{ll_buffers}), which have no effect on semantics.
\end{proof}

\begin{proof}[Parameters~\texttt{D4},~\texttt{S4},~\texttt{S6},~\texttt{R4}]
    The parameters change the memory layout of low-level BUFs (Definition~\ref{ll_buffers}), which does not affect extensional equality.
\end{proof}

\begin{proof}[Parameters~\texttt{S1}]
    This parameter sets the order in which function $f$ is applied to parts, which is trivially sound for any order.
\end{proof}

\begin{proof}[Parameters~\texttt{R1}]
Similarly to parameter~\texttt{D1}, parameter~\texttt{R1} reorders combine operators $\co{1},\dotsc,\co{D}$, but
the operators
are not restricted to be concatenation.
We prove our assumption by exploiting the MDH property (Definition~\ref{def_mdh}) together with the proof of parameter~\texttt{D1}, as follows:
\begin{align*}
    & &                & & (  a_1   && \co{d_1}     &&  a_2  )  && \co{d_2}    &&   (  a_3  && \co{d_1}     &&  a_4  ) &&  \\
    &=&  \mdh(\dotsc)  &(& (  a_1   && \dplus_{d_1} &&  a_2  )  && \dplus_{d_2}&&   (  a_3  && \dplus_{d_1} &&  a_4  ) &&) \\
    &=&  \mdh(\dotsc)  &(& (  a_1   && \dplus_{d_2} &&  a_3  )  && \dplus_{d_1}&&   (  a_2  && \dplus_{d_2} &&  a_4  ) &&) \\
    &=&                & & (  a_1   && \co{d_2}     &&  a_3  )  && \co{d_1}    &&   (  a_2  && \co{d_2}     &&  a_4  ) &&  \checkmark
\end{align*}
\end{proof}

\section{Experimental Results}
\label{ch:eval}

We experimentally evaluate our approach by comparing it to popular representatives of four important classes:~%

\begin{enumerate}
\setlength\itemsep{0.3em}
    \item \emph{Scheduling Approach}: TVM~\cite{222575} which generates GPU and CPU code from programs expressed in TVM's own high-level program representation;

    \item \emph{Polyhedral Compilers}: PPCG~\cite{10.1145/2400682.2400713} for GPUs\footnote{
        We cannot compare to polyhedral compiler TC~\cite{10.1145/3355606} which is optimized toward deep learning computations on GPUs, because
        TC is not under active development anymore and thus is not working for newer CUDA architectures~\cite{tc_git}.
        \citet{8891668} show that our approach~--~already in its proof-of-concept version~--~achieves higher performance than TC for popular computations on real-world data sets.
    } and Pluto~\cite{bondhugula2008pluto} for CPUs, which automatically generate executable program code in CUDA (PPCG) or OpenMP (Pluto) from straightforward, unoptimized C programs;

    \item \emph{Functional Approach}: Lift~\cite{10.1145/2784731.2784754} which generates OpenCL code from a Lift-specific, functional program representation;

    \item \emph{Domain-Specific Libraries}: NVIDIA cuBLAS~\cite{cublas} and NVIDIA cuDNN~\cite{cudnn}, as well as Intel oneMKL~\cite{onemkl} and Intel oneDNN~\cite{onednn}, which offer the user easy-to-use, domain-specific building blocks for programming.
    The libraries internally rely on pre-implemented assembly code that is optimized by experts for their target application domains:~linear algebra (cuBLAS and oneMKL) or convolutions (cuDNN and oneDNN), respectively.
    To make comparison against the libraries challenging for us, we compare to all routines provided by the libraries.
    For example, the cuBLAS library offers three, semantically equal but differently optimized routines for computing \texttt{MatMul}:~\texttt{cublasSgemm} (the default \texttt{MatMul} implementation in cuBLAS), \texttt{cublasGemmEx} which is part of the \texttt{cuBLASEx} extension of cuBLAS~\cite{cublas_ex}, and the most recent \texttt{cublasLtMatmul} which is part of the \texttt{cuBLASLt} extension~\cite{cublas_lt};~each of these three routines may perform differently on different problem sizes (NVIDIA usually recommends to naively test which routine performs best for the particular target problem).
    To make comparison further challenging for us, we exhaustively test for each routine all of its so-called \texttt{cublasGemmAlgo\_t} variants, and report the routine's runtime for the best performing variant.
    In the case of oneMKL, we compare also to its \emph{JIT engine}~\cite{mkl_jit} which is specifically designed and optimized for small problem sizes.
    We also compare to library \emph{EKR}~\cite{hentschel2008krebsregister} which computes data mining example \texttt{PRL} (Figure~\ref{fig_hl_examples}) on CPUs~--~the library is implemented in the Java programming language and parallelized via \emph{Java Threads}, and the library is used in practice by the \emph{Epidemiological Cancer Registry} in North Rhine-Westphalia (Germany) which is the currently largest cancer registry in Europe.
\end{enumerate}

\vspace*{5px}

\noindent
We compare to the approaches experimentally in terms of:~%
\begin{itemize}
    \setlength\itemsep{0.3em}
    \item[i)] \emph{Performance}: via a runtime comparison of our generated code against code that is generated according to the related approaches;

    \item[ii)]\emph{Portability}: based on the \emph{Pennycook Metric}~\cite{PENNYCOOK2019947} which mathematically defines portability\footnote{
        Pennycook's metric is actually called \emph{Performance Portability (PP)}. Since performance portability particularly includes functional portability, we refer to Pennycook's PP also more generally as \emph{Portability} only.
    } as:
    \[
        \Phi(a,p,H) = \begin{cases}
            \frac{|H|}{\sum_{i\in H}\frac{1}{e_i(a,p)}} & \text{if $i$ is supported, $\forall i\in H$} \\
            0 & \text{otherwise}
        \end{cases}
    \]
    In words: "for a given set of platforms~$H$, the \emph{performance portability}~$\Phi$ of an application~$a$ solving problem~$p$ is defined as~$\Phi(a,p,H)$, where~$e_i(a,p)$ is the performance efficiency (i.e. a ratio of observed performance relative to some proven, achievable level of performance) of application~$a$ solving problem~$p$ on platform~$i$;~value $\Phi(a,p,H) $ is $0$, if any platform in $H$ is unsupported by $a$ running $p$."~\cite{PENNYCOOK2019947}.
    Consequently, Pennycook defines portability as
    a real value in the interval $[0,1]_\IR$ such that a value close to $1$ indicates \emph{high} portability and a value close to $0$ indicates \emph{low} portability.
    Here, platforms $H$ represents a set of devices (CPUs, GPUs, $\dotsc$), an application $a$ is in our context a framework (such as TVM, a polyhedral compiler, or our approach), problems $p$ are our case studies, and $e_i(a,p)$ is computed as the
    runtime $a^\text{best}_{p,i}$ of the application that achieves the best observed runtime for problem $p$ on platform $i$, divided by the
    runtime of application $a$ for problem $p$ running on platform $i$.

    \item[iii)] \emph{Productivity}: by intuitively arguing that our approach achieves the same/lower/higher productivity as the related approaches, using the representative example computation \emph{Matrix-Vector Multiplication (\texttt{MatVec})} (Figure~\ref{fig:intro_example}).
    Classical code metrics, such as \emph{Lines of Code (LOC)}, COCOMO~\cite{cocomo}, McCabe's cyclomatic complexity~\cite{1702388}, and Halstead development effort~\cite{halstead1977elements} are not meaningful for comparing the short and concise
    programs in high-level languages
    as proposed by
    the related work as well as our approach.

\end{itemize}

\vspace*{3px}

In the following, after discussing our application case studies, experimental setup, auto-tuning system, and code generator, we compare our approach to each of the four above mentioned classes of approaches (1)-(4) in Sections~\ref{ch:eval:scheduling}-\ref{ch:eval:domain}.

\subsection*{Application Case Studies}

We use for experiments in this section popular example computations from Figure~\ref{fig_hl_examples} that belong to different classes of computations:

\begin{itemize}
\setlength\itemsep{0.5em}
    \item {Linear Algebra Subroutines (BLAS)}: \emph{Matrix Multiplication} (\texttt{MatMul}) and \emph{Matrix-Vector Multiplication} (\texttt{MatVec});

    \item {Stencil Computations}: \emph{Jacobi Computation} (\texttt{Jacobi3D}) and \emph{Gaussian Convolution} (\texttt{Conv2D}) which differ from linear algebra routines by accessing neighboring elements in their input data;

    \item {Quantum Chemistry}: \emph{Coupled Cluster (\texttt{CCSD(T)})}
    computations which differ from linear algebra routines and stencil computations by accessing their high-dimensional input data in complex, transposed fashions;

    \item {Data Mining}: \emph{Probabilistic Record Linkage} (\texttt{PRL}) which differs from the previous computations by relying on a PRL-specific combine operator and scalar function (instead of straightforward additions or multiplications as the previous computations);

    \item {Deep Learning}: the most time-intensive computations within the popular neural networks \texttt{ResNet-50}~\cite{DBLP:journals/corr/HeZRS15}, \texttt{VGG-16}~\cite{https://doi.org/10.48550/arxiv.1409.1556}, and~\texttt{MobileNet}~\cite{DBLP:journals/corr/HowardZCKWWAA17}, according to their TensorFlow implementations~\cite{resnet_git,vgg_git,mobilenet_git}.
    Deep learning computations rely on advanced variants of linear algebra routines and stencil computations, e.g., \texttt{MCC} and \texttt{MCC\_Capsule} for computing convolution-like stencils, instead of the classical \texttt{Conv2D} variant of convolution (Figure~\ref{fig_hl_examples})~--~the deep learning variants are considered as significantly more challenging to optimize than their classical variants~\cite{10.1145/3317550.3321441}.
\end{itemize}
We use for experiments this subset of computations from Figure~\ref{fig_hl_examples} to make experimenting challenging for us:~%
the computations differ in major characteristics (as discussed in Section~\ref{ch_hl_examples}), e.g., accessing neighboring elements in their input data (as stencil computations) or not (as linear algebra routines), thus usually requiring fundamentally different kinds of optimizations.
Consequently, we consider it challenging for our approach to achieve high performance for our studies, because our approach relies on a generalized optimization process (discussed in Section~\ref{ch:lowering})
that uniformly applies to any kind of data-parallel computation and also parallel architecture.
In contrast, the optimization processes of the related approaches are often specially designed and tied to a particular application class and often also architecture.
For example, NVIDIA cuBLAS and Intel oneMKL are highly optimized specifically for linear algebra routines on either GPU or CPU, respectively, and TVM is specifically designed and optimized for deep learning computations.

To make experimenting further challenging for us, we consider data sizes and characteristics either taken from real-world computations (e.g., from the \emph{TCCG} benchmark suite~\cite{tccg2016a} for quantum chemistry computations) or sizes that are preferable for our competitors, e.g., powers of two for which
many competitors are highly optimized, e.g., vendor libraries.
For the deep learning case studies, we use data characteristics (sizes, strides, padding strategy, image/filter formats, etc.) taken from the particular implementations of the neural networks when computing the popular \emph{ImageNet}~\cite{NIPS2012_c399862d} data set (the particular characteristics are listed in our Appendix, Section~\ref{app_data_characteristics}, for the interested reader).
For all experiments, we use single precision floating point numbers (a.k.a. \texttt{float} or \texttt{fp32}), as such precision is the default in TensorFlow and many other frameworks.

\subsection*{Experimental Setup}

We run our experiments on a cluster containing two different kinds of GPUs and CPUs:
\begin{itemize}
\setlength\itemsep{0.2em}
    \item \texttt{NVIDIA Ampere GPU A100-PCIE-40GB}
    \item \texttt{NVIDIA Volta GPU V100-SXM2-16GB}
    \item \texttt{Intel Xeon Broadwell CPU E5-2683 v4 @ 2.10GHz}
    \item \texttt{Intel Xeon Skylake CPU Gold-6140 @ 2.30GHz}
\end{itemize}

\noindent
We represent the two CUDA GPUs in our formalism
using model $\texttt{ASM}_\texttt{CUDA+WRP}$ (Example~\ref{example_asm}).
We rely on model $\texttt{ASM}_\texttt{CUDA+WRP}$, rather than the CUDA's standard model $\texttt{ASM}_\texttt{CUDA}$ (also in Example~\ref{example_asm}), to exploit CUDA's (implicit) warp level for a fair comparison to the related approaches:~warp-level optimizations are exploited by the related approaches (such as TVM), e.g., for \emph{shuffle operations}~\cite{cuda_shuffles} which combine the results of threads within a warp with high performance.
To fairly compare our approach to TVM and PPCG, we avoid exploiting warps' \emph{tensor core intrinsics}~\cite{cuda_tensor_cores}, in all experiments, which compute the multiplication of small matrices with high performance~\cite{10.1145/3575693.3576933}, because these intrinsics are not used in the TVM- and PPCG-generated CUDA code.
For the two CPUs, we rely on model $\texttt{ASM}_\texttt{OpenCL}$ (Example~\ref{example_asm}) for generating OpenCL code.
The same as our approach, TVM also generates OpenCL code for CPUs; Pluto relies on the OpenMP approach to target CPUs.

For all experiments, we use the currently newest versions of frameworks, libraries, and compilers,
as follows.
We compile our generated GPU code using library \texttt{CUDA\;NVRTC}~\cite{nvrtc} from \texttt{CUDA\;Toolkit\;11.4}, and we use \texttt{Intel's OpenCL runtime} version \texttt{18.1.0.0920} for compiling CPU code.
For both compilers, we do not set any flags so that they run in their default modes.
For the related approaches, we use the following versions of frameworks, libraries, and compilers:

\newpage

\begin{itemize}
\setlength\itemsep{0.3em}
    \item \texttt{TVM}~\cite{tvm_git} version \texttt{0.8.0} which also uses our system's \texttt{CUDA\;Toolkit} version \texttt{11.4} for GPU computations and \texttt{Intel's runtime} version \texttt{18.1.0.0920} for computations on~CPU;

    \item \texttt{PPCG}~\cite{ppcg_git} version \texttt{0.08.04} using flag \texttt{--target=cuda} for generating CUDA code, rather than OpenCL, as CUDA is usually better performing than OpenCL on NVIDIA GPUs, and we use flag \texttt{--sizes} followed by auto-tuned tile sizes~--~we rely on the \emph{Auto-Tuning Framework (ATF)}~\cite{10.1145/3427093} for choosing optimized tile size values (as we discuss in the next subsection);

    \item \texttt{Pluto}~\cite{pluto_git} commit \texttt{12e075a} using flag \texttt{--parallel} for generating OpenMP-parallelized C code (rather than sequential C), as well as flag \texttt{--tile} to use ATF-tuned tile sizes for Pluto;~%
    the Pluto-generated OpenMP code is compiled via  \texttt{Intel's\;icx} compiler version \texttt{2022.0.0} using the Pluto-recommended optimization flags \texttt{-O3\;-qopenmp};

    \item \texttt{NVIDIA\;cuBLAS}~\cite{cublas} from \texttt{CUDA\;Toolkit\;11.4}, using the NVIDIA-recommended compiler flags \texttt{-fast\;-O3\;-DNDEBUG};

    \item \texttt{NVIDIA\;cuDNN}~\cite{cudnn} from \texttt{CUDA\;Toolkit\;11.4}, using the NVIDIA-recommended compiler flags \texttt{-fast\;-O3\;-DNDEBUG};

    \item \texttt{Intel\;oneMKL}~\cite{onemkl} compiled with \texttt{Intel's\;icpx} compiler version \texttt{2022.0.0}, using flags \texttt{-DMKL\_ILP64 -qmkl=parallel -L\$\{MKLROOT\}/lib/intel64 -liomp5 -lpthread -lm -ldl}, as recommended for \texttt{oneMKL} by Intel's \emph{Link Line Advisor} tool~\cite{link_line_advisor}, as well as standard flags \texttt{-O3\;-NDEBUG};

    \item \texttt{Intel\;oneDNN}~\cite{onednn} also compiled with \texttt{Intel's\;icpx} compiler version \texttt{2022.0.0}, using flags \texttt{-I\$\{DNNLROOT\}/include -L\$\{DNNLROOT\}/lib -ldnnl}, according to oneDNN's documentation, as well as standard flags \texttt{-03\;-NDEBUG};

    \item \texttt{EKR}~\cite{hentschel2008krebsregister} executed via \texttt{Java\;SE\;1.8.0\;Update\;281}.
\end{itemize}

We profile runtimes of CUDA and OpenCL programs
using the corresponding, event-based profiling APIs provided by CUDA and OpenCL.
For \texttt{Pluto} which generates OpenMP-annotated C code, we measure runtimes via system call \texttt{clock\_gettime}~\cite{c_profiling}.
In the case of C++ libraries \texttt{Intel\:oneMKL} and \texttt{Intel\:oneDNN}, we use the C++ \texttt{chrono} library~\cite{cpp_profiling} for profiling.
Libraries \texttt{NVIDIA\:cuBLAS} and \texttt{NVIDIA\:cuDNN} are also based on the CUDA programming model;~thus, we profile them also via CUDA events.
To measure the runtimes of the \texttt{EKR} Java library, we use Java function \texttt{System.currentTimeMillis()}.

All measurements of CUDA and OpenCL programs contain the pure program runtime only (a.k.a. \emph{kernel runtime}).
The runtime of \emph{host code}\footnote{
    \label{footnote_hostcode}
    \emph{Host code} is required in approaches CUDA and OpenCL for program execution:~it compiles the CUDA and OpenCL programs, performs data transfers between host and device, etc. We rely on the high-level library \emph{dOCAL}~\cite{rasch_docal,8644541} for host code programming in this work.
} is not included in the reported runtimes, as performance of host code is not relevant for this work and the same for all approaches.

In all experiments, we collect measurements until the 99\% confidence interval was within 5\% of our reported means, according to the guidelines for \emph{scientific benchmarking of parallel computing systems} by \citet{10.1145/2807591.2807644}.

\subsection*{Auto-Tuning}

The auto-tuning process of
our approach relies on the generic \emph{Auto Tuning Framework~(ATF)}~\cite{10.1145/3427093}.
The ATF framework has proven to be efficient for exploring
large search spaces of constrained tuning parameters (as our space introduced in Section~\ref{ch_gen_ll_exp}).
We use~ATF, out of the box, exactly as described by~\citet{10.1145/3427093}:~%
1)~we straightforwardly represent in ATF our search space (Table~\ref{tab_tps}) via \emph{tuning parameters} which express the parameters in the table and their constraints;
2)~we use ATF's pre-implemented cost functions for CUDA and OpenCL to measure the cost of our generated OpenCL and CUDA codes (in this paper, we consider as cost program's runtime, rather than its energy consumption or similar);
3)~we start the tuning process using ATF's default search technique (\emph{AUC bandit}~\cite{10.1145/2628071.2628092}).
ATF then fully automatically determines a well-performing tuning parameter configuration for the particular combination of a case study, architecture, and input/output characteristics (size, memory layout,~etc).

For scheduling approach \texttt{TVM}, we use its \emph{Ansor}~\cite{258858} optimization engine which is specifically designed and optimized toward generating optimized \texttt{TVM} schedules.
Polyhedral compilers \texttt{PPCG} and \texttt{Pluto} do not provide own auto-tuning systems;~thus, we use for them also ATF for auto-tuning, the same as for our approach.
For both compilers, we additionally also report their runtimes when relying on their internal heuristics, rather than on auto-tuning, \mbox{to fairly compare to them.}

To achieve the best possible performance results for \texttt{TVM}, \texttt{PPCG}, and \texttt{Pluto}, we auto-tune each of these frameworks individually, for each particular combination of case study, architecture, and input/output characteristics, the same as for our approach.
For example, we start for \texttt{TVM} one tuning run when auto-tuning case study \texttt{MatMul} for the \texttt{NVIDIA\:Ampere\:GPU} on one input size, and another, new tuning run for a new input size, etc.

Hand-optimized libraries \texttt{NVIDIA\:cuBLAS/cuDNN} and \texttt{Intel\:oneMKL/oneDNN} rely on heuristics provided by experts, rather than auto-tuning.
By relying on heuristics, the libraries avoid the time-intensive process of auto-tuning.
However, auto-tuning is well amortized in many application areas (e.g., deep learning), because the auto-tuned implementations are re-used in many program runs.
Moreover, auto-tuning avoids the complex and costly process of hand optimization by experts, and it often achieves higher performance than hand-optimized code (as we confirm later in our experiments), because well-performing optimizations are often not intuitive.

For a fair comparison,
we use for each tuning run uniformly the same tuning time of $12$h.
Even though for many computations well-performing tuning results could often be found in less than $12$h, for our approach as well as for other frameworks, we use such generous tuning time for all frameworks
to avoid auto-tuning issues in our reported results~--~analyzing, improving, and accelerating the auto-tuning process is beyond the scope of this work and intended for our future work (as also outlined in Section~\ref{ch:fw}).
In particular, TVM's Ansor optimizer was often able to find well performing optimizations in $6$h of tuning time or less.
This is because Ansor explores a small search space that is designed and optimized for deep learning computations~--~Ansor's space is a proper subset of our space, as our space aims to capture general optimizations that apply to arbitrary data-parallel computations.
However, the focus on deep learning causes Ansor to have difficulties with optimizing computations not taken from the deep learning area, \mbox{as we confirm in our experiments.}

To improve the auto-tuning efficiency for our implementations, we rely on a straightforward
cost model that shrinks our search space in Table~\ref{tab_tps} before starting our ATF-based auto-tuning process:~%
i)~we always use the same values for Parameters~\texttt{D1},~\texttt{S1},~\texttt{R1} as well as for Parameters~\texttt{D2},~\texttt{S2},~\texttt{R2}, thereby generating the same loop structure for all three phases (de-composition, scalar, and re-composition) such that the structures can be generated as a fused loop nest;~%
ii)~we restrict Parameters~\texttt{D2},~\texttt{S2},~\texttt{R2} to two values~--~one value that let threads process outer parts (a.k.a. \emph{blocked access} or \emph{outer parallelism}, respectively) and one to let threads process inner parts (\emph{strided access} or \emph{inner parallelism});~all other permutations are currently ignored for simplicity or because they have no effect on the generated code (e.g., permutations of Parameters~\texttt{D2},~\texttt{S2},~\texttt{R2} that only differ in dimension tags belonging to memory layers, as discussed in the previous section);~%
iii)~we restrict Parameters~\texttt{D3},~\texttt{S3},~\texttt{S5},~\texttt{R3} such that each parameter is invariant under different values of $d$ of its input pairs $(l,d)\in\MDHLVL$, i.e., we always copy full tiles in memory regions (and not a full tile of one input buffer and a half tile of another input buffer, which sometimes might achieve higher performance when memory is a limited resource).

\newpage

Our cost model is straightforward and might filter out configurations from our search space that achieve potentially higher performance than we report for our approach in Sections~\ref{ch:eval:scheduling}-\ref{ch:eval:domain}.
We aim to substantially improve our naive cost model in future work, based on \emph{operational semantics} for our low-level representation, to improve the auto-tuning quality and to reduce (or even avoid) tuning time.

\subsection*{Code Generation}

We provide an open source \emph{MDH compiler}~\cite{mdh_in_python} for generating executable program code from expressions in our high-level representation (as illustrated in Figure~\ref{contributions_overall}).
Our compiler takes as input the high-level representation of the target computation (Figure~\ref{fig_hl_examples}), in the form of a Python program,
and it fully automatically generates auto-tuned program code, based on the concepts and methodologies introduced and discussed in this paper and the ATF auto-tuning framework~\cite{10.1145/3427093}.

In our future work, we aim to integrate our code generation approach into the \emph{MLIR} compiler framework~\cite{9370308}, building on work-in-progress results~\cite{mlir_mdh_meeting}, thereby making our work better accessible to the community.
We consider approaches such as \emph{AnyDSL}~\cite{10.1145/3276489} and \emph{BuildIt}~\cite{9370333} as further, interesting frameworks in which our compiler could be implemented.

\subsection{Scheduling Approaches}
\label{ch:eval:scheduling}

\subsubsection*{Performance}

Figures~\ref{img_eval_la}-\ref{img_eval_dl_cpu} report the performance of the \texttt{TVM}-generated code, which is in CUDA for GPUs and in OpenCL for CPUs.
We observe that we usually achieve with our approach the high performance of \texttt{TVM} and often perform even better.
For example, in Figure~\ref{img_eval_dl_gpu}, we achieve a speedup $>2\times$ over \texttt{TVM} on \texttt{NVIDIA\:Ampere\:GPU} for matrix multiplications as used in the inference phase of the \texttt{ResNet-50} neural network~--~an actually favorable example for \texttt{TVM} which is designed and optimized toward deep learning computations executed on modern NVIDIA GPUs.
Our performance advantage over \texttt{TVM} is because we parallelize and optimize more efficiently
reduction-like computations~--~in the case of \texttt{MatMul} (Figure~\ref{fig_hl_examples}), its $3$rd-dimension (a.k.a. $k$-dimension).
The difficulties of TVM with reduction computations become particularly obvious when computing dot products (\texttt{Dot}) on GPUs (Figure~\ref{img_eval_la}):~the \texttt{Dot}'s main computation part is a reduction computation (via point-wise addition, see Figure~\ref{fig_hl_examples}), thus requiring reduction-focussed optimization, in particular when targeting the highly parallel architecture of GPUs:~%
in the case of \texttt{Dot} (Figure~\ref{img_eval_la}), our generated CUDA code exploits parallelization over CUDA blocks, whereas the Ansor-generated TVM code exploits parallelization over threads within in a single block only, because \texttt{TVM} currently cannot use blocks for parallelizing reduction computations~\cite{tvm_issue_sync_blocks}.
Furthermore, while \texttt{TVM}'s Ansor rigidly parallelizes outer dimensions~\cite{258858}, our ATF-based tuning process has auto-tuned our tuning parameters~\texttt{D2},~\texttt{S2},~\texttt{R2} in Table~\ref{tab_tps} to exploit parallelism for inner dimensions, which achieves higher performance for this particular \texttt{MatMul} example used in \texttt{ResNet-50}.
Also, for \texttt{MatMul}-like computations, Ansor always caches parts of the input in GPU's shared memory, and it computes these cached parts always in register memory.
In contrast, our caching strategy is auto-tunable (via parameters~\texttt{D3},~\texttt{S3}~\texttt{S5},~\texttt{R3} in Table~\ref{tab_tps}), and ATF has determined to not cache the input matrices into fast memory resources for the \texttt{MatMul} example in \texttt{ResNet-50}.
Surprisingly, Ansor does not exploit fast memory resources for Jacobi stencils (Figure~\ref{img_eval_stencils}), as required to achieve high performance for them:~our generated and auto-tuned CUDA kernel for Jacobi uses register memory for both inputs (image buffer and filter) when targeting \texttt{NVIDIA\:Ampere\:GPU} (small input size), thereby achieving a speedup over TVM+Ansor of $1.93\times$ for Jacobi.
Most likely, Ansor fails to foresee the potential of exploiting fast memory resources for Jacobi stencils, because Jacobi's index functions used for memory accesses (Figure~\ref{fig_hl_examples}) are injective.
For the \texttt{MatMul} example of \texttt{ResNet-50}'s training phase (Figure~\ref{img_eval_dl_gpu}), we achieve a speedup over \texttt{TVM} on \texttt{NVIDIA\:Ampere\:GPU} of $1.26\times$, because auto-tuning determined to store parts of input matrix $A$ as transposed into fast memory (via parameter~\texttt{D4} in Table~\ref{tab_tps}).
Storing parts of the input/output data as transposed is not considered by Ansor as optimization, perhaps because such optimization must be expressed in TVM's high-level language, rather than in its scheduling language~\cite{tvm_issue_layout}.
For \texttt{MatVec} on \texttt{NVIDIA\:Ampere\:GPU} (Figure~\ref{img_eval_la}), we achieve a speedup over \texttt{TVM} of $1.22\times$ for the small input size, by exploiting a so-called \emph{swizzle pattern}~\cite{10.1145/3297858.3304059}:~our ATF tuner has determined to assign threads that are consecutive in CUDA's \texttt{x}-dimension to the second MDA dimension (via parameters~\texttt{D2},~\texttt{S2},~\texttt{R2} in Table~\ref{tab_tps}), thereby accessing the input matrix in a GPU-efficient manner (a.k.a \emph{coalesced global memory accesses}~\cite{cuda-programming-guide}).
In contrast, for \texttt{MatVec} computations, Ansor assigns threads with consecutive \texttt{x}-ids always to the first data dimension, in a non-tunable manner, causing lower performance.

Our positive speedups over \texttt{TVM} on CPU are for the same reasons as discussed above for GPU.
For example, we achieve a speedup of $>3\times$ over \texttt{TVM} on \texttt{Intel\:Skylake\:CPU} for \texttt{MCC} (Figure~\ref{img_eval_dl_cpu}) as used in the training phase of the \texttt{MobileNet} neural network, because we exploit fast memory resources more efficiently than TVM:~our auto-tuning process has determined to use register memory for the \texttt{MCC}'s second input (the filter buffer \texttt{F}, see Table~\ref{fig_hl_examples}) and using no fast memory for the first input (image buffer \texttt{I}), whereas Ansor uses shared memory rigidly for both inputs of \texttt{MCC}.
Moreover, our auto-tuning process has determined to parallelize the inner dimensions of \texttt{MCC}, while Ansor always parallelizes outer dimensions.
We achieve the best speedup over \texttt{TVM} for \texttt{MCC} on an input size taken from \texttt{TVM}'s own tutorials~\cite{tvm_mcc_input_size} (Figure~\ref{img_eval_stencils}), rather than from neural networks (as in Figures~\ref{img_eval_dl_gpu} and~\ref{img_eval_dl_cpu}).
This is because \texttt{TVM}'s \texttt{MCC} size includes large reduction computations, which are not efficiently optimized by TVM (as discussed above).

The \texttt{TVM} compiler achieves higher performance than our approach for some examples in Figures~\ref{img_eval_la}-\ref{img_eval_dl_cpu}.
However, in most cases, this is for a technical reason only:~\texttt{TVM} uses the \texttt{NVCC} compiler for compiling CUDA code, whereas our proof-of-concept code generator currently relies on NVIDIA's \texttt{NVRTC} library which surprisingly generates less efficient CUDA assembly than \texttt{NVCC}.
In three cases, the higher performance of \texttt{TVM} over our approach is because our ATF auto-tuning framework was not able to find a better performing tuning configuration than \texttt{TVM}'s Ansor optimization engine during our $12$h tuning time;~the three cases are:~%
1)~\texttt{MCC}                   from \texttt{VGG-16}'s inference phase on \texttt{NVIDIA\:Ampere\:GPU} (Figure~\ref{img_eval_dl_gpu}),~%
2)~\texttt{MCC} (capsule variant) from \texttt{VGG-16}'s training phase on \texttt{NVIDIA\:Ampere\:GPU} (Figure~\ref{img_eval_dl_gpu}), and~%
3)~\texttt{MCC} (capsule variant) from \texttt{ResNet-50}'s training phase on \texttt{Intel\:Skylake\:CPU} (Figure~\ref{img_eval_dl_cpu}).
However, when we manually set the Ansor-found tuning configurations also for our approach (analogously as done in Section~\ref{ll_examples_matmul_resnet}), instead of using the ATF-found configurations, we achieve for these three cases exactly the same high performance as \texttt{TVM}+Ansor, i.e., the well-performing configurations are contained in our search space (Table~\ref{tab_tps}).
Most likely, Ansor was able to find this well-performing configuration within the $12$h tuning time, because it explores a significantly smaller search space that is particularly designed for deep learning computations.
To avoid such tuning issues in our approach, we aim to substantially improve our auto-tuning process in future work:~we plan to introduce an analytical cost model that assists (or even replaces) our auto-tuner, as we also outline in Section~\ref{ch:fw}.

Note that the \texttt{TVM} compiler crashes
for our data mining example \texttt{PRL}, because \texttt{TVM} has difficulties with computations relying on user-defined combine operators~\cite{tvm_issue_cos_3}.

\newpage

\subsubsection*{Portability}
Figure~\ref{img_eval_port_cpu_gpu} reports the portability of the \texttt{TVM} compiler.
Our portability measurements are based on the Pennycook metric where a value close to $1$ indicates high portability and a value close to $0$ indicates low portability, correspondingly.
We observe that except for the example of transposed matrix multiplication $\texttt{GEMM}^\texttt{T}$, we always achieve higher portability than \texttt{TVM}.
The higher portability of \texttt{TVM} for $\texttt{GEMM}^\texttt{T}$ is because \texttt{TVM} achieves for this example higher performance than our approach on \texttt{NVIDIA\:Volta\:GPU}.
However, the higher performance of \texttt{TVM} is only due to the fact that \texttt{TVM} uses NVIDIA's \texttt{NVCC} compiler for compiling CUDA code, while we currently rely on NVIDIA's \texttt{NVRTC} library which surprisingly generates less efficient CUDA assembly, as discussed above.

\subsubsection*{Productivity}

Listing~\ref{lst_matmul_tvm} shows how matrix-vector multiplication (\texttt{MatVec}) is implemented in TVM's high-level program representation which is embedded into the Python programming language.
In line~\ref{lst_matmul_tvm_1}, the input size $(I,K)\in\IN\times\IN$ of matrix $M\in T^{I\times K}$ (line~\ref{lst_matmul_tvm_2}) and vector $v\in T^K$ (line~\ref{lst_matmul_tvm_3}) are declared, in the form of function parameters;~the matrix and vector are named \texttt{M} and \texttt{v} and both are assumed to contain elements of scalar type $T=\texttt{float32}$ (floating point numbers).
Line~\ref{lst_matmul_tvm_4} defines a so-called \emph{reduction axis} in \texttt{TVM}, in which all values are combined in line~\ref{lst_matmul_tvm_7} via \texttt{te.sum} (addition).
The basic computation part of \texttt{MatVec}~--~multiplying matrix element \texttt{M[i,k]} with vector element \texttt{v[k]}~--~is also specified in line~\ref{lst_matmul_tvm_7}.

\begin{lstlisting}[
    language=C,
    morekeywords={__kernel__},
    mathescape=true,
    numbers=left,
    xleftmargin=1.5em,
    float=h!,
    frame=lines,
    caption={\texttt{TVM} program expressing Matrix-Vector Multiplication (\texttt{MatVec})},
    label={lst_matmul_tvm}
    ]
def MatVec(I, K):$\label{lst_matmul_tvm_1}$
    M = te.placeholder((I, K), name='M', dtype='float32')$\label{lst_matmul_tvm_2}$
    v = te.placeholder((K,), name='v', dtype='float32')$\label{lst_matmul_tvm_3}$

    k = te.reduce_axis((0, K), name='k')$\label{lst_matmul_tvm_4}$
    w = te.compute($\label{lst_matmul_tvm_5}$
        (I,),$\label{lst_matmul_tvm_6}$
        lambda i: te.sum(M[i, k] * v[k], axis=k)$\label{lst_matmul_tvm_7}$
    )
    return [M, v, w]$\label{lst_matmul_tvm_9}$
  \end{lstlisting}

While we consider the \texttt{MatVec} implementations of \texttt{TVM} (Listing~\ref{lst_matmul_tvm}) and our approach (Figure~\ref{fig:intro_example}) basically on the same level of abstraction, we consider our approach as more expressive in general.
This is because our approach supports multiple reduction dimensions that may rely on different combine operators, e.g., as required for expressing the \texttt{MBBS} example in Figure~\ref{fig_hl_examples}.
In contrast, \texttt{TVM} is struggling with different combine operators~--~adding support for multiple, different reduction dimensions is considered in the TVM community as a non-trivial extension of TVM~\cite{tvm_issue_cos_1,tvm_issue_cos_2}.
Also, we consider our approach as slightly less error-prone:~%
we automatically compute the expected sizes of matrix $M$ (as $I\times K$) and vector $v$ (as $K$), based on the user-defined input size $(I,K)$ in line~\ref{lst_matmul_tvm_1} and index functions $(i,k)\mapsto(i,k)$ for the matrix and $(i,k)\mapsto(k)$ for the vector in line~\ref{lst_matmul_tvm_7} (see Definition~\ref{def_iv}).
In contrast, \texttt{TVM} redundantly requests these matrix and vector sizes from the user:~once in lines~\ref{lst_matmul_tvm_2} and~\ref{lst_matmul_tvm_3} of Listing~\ref{lst_matmul_tvm}, and again in lines~\ref{lst_matmul_tvm_4} and~\ref{lst_matmul_tvm_6}.
\texttt{TVM} uses these sizes for generating the function specification of its generated \texttt{MatVec} code, which lets \texttt{TVM} generate incorrect low-level code~--~without issuing an error message~--~when the user sets non-matching sizes in lines~\ref{lst_matmul_tvm_2}/\ref{lst_matmul_tvm_3} and lines~\ref{lst_matmul_tvm_4}/\ref{lst_matmul_tvm_6}

\newpage

\subsection{Polyhedral Compilers}
\label{ch:eval:polyhedral}

\subsubsection*{Performance}

Figures~\ref{img_eval_la}-\ref{img_eval_dl_cpu} report
the performance achieved by the \texttt{PPCG}-generated CUDA code for GPUs and of the OpenMP-annotated C code generated by polyhedral compiler \texttt{Pluto} for CPUs.
For a fair comparison, we report for both polyhedral compilers their performance achieved for ATF-tuned tile sizes (denoted as \texttt{PPCG+ATF}/\texttt{Pluto+ATF} in the figures), as well as the performance of the two compilers when relying on their internal heuristics instead of auto-tuning (denoted as \texttt{PPCG} and \texttt{Pluto}).
In some cases, PPCG's heuristic crashed with error \texttt{"too\;many\;resources\;requested\;for launch"}, because the heuristic seems to not take into account device-specific constraints, e.g., limited availability of GPUs' fast memory resources.

We observe from Figures~\ref{img_eval_la}-\ref{img_eval_dl_cpu} that in all cases, our approach achieves better performance than \texttt{PPCG} and \texttt{Pluto}~--~sometimes by multiple orders of magnitude, in particular for deep learning computations (Figures~\ref{img_eval_dl_gpu} and~\ref{img_eval_dl_cpu}).
This is caused by the rigid optimization goals of \texttt{PPCG} and \texttt{Pluto}, e.g., always parallelizing outer dimensions, which causes severe performance losses.
For example, we achieve a speedup over \texttt{PPCG} of ${>13\times}$ on \texttt{NVIDIA\:Ampere\:GPU} and of ${>60\times}$ over \texttt{Pluto} on \texttt{Intel\:Skylake\:CPU} for \texttt{MCC} as used in the inference phase of the real-world \texttt{ResNet-50} neural network.
Compared to \texttt{PPCG}, our better performance for this \texttt{MCC} example is because \texttt{PPCG} has difficulties with efficiently parallelizing computations relying on more than $3$ dimension.
Most likely, this is because CUDA offers per default  $3$ dimensions for parallelization (called \texttt{x},~\texttt{y},~\texttt{z} dimension in CUDA).
However, \texttt{MCC} relies on $7$ parallelizable dimensions (as shown in Figure~\ref{fig_hl_examples}), and exploiting the parallelization opportunities of the $4$ further dimensions (as done in our generated CUDA code) is essential to achieve high performance for this
\texttt{MCC} example from \texttt{ResNet-50}.
Our performance advantage over \texttt{Pluto} for the \texttt{MCC} example is because \texttt{Pluto} parallelizes the outer dimensions of \texttt{MCC} only (whereas our approach has the potential to parallelize all dimensions);~however, the dimension has a size of only $1$ for this real-world example, resulting in starting only $1$ thread in the \texttt{Pluto}-generated OpenMP code.

For dot products \texttt{Dot} (Figure~\ref{img_eval_la}), we can observe that \texttt{PPCG} fails to generate parallel CUDA code, because \texttt{PPCG} cannot parallelize and optimize computations which rely solely on combine operators different from concatenation, as we also discuss in Section~\ref{sec_rw_poly}.
In Section~\ref{sec_rw_poly}, we particularly discuss that we do not consider the performance issues of \texttt{PPCG} and \texttt{Pluto} as weaknesses of the polyhedral approach in general, but of the particular polyhedral transformations chosen for \texttt{PPCG} and \texttt{Pluto}.

Note that \texttt{Pluto} crashes for our data mining example (Figure~\ref{img_eval_prl}), with \texttt{"Error\;extracting polyhedra\;from\;source\;file"}, because the scalar function of this example is too complex for \texttt{Pluto} (it contains \texttt{if}-statements).
Moreover, Intel's \texttt{icx} compiler struggles with compiling the \texttt{Pluto}-generated OpenMP code for quantum chemistry computations (Figure~\ref{img_eval_quantum}):~we aborted \texttt{icx}'s compilation process
after $24$h compilation time.
The \texttt{icx}'s issue with the \texttt{Pluto}-generated code
is most likely because of too aggressive loop unrolling of \texttt{Pluto}~--~the \texttt{Pluto}-generated OpenMP code has often a size $>50$MB for our real-world quantum chemistry examples.

\subsubsection*{Portability}

Since \texttt{PPCG} and \texttt{Pluto} are each designed for particular architectures only, they achieve the lowest portability of $0$ for all our studies in Figure~\ref{img_eval_port_cpu_gpu}, according to the Pennycook metric.
To simplify for \texttt{PPCG} and \texttt{Pluto} the portability comparison with our approach, we compute the Pennycook metric additionally also for two restricted sets of devices:~
only GPUs to make comparison against our approach easier for \texttt{PPCG}, and only CPUs to make comparison easier for~\texttt{Pluto}.

Figures~\ref{img_eval_port_la}-\ref{img_eval_port_dl} report the portability of \texttt{PPCG} when considering only GPUs, as well as the portability of \texttt{Pluto} for only CPUs.
We observe that we achieve higher portability for all our studies, as we constantly achieve higher performance than the two polyhedral compilers for the studies.

Note that even when restricting our set of devices to only GPUs for \texttt{PPCG} or only CPUs for \texttt{Pluto}, the two polyhedral compilers still achieve a portability of $0$ for some examples, because they fail to generate code for them (as discussed above).

\subsubsection*{Productivity}

Listing~\ref{lst_matmul_tc} shows the input program of polyhedral compilers \texttt{PPCG} and \texttt{Pluto} for \texttt{MatVec}.
Both take as input easy-to-implement, straightforward, sequential C code.
We consider these two polyhedral compilers as more productive than our approach (as well as scheduling and functional approaches, and also polyhedral compilers that take DSL programs as input, such as TC~\cite{10.1145/3355606}), because both compilers fully automatically generate optimized parallel code from unoptimized, sequential program code.

\citet{mdpoly,mdpoly_src} show that our approach can achieve the same high user productivity as polyhedral compilers, by using a polyhedral frontend for our approach:~we can alternatively take as input the same sequential program code as \texttt{PPCG} and \texttt{Pluto}, instead of programs implemented in our high-level program representation (as in Figure~\ref{fig:intro_example}).
The sequential input program is then transformed via polyhedral tool \emph{pet}~\cite{verdoolaege2012polyhedral} to its polyhedral representation which is then automatically transformed to our high-level program representation, according to the methodology presented by \citet{mdpoly,mdpoly_src}.

\begin{lstlisting}[
    language=C,
    morekeywords={__kernel__},
    mathescape=true,
    numbers=left,
    xleftmargin=1.5em,
    float=h!,
    frame=lines,
    caption={\texttt{PPCG}/\texttt{Pluto} program expressing Matrix-Vector Multiplication (\texttt{MatVec})},
    label={lst_matmul_tc}
    ]
for( int i = 0 ; i < I ; ++i )
  for( int k = 0 ; k < K ; ++k )
      w[i] += M[i][k] * v[k];
  \end{lstlisting}

\subsection{Functional Approaches}
\label{ch:eval:functional}

Our previous work~\cite{8891668} already shows that while functional approaches provide a solid formal foundation for computations, they typically suffer from performance and portability issues.
For this, our previous work compares our approach (in its original, proof-of-concept implementation~\cite{8891668}) to the state-of-the-art \texttt{Lift}~\cite{10.1145/2784731.2784754} framework which, to the best of our knowledge, has so far not been improved toward higher performance and/or better portability.
Consequently, we refrain from a further performance and portability evaluation of \texttt{Lift} and focus in the following on analyzing and discussing the productivity potentials of functional approaches, using again the state-of-the-art \texttt{Lift} approach as running example.
In Section~\ref{ch:rw:functional}, we discuss the performance and portability issues of functional approaches from a general perspective.

\paragraph*{Performance/Portability} Already experimentally evaluated in previous work~\cite{8891668} and discussed in general terms in Section~\ref{ch:rw:functional}.

\paragraph*{Productivity}

Listing~\ref{lst_matmul_lift} shows how \texttt{MatVec} is implemented in \texttt{Lift}.
In line~\ref{lst_matmul_lift_1}, type parameters \texttt{n} and \texttt{m} are declared, via the \texttt{Lift} building block \texttt{nFun}.
Line~\ref{lst_matmul_lift_2} declares a function \texttt{fun} that takes as input a matrix of size $\texttt{m}\times\texttt{n}$ and a vector of size \texttt{n}, both consisting of floating point numbers (\texttt{float}).
The computation of \texttt{MatVec} is specified in lines~\ref{lst_matmul_lift_3} and~\ref{lst_matmul_lift_4}.
In line~\ref{lst_matmul_lift_3}, \texttt{Lift}'s \texttt{map} pattern iterates over all rows of the matrix, and the \texttt{zip} pattern in line~\ref{lst_matmul_lift_4} combines each row pair-wise with the input vector.
Afterward, multiplication \texttt{*} is applied to each pair, using \texttt{Lift}'s \texttt{map} pattern again, and the obtained products are finally combined via addition \texttt{+} using \texttt{Lift}'s \texttt{reduce} pattern.

\begin{lstlisting}[
    language=C,
    morekeywords={nFun,fun,map,zip,reduce},
    mathescape=true,
    numbers=left,
    xleftmargin=1.5em,
    float=h!,
    frame=lines,
    caption={\texttt{Lift} program expressing Matrix-Vector Multiplication (\texttt{MatVec})},
    label={lst_matmul_lift}
    ]
nFun(n => nFun(m =>$\label{lst_matmul_lift_1}$
  fun(matrix: [[float]n]m => fun(xs: [float]n =>$\label{lst_matmul_lift_2}$
    matrix :>> map(fun(row =>$\label{lst_matmul_lift_3}$
      zip(xs, row) :>> map(*) :>> reduce(+, 0)$\label{lst_matmul_lift_4}$
    )) )) ))$\label{lst_matmul_lift_5}$
\end{lstlisting}

\newpage

Already for expressing \texttt{MatVec}, we can observe that \texttt{Lift} relies on a vast set of small, functional building blocks (five building blocks for \texttt{MatVec}:~\texttt{nFun}, \texttt{fun}, \texttt{map}, \texttt{zip}, and \texttt{reduce}), and the blocks have to be composed and nested in complex ways for expressing computations.
Consequently, we consider programming in \texttt{Lift} and \texttt{Lift}-like approaches as complex and their productivity for the user as limited.
Moreover, the approaches often need fundamental extension for targeting new kinds of computations, e.g., so-called \emph{macro-rules} which had to be added to \texttt{Lift} to efficiently target matrix multiplications~\cite{10.1145/2884045.2884046} and primitives \texttt{slide} and \texttt{pad} together with optimization \emph{overlapped tiling} for expressing stencil computations~\cite{10.1145/3168824}.
This need for extensions limits the expressivity of the \texttt{Lift} language and thus further hinders productivity.

In contrast to \texttt{Lift}, our approach relies on exactly three higher-order functions
(Figure~\ref{hl_overview})
to express various kinds of data-parallel computations (Figure~\ref{fig_hl_examples}):~%
1)~$\iv$~(Definition~\ref{def_iv}) which prepares the input data;~our $\iv$ function is designed as general enough to subsume, in a structured way,
the subset of all \texttt{Lift} patterns intended to change the view on input data,
including patterns \texttt{zip}, \texttt{pad}, and \texttt{slide};
2)~$\mdh$~(Definition~\ref{def_mdh}) expresses the actual computation part, and it and subsumes the \texttt{Lift} patterns performing actual computations (\texttt{fun}, \texttt{map}, \texttt{reduce}, $\dotsc$);~%
3)~$\ov$~(Definition~\ref{def_ov}) expresses the view on output data and is
designed
to work
similarly as
function $\iv$~(Lemma~\ref{theorem_views}).
Our three functions are always composed straightforwardly, in the same, fixed order (Figure~\ref{hl_overview}), and they do not rely on complex function nesting for expressing computations.

Note that even though our language is designed as minimalistic,
it
should cover the expressivity of the \texttt{Lift} language\footnote{
  This work is focussed on dense computations. \texttt{Lift} supports sparse computations~\cite{10.1145/3377555.3377896}
  which we consider as future work for our approach
  (as also outlined in Section~\ref{ch:fw}). We consider \texttt{Lift}'s approach, based on their so-called \emph{position dependent arrays}, as a great inspiration for our future goal.
}
and beyond:~for example, we are currently not aware of any \texttt{Lift} program being able to express the prefix-sum examples in Figure~\ref{fig_hl_examples}.
For the above reasons, we consider programming in our high-level language
as more productive for the user than programming in \texttt{Lift}-like, functional-style languages.
Furthermore, as discussed in Section~\ref{ch:eval:polyhedral}, our approach can take as input also straightforward, sequential program code, which further  contributes to the productivity of our approach.

\begin{figure}[p!]
    \centering
    \includegraphics[scale=0.27]{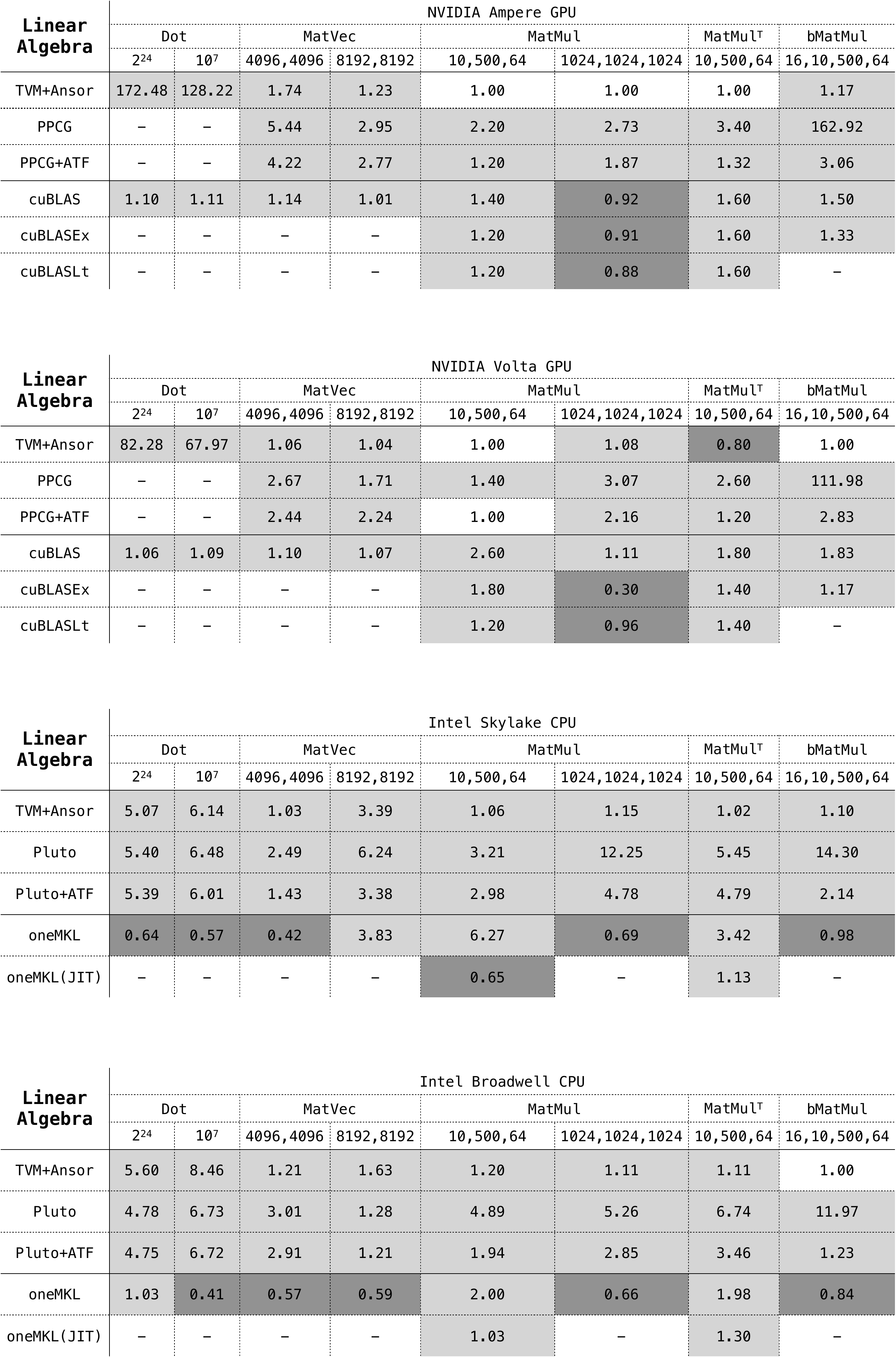}
    \caption{Speedup (higher is better) of our approach for linear algebra routines on GPUs and CPUs over:
    i)~scheduling approach \texttt{TVM},
    ii)~polyhedral compilers \texttt{PPCG} (GPU) and \texttt{Pluto} (CPU), as well as
    iii)~hand-optimized libraries provided by vendors.
    Dash symbol "-" means that this framework does not support this particular combination of architecture, computation, and data characteristic.
    }
    \label{img_eval_la}
\end{figure}
\clearpage

\begin{figure}[p!]
    \centering
    \includegraphics[scale=0.27]{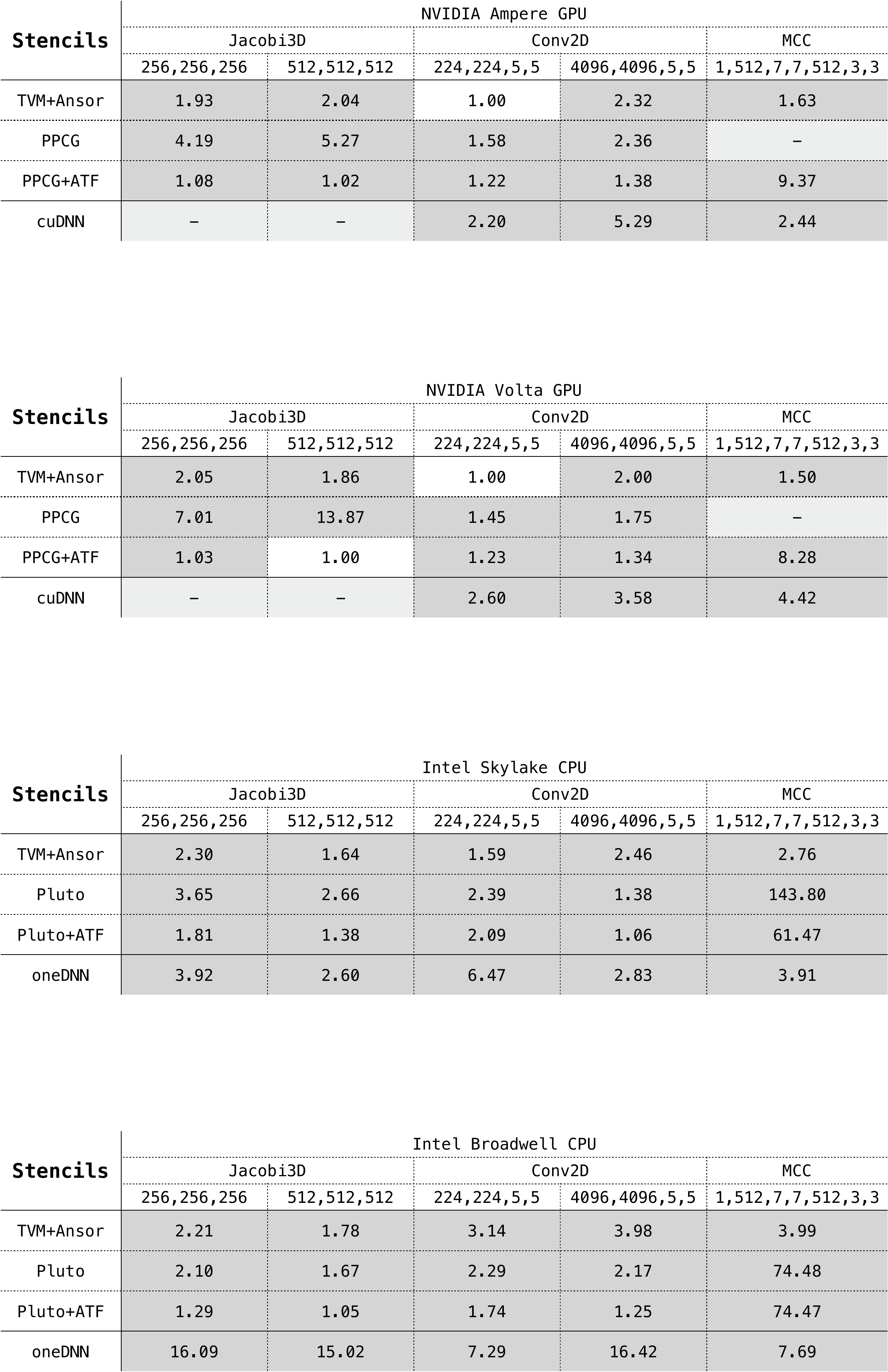}
    \caption{Speedup (higher is better) of our approach for stencil computations on GPUs and CPUs over:
    i)~scheduling approach \texttt{TVM},
    ii)~polyhedral compilers \texttt{PPCG} (GPU) and \texttt{Pluto} (CPU), as well as
    iii)~hand-optimized libraries provided by vendors.
    Dash symbol "-" means that this framework does not support this particular combination of architecture, computation, and data characteristic.
    }
    \label{img_eval_stencils}
\end{figure}
\clearpage

\begin{figure}[p!]
    \centering
    \includegraphics[scale=0.27]{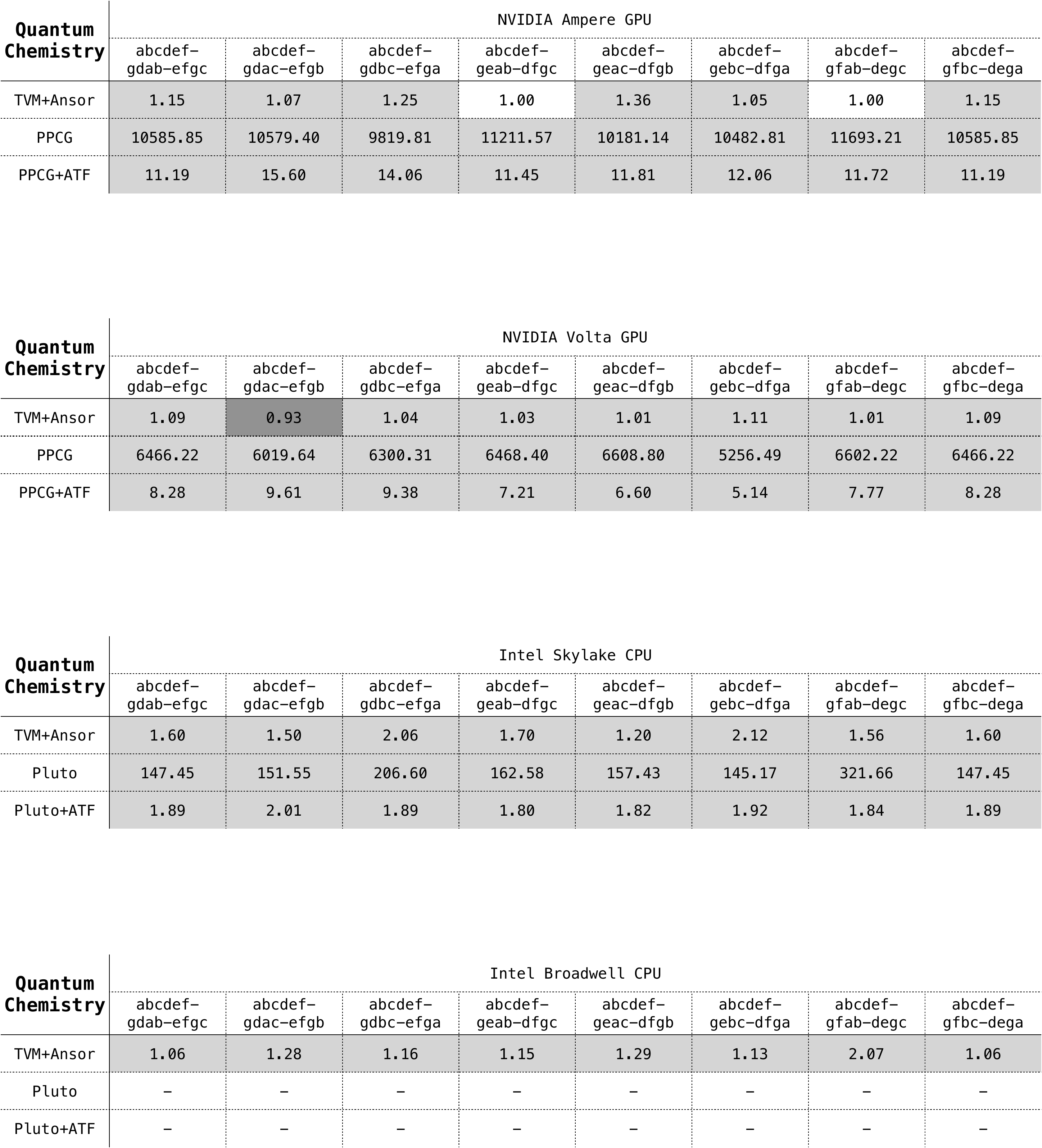}
    \caption{Speedup (higher is better) of our approach for quantum chemistry computations Coupled Cluster (\texttt{CCSD(T)}) on GPUs and CPUs over:
    i)~scheduling approach \texttt{TVM}, and
    ii)~polyhedral compilers \texttt{PPCG} (GPU) and \texttt{Pluto} (CPU).
    Dash symbol "-" means that this framework does not support this particular combination of architecture, computation, and data characteristic.
    }
    \label{img_eval_quantum}
\end{figure}
\clearpage

\begin{figure}[p!]
    \centering
    \includegraphics[scale=0.27]{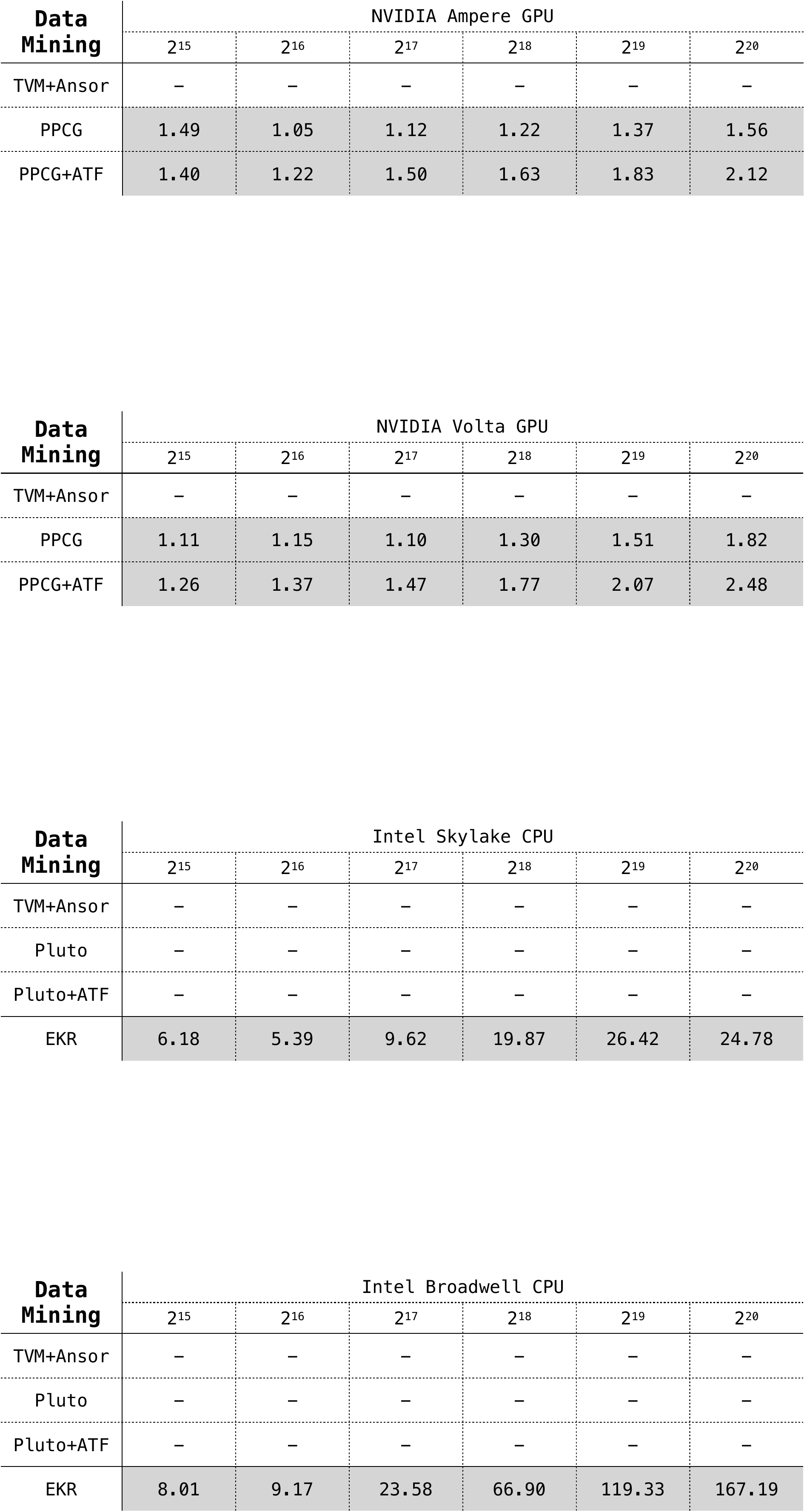}
    \caption{Speedup (higher is better) of our approach for data mining algorithm Probabilistic Record Linkage (\texttt{PRL}) on GPUs and CPUs over:
    i)~scheduling approach \texttt{TVM}, and
    ii)~polyhedral compilers \texttt{PPCG} (GPU) and \texttt{Pluto} (CPU), as well as the
    iii)~hand-implemented Java CPU implementation used by \emph{EKR}~--~the largest cancer registry in Europa.
    Dash symbol "-" means that this framework does not support this particular combination of architecture, computation, and data characteristic.
    }
    \label{img_eval_prl}
\end{figure}
\clearpage

\begin{figure}[p!]
    \centering
    \includegraphics[scale=0.27]{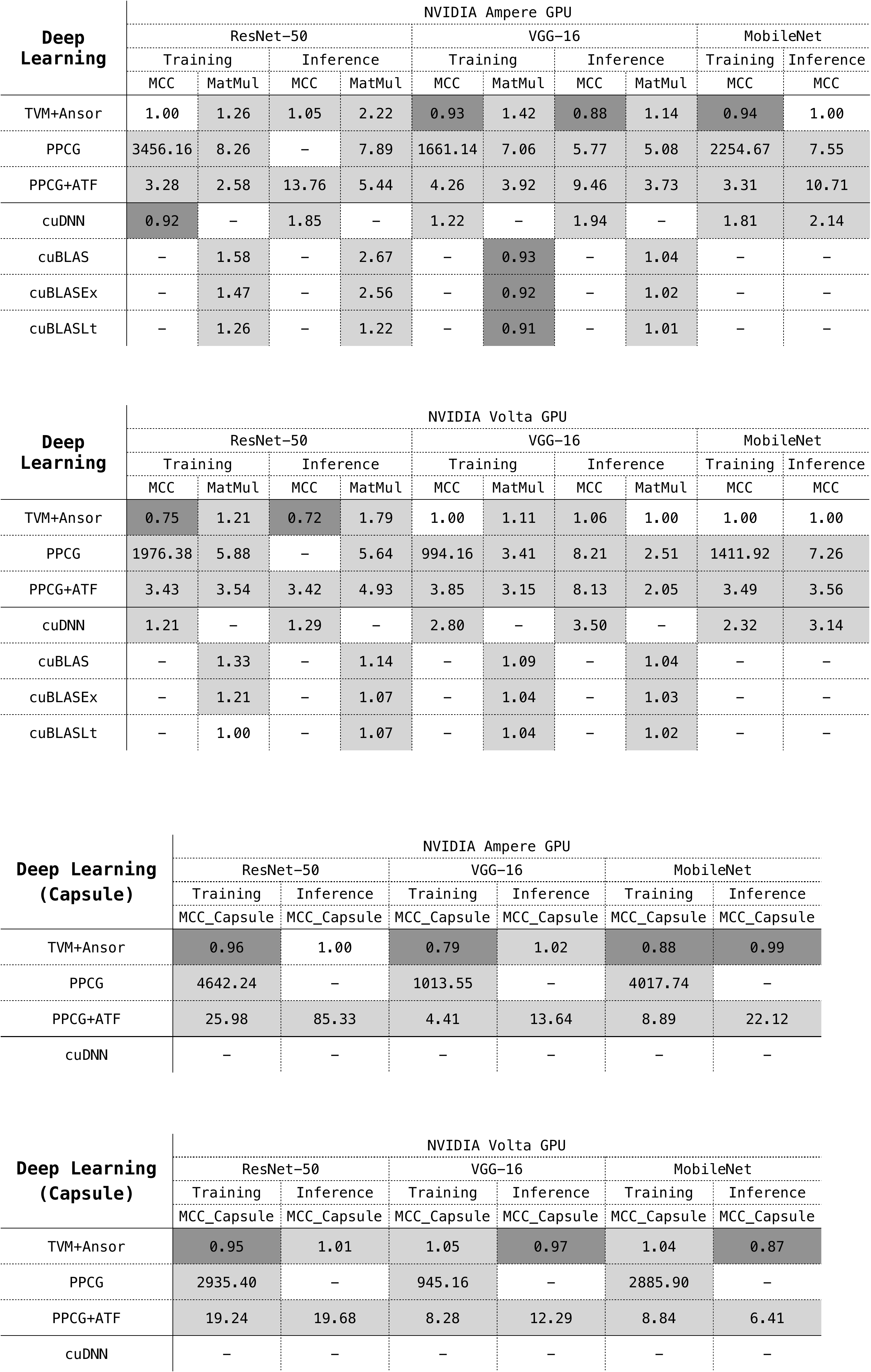}
    \caption{Speedup (higher is better) of our approach for the most time-intensive computations used in deep learning neural networks \texttt{ResNet-50}, \texttt{VGG-16}, and \texttt{MobileNet} on GPUs over:
    i)~scheduling approach \texttt{TVM},
    ii)~polyhedral compilers \texttt{PPCG} (GPU), as well as
    iii)~hand-optimized libraries provided by vendors.
    Dash symbol "-" means that this framework does not support this particular combination of architecture, computation, and data characteristic.
    }
    \label{img_eval_dl_gpu}
\end{figure}
\clearpage

\begin{figure}[p!]
    \centering
    \includegraphics[scale=0.27]{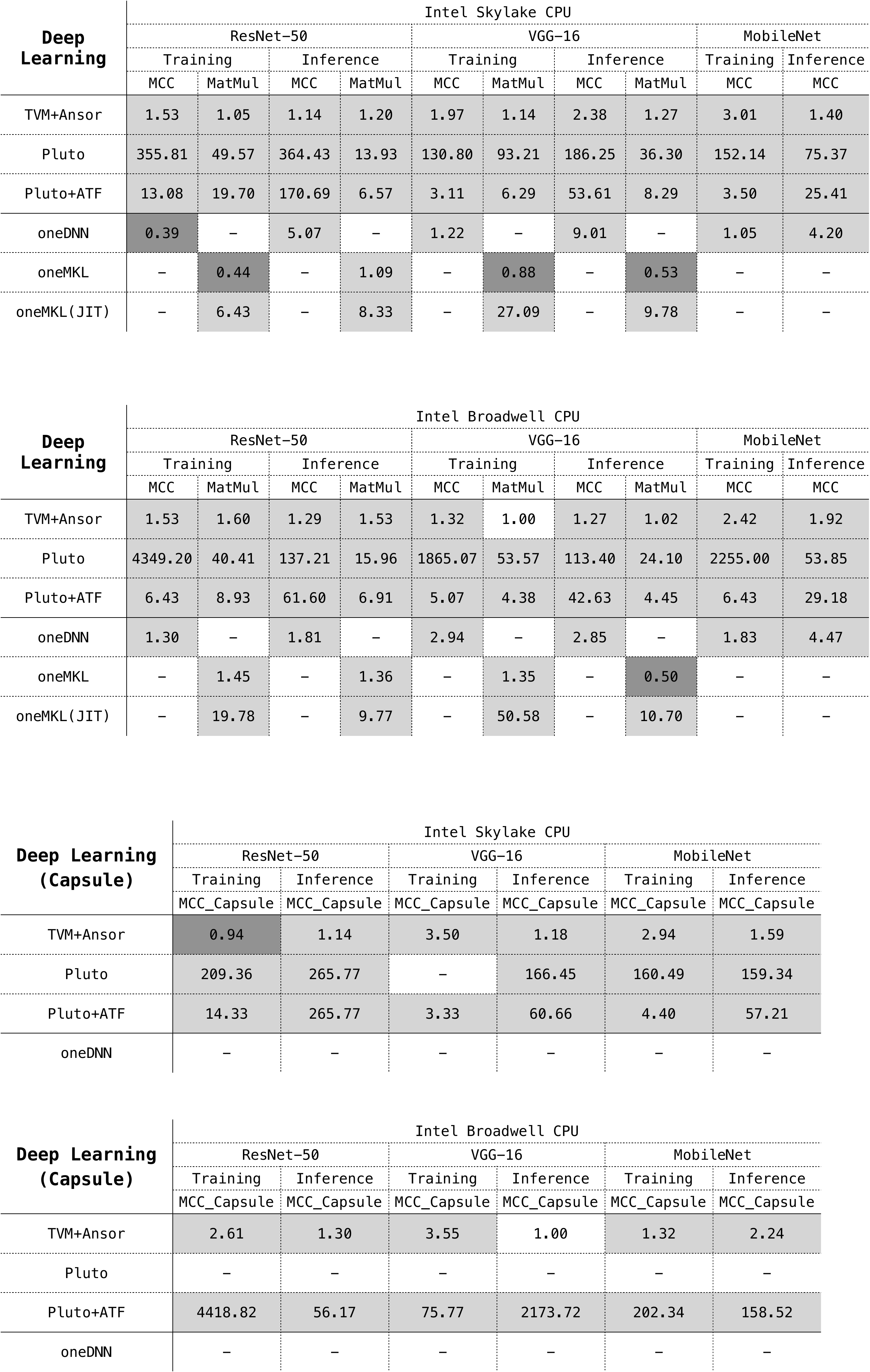}
    \caption{Speedup (higher is better) of our approach for the most time-intensive computations used in deep learning neural networks \texttt{ResNet-50}, \texttt{VGG-16}, and \texttt{MobileNet} on CPUs over:
    i)~scheduling approach \texttt{TVM},
    ii)~polyhedral compilers \texttt{Pluto} (CPU), as well as
    iii)~hand-optimized libraries provided by vendors.
    Dash symbol "-" means that this framework does not support this particular combination of architecture, computation, and data characteristic.
    }
    \label{img_eval_dl_cpu}
\end{figure}
\clearpage

\begin{figure}[p!]
    \centering
    \includegraphics[scale=0.27]{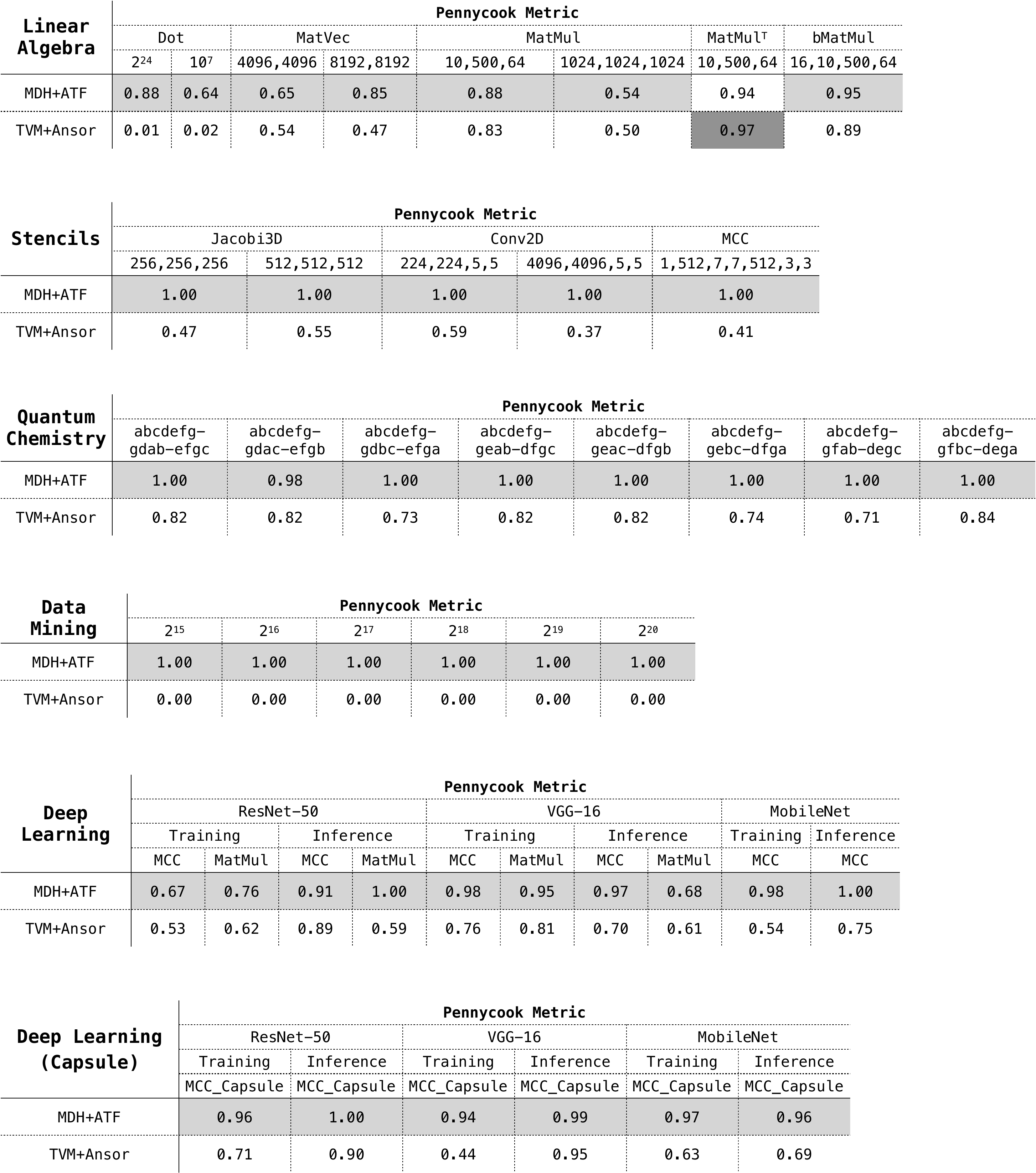}
    \caption{Portability (higher is better), according to Pennycook metric,
    of our approach and \texttt{TVM} over GPUs and CPUs for case studies.
    Polyhedral compilers \texttt{PPCG}/\texttt{Pluto} and vendor libraries by NVIDIA and Intel are not listed:~due to their limitation to certain architectures,
    all of them
    achieve the lowest portability of $0$ only.}
    \label{img_eval_port_cpu_gpu}
\end{figure}
\clearpage

\begin{figure}[h!]
    \centering
    \includegraphics[scale=0.27]{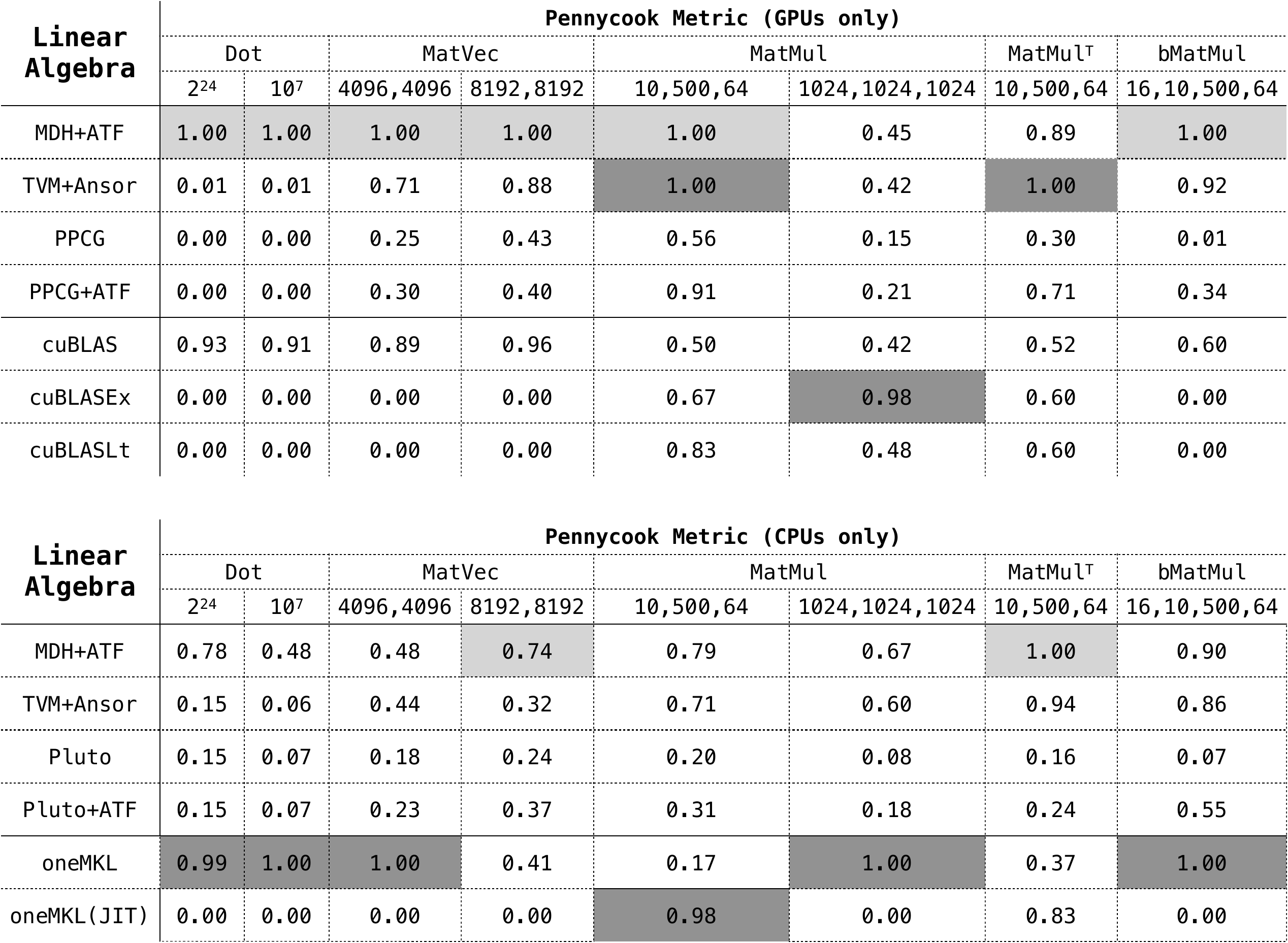}
    \caption{Portability (higher is better), according to Pennycook metric, for linear algebra routines computed on only GPUs or CPUs, respectively.
    The restriction simplifies for frameworks with limited architectural support (such as polyhedral compilers and vendor libraries) the portability comparisons against our approach.}
    \label{img_eval_port_la}
\end{figure}

\begin{figure}[h!]
    \centering
    \includegraphics[scale=0.27]{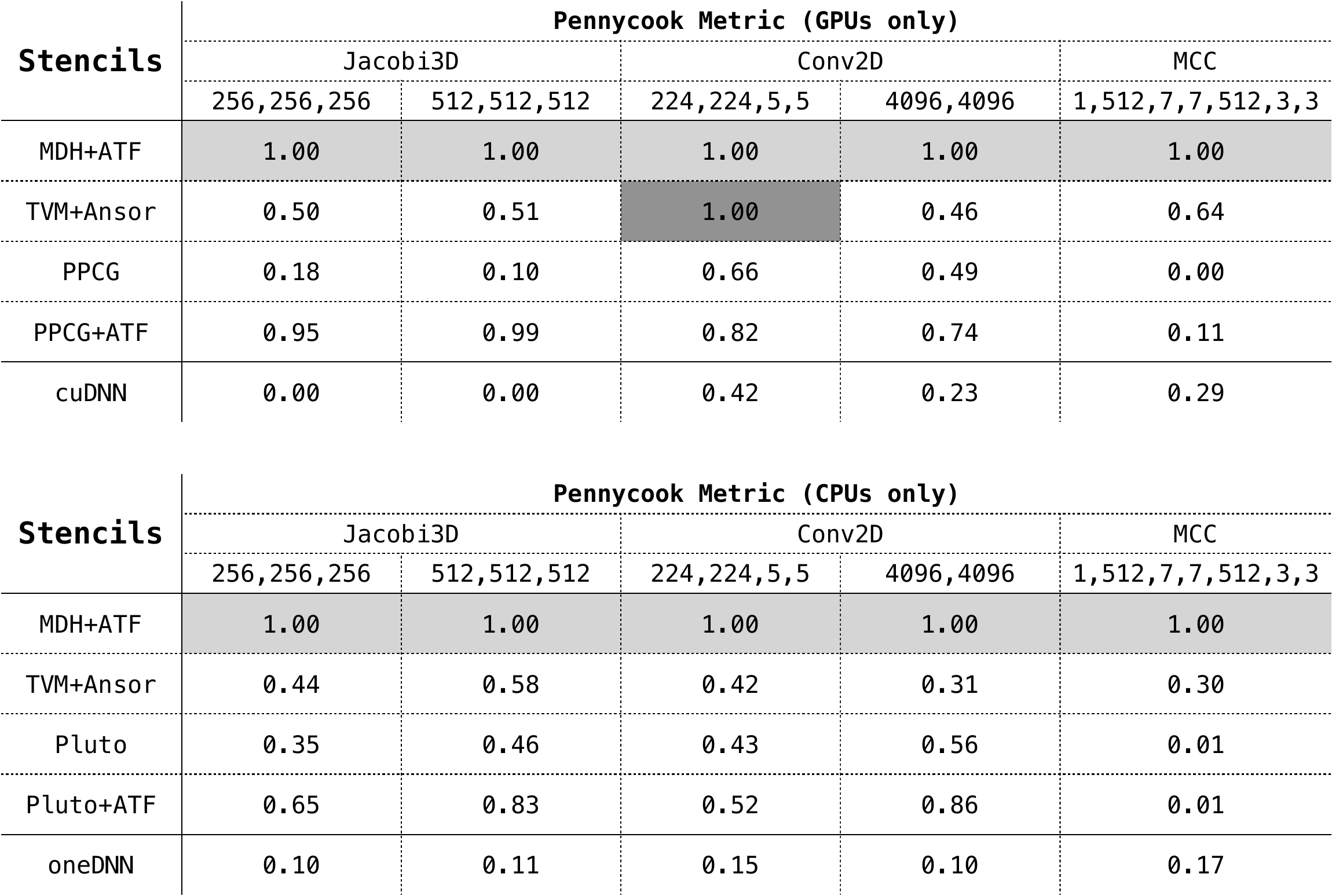}
    \caption{Portability (higher is better), according to Pennycook metric, for stencil computations computed on only GPUs or CPUs, respectively.
    The restriction simplifies for frameworks with limited architectural support (such as polyhedral compilers and vendor libraries) the portability comparisons against our approach.}
    \label{img_eval_port_stencils}
\end{figure}
\clearpage

\vspace{10px}

\begin{figure}[h!]
    \centering
    \includegraphics[scale=0.27]{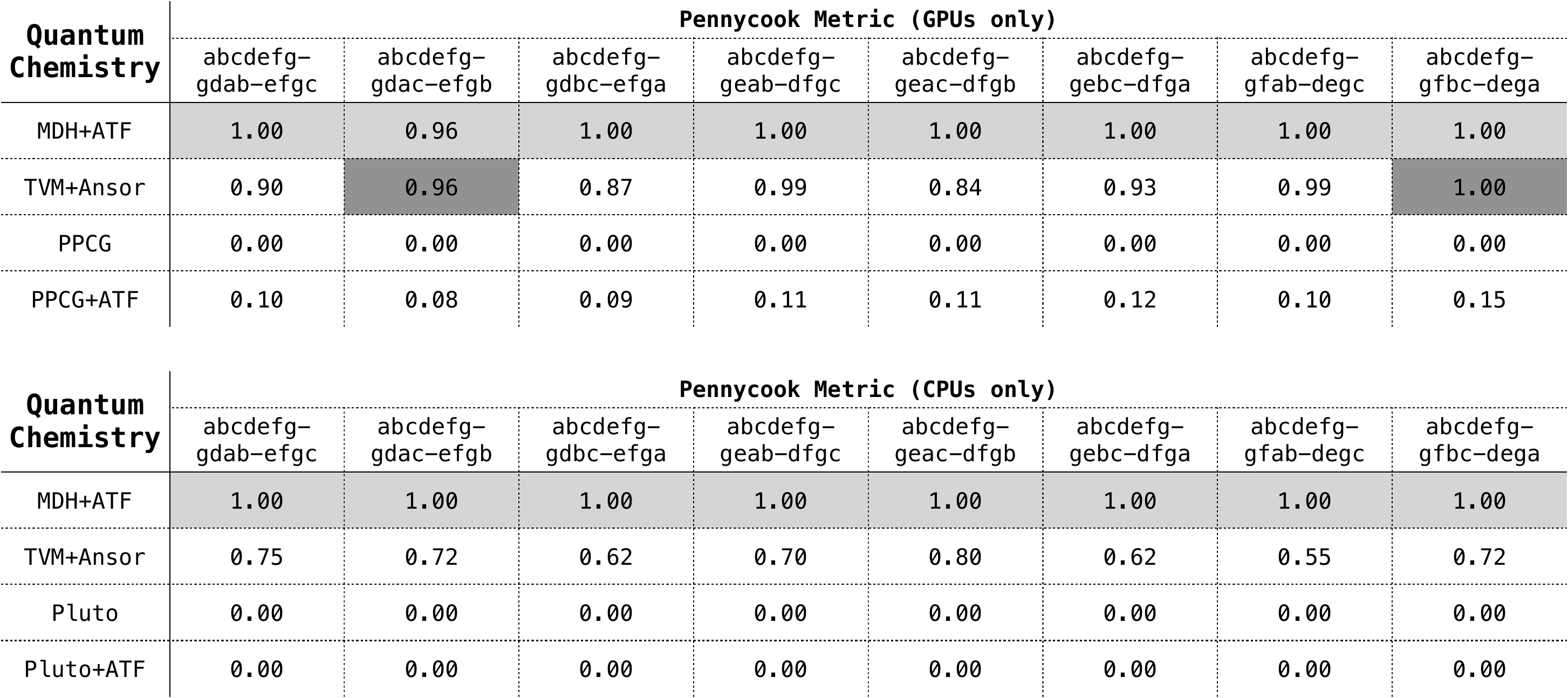}
    \caption{Portability (higher is better), according to Pennycook metric, for quantum chemistry computation Coupled Cluster (\texttt{CCSD(T)}) computed on only GPUs or CPUs, respectively.
    The restriction simplifies for frameworks with limited architectural support (such as polyhedral compilers and vendor libraries) the portability comparisons against our approach.}
    \label{img_eval_port_quantum}
\end{figure}

\vspace{30px}

\begin{figure}[h!]
    \centering
    \includegraphics[scale=0.27]{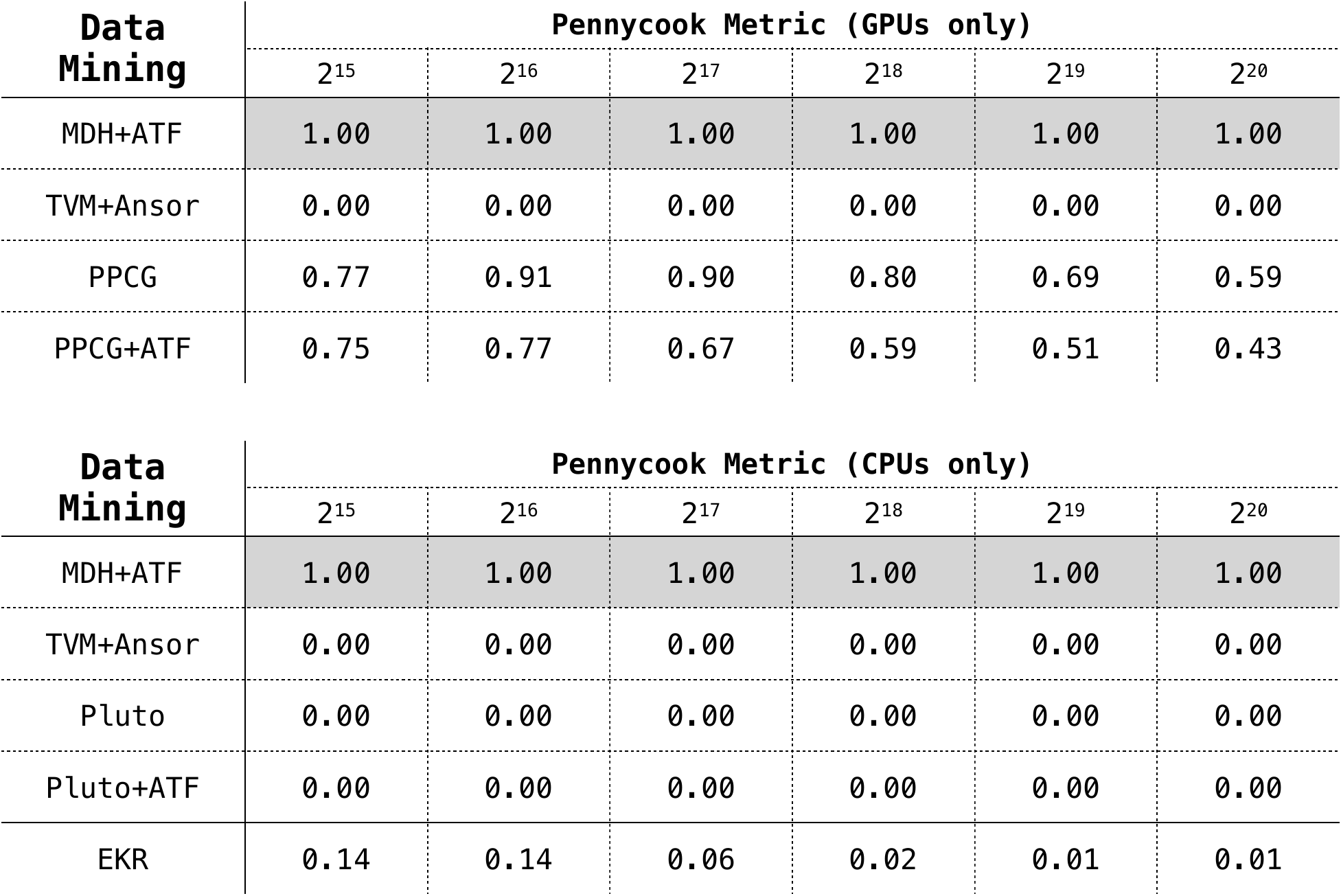}
    \caption{Portability (higher is better), according to Pennycook metric, for data mining algorithm Probabilistic Record Linkage (\texttt{PRL}) computed on only GPUs or CPUs, respectively.
    The restriction simplifies for frameworks with limited architectural support (such as polyhedral compilers and vendor libraries) the portability comparisons against our approach.}
    \label{img_eval_port_prl}
\end{figure}
\clearpage

\begin{figure}[h!]
    \centering
    \includegraphics[scale=0.27]{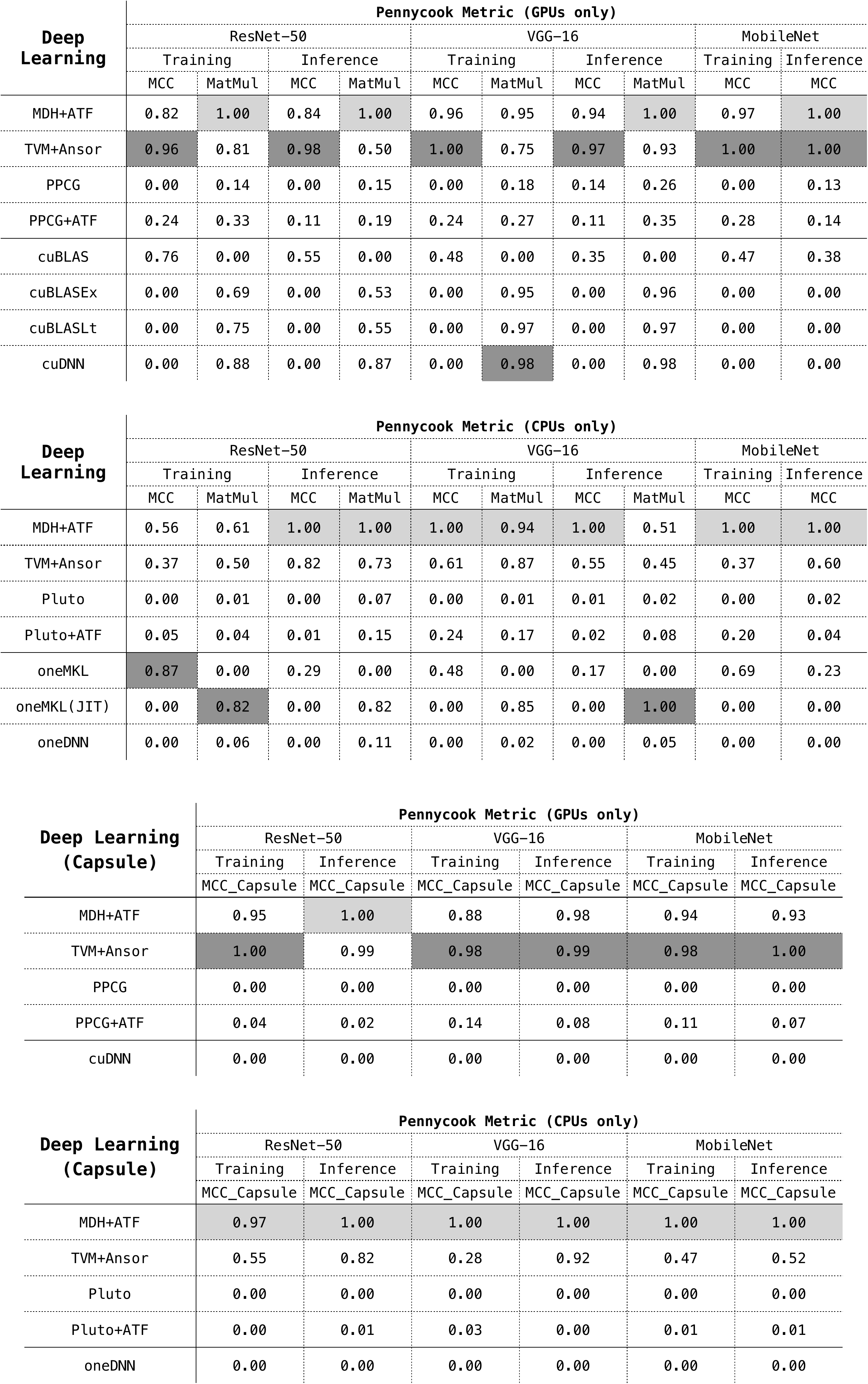}
    \caption{Portability (higher is better), according to Pennycook metric, for deep learning computations computed on only GPUs or CPUs, respectively.
    The restriction simplifies for frameworks with limited architectural support (such as polyhedral compilers and vendor libraries) the portability comparisons against our approach.}
    \label{img_eval_port_dl}
\end{figure}

\clearpage

\subsection{Domain-Specific Approaches}
\label{ch:eval:domain}

\subsubsection*{Performance}

Figures~\ref{img_eval_la}-\ref{img_eval_dl_cpu}
report for completeness also performance results achieved by domain-specific approaches.
Since domain-specific approaches are specifically designed and optimized for particular applications domains and often also architectures (e.g., only linear algebra routines on only GPU), we consider comparing to them as most challenging for us:~our approach is designed and optimized for data-parallel computations in general, from arbitrary application domains (the same as also polyhedral compilers and many functional approaches), and our approach is also designed as generic in the target parallel architecture.

We observe in Figures~\ref{img_eval_la}-\ref{img_eval_dl_cpu} that the
domain-specific libraries \texttt{NVIDIA\:cuBLAS/cuDNN} (for linear algebra routines and convolutions on GPUs) and \texttt{Intel\:oneMKL/oneDNN} (for linear algebra routines and convolutions on CPUs) sometimes perform better and sometimes worse than our approach.

The better performance of libraries over our approach is most likely\footnote{
Since the Intel and NVIDIA libraries are not open source, we cannot explain their performance behavior with certainty.
} because the libraries internally rely on assembly-level optimizations, while we currently focus on the higher \mbox{CUDA/OpenCL} level of abstraction which offers less optimization opportunities~\cite{6494986,10.1145/1356052.1356053}.
The \texttt{cuBLASEx} extension of \texttt{cuBLAS} achieves in one case~--~\texttt{MatMul} on \texttt{NVIDIA\:Volta\:GPU} for square $1024\times1024$ input matrices~--~significantly higher performance than our approach.
The high performance is achieved by \texttt{cuBLASEx} when using its \texttt{CUBLAS\_GEMM\_ALGO1\_TENSOR\_OP} algorithm variant, which casts the \texttt{float}-typed inputs implicitly to the half precision type (a.k.a. \texttt{half} or \texttt{fp16}), allowing \texttt{cuBLASEx} to exploit the GPU's tensor core extension~\cite{cuda_tensor_cores}.
Thereby, \texttt{cuBLASEx} achieves significantly higher performance than our approach, because tensor cores compute small matrix multiplication immediately in hardware;~however, at the cost of a significant precision loss:~%
the \texttt{half} scalar type achieves only half the accuracy achieved by scalar type \texttt{float}.
When using \texttt{cuBLASEx}'s default algorithm \texttt{CUBLAS\_GEMM\_DEFAULT} (rather than algorithm \texttt{CUBLAS\_GEMM\_ALGO1\_TENSOR\_OP}), which
retains the \texttt{float} type and thus meets the accuracy expected from the computation, we achieve a speedup of $1.11\times$ over \texttt{cuBLASEx}.\footnote{
  For the interested reader, we report in our Appendix, Section~\ref{app_cuBLASEx_runtime_accuracy}, the runtime of \texttt{cuBLASEx} for all its algorithm variants, including reports for the accuracy achieved by the different variants.
}

The reason for the better performance of our approach over NVIDIA and Intel libraries is most likely because our approach allows generating code that is also optimized (auto-tuned) for data characteristics, which is important for high performance~\cite{10.1145/3126908.3126939}.
In contrast, the vendor libraries usually rely on pre-implemented code that is optimized toward only average high performance for a range of data characteristics (size, memory layout, etc).
By relying on these fixed, pre-implemented code, the libraries avoid the auto-tuning overhead.
However, auto-tuning is often amortized, particularly for deep learning computations~--~the main target of libraries \texttt{NVIDIA\:cuDNN} und \texttt{Intel\:oneDNN}~--~because the auto-tuned implementations are re-used in many program runs.
Moreover, we achieve better performance for convolutions (Figure~\ref{img_eval_stencils}), because the libraries re-use optimizations for these computations originally intended for linear algebra routines~\cite{7573804}, whereas our optimization space (Table~\ref{tab_tps}) is designed for data-parallel computations in general and not as specifically oriented toward linear algebra.

Compared to the \texttt{EKR} library (Figure~\ref{img_eval_prl}), we achieve higher performance, because the \texttt{EKR}'s Java implementation inefficiently handles memory:~the library is implemented using Java's \texttt{ArrayList} data structure which is convenient to use for the Java programmer, but inefficient in terms of performance, because the structure internally performs costly memory re-allocations.

\subsubsection*{Portability}

Similar to polyhedral compilers \texttt{PPCG} and \texttt{Pluto}, the domain-specific approaches work for certain architectures only and thus achieve the lowest portability of $0$ only in Figure~\ref{img_eval_port_cpu_gpu} for our studies.
The domain-specific approaches are also restricted to a narrow set of studies, e.g., only linear algebra routines as \texttt{NVIDIA\:cuBLAS} and \texttt{Intel\:oneMKL} or only data mining example \texttt{PRL} as~\texttt{EKR}.
Consequently, the approaches achieve for these unsupported studies also a portability of only $0$ in Figures~\ref{img_eval_port_la}-\ref{img_eval_port_dl} in which our portability evaluation is limited to only GPUs or CPUs, respectively, to make comparison against our approach easier for the vendor libraries.

For their target studies, domain-specific approaches achieve high portability.
This is because the approaches are specifically designed and optimized toward these studies, e.g., via assembly-level optimizations which are currently beyond the scope of our work and considered as future work for our approach (see Section~\ref{ch:fw}).

\subsubsection*{Productivity}

Listing~\ref{lst_matmul_cublas} shows the implementation of \texttt{MatVec} in domain-specific approach \texttt{NVIDIA} \texttt{cuBLAS};~the implementation of \texttt{MatVec} in other domain-specific approaches, e.g., \texttt{Intel\:oneMKL}, is analogous to the implementation in Listing~\ref{lst_matmul_cublas}.

We consider domain-specific approaches as most productive for their target domain:~in the case of \texttt{MatVec}, the user simply calls the high-level function \texttt{cublasSgemv} and passes to it the input matrices (omitted via ellipsis in the listing) together with some meta information (memory layout of matrices, etc);~\texttt{cuBLAS} then automatically starts the GPU computation for
\texttt{MatVec}.

Besides the fact that domain-specific approaches typically target only particular target architectures, a further fundamental productivity issue of domain-specific approaches is that they can only be used for a narrow class of computations, e.g., only linear algebra routines as \texttt{NVIDIA\:cuBLAS} and \texttt{Intel\:oneMKL}.
Moreover, in the case of domain-specific libraries from NVIDIA and Intel, it is often up to the user to manually choose among different, semantically equal but differently performing implementations for high performance.
For example, the \texttt{cuBLAS} library offers three different routines for computing matrix multiplications:~%
1)~\texttt{cublasSgemm} (part of standard \texttt{cuBLAS}),~%
2)~\texttt{cublasGemmEx} (part of the \texttt{\texttt{cuBLASEx}} extension of \texttt{cuBLAS}), and~%
3)~routine \texttt{cublasLtMatmul} (part of the \texttt{cuBLASLt} extension).
These routines often also offer different, so-called \emph{algorithms} (e.g., $42$ algorithm variants in the case \texttt{\texttt{cuBLASEx}}) which
impact the internal optimization process.
When striving for the highest performance potentials of libraries, the user is in charge of naively testing each possible combination of routine and algorithm variant (as we have done in Figures~\ref{img_eval_la}-\ref{img_eval_dl_cpu} to make experimenting challenging for us).
In addition, the user must be aware that different combinations of routines and algorithms can produce results of reduced accuracy (as discussed above), which can be critical for accuracy-sensitive use cases.

\begin{lstlisting}[
    language=C,
    morekeywords={__kernel__},
    mathescape=true,
    numbers=left,
    xleftmargin=1.5em,
    float=h!,
    frame=lines,
    caption={\texttt{cuBLAS} program expressing Matrix-Vector Multiplication (\texttt{MatVec})},
    label={lst_matmul_cublas}
    ]
cublasSgemv( /* ... */ );
\end{lstlisting}

\section{Related Work}
\label{sec_rw}

Three major classes of approaches currently focus on code generation and optimization for data-parallel computations:~%
1)~scheduling,
2)~polyhedral, and
3)~functional.
In the following, we compare in Sections~\ref{sec_rw_scheduling}-\ref{ch:rw:functional} our approach to each of these three classes~--~in terms of \emph{performance}, \emph{portability}, and \emph{productivity}.
In contrast to Section~\ref{ch:eval}, which has compared our approach against these classes experimentally, this section is focussed on discussions in a more general, non-experimental context.
Afterward, we outline domain-specific approaches in Section~\ref{sec_rw_dsa}, which are specifically designed and optimized toward their target application domains.
In Section~\ref{sec_rw_hla}, we outline approaches focussing on optimizations that operate at the algorithmic level of abstraction (and thus at a higher abstraction level than our approach);~we consider these higher level approaches as greatly combinable with our work.
Finally, we discuss in Section~\ref{sec_rw_mdh} the differences between our approach introduced in this paper and the already existing work on MDHs.

\subsection{Scheduling Approaches}
\label{sec_rw_scheduling}

Popular examples of scheduling approaches include
\emph{UTF}~\cite{kelly1998framework},
\emph{URUK}~\cite{uruk},
\emph{CHill}~\cite{chen2008chill,10.1145/2400682.2400690},
\emph{Halide}~\cite{10.1145/2491956.2462176},
\emph{Clay}~\cite{10.1145/2854038.2854048},
\emph{TVM}~\cite{222575},
\emph{TeML}~\cite{10.1145/3393934.3278131},
\emph{Tiramisu}~\cite{8661197},
\emph{DaCe}~\cite{dace},
\emph{Fireiron}~\cite{10.1145/3410463.3414632},
\emph{Elevate}~\cite{10.1145/3408974},
\emph{DISTAL}~\cite{10.1145/3519939.3523437}, and
\emph{LoopStack}~\cite{https://doi.org/10.48550/arxiv.2205.00618}.
While scheduling approaches usually achieve high performance, they often have difficulties with achieving portability
and productivity, as we discuss in the following.\footnote{
    \citet{10.1145/3578360.3580269} introduce (optionally) a scheduling language for MDH to incorporate expert knowledge into MDH's optimization process, e.g., to achieve~%
    1)~better optimization, as an auto-tuning system might not always make the same high-quality optimization decisions as a human expert, and/or~%
    2)~faster auto-tuning, as some (or even all) optimization decisions might be made by the expert user and thus are not left to the costly auto-tuner.
}

\paragraph{Performance}
Scheduling approaches usually achieve high performance.
For this, the approaches incorporate human expert knowledge into their optimization process which is based on two major steps:~%
1)~a human expert implements an optimization program (a.k.a \emph{schedule}) in a so-called \emph{scheduling language}~--~the program specifies the basic optimizations to perform, such as tiling and parallelization;~%
2)~an auto-tuning system (or a human hardware expert) chooses values of performance-critical parameter of the optimizations implemented in the schedule, e.g., particular values of tile sizes and concrete numbers of threads.

Our experiments in Section~\ref{ch:eval} show that compared to the state-of-the-art scheduling approach TVM (using its recent Ansor optimizer~\cite{258858} for schedule generation), our approach achieves competitive and sometimes even better performance, e.g., speedups up to $2.22\times$ on GPU and $3.55\times$ on CPU over TVM+Ansor for
computations
taken from TVM's
favorable
application domain (deep learning).
Section~\ref{ch:eval} discusses
that our better performance is due to the design and structure of our general
optimization space (Table~\ref{tab_tps}) which can be efficiently explored~--~fully automatically~--~using state-of-the-art auto tuning techniques~\cite{10.1145/3427093}.
We focus on TVM in our experiments (rather than, e.g. Halide) to make experimenting challenging for us:~%
TVM+Ansor has proved to achieve higher performance on GPUs and CPUs than popular state-of-practice approaches~\cite{258858}, including Halide, pyTorch~\cite{NEURIPS2019_bdbca288}, and the recent FlexTensor optimizer~\cite{10.1145/3373376.3378508}.

The recent approach TensorIR~\cite{10.1145/3575693.3576933} is a compiler for deep learning computations that achieves higher performance than TVM on NVIDIA GPUs.
However, this performance gain over TVM is mainly achieved by exploiting the domain-specific \emph{tensor core}~\cite{cuda_tensor_cores} extensions of NVIDIA GPUs, which compute in hardware the multiplications of small, low-precision $4\times4$ matrices.
For this, TensorIR introduces the concept of \emph{blocks} which represent sub-computations, e.g., multiplying $4\times4$ matrices.
These blocks are than mapped by TensorIR to domain-specific hardware extensions, which often leads to high performance.

While domain-specific hardware extensions are not targeted in this paper, we can naturally exploit them in our approach, similar to TensorIR, as we plan for our future work:~the sub-computations targeted by the current hardware extensions, such as matrix multiplication on $4\times4$ matrices, can be straightforwardly expressed in our approach (Figure~\ref{fig_hl_examples}).
Thus, we can match these sub-expressions in our low-level representation and map them to hardware extensions in our generated code.
For this, instead of relying on a full partitioning in our low-level representation (as in Figure~\ref{fig_ll_example}) such that we can apply scalar function $f$ to the fully de-composed data (consisting of a single scalar value only in the case of a full partitioning), we plan to rely on a coarser-grained partitioning schema, e.g., down to only $4\times 4$ matrices (rather than $1\times1$ matrices, as in the case of a full partitioning).
This allows us replacing scalar function $f$ (which in the case of matrix multiplication is a simple scalar multiplication $*$) with the operation supported by the hardware extension, such as matrix multiplication on $4\times4$ matrices.
We expect for our future work to achieve the same advantages over TensorIR as over TVM, because apart from supporting domain-specific hardware extensions, TensorIR is very similar to TVM.

\paragraph{Portability}
While scheduling approaches achieve high performance, they tend to struggle with achieving portability.
This is because even though the approaches often offer different, pre-implemented backends (e.g., a CUDA backend to target NVIDIA GPUs and an OpenCL backend for CPUs), they do not propose
any structured methodology about how new backends can be added, e.g., for potentially upcoming architectures, with potentially deeper memory and core hierarchies than GPUs and CPUs.
This might be particularly critical (or requiring significant development effort) for the application area of deep learning which is the main target of many scheduling approaches, e.g., TVM and TensorIR, and for which new architectures are arising continuously~\cite{10.1145/3282307}.

In contrast, we introduce in this paper a formally precise recipe for correct-by-construction code generation in different backends (including OpenMP, CUDA, and OpenCL), generically in the target architecture:~we introduce an architecture-agnostic low-level representation (Section~\ref{ch:low_level}) as target for our high-level programs (Section~\ref{ch:high_level}), and we describe formally how our high-level programs are automatically lowered to our low-level representation (Section~\ref{ch:lowering}), based on the architecture-agnostic optimization space in Table~\ref{tab_tps}.
Our Appendix, Section~\ref{app_sec_code_generation}, outlines how executable, imperative-style program code is straightforwardly generated from low-level expressions, which we plan to discuss and illustrate in detail in our future work.

\paragraph{Productivity}
Scheduling approaches rely on a two-step optimization process, as discussed above:~implementing a schedule (first step) and choosing optimized values of performance-critical parameters within that schedule (second step).
While the second step often can be easily automatized, e.g., via auto-tuning~\cite{NEURIPS2018_8b570001}, the first step~--~implementing a schedule~--~usually has to be conducted manually by the user to achieve high performance, which requires expert knowledge and thus hinders productivity.
The lack of formal foundation of many scheduling approaches further complicates implementing schedules for the user, as implementation becomes error prone and hardly predictable.
For example, Fireiron's schedules can achieve high performance, close to GPUs' peak, but schedules in Fireiron can easily generate incorrect low-level code:~Fireiron cannot guarantee that optimizations expressed in its scheduling language are semantics preserving, e.g., based on a formal foundation as done in this work, making programming Fireiron's schedules error prone and complex for the user.
Similarly, TVM is sometimes unable to detect user errors in both its high-level language (as discussed in Section~\ref{ch:eval:scheduling}) as well as scheduling language%
~\cite{tvm_issue_double_bind}.
Safety in parallel programming is an ongoing major demand, in particular from industry~\cite{vulkan_safety}.

Auto schedulers, such as Halide's optimization engine~\cite{10.1145/2897824.2925952} and TVM's recent Ansor~\cite{258858}, aim to automatically generate well-performing, correct schedules
for the user.
However, a major flaw of the current auto schedulers is that even though they work well for some computations (e.g., from deep learning, as TVM's Ansor), they may perform worse for others.
For example, our approach achieves a speedup over TVM+Ansor of $>100\times$ already for straightforward dot products (Figure~\ref{img_eval_la}).
This is because Ansor does not exploit multiple thread blocks and uses only a small number of threads for reduction computations.
While such optimization decisions are often beneficial for reductions as used in deep learning (e.g., within the computations of convolutions and matrix multiplications on deep learning workloads, because parallelization can be better exploited for outer loops of these computations), these rigid optimization decisions of Ansor may perform worse in other contexts (e.g., for computing dot product).

To avoid the productivity issues of scheduling approaches, we have designed our optimization process as fully auto-tunable, thereby freeing the user from the burden and complexity of making complex optimization decisions.
Our optimization space (Table~\ref{tab_tps}) is designed as generic in the target application area and hardware architecture, thereby achieving high performance for various combinations of data-parallel computations and architectures (Section~\ref{ch:eval}).
Correctness of optimizations is ensured in our approach by introducing a formal foundation that enables mathematical reasoning about correctness (Section~\ref{ch:lowering}).
Particularly, our optimization process is designed as \emph{correct-by-construction}, meaning that any valid optimization decisions (i.e., a particular choice of tuning parameters in Table~\ref{tab_tps} that satisfy the constraints) leads to a correct expression in our low-level expression (as in Figure~\ref{fig_ll_example}).
In contrast, approaches such as introduced by~\citet{clement:hal-03653857} formally validate optimization decisions of scheduling approaches in already generated low-level code.
Thereby, such approaches work potentially for arbitrary scheduling approaches (Halide, TVM, $\dotsc$), but the approaches cannot save the user
at the high abstraction level
from implementing incorrect optimizations (e.g., via easy-to-understand, high-level error messages indicating that an invalid optimization decisions is made)
or restricting the optimization space otherwise to valid decisions only, e.g., for an
efficient
auto-tuning process, because the approaches check already generated program code.

Scheduling approaches often also suffer from expressivity issues.
For example,
Fireiron is limited to computing only matrix multiplications on only NVIDIA GPUs, and
TVM does not support computations that rely on multiple combine operators different from concatenation~\cite{tvm_issue_cos_1,tvm_issue_cos_2}, e.g., as required for expressing the \emph{Maximum Bottom Box Sum} example in Figure~\ref{fig_hl_examples}.
Also, TVM has difficulties with user-defined combine operators~\cite{tvm_issue_cos_3} and thus crashes for example~\emph{Probabilistic Record Linkage} in Figure~\ref{fig_hl_examples}.
In contrast to TVM, we introduce a formal methodology about of how to manage different kinds of arbitrary, user-defined combine operators (Section~\ref{ch:low_level}), which is considered challenging~\cite{tvm_issue_cos_1}.

\subsection{Polyhedral Approaches}
\label{sec_rw_poly}

Polyhedral approaches, as introduced by \citet{Feautrier1992}, as well as
\emph{Pluto}~\cite{bondhugula2008pluto},
\emph{Polly}~\cite{doi:10.1142/S0129626412500107},
\emph{PPCG}~\cite{10.1145/2400682.2400713},
\emph{Polyhedral Tensor Schedulers}~\cite{9188233},
\emph{TC}~\cite{10.1145/3355606}, and
\emph{AKG}~\cite{9741260} rely on a formal, geometrically inspired program representation, called \emph{polyhedral model}.
Polyhedral approaches often achieve high user productivity, e.g., by automatically parallelizing and optimizing straightforward sequential code.
However, the approaches tend to have difficulties with achieving high performance and portability when used for generating low-level program code,
as we outline in the following.
In Section~\ref{sec_rw_hla}, we revisit the polyhedral approach as a potential frontend for our approach, as polyhedral transformations have proven to be efficient when used for high-level code optimizations (e.g., \emph{loop skewing}~\cite{97902}), rather than low-level code generation.

\paragraph{Performance}
Polyhedral compilers tend to struggle with achieving their full performance potential.
We argue that this performance issue of polyhedral compilers is mainly caused by the following two major reasons.

\vspace*{0px}

While we consider the set of polyhedral transformation (so-called \emph{affine transformation}) as broad, expressive, and powerful, each polyhedral compiler implements a subset of expert-chosen transformations. ~%
This subset of transformations, as well as the application order of transformations, are usually fixed in a particular polyhedral compiler and chosen toward specific optimization goals only,
e.g., coarse-grained parallelization and locality-aware data accesses (a.k.a. \emph{Pluto algorithm}~\cite{10.1007/978-3-540-78791-4_9}), causing the search spaces of polyhedral compilers to be a proper subset of our space in Table~\ref{tab_tps}.
Consequently, computations that require for high performance other subsets of polyhedral transformations and/or application orders of transformations (e.g., transformations toward fine-grained parallelization) might not achieve their full performance potential when compiled with a particular polyhedral compiler~\cite{10444791}.

In contrast to the currently existing polyhedral compilers, we have designed our optimization process as generic in goals:~%
for example, our space is designed such that
the
degree of parallelization (coarse, fine, $\dotsc$) is fully auto-tunable for the particular combination of target architecture and computation to optimize.
We consider it as an interesting future work to
investigate the strength and weaknesses of the polyhedral model for expressing our generic optimization space.

\vspace*{0px}

We see the second reason for potential performance issues in polyhedral compilers in their difficulties with reduction-like computations.
This is mainly caused by the fact that the polyhedral model captures less semantic information than the high-level program representation introduced in Section~\ref{ch:high_level} of this paper:~combine operators which are used to combine the intermediate results of computations (e.g., operator $+$ from Example~\ref{def:mda_pw} for combining the intermediate results of the dot products within matrix multiplication) are not explicitly represented in the polyhedral model;~the polyhedral model is rather focussed on modeling memory accesses and their relative order only.
Most likely, this semantic information is missing in the polyhedral model, because polyhedral approaches were originally intended to fully automatically optimize loop-based, sequential code (such as Pluto and PPCG)~--~extracting combine operators automatically from sequential code is challenging and often even impossible (Rice's theorem).

In contrast, our proposed high-level representation explicitly captures combine operators (Figure~\ref{fig_hl_examples}), by requesting these operators explicitly from the user.
This is important, because the operators are often required for generating code that fully utilizes the highly parallel hardware of state-of-the-art architectures (GPUs, etc), as discussed in Section~\ref{ch:eval}.
Similarly to our approach, the polyhedral compiler TC also requests combine operators explicitly from the user.
However, TC is restricted to operators \texttt{+} (addition), \texttt{*} (multiplication), \texttt{min} (minimum), and \texttt{max} (maximum) only, thereby TC is not able to express important examples in Figure~\ref{fig_hl_examples}, e.g., PRL which is popular in data mining.
Moreover, TC outsources the computation of its combine operators to the NVIDIA CUB library~\cite{cub};~most likely as a workaround, because TC relies on the polyhedral model which is not designed to capture and exploit semantic information about combine operators for optimization.
Thereby, TC is dependent on external approaches for computing combine operators, which might not always be available (e.g., for upcoming architectures).

Workarounds have been proposed by the polyhedral community to target reduction-like computations~\cite{DBLP:journals/corr/DoerfertSHB15,10.1145/2967938.2967950}.
However, these approaches are limited to a subset of computations, e.g., by not supporting user-defined scalar types~\cite{DBLP:journals/corr/DoerfertSHB15} (as required for our PRL example in Figure~\ref{fig_hl_examples}), or by being limited to GPUs only~\cite{10.1145/2967938.2967950}.
Comparing the semantic information captured in the polyhedral model vs our MDH-based representation have been the focus of discussions between polyhedral experts and MDH developers~\cite{mlir_mdh_meeting}.

\paragraph{Portability}
The polyhedral approach, in its general form, is a framework offering transformation rules (affine transformations), and each individual polyhedral compiler implements a set of such transformations which are then instantiated (e.g., with particular tile sizes) and applied when compiling a particular application.
However, individual polyhedral compilers (e.g., PPCG and Pluto) apply a fixed set of affine transformations, thereby rigidly optimizing for a particular target architecture only, e.g., only GPU (as PPCG) or only CPU (as Pluto), and it remains open which affine transformations have to be used and how for other architectures, e.g., upcoming accelerators for deep learning computations~\cite{10.1145/3282307} with potentially
more complex
memory and core hierarchies than GPUs and CPUs.
Moreover, while we introduce an explicit low-level representation (Section~\ref{ch:low_level}), the polyhedral approach does not introduce representations on different abstraction levels:~the model relies on one representation that is transformed via affine transformations.
Apart from the ability of our low-level representation to handle combine operators (which we consider as complex and important), we see the advantages of our explicit low-level representation in, for example, explicitly representing memory regions, which allows formally defining important correctness constraints, e.g., that GPU architectures allow combining the results of threads in designated memory regions only.
Furthermore, our low-level representation also allows straightforwardly generating executable code from it (shown in Section~\ref{app_sec_code_generation} of our Appendix, and planned to be discussed thoroughly in future work).
In contrast, code generation from the polyhedral model has proven challenging~\mbox{\cite{9741260,vasilache2022composable,10.1145/2743016}.}

\paragraph{Productivity}
Most polyhedral compilers achieve high user productivity, by fully automatically parallelizing and optimizing straightforward sequential code (as Pluto and PPCG).
Our approach currently relies on a DSL (Domain-Specific Language) for expressing computations, as discussed in Section~\ref{ch:high_level};~thus, our approach can be considered as less productive than many polyhedral compilers.
However, \citet{mdpoly,mdpoly_src} show that DSL programs in our approach can be automatically generated from sequential code (optionally annotated with simple, OpenMP-like directives for expressing combine operators, enabling advanced optimizations), by using polyhedral tool \texttt{pet}~\cite{verdoolaege2012polyhedral} as a frontend for our approach.
Thereby, we are able to achieve the same, high user productivity as polyhedral compilers.
We consider this direction~--~combing the polyhedral model with our approach~--~as promising, as it enables benefitting from the advantages of both directions:~%
optimizing sequential programs and making them parallelizable using polyhedral techniques (like \emph{loop skewing}, as also outlined in Section~\ref{sec_rw_hla}), and
mapping the optimized and parallelizable code eventually to parallel architectures based on the concepts and methodologies introduced in this paper.

\subsection{Functional Approaches}
\label{ch:rw:functional}

Functional approaches map data-parallel computations
that are expressed via small, formally defined building blocks (a.k.a. patterns~\cite{gorlatch2011parallel}, such as \texttt{map} and \texttt{reduce}) to the memory and core hierarchies of parallel architectures, based on a strong formal foundation.
Notable functional approaches include
Accelerate~\cite{10.1145/1926354.1926358}, Obsidian~\cite{10.1007/978-3-642-24452-0_9},
so-called \emph{skeleton libraries}~\cite{6008967,aldinucci2017fastflow,10.1145/1863482.1863487,skepu2},
and the modern Lift approach~\cite{10.1145/2784731.2784754} (recently also known as RISE~\cite{DBLP:journals/corr/abs-2201-03611}).

In the following, as functional approaches usually follow the same basic concepts and methodologies, we focus on comparing to Lift, because Lift is more recent
than, e.g., Accelerate and Obsidian.

\paragraph{Performance}
Functional approaches tend to struggle with achieving their full performance potential, often caused by the design of their optimization spaces.
For example, analogously to our approach, functional approach Lift
relies on an internal low-level representation~\cite{7863730} that is used as target for Lift's high-level programs.
However, Lift's transformation process, from high level to low level, turned out to be challenging:~%
Lift's lowering process relies on an infinitely large optimization space~--~identifying a well-performing configuration within that space is too complex to be done automatically in general, due to the space's large and complex structure.
As a workaround, Lift currently uses approach Elevate~\cite{10.1145/3408974} to incorporate user knowledge into the optimization process;~however, at the cost of productivity, as manually expressing optimization is challenging, particularly for non-expert users.

In contrast, our optimization process is designed as auto-tunable (Table~\ref{tab_tps}), thereby achieving fully automatically high performance, as confirmed in our experiments (Section~\ref{ch:eval}), without involving the user for optimization decisions.
In particular, our previous work already showed that our approach~--~even in its original, proof-of-concept implementation~\cite{8891668}~--~can significantly outperform Lift on GPU and CPU~\cite{8891668}.
Our performance advantage over Lift is mainly caused by the design of our optimization process:~relying on formally defined tuning parameters (Table~\ref{tab_tps}), rather than on formal transformation rules that span a too large and complex search space (as Lift), thereby contributing to a simpler, fully auto-tunable optimization process.

\paragraph{Portability}
The current functional approaches usually are designed and optimized toward code generation in a particular programming model only.
For example, Lift inherently relies on the OpenCL programming model, because OpenCL works for multiple kinds of architectures: NVIDIA GPU, Intel CPU, etc.
However,
we see two major disadvantages in addressing the portability issue via OpenCL only:~%
1)~GPU-specific optimizations (such as \emph{shuffle operations}~\cite{cuda_shuffles}) are available only in the CUDA programming model, but not in OpenCL;~%
2)~the set of OpenCL-compatible devices is broad but still limited;~in particular, in the \emph{new golden age for computer architectures}~\cite{10.1145/3282307}, upcoming architectures are arising continuously and may not support the OpenCL standard.
We consider targeting new programming models as challenging for Lift, as its formal low-level representation is inherently designed for OpenCL~\cite{7863730};~targeting further programming models with Lift would require the design and implementation of new low-level representations, which we do not consider as straightforward.

To allow easily targeting new programming models with our approach, we have designed our formalism as generic in the target model:~our low-level representation (Figure~\ref{fig_gen_ll}) and optimization space (Table~\ref{tab_tps}) are designed and optimized toward an \emph{Abstract System Model} (Definition~\ref{def_asm}) which is capable of representing the device models of important programming approaches, including OpenMP, CUDA, and OpenCL (Example~\ref{example_asm}).
Furthermore, we have designed our high- and low-level representations as minimalistic (Figures~\ref{fig_generic_hl} and~\ref{fig_gen_ll}), e.g., by relying on three higher-order functions only for expressing programs at the high abstraction level, which simplifies and reduces the development effort for implementing code generators for programming models.

In addition, we believe that compared to our approach, the following basic design decisions of Lift (and similar functional approaches) complicate the process of code generation for them and increase the development effort for implementing code generators:~%
1)~relying on a vast set of small patterns for expressing computations,
rather than aiming at a minimalistic design as we do
(as also discussed in Section~\ref{ch:eval:functional});~%
2)~relying on complex function nestings and compositions for expressing computations,
rather than avoiding nesting and relying on a fixed composition structure of functions, as in our approach (Figure~\ref{hl_overview});~%
3)~requiring new patterns for targeting new classes of data-parallel computations (such as patterns \texttt{slide} and \texttt{pad} for stencils~\cite{10.1145/3168824}), which have to be non-trivially integrated into Lift's type and optimization system (often via extensions of the systems~\cite{10.1145/2884045.2884046,10.1145/3168824}), instead of relying on a fixed set of expressive patterns (Figure~\ref{fig_generic_hl}) and generalized optimizations (Table~\ref{tab_tps}) that work for various kinds of data-parallel computations (Figure~\ref{fig_hl_examples});~%
4)~expressing high-level and low-level concepts in the same language, instead of separating high-level and low-level concepts
for a more structured and thus simpler code generation process
(Figure~\ref{contributions_overall}).
We consider these four design decisions as disadvantageous for code generation, because they require from a code generator handling various kinds of patterns (decision 1), and the patterns need to be translated to significantly different code variants, depending on their nesting level and composition order (decision 2).
Moreover, each extension of patterns (decision 3) might affect code generation also for the already supported patterns, because the existing patterns need to be combined with the new ones via composition and nesting (decision 2).
We consider mixing up high-level and low-level concepts in the same language (decision 4) as further complicating the code generation process, because code generators cannot be implemented in clear, distinct stages:~\emph{high-level language} $\rightarrow$ \emph{low-level language} $\rightarrow$ \emph{executable program code}.

\paragraph{Productivity}

Functional approaches are expressive frameworks~--~to the best of our knowledge, the majority of these approaches
should also be able
to express
(possibly after some extension)
many of the
high-level programs that can also be expressed via our high-level representation (e.g., those presented in Figure~\ref{fig_hl_examples}).

A main difference we see between the high-level representations of existing functional approaches and the representation introduced by our approach is that the existing approaches rely on a vast set of higher-order functions for expressing computations;~these functions have to be functionally composed and nested in complex ways for expressing computations.
For example, expressing matrix multiplication in Lift requires also involving Lift's pattern \texttt{transpose} (also when operating on non-transposed input matrices)~\cite{10.1145/2884045.2884046}, as per design in Lift, multi-dimensional data is considered as an array of arrays (rather than a multi-dimensional array, as in our approach as well as polyhedral approaches).
In contrast, we aim to keep our high-level language minimalistic, by expressing data-parallel computations using exactly three higher-order functions and which are always used in the same, fixed order (shown in Figure~\ref{hl_overview}).
\citet{mdpoly,mdpoly_src} confirm that due to the minimalistic and structured design of our high-level representation, programs in our representation can even be systematically generated from straightforward, \mbox{sequential program code}.

Functional approaches also tend to require extension when targeting new application areas, which hinders the expressivity of the frameworks and thus also their productivity.
For example, functional approach Lift~\cite{10.1145/2784731.2784754} required notable extension for targeting, e.g., matrix multiplications (so-called \emph{macro-rules} had to be added to Lift~\cite{10.1145/2884045.2884046}) and stencil computations (primitives \texttt{slide} and \texttt{pad} were added, and Lift's tiling optimization had to be extended toward \emph{overlapped tiling}~\cite{10.1145/3168824}).
In contrast, we have formally defined our class of targeted computations (as MDH functions, Definition~\ref{def_mdh}), and the generality of our approach allows expressing matrix multiplications and stencils out of the box, without relying on domain-specific building blocks.

\subsection{Domain-Specific Approaches}
\label{sec_rw_dsa}
Many approaches focus on code generation and optimization for particular domains.
A popular domain-specific approach is \emph{ATLAS}~\cite{1437325}
for linear algebra routines on CPUs\footnote{
  Previous work~\cite{10.1145/3427093} shows that MDH (already in its original, proof-of-concept implementation) achieves higher performance than ATLAS.
}.
Similar to ATLAS, approach \emph{FFTW}~\cite{681704} targets \emph{Fast Fourier Transform (FFT)}, and \emph{SPIRAL}~\cite{1386651} works for \emph{Digital Signal Processing (DSP)}.

Nowadays, the best performing, state-of-practice domain-specific approaches are often provided by vendors and specifically designed and optimized toward their target application domain and also architecture.
For example, the popular vendor library NVIDIA cuBLAS~\cite{cublas} is optimized by hand, on the assembly level, toward computing linear algebra routines on NVIDIA GPUs~--~cuBLAS is considered in the community as gold standard for computing linear algebra routines on GPUs.
Similarly, Intel's oneMKL library~\cite{onemkl} computes with high performance linear algebra routines on Intel CPUs, and libraries NVIDIA cuDNN~\cite{cudnn} and Intel oneDNN~\cite{onednn} work well for convolution computations on either NVIDIA GPU (cuDNN) or Intel CPU (oneDNN), respectively.

\newpage

In the following, we discuss domain-specific approaches in terms of \emph{performance}, \emph{portability}, and \emph{productivity}.

\paragraph{Performance}
Domain-specific approaches, such as cuBLAS and cuDNN, usually achieve high performance.
This is because the approaches are hand-optimized by performance experts~--~on the assembly level~--~to exploit the full performance potential of their target architecture.
In our experiments (Section~\ref{ch:eval}), we show that our approach often achieves competitive and sometimes even better performance than domain-specific approaches provided by NVIDIA and Intel, which is mainly caused by their portability issues across different data characteristics, as we discuss in the next paragraph.

\paragraph{Portability}
Domain-specific approaches usually struggle with achieving portability across different architectures.
This is because the approaches are often implemented in architecture-specific assembly code to achieve high performance, but thereby also being limited to their target architecture.
The domain-specific approaches often also struggle with achieving portability of performance across different data characteristics (e.g., their sizes):~the approaches usually rely on a set of pre-implemented implementations that are each designed and optimized toward average high performance across a range of data characteristic.
In contrast, our approach (as well as many scheduling and polyhedral approaches) allow automatically optimizing (auto-tuning) computations for particular data characteristics, which is important for achieving high performance~\cite{10.1145/3126908.3126939}.
Thereby, our approach often outperforms domain-specific approaches (as confirmed in Section~\ref{ch:eval}), particularly for advanced data characteristics (small, uneven, irregularly shaped, $\dots$), e.g., as used in deep learning.
The costly time for auto-tuning is well amortized in many application areas, because the auto-tuned implementations are re-used in many program runs.
Furthermore, auto-tuning avoids the time-intensive and costly process of hand-optimization by human experts.

\paragraph{Productivity}
Domain-specific approaches usually achieve highest productivity for their target domain (e.g., linear algebra), by providing easy to use high-level abstractions.
However, the approaches suffer from significant expressivity issues, because~--~per design~--~they are inherently restricted to their target application domain only.
Also, the approaches are often inherently bound to only particular architectures, e.g., only GPU (as NVIDIA cuBLAS and cuDNN) or only CPU (as Intel oneMKL and oneDNN).
Domain-specific vendor libraries, such as NVIDIA cuBLAS and Intel oneMKL, also tend to offer the user differently performing variants of computations;~the variants have to be naively tested by the user when striving for the full performance potentials of approaches (as discussed in Section~\ref{ch:eval:domain}), which is cumbersome for the user.

\subsection{Higher-Level Approaches}
\label{sec_rw_hla}

There is a broad range of existing work that is focused on higher-level optimizations than proposed by this work.
We consider such higher-level approaches as greatly combinable with our approach.
For example, the polyhedral approach is capable of expressing algorithmic-level optimizations, like \emph{loop skewing}~\cite{97902}, to make programs parallelizable;~such optimizations are beyond the scope of this work, but they can be combined with our approach as demonstrated by \citet{mdpoly,mdpoly_src}.
Similarly, we consider the approaches introduced by~\citet{814600,10.1145/3434301,10.1145/3314221.3314612,10.1145/504210.504213}, which also focus on algorithmic-level optimizations, as greatly combinable with our approach:~algorithmically optimizing user code according to the approaches' techniques, and using our methodologies to eventually map the optimized code to executable program code for parallel architectures.

Futhark~\cite{10.1145/3062341.3062354}, Dex~\cite{10.1145/3473593}, and ATL~\cite{10.1145/3498717} are further approaches focussed on high-level program transformations, like advanced \emph{flattening} mechanisms~\cite{10.1145/3293883.3295707}, thereby optimizing programs at the algorithmic level of abstraction.
We consider using our work as backend for these approaches as promising:~the three approaches often struggle with mapping their algorithmically optimized program variants eventually to the multi-layered memory and core hierarchies of state-of-the-art parallel architectures, which is exactly the focus of this work.

\subsection{Existing Work on MDH}
\label{sec_rw_mdh}

Our work is inspired by the algebraic formalism of Multi-Dimensional Homomorphisms (MDHs) which is introduced in the work-in-progress paper~\cite{rasch2018multi}.
The MDH approach, as presented in the previous work, relies on a semi-formal foundation and focuses on code generation for
the OpenCL programming model only~\cite{8891668}.
This work makes major contributions over the existing work on MDHs and its OpenCL code generation approach.

We introduce a full formalization of MDH's high-level program representation.
In our new formalism, we rely on expressive typing:~for example, we encode MDHs' data sizes into our type system, e.g., by introducing both \emph{index sets} for MDAs (Definition~\ref{def_mda}) and \emph{index set functions} for combine operators (Definition~\ref{def_combine_op}), and we respect and maintain these sets and functions thoroughly during MDH computations (Definition~\ref{def_mdh}).
Our expressive typing significantly contributes
to correct and simplified code generation,
as all relevant type and data size information are contained in our formal, low-level program representation (Figure~\ref{fig_gen_ll}) from which we eventually generate executable program code~(Section~\ref{ch:low_level}).
In contrast, the existing MDH work considers multi-dimensional arrays~(MDAs) of arbitrary sizes and dimensionalities to be all of the same, straightforward type, which has greatly simplified the design of the proof-of-concept MDH formalism introduced by~\citet{rasch2018multi} (in particular, the definition and usage of combine operators), but at the cost of significantly harder and error-prone code generation:~all the missing, type-relevant information need to be elaborated by the implementer of the code generator in the existing MDH work, e.g., allocation sizes of fast memory resources used for caching input data or for storing computed intermediate results.
Furthermore, while the original MDH work~\cite{rasch2018multi} is focused on introducing higher-order function $\mdh$ only, this work particularly also introduces higher order functions $\iv$ and $\ov$ (Section~\ref{sec_views}) which express input and output views in a formally structured and concise manner, and which are central building blocks in our new approach for expressing computations (Figure~\ref{fig_hl_examples}).
Also, by introducing and exploiting the index set concept for MDAs, we have improved the definition of the concatenation operator $\dplus$ (Example~\ref{def:mda_concat}) toward commutativity, which is required for important optimizations. e.g., loop permutations (expressed via Parameters~\texttt{D1},~\texttt{S1},~\texttt{R1} in Table~\ref{tab_tps}).

A further substantial improvement is the introduction of our low-level representation (Section~\ref{ch:low_level}).
It relies on a novel combination of tuning parameters (Table~\ref{tab_tps}) that enhance, generalize, and extend the existing, proof-of-concept MDH parameters which capture
a subset of OpenCL-orientated features only~\cite{8891668}.
Moreover, while the existing MDH work introduces formally only parameters for flexibly choosing numbers of threads~\cite{rasch2018multi} (which corresponds to a very limited variant of our tuning parameter~\texttt{0} in Table~\ref{tab_tps}, because our parameter~\texttt{0} also choses numbers of memory tiles and is not restricted to OpenCL), the other OpenCL parameters are introduced and discussed by~\citet{8891668} only informally, from a technical perspective.
With our novel parameter set, we are able to target various kinds of programming models (e.g., also CUDA, as in Section~\ref{ch:eval}) and also to express important optimizations that are beyond the existing work on MDH, e.g., optimizing the memory access pattern of computations:~for example, we achieve speedups $>2\times$
over existing MDH for the deep learning computations discussed in Section~\ref{ch:eval}.
Our new tuning parameters are expressive enough to represent state-of-the-art, data-parallel implementations, e.g., as generated by scheduling and polyhedral approaches (Figures~\ref{tp_ll_tab_tvm_gpu}-\ref{tp_ll_tab_pluto_cpu}), and our experiments in Section~\ref{ch:eval} confirm that auto-tuning our parameters enables performance beyond the state of the art, including hand-optimized solutions provided by vendors, which is not possible
when using the existing MDH approach.
The expressivity of our parameters particularly also enables comparing significantly differently optimized implementations (e.g., scheduling-optimized vs. polyhedral-optimized, as in Section~\ref{ll_examples_matmul_resnet}), based on the values of formally specified tuning parameters,
which we consider as promising for structured performance analysis in future work.
Moreover, our new low-level representation targets architectures that may have arbitrarily deep memory and core hierarchies, by having optimized our representation toward an \emph{Abstract System Model}~(Definition~\ref{def_asm});~in contrast, the existing MDH work is focused on OpenCL-compatible architectures only.

Our experimental evaluation extends the previous MDH experiments by comparing also to the popular state-of-practice approach TVM which is attracting increasing attention from both academia~\cite{tvmcon} and industry~\cite{octoml}.
Also, we compare to the popular polyhedral compilers PPCG and Pluto, as well as the currently newest versions of hand-optimized, high-performance libraries provided by vendors.
Furthermore, we have included a real-world case study in our experiments, considering the most time-intensive computations within the three popular deep learning neural networks \texttt{ResNet-50}, \texttt{VGG-16}, and \texttt{MobileNet};~the study also includes Capsule-style convolution computations, which are considered challenging to optimize~\cite{10.1145/3317550.3321441}.
Moreover, Table~\ref{fig_hl_examples} analyzes MDH's expressivity using new examples:~it shows that MDH~--~based on the new contributions of this work (e.g., view functions)~--~is capable of expressing computations \texttt{bMatMuL}, \texttt{MCC\_Capsule}, \texttt{Histo}, \texttt{scan}, and \texttt{MBBS}, which have not been expressed via MDH in previous work.
Our experiments confirm that we achieve high performance for \texttt{bMatMuL} and \texttt{MCC\_Capsule} on GPUs and CPUs, and our future work aims to thoroughly analyze our approach for computations \texttt{Histo}, \texttt{scan}, and \texttt{MBBS} in terms of performance, portability, and productivity.

\section{Conclusion}
\label{ch:conclusion}

We introduce a formal (de/re)-composition approach for data-parallel computations targeting state-of-the-art parallel architectures.
Our approach aims to combine three major advantages over related approaches~--~{performance}, {portability}, and {productivity}~--~by introducing formal program representations on both:~%
1)~\emph{high level}, for conveniently expressing~--~in one uniform formalism~--~various kinds of data-parallel computations (including linear algebra routines, stencil computations, data mining algorithms, and quantum chemistry computations), agnostic from hardware and optimization details, while still capturing all information relevant for generating high-performance program code;
2)~\emph{low level}, which
allows uniformly reasoning~--~in the same formalism~--~about optimized \mbox{(de/re)-compositions} of data-parallel computations targeting different kinds of parallel architectures (GPUs, CPUs, etc).
We \emph{lower} our high-level representation to our low-level representation, in a formally sound manner, by introducing a generic search space that is based on performance-critical parameters.
The parameters of our lowering process enable fully automatically optimizing (auto-tuning) our low-level representations for a particular target architecture and characteristics of the input and output data,
and our low-level representation is designed such that it can be straightforwardly transformed to executable program code in imperative-style programming languages (including OpenMP, CUDA, and OpenCL).
Our experiments confirm that due to the design and structure of our generic search space in combination with auto-tuning, our approach achieves higher performance on GPUs and CPUs than popular state-of-practice approaches, including hand-optimized libraries provided by vendors.

\section{Future Work}
\label{ch:fw}

We consider this work as a promising starting point for future directions.
A major future goal is to extend our approach toward expressing and optimizing simultaneously multiple data-parallel computations (e.g., matrix multiplication followed by convolution), rather than optimizing computations individually and thus independently from each other (e.g., only matrix multiplication or only convolution).
Such extension enables optimizations, such as \emph{kernel fusion}, which is important for the overall application performance and considered challenging~\cite{7013003,9741270,10.1145/3519941.3535078}.
We see this work as a promising
foundation
for our future goal, because it enables expressing and reasoning about different computations in the same formal framework.
Targeting computations on sparse input/output data formats, inspired by~\citet{10.1145/3377555.3377896,9407251,10.1145/3155284.3018756,10.1145/3133901}, is a further major goal, which requires extending our approach toward irregularly-shaped input and output data, similarly as done by~\citet{10.1145/3377555.3377896}.
Regarding our optimization process, we aim to introduce an analytical cost model for computations expressed in our formalism~--~based on operational semantics~--~%
thereby accelerating (or even avoiding) the auto-tuning overhead, similarly as done by~\citet{10.1145/3434306,10.1145/3445814.3446759}.
Moreover, we aim to incorporate methods from machine learning into our optimization process~\cite{10.1145/2536688,merouani2024looper}, instead of relying on empirical auto-tuning methods only.
To make our work better accessible for the community, we aim to implement our approach into \emph{MLIR}~\cite{9370308} which offers a reusable compiler infrastructure.
The contributions of this work give a precise, formal recipe of how to implement our introduced methods into approaches such as MLIR.
Moreover, relying on the MLIR framework will contribute to a structured code generation process in assembly-level programming models, such as LLVM~\cite{1281665} and NVIDIA PTX~\cite{ptx}.
We consider targeting assembly languages as important for our future work:~assembly code offers further, low-level optimization opportunities~\cite{6494986,10.1145/1356052.1356053}, thereby enabling our approach to potentially achieve higher performance than reported in Section~\ref{ch:eval} for our generated CUDA and OpenCL code.
Also, we aim to extend our approach toward distributed multi-device systems that are heterogeneous, inspired by dynamic load balancing approaches~\cite{5470413} and advanced data distributions techniques~\cite{10.1145/3519939.3523437}.
Targeting domain-specific hardware extensions, such as \emph{NVIDIA Tensor Cores}~\cite{cuda_tensor_cores} is also an important goal for our future work, as such extensions allow significantly accelerating computations for the target of the extensions (e.g., deep learning~\cite{8425458}).
Finally, we aim to support more target backends (additionally to OpenMP, CUDA, and OpenCL), e.g., AMD's \emph{HIP}~\cite{amd_hip} which is efficient for programming AMD GPUs.
Similarly, we consider \emph{Triton}~\cite{10.1145/3315508.3329973}, \emph{AMOS}~\cite{10.1145/3470496.3527440}, and \emph{Graphene}~\cite{10.1145/3582016.3582018} as further, promising backends for our approach.

\begin{acks}
This work was funded by
the \emph{Deutsche Forschungsgemeinschaft (DFG, German Research Foundation)} -- project \emph{PPP-DL (470527619)}.
I would like to thank Richard Schulze for conducting the extensive set of experiments and the reviewers for their thorough reading of the paper and their comments and remarks that helped to improve this work.
\end{acks}

\newpage

\bibliography{references.bib}

\newpage

\begin{appendix}
\section*{Appendix}

Our appendix provides details for the interested reader that
should not
be required for understanding the basic concepts and ideas introduced in this paper.

\section{Mathematical Foundation}
\label{sec_math_foundation}

We rely on a set theoretical foundation, based on ZFC set theory~\cite{ciesielski1997set}.
We avoid class theory, such as NBG~\cite{Neumann+1925+219+240}, by assuming, for example, that our universe of types contains all relevant representatives (\texttt{int}, \texttt{float}, \texttt{struct}, etc), but is not the "class of all types".
Thereby, we avoid fundamental issues~\cite{russell2020principles}
which are not relevant for this work.

\subsection{Family}

\begin{definition}[Family]
Let $I$ and $A$ be two sets.
A \emph{family} $F$ from $I$ to $A$ is any set
\[
F := \{\ (i,a) \ | \ i\in I \wedge a\in A \ \}
\]
such that the following two properties are satisfied:
\begin{itemize}

\item  \emph{left-total}: $\forall i\in I:\,\exists a\in A:\,(i,a)\in F$
\item  \emph{right-unique}: $(i,a)\in F \ \wedge \ (i,a')\in F \ \Rightarrow \ a=a'$

\end{itemize}
We refer to $I$ also as \emph{index set} of family $F$ and to $A$ as $F$'s \emph{image set}.
If $I$ has a strict total order $<$,
we refer to $F$ also as \emph{ordered family}.
\end{definition}

\begin{notation}[Family]
\label{app_not_family}
  Let $F$ be a family from $I$ to $A$.

  We write:
  \begin{itemize}
    \item $F_i$ for the unique $a\in A$ such that $(i,a)\in F$;

    \item $(F_i)_{i\in I}$ instead of $F$ to explicitly state $F$'s index and image sets in our notation;

    \item $(\,F_{i_1,\dotsc,i_n}\,)_{i_1\in I_1,\dotsc,i_n\in I_n}$
    instead of $(\dotsc(\,F_{i_1,\dotsc,i_n}\,)_{i_n\in I_n}\,\dotsc)_{i_1\in I_1}$.

  \end{itemize}

  Alternatively, depending on the context, we use the following notation:
  \begin{itemize}
    \item $F^{<i>}$  \ instead of $F_i$;

    \item $(F_i)^{<i\in I>}$  \ instead of $(F_i)_{i\in I}$;

    \item $(\,F_{i_1,\dotsc,i_n}\,)^{<i_1\in I_1\,|\,\dotsc\,|\,i_n\in I_n>}$ \ instead of $(\,F_{i_1,\dotsc,i_n}\,)_{i_1\in I_1,\dotsc,i_n\in I_n}$.
  \end{itemize}
  For nested families, each index set $I_k$ may depend on the earlier-defined values $i_1,\dotsc,i_{k-1}$ (not explicitly stated above for brevity).
\end{notation}

\vspace*{5px}

\begin{definition}[Tuple]
We identify $n$-tuples as families that have index set $[1,n]_\IN$.
\end{definition}

\begin{example}[Tuple]
A $2$-tuple $(a,b)$ (a.k.a \emph{pair}) is a family $(F_i)_{i\in I:=[1,2]_\IN}$ for which $F_1=a$ and $F_2=b$.
\end{example}

\subsection{Scalar Types}
\label{app_scalar_type}
We denote by
\[
\type := \{ \, \texttt{int}, \texttt{int8}, \,\texttt{int16}, \,\dotsc, \ \texttt{float},\, \texttt{double},\, \dotsc, \ \texttt{struct}, \dotsc  \ \}
\]
our set of \emph{scalar types}, where
\texttt{int8} and \texttt{int16} represent 8-bit/16-bit integer numbers,
\texttt{float} and \texttt{double} are the types of single/double precision floating point numbers (IEEE 754 standards),
\texttt{struct}s contain a fixed set of other scalar types, etc.
For simplicity, we interpret integer types (\texttt{int}, \texttt{int8}, \,\texttt{int16}, $\dotsc$) uniquely as integers $\IZ$, floating point number types (\texttt{float} and \texttt{double}) as rationale numbers $\IQ$, etc.

For high flexibility, we avoid fixing $\type$ to a particular set of scalar types, i.e., we assume that $\type$ contains all practice-relevant types.
This is legal, because our formalism makes no assumptions on the number and kinds of scalar types.

We consider operations on scalar types (addition, multiplication, etc) to be:~%
1)~\emph{atomic}: we do not aim at parallelizing or otherwise optimizing operations on scalar values in this work;~%
2)~\emph{size preserving}: we assume that all values of a scalar type have the same arbitrary but fixed size.

Note that we can potentially also define, for example, the set of arbitrarily sized matrices $\{ T^{m\times n} \ | \ m,n\in\IN, T\in\type \}$ as scalar type in our approach.
However, this would prevent any kind of formal reasoning about type correctness and performance of matrix-related operations (e.g., matrix multiplication), such as parallelization (due to our atomic assumption above) or type correctness (e.g., assuring in matrix multiplication that number of columns of the first input matrix coincides with and number of rows of the second matrix:~due to our size preservation assumption above, we would not be able to distinguish matrices based on their sizes).

\subsection{Functions}

\begin{definition}[Function]
Let $A\in\type$ and $B\in\type$ be two scalar types.

A \emph{(total) function}
$f$
is a tuple of the form
\[
f \ \in \ \{ \ \ ( \ \underbrace{(\,A, \ B\,)}_{\texttt{function type}}\,, \ \underbrace{G_f}_{\substack{\texttt{function}\\[2pt]\texttt{graph}}} \ ) \ | \ G_f\subseteq\{ \ (a,b) \ | \ a\in A \wedge b\in B \ \} \ \ \}
\]
that satisfies the following two properties:~%
\begin{itemize}
  \item \emph{left-total}:~$\forall a\in A:\,\exists b\in B: \, (a,b)\in G_f$;
  \item \emph{right-unique}:~$(a,b)\in G_f \ \wedge \ (a,b')\in G_f \ \Rightarrow \ b=b'$.
\end{itemize}
We write $f(a)$ for the unique $b\in B$ such that $(a,b)\in G_f$.
Moreover, we denote $f$'s function type as $A\to B$, and we write $f:A\to B$ to state that $f$ has function type $A\to B$.
\vspace*{3px} \\
We refer to:
\begin{itemize}
\item $A$ as the \emph{domain} of $f$
\item $B$ as the \emph{co-domain} (or \emph{range}) of $f$
\item $(A,B)$ as the \emph{type} of $f$
\item $G_f$ as the \emph{graph} of $f$
\end{itemize}
If $f$ does not satisfy the left total property, we say $f$ is \emph{partial}, and we denote $f$'s type as $f:A\to_p B$.
\end{definition}

We allow functions to have so-called \emph{dependent types}~\cite{10.1145/292540.292560} for expressive typing.
For example, dependent types enable encoding the sizes of families into the type system, which contributes to better error checking.
We refer to dependently typed functions as \emph{meta-functions}, as outlined in the following.

\begin{definition}[Meta-Function]
We refer to any family of functions
\[
( \ f^{<i>}:A^{<i>}\to B^{<i>} \ )^{<i\in I>}
\]
as \emph{meta-function}.
In the context of meta-functions, we refer to index $i\in I$ also as \emph{meta-parameter}, to index set $I$ as \emph{meta-parameter type}, to $A^{<i\in I>}$ and $B^{<i\in I>}$ as meta-types (as both are generic in meta-parameter $i\in I$), and to $A^{<i>}\to B^{<i>}$ for concrete $i$ as meta-function $f$'s ordinary function type.

\noindent
In the following, we often write:
\begin{itemize}
  \item $f^{<i\in I>}:A^{<i>}\to B^{<i>}$ instead of $( \, f^{<i>}:A^{<i>}\to B^{<i>} \, )^{<i\in I>}$;
  \item $f^{<i>}:A'\to B^{<I>} $ (or $f^{<i>}:A^{<i>}\to B'$) iff $A^{<i>}=A'$ (or $B^{<i>}=B'$) for all $i\in I$.
\end{itemize}
\end{definition}

We use \emph{multi-stage meta-functions} as a concept analogous to \emph{multi staging}~\cite{10.1145/258993.259019} in programming and similar to \emph{currying} in mathematics.

\begin{definition}[Multi-Stage Meta-Function]
\label{app_multi_stage_meta_function}

A \emph{multi-stage meta-function} is a nested family of functions:
\begin{align*}
&
f^{
  \overbrace{<i_1\in I_1^{<>}>}^{\texttt{stage $1$}}
  \ \dotsc \
  \overbrace{<i_S\in I_S^{<i_1,\dotsc,i_{S-1}>}>}^{\texttt{stage $S$}}
  }: \
\underbrace{A^{<i_1,\dotsc,i_S>} \ \to \  B^{<i_1,\dotsc,i_S>}}_{\texttt{function instance}}
\end{align*}
Here,
$I_s^{<i_1,\dotsc,i_{s-1}>}$, $s\in[1,S]_\IN$,
is the meta-parameter type on stage $s$, which may depend on all meta-parameters of the previous stages $i_1,\dotsc,i_{s-1}$.
We refer to such meta-functions also as \emph{$S$-stage meta-functions}, and we denote their type also
as
\[
f^{<i_1\in I_1^{<>} \ | \ \dotsc \ | \ i_S\in I_S^{<i_1,\dotsc,i_{S-1}>}>}:A^{<i_1,\dotsc,i_S>} \ \to \  B^{<i_1,\dotsc,i_S>}
\]
and access to them as
\[
f^{<i_1 \ | \ \dotsc \ | \ i_S>}(x)
\]
where different stages are separated by vertical bars.
\end{definition}

\vspace*{5px}

We allow partially applying parameters (meta and ordinary) of meta-functions.

\begin{definition}[Partial Meta-Function Application]
\label{app_partial_meta_func_app}

Let
\begin{align*}
f^{<i_1\in I_1 \, | \,\dotsc \, | \,i_S\in I_S>}:A^{<i_1,\dotsc,i_S>} \ \to \  B^{<i_1,\dotsc,i_S>}
\end{align*}
be a meta-function (meta-parameters of meta-types $I_1,\dotsc,I_S$ omitted for brevity).

\begin{itemize}
  \item The \emph{partial application} of meta-function $f$ on stage $s$ to meta-parameter $\hat{i}_s$ is the meta-function
  \[
  f'^{<{i}_1\in \hat{I}_1 \, | \,\dotsc \, | \,{i}_{s-1}\in \hat{I}_{s-1} \, | \,i_{s+1}\in I_{s+1} \, | \,\dotsc \, | \,i_S\in I_S>}:A^{<i_1,\dotsc,i_{s-1},\,\hat{i}_s\,,i_{s+1}\dotsc,i_S>} \ \to \  B^{<i_1,\dotsc,i_{s-1},\,\hat{i}_s\,,i_{s+1},\dotsc,i_S>}
  \]
  where $\hat{I}_1\subseteq I_1\:,\dotsc,\:\hat{I}_{s-1}\subseteq I_{s-1}$ are the largest sets such that $\hat{i}_s\in I_s^{<{i}_1,\dotsc,{i}_{s-1}>}$ for all ${i}_1\in \hat{I}_1\:,\dotsc,\:{i}_{s-1}\in \hat{I}_{s-1}$.
  The function is defined as:
  \begin{align*}
  f'^{<{i}_1 \, | \,\dotsc \, | \,{i}_{s-1} \, | \,i_{s+1} \, | \,\dotsc \, | \,i_S>}(x) := f^{<{i}_1 \, | \,\dotsc \, | \,{i}_{s-1} \, | \,\hat{i}_s \, | \,i_{s+1} \, | \,\dotsc \, | \,i_S>}(x)
  \end{align*}
  We write for $f'$'s type also
  \[
  f^{<{i}_1\in \hat{I}_1 \, | \,\dotsc \, | \,{i}_{s-1}\in \hat{I}_{s-1} \, | \,\hat{i}_s \, | \,i_{s+1}\in I_{s+1} \, | \,\dotsc \, | \,i_S\in I_S>}:A^{<i_1,\dotsc,i_{s-1},\,\hat{i}_s\,,i_{s+1}\dotsc,i_S>} \ \to \  B^{<i_1,\dotsc,i_{s-1},\,\hat{i}_s\,,i_{s+1},\dotsc,i_S>}
  \]
  where $f'$ is replaced by $f$ and $i_s\in I_s$ is replaced by the concrete value $\hat{i}_s$.

    \item The \emph{partial application} of meta-function $f$ to ordinary parameter $x$ is the meta-function
    \begin{align*}
    f'^{<{i}_1\in\hat{I}_1 \, | \,\dotsc \, | \,{i}_S\in\hat{I}_S>}:\underbrace{B_1^{<{i}_1,\dotsc,{i}_S>} \ \to \  B_2^{<{i}_1,\dotsc,{i}_S>}}_{=\,B^{<{i}_1,\dotsc,{i}_S>}}
    \end{align*}
    where $\hat{I}_1\subseteq I_1\:,\dotsc,\:\hat{I}_{S}\subseteq I_{S}$ are the largest sets such that $x\in A^{<{i}_1,\dotsc,{i}_S>}$ for all ${i}_1\in\hat{I}_1\:,\dotsc,\:{i}_{S}\in\hat{I}_{S}$.
    The function is defined as:
    \[
      f'^{<{i}_1 \, | \,\dotsc \, | \,{i}_S>}(\underbrace{x'}_{\in B_1^{<...>}}) := f^{<{i}_1 \, | \,\dotsc \, | \,{i}_S>}(\underbrace{x}_{\in A^{<...>}})(\underbrace{x'}_{\in B_1^{<...>}})
    \]
\end{itemize}
\end{definition}

We allow \emph{generalizing} meta-parameters.
For example, when generalizing meta-parameters that express input sizes, we allow using the corresponding meta-function on arbitrarily sized inputs (a.k.a. \emph{dynamic size} in programming).
In our generated code, meta-parameters are available at compile time such that concrete meta-parameter values can be exploited for generating well-performing code (e.g., for setting static loop boundaries).
Consequently, meta-parameter generalization increases the expressivity of the generated code (e.g., by being able to process differently sized inputs without re-compilation for unseen input sizes), but usually at the cost of performance, because generalized meta-parameters cannot be exploited during code generation.

\begin{definition}[Generalized Meta-Parameters]
\label{app_gen_meta_func}

Let
\begin{align*}
f^{<i_1\in I_1 \, | \,\dotsc \, | \,i_s\in I_s \, | \,\dotsc \, | \,i_S\in I_S>}:A^{<i_1,\dotsc,i_s,\dotsc,i_S>} \ \to \  B^{<i_1,\dotsc,i_s,\dotsc,i_S>}
\end{align*}
be a meta-function (meta-parameters of $I_1,\dotsc,I_S$ omitted for brevity) such that
\[
    f^{<i_1 \, | \,\dotsc \, | \,i_s \, | \,\dotsc \, | \,i_S>}(x)
    =
    f^{<i_1 \, | \,\dotsc \, | \,i'_s \, | \,\dotsc \, | \,i_S>}(x)
\]
i.e., $f$'s behavior is invariant under different values of meta-parameter $i_s$ in stage $s$.

  The \emph{generalization} of $f$ in meta-parameter $s\in[1,S]_\IN$ is the meta-function
  \begin{align*}
  &
  f'^{<i_{1}\in I_{1} \ | \ \dotsc \ | \ i_{s-1}\in I_{s-1} \ | \ i_{s+1}\in I_{s+1} \ | \ \dotsc \ | \ i_S \in I_{S}>}: \\[0px]
  &
  \hspace*{10px}
  \bigcup_{i_s\in I_s^{<i_1,\dotsc,i_{s-1}>}} A^{<i_{1} \ | \ \dotsc \ | \ i_{s-1} \ | \ i_s \ | \ i_{s+1} \ | \ \dotsc \ | \ i_S>}
  \to
  \bigcup_{i_s\in I_s^{<i_1,\dotsc,i_{s-1}>}} B^{<i_{1} \ | \ \dotsc \ | \ i_{s-1} \ | \ i_s \ | \ i_{s+1} \ | \ \dotsc \ | \ i_S>}
  \end{align*}
  which is defined as:
  \[
    f'^{<i_{1} \ | \ \dotsc \ | \ i_{s-1} \ | \ i_{s+1} \ | \ \dotsc \ | \ i_S>}(x) := f'^{<i_{1} \ | \ \dotsc \ | \ i_{s} \ | \ i_{s+1} \ | \ \dotsc \ | \ i_S>}(x)
  \]
  for an arbitrary $i_s\in I_s$ such that $x\in A^{<i_{1} \ | \ \dotsc \ | \ i_{s-1} \ | \ i_s \ | \ i_{s+1} \ | \ \dotsc \ | \ i_S>}$.

  We write for $f$'s type also
  \[
  f^{<i_{1}\in I_{1} \ | \ \dotsc \ | \ i_{s-1}\in I_{s-1} \ | \ *\in I_s \ | \ i_{s+1}\in I_{s+1} \ | \ \dotsc \ | \ i_S \in I_{S}>}: A^{<i_1,\dotsc,i_S>}\to B^{<i_1,\dotsc,i_S>}
  \]
  where $i_s$ is replaced by $*$, and for access to $f$
  \[
  f^{<i_{1} \ | \ \dotsc \ | \ i_{s-1} \ | \ * \ | \ i_{s+1} \ | \ \dotsc \ | \ i_S>}(x)
  \]
  \end{definition}

We use \emph{postponed meta-parameters} to change the order of meta-parameters of already defined meta-functions.
For example, we use postponed meta-parameters in Definition~\ref{def_iv} to compute the values of meta-parameters based on the particular meta-parameter values of later stages.

\begin{definition}[Postponed Meta-Parameters]
\label{app_del_meta_func}
Let
\begin{align*}
f^{<i_1\in I_1 \, | \,\dotsc \, | \,i_s\in I_s \, | \,\dotsc \, | \,i_S\in I_S>}:A^{<i_1,\dotsc,i_s,\dotsc,i_S>} \ \to \  B^{<i_1,\dotsc,i_s,\dotsc,i_S>}
\end{align*}
be a meta-function (meta-parameters of $I_1,\dotsc,I_S$ omitted via ellipsis for brevity) such that for each $k\in(s,S]_\IN$, it holds:
\[
I_{k}^{<i_1 \ | \ \dotsc \ | \ i_s   \ | \ \dotsc \ | \ i_{k-1}>} =
I_{k}^{<i_1 \ | \ \dotsc \ | \ i_s'   \ | \ \dotsc \ | \ i_{k-1}>}
\]
i.e., the $I_k$ are invariant under different values of meta-parameter $i_s$ in stage $s$ (i.e., $i_s$ can be ignored in the parameter list of $I_{k}$).

Function $f'$ is function $f$ \emph{postponed} on stage $s$ to meta-type
\[
\hat{I}_s^{<i_1 \ | \ \dotsc \ | \ i_{s-1} \ | \ i_{s+1} \ | \ \dotsc \ | \ i_S>}\subseteq I_s^{<i_1 \ | \ \dotsc \ | \ i_{s-1}>}
\]
which, in contrast to $I_s$, may also depend on meta-parameter values $i_{s+1},\dotsc,i_S$,
iff $f'$ is of type
\begin{align*}
f'^{<i_1\in I_1^{<\dotsc>} \ | \ \dotsc \ | \ i_{s-1}\in I^{<\dotsc>}_{s-1} \ | \ i_{s+1}\in I^{<\dotsc>}_{s+1} \ | \ \dotsc \ | \ i_S\in I^{<\dotsc>}_S><i_s\in\hat{I}^{<i_1,\dotsc,i_S>}_s>}:A^{<i_1,\dotsc,i_s,\dotsc,i_S>} \ \to \  B^{<i_1,\dotsc,i_s,\dotsc,i_S>}
\end{align*}
and defined as:
\[
    f'^{<i_1\,|\,\dotsc\,|\,i_{s-1}\,|\,i_{s+1}\,|\,\dotsc\,|\,i_S><i_s>}(a) =
    f^{<i_1\,|\,\dotsc\,|\,i_{s-1}\,|\,i_{s}\,|\,i_{s+1}\,|\,\dotsc\,|\,i_S>}(a)
\]

\noindent
We write for $f'$'s type also
\begin{align*}
f^{<i_1\in I_1^{<\dotsc>} \ | \ \dotsc \ | \ i_{s-1}\in I^{<\dotsc>}_{s-1} \ | \ \rightarrow \ | \ i_{s+1}\in I^{<\dotsc>}_{s+1} \ | \ \dotsc \ | \ i_S\in I^{<\dotsc>}_S><i_s\in\hat{I}^{<\dotsc>}_s>}:A^{<i_1,\dotsc,i_s,\dotsc,i_S>} \ \to \  B^{<i_1,\dotsc,i_s,\dotsc,i_S>}
\end{align*}
where $f'$ is replaced by $f$ and $i_s$ by symbol "$\rightarrow$".
For access to $f'$, we write
\begin{align*}
f^{<i_1 \ | \ \dotsc \ | \ i_{s-1} \ | \ \rightarrow \ | \ i_{s+1} \ | \ \dotsc \ | \ i_S><i_s>}(x)
\end{align*}
\end{definition}

When using a binary function for combining a family of elements, we often use the following notation.

\begin{notation}[Iterative Function Application]
\label{app_not_iterative_func_app}
  Let $\circledast:T\times T\to T$ be an arbitrary associative and commutative binary function on scalar type $T\in\type$.
  Let further
  $x$ be an arbitrary family that has index set $I:=\{i_1,\dotsc,i_N\}$ and image set $\{x_i\}_{i\in I}\subseteq T$.

  We write
  $
    \underset{i\in I}{\circledast} \ x_i
  $
  instead of
  $
    x_{i_1} \circledast \dotsc \circledast x_{i_N}
  $ (infix notation).
\end{notation}

\vspace*{10px}

\subsection{\texttt{MatVec} Expressed in MDH DSL}
\label{matvec_in_dsl}

Our \texttt{MatVec} example from Figure~\ref{fig:intro_example} is expressed in MDH's Python-based high-level \emph{Domain-Specific Language~(DSL)}, used as input by our \emph{MDH compiler}~\cite{mdh_in_python}, as follows:

\begin{lstlisting}[
    language=Python,
    morekeywords={__kernel__},
    mathescape=true,
    numbers=left,
    xleftmargin=1.5em,
    float=h!,
    frame=lines,
    caption={\texttt{MatVec} expressed in MDH's Python DSL},
    label={lst_matvec_in_mdh_dsl}
    ]
def matvec(T: ScalarType, I: int, K: int):
    @mdh( out( w = Buffer[T, [I]]                        ) ,
          inp( M = Buffer[T, [I, K]], v = Buffer[T, [K]] ) )
    def mdh_matvec():
        def mul(out, inp):
            out['w'] = inp['M'] * inp['v']

        def scalar_plus(res, lhs, rhs):
            res['w'] = lhs['w'] + rhs['w']

        return (
            out_view[T]( w = [lambda i, k: (i)] ),
              md_hom[I, K]( mul, ( CC, PW(scalar_plus) ) ),
                inp_view[T, T]( M = [lambda i, k: (i, k)] ,
                                v = [lambda i, k: (k)   ] )
        )
\end{lstlisting}

\noindent
Our MDH compiler takes an expression as in Listing~\ref{lst_matvec_in_mdh_dsl} as input, and it fully automatically generates auto-tuned program code from it, according to the methodologies presented in this paper (particularly in Section~\ref{app_sec_code_generation}).

We rely on a Python-based DSL, because Python is becoming increasingly popular in both academia and industry~\cite{tiobe}.
Our future work aims to offer our MDH-based DSL in further approaches, e.g., the \emph{MLIR} compiler framework~\cite{9370308} to make our approach better accessible to the community.

\newpage

\section{Addendum Section~\ref{ch:high_level}}

\subsection{Design Decisions: Combine Operators}
\label{app_design_dec_co}

We list some design decisions for combine operators (Definition~\ref{def_combine_op}).

\begin{note}

\begin{itemize}
\setlength\itemsep{0.2em}

\item[]

\item
We deliberately restrict index set function $\size{}{MDA}{MDA}$ to compute the index set in the particular dimension $d$ only, and not of all $D$ Dimensions
(i.e., the function's output is in $\IDXs$ and not $\IDXs^D$),
because this enables applying combine operator $\co{}$ iteratively:
\[
(\dotsc(\,(\MDA_1 \co{}^{<(P,Q)>} \MDA_2) \, \co{}^{<(P\cupdot Q,R)>} \MDA_3) \, \co{}^{<(P\cupdot Q\cupdot R,\dotsc)>} \, \dotsc
\]
for MDAs $\MDA_1,\MDA_2,\MDA_3,\dotsc$ that have index sets $\size{}{MDA}{MDA}(P), \size{}{MDA}{MDA}(Q), \size{}{MDA}{MDA}(R),\dotsc$ in dimension $d$.
This is because the index set of the output MDA changes only in dimension $d$, to the new index set $\size{}{MDA}{MDA}(P\cupdot Q)$, \ $\size{}{MDA}{MDA}(\size{}{MDA}{MDA}(P\cupdot Q)\cupdot R)$,~~$\dotsc$, so that the output MDA can be used as input for a new application of $\co{}$.

\item It is a design decision that a combine operator's index set function $\size{}{MDA}{MDA}$ takes as input the MDA index set $P$ or $Q$ in the particular dimension $d$ only, rather than the all sets $(I_1,\dotsc,I_D)$.
Our approach can be easily extended to index set functions $\size{}{MDA}{MDA}:\IDXs^D\to\IDXs$ that take the entire MDA's index sets as input.
However, we avoid this additional complexity, because we are currently not aware of any real-world application that would benefit from such extension.

\item
For better convenience, we could potentially define the meta-type of combine operators (Definition~\ref{def_combine_op}) such that meta-parameter $(I_1,\dotsc,I_{d-1},I_{d+1},\dotsc,I_D)$ is separated from parameter $(P,Q)$ in a distinct, earlier stage (Definition~\ref{app_multi_stage_meta_function}).
This would allow automatically deducing $(I_1,\dotsc,I_{d-1},I_{d+1},\dotsc,I_D)$ from the input MDAs' types, whereas for meta-parameter $(P,Q)$, automatic deduction is usually not possible:~function $\size{}{MDA}{MDA}$ has to be either invertible for automatically deducing $P$ and $Q$ from the input MDAs
or
invariant under different values of $P$ and $Q$.
Consequently, separating parameter $(I_1,\dotsc,I_{d-1},I_{d+1},\dotsc,I_D)$ in a distinct, earlier stage would allow avoiding explicitly stating this parameter, by deducing it from the input MDAs' type, and only explicitly stating parameter $(P,Q)$, e.g.,
$\circledast^{<\,(P,Q)\,>}_{2}(a,b)$ instead of
$\circledast^{<\,(I_1)\,|\,(P,Q)\,>}_{2}(a,b)$
for
$a\in T[I_1,\size{}{MDA}{MDA}(P)]$ and
$b\in T[I_1,\size{}{MDA}{MDA}(Q)]$.

We avoid separating $(I_1,\dotsc,I_{d-1},I_{d+1},\dotsc,I_D)$ and $(P,Q)$ in this work, as we focus on concatenation (Example~\ref{def:mda_concat}), prefix-sum (Example~\ref{def:mda_scan}), and point-wise combination (Example~\ref{def:mda_pw}) only, which have invertible or $P$/$Q$-invariant index set functions, respectively.
Consequently, for the practice-relevant combine operators considered in this work, we can deduce all meta-parameters automatically.

\end{itemize}
\end{note}

\newpage

\subsection{Generalized Notion of MDHs}
\label{app_gen_notion_mdh}

The MDH Definition~\ref{def_mdh} can be generalized to have an arbitrary algebraic structure as input.

\begin{definition}[Multi-Dimensional Homomorphism]
\label{def_gen_mdh}

Let
\[
\mathcal{A}^\downarrow:=(\ T^\texttt{INP}[\fsize{1}{}{MDA}{MDA}{^\downarrow}(*),\dotsc,\fsize{D}{}{MDA}{MDA}{^\downarrow}(*)] \ , \ ({^\downarrow}{\circledast_d})_{d\in[1,D]_\IN} \ )
\]
and
\[
\mathcal{A}^\uparrow:=(\ T^\texttt{OUT}[\fsize{1}{}{MDA}{MDA}{^\uparrow}(*),\dotsc,\fsize{D}{}{MDA}{MDA}{^\uparrow}(*)] \ , \ ({^\uparrow}{\circledast_d})_{d\in[1,D]_\IN} \ )
\]
be two algebraic structures, where
\[
({^\downarrow}{\circledast_d}\in{\CO}{^{\scriptstyle<\sfsize{d}{}{MDA}{MDA}{^\downarrow}\,|\,T^\texttt{INP}\,|\,D\,|\,d>}})_{d\in[1,D]_\IN}
\]
and
\[
({^\uparrow}{\circledast_d}\in{\CO}^{<\sfsize{d}{}{MDA}{MDA}{^\uparrow}\,|\,T^\texttt{OUT}\,|\,D\,|\,d>})_{d\in[1,D]_\IN}
\]
are tuples of combine operators,
for $D\in\IN$, \
$T^\texttt{INP},T^\texttt{OUT}\in\type$, \
$
\fsize{d}{}{MDA}{MDA}{^\downarrow}, \fsize{d}{}{MDA}{MDA}{^\uparrow}:\IDXs$ $\to\IDXs
$, and the two structures' carrier sets
\[
T^\texttt{INP}[\fsize{1}{}{MDA}{MDA}{^\downarrow}(*),\dotsc,\fsize{D}{}{MDA}{MDA}{^\downarrow}(*)]
\]
and
\[
T^\texttt{OUT}[\fsize{1}{}{MDA}{MDA}{^\uparrow}(*),\dotsc,\fsize{D}{}{MDA}{MDA}{^\uparrow}(*)]
\]
denote the set of MDAs that are in the function domain of combine operators
(the star symbol is used for indicating the function range of index functions).

A \emph{Multi-Dimensional Homomorphism (MDH)} from the algebraic structure $\mathcal{A}^\downarrow$ to the structure $\mathcal{A}^\uparrow$
is any function
\[
h^{<I_1,\dotsc,I_D\in\IDXs>}:T^\texttt{INP}[\fsize{1}{}{MDA}{MDA}{^\downarrow}(I_1),\dotsc,\fsize{D}{}{MDA}{MDA}{^\downarrow}(I_D)]\to T^\texttt{OUT}[\fsize{1}{}{MDA}{MDA}{^\uparrow}(I_1),\dotsc,\fsize{D}{}{MDA}{MDA}{^\uparrow}(I_D)]
\]
that satisfies the \emph{homomorphic property}:
\[
h( \ \MDA_1 {^\downarrow}{\circledast_d} \ \MDA_2 \ ) \ = \ h(\MDA_1)  \, {^\uparrow}{\circledast_d} \, h(\MDA_2)
\]
\end{definition}

The MDH Definition~\ref{def_mdh} is a special case of our generalized MDH Definition~\ref{def_gen_mdh}, for ${^\downarrow}{\circledast_d}=\dplus^{<T^\texttt{INP}\,|\,D\,|\,d>}$~(Example~\ref{def:mda_concat}).

\vspace*{5px}

Higher-order function $\mdh$ (originally introduced in Definition~\ref{def_md_hom}) is defined for the generalized MDH Definition~\ref{def_gen_mdh} as follows.

\newpage

\begin{definition}[Higher-Order Function $\mdh$]
\label{def_gen_md_hom}
The higher-order function $\mdh$ is of type
\begin{align*}
&
\hspace*{0px}
\texttt{md\_hom}^{<T^\texttt{INP},T^\texttt{OUT}\in\type \ | \ D\in\IN \ | \ (\ssize{d}{MDA}{MDA}{^\downarrow}:\IDXs\to\IDXs)_{d\in[1,D]_\IN},} \\[-5pt]
&\hspace*{111px}^{(\ssize{d}{MDA}{MDA}{^\uparrow}:\IDXs\to\IDXs)_{d\in[1,D]_\IN}>}: \\[2pt]
& \hspace*{20px}
\underbrace{( \ \CO^{<\ssize{1}{MDA}{MDA}{^\downarrow}\,|\,T^\texttt{OUT}\,|\,D\,|\,1>}\times\dotsc\times\CO^{<\ssize{D}{MDA}{MDA}{^\downarrow}\,|\,T^\texttt{OUT}\,|\,D\,|\,D>} \ )}_{{^\downarrow}{\circledast_1} \ ,\,\dotsc\,, \ {^\downarrow}{\circledast_D}} \ \ \times \\[5pt]
&\hspace*{20px}
\underbrace{\mathrm{SF}^{<T^\texttt{INP},T^\texttt{OUT}>}}_{f} \ \ \times \\[5pt]
&\hspace*{20px}
\underbrace{( \ \CO^{<\ssize{1}{MDA}{MDA}{^\uparrow}\,|\,T^\texttt{OUT}\,|\,D\,|\,1>}\times\dotsc\times\CO^{<\ssize{D}{MDA}{MDA}{^\uparrow}\,|\,T^\texttt{OUT}\,|\,D\,|\,D>} \ )}_{{^\uparrow}{\circledast_1} \ ,\,\dotsc\,, \ {^\uparrow}{\circledast_D}} \\[5pt]
&\hspace*{180px}
\to_p \ \
\underbrace{\MDH^{<T^\texttt{INP},T^\texttt{OUT}\,|\,D\,|\,(\ssize{d}{MDA}{MDA}{^\downarrow})_{d\in[1,D]_\IN},(\ssize{d}{MDA}{MDA}{^\uparrow})_{d\in[1,D]_\IN}>}}_{
\texttt{md\_hom}( \ (\circledast^\downarrow_1,\dotsc, \circledast^\downarrow_D) \ , \ f \ , \ (\circledast^\uparrow_1,\dotsc, \circledast^\uparrow_D) \ )
}
\end{align*}
The function takes as input a scalar function $f$ and two tuples of $D$-many combine operators $({^\downarrow}{\circledast_1},\dotsc,{^\downarrow}{\circledast_D})$ and $({^\uparrow}{\circledast_1},\dotsc,{^\uparrow}{\circledast_D})$,
and it yields a function $\texttt{md\_hom}( \ ({^\downarrow}{\circledast_1},\dotsc, {^\downarrow}{\circledast_D}) \ , \ f \ , \ ({^\uparrow}{\circledast_1},\dotsc, {^\uparrow}{\circledast_D}) \ )$ which is defined as:
\begin{align*}
                        \hspace*{10px}& {^\downarrow}{\MDA}
                        &&\in T^\texttt{INP}[\,\sfsize{1}{}{MDA}{MDA}{^\downarrow}(\,I_1\,)\,,\,\dotsc\,,\,\sfsize{D}{}{MDA}{MDA}{^\downarrow}(\,I_D\,)\,] &\\
=:                    \hspace*{10px}& \\
                        \hspace*{10px}& \underset{i_1\in I_1}{{^\downarrow}{\circledast_1}}\dotsc \underset{i_D\in I_D}{{^\downarrow}{\circledast_D}} {^\downarrow}{\MDA}{^{<i_1,\dotsc,i_d>}}
                        &&\in T^\texttt{INP}[\,{\sfsize{1}{}{MDA}{MDA}{^\downarrow}(\,\{i_1\}\,)\,,\,\dotsc\,,\,\sfsize{D}{}{MDA}{MDA}{^\downarrow}(\,\{i_D\}\,)}\,] &\\
\mapsto              \hspace*{10px}& \\
                        \hspace*{10px}& \underset{i_1\in I_1}{\dplus_1}\dotsc \underset{i_D\in I_D}{\dplus_D} {^\downarrow}{\MDA}_f{^{<i_1,\dotsc,i_d>}}
                        &&\in T^\texttt{INP}[\, \{i_1\} \,,\,\dotsc\,,\, \{i_D\} \,] &\\
\overset{\vec{f}}{\mapsto} \hspace*{10px}& \\
                        \hspace*{10px}& \underset{i_1\in I_1}{\dplus_1}\dotsc \underset{i_D\in I_D}{\dplus_D} {^\uparrow}{\MDA}_f{^{<i_1,\dotsc,i_d>}}
                        &&\in T^\texttt{OUT}[\, \{i_1\} \,,\,\dotsc\,,\, \{i_D\} \,] &\\
\mapsto              \hspace*{10px}& \\
                        \hspace*{10px}& \underset{i_1\in I_1}{{^\uparrow}{\circledast_1}}\dotsc \underset{i_D\in I_D}{{^\uparrow}{\circledast_D}} {^\uparrow}{\MDA}{^{<i_1,\dotsc,i_d>}}
                        &&\in T^\texttt{OUT}[\,{\sfsize{1}{}{MDA}{MDA}{^\uparrow}(\,\{i_1\}\,)\,,\,\dotsc\,,\,\sfsize{D}{}{MDA}{MDA}{^\uparrow}(\,\{i_D\}\,)}\,] &\\
=:                    \hspace*{10px}& \\
                        \hspace*{10px}& {^\uparrow}{\MDA}
                        &&\in T^\texttt{OUT}[ \sfsize{1}{}{MDA}{MDA}{^\uparrow}(\,I_1\,),\dotsc,\sfsize{D}{}{MDA}{MDA}{^\uparrow}(\,I_D\,)]&
\end{align*}
\vspace*{2px}
Here, $\vec{f}$ denotes the adaption of function $f$ to operate on MDAs comprising a single scalar value only~--~the function is of type
\[
\vec{f}^{<i_1,\dotsc,i_D\in\IN>}:T^\texttt{INP}[ \, \{i_1\},\dotsc,\{i_D\} \, ]\to T^\texttt{OUT}[ \ \{i_1\},\dotsc,\{i_D\} \ ]
\]
and defined as~%
\[
\vec{f}(x)[\,i_1,\dotsc,i_D\,]:=f(x[\,i_1,\dotsc,i_D\,])
\]

We refer to the
first application of $\mapsto$ as \emph{de-composition}, to the
application of $\overset{\vec{f}}{\mapsto}$ as \emph{scalar function application}, and to
the second application of $\mapsto$ as \emph{re-composition}.

\vspace*{2px}

\noindent
For $\texttt{md\_hom}( \, (\circledast^\downarrow_1,\dotsc, \circledast^\downarrow_d) \, , \, f , \, (\circledast^\uparrow_1,\dotsc, \circledast^\uparrow_d) \, )$, we require by definition the homomorphic property (Definition~\ref{def_gen_mdh}),
i.e., for each $d\in[1,D]_\IN$, it must hold:
\begin{align*}
&
\texttt{md\_hom}( \ ({^\downarrow}{\circledast_1},\dotsc, {^\downarrow}{\circledast_D}) \ , \ f \ , \ ({^\uparrow}{\circledast_1},\dotsc, {^\uparrow}{\circledast_D}) \ )
(\ \MDA_1 {^\downarrow}{\circledast_d} \ \MDA_2 \ ) \ = \\[5pt]
&\hspace*{165px}
\texttt{md\_hom}( \ ({^\downarrow}{\circledast_1},\dotsc, {^\downarrow}{\circledast_D}) \ , \ f \ , \ ({^\uparrow}{\circledast_1},\dotsc, {^\uparrow}{\circledast_D}) \ )
( \, \MDA_1 \, ) \\
&\hspace*{175px} {^\uparrow}{\circledast_d} \\
&\hspace*{165px} \texttt{md\_hom}( \ ({^\downarrow}{\circledast_1},\dotsc, {^\downarrow}{\circledast_D}) \ , \ f \ , \ ({^\uparrow}{\circledast_1},\dotsc, {^\uparrow}{\circledast_D}) \ )
( \, \MDA_2 \, )
\end{align*}
\end{definition}

\subsection{Simple MDH Examples}
\label{sec_mdh_examples}

\paragraph*{Function Mapping}

Function
$\texttt{map}^{<T^\texttt{INP},T^\texttt{OUT}\,|\,D\,|\,(I_1,\dotsc,I_D)>}(f)$
maps a function $f:T^\texttt{INP}\to T^\texttt{OUT}$ to each element of an MDA that has scalar type $T^\texttt{INP}\in\type$, dimensionality $D\in\IN$, and index sets $I:=(I_1,\dotsc,I_D)\in\IDXs^D$.
The function is of type
\begin{align*}
&
\texttt{map}^{<T^\texttt{INP},T^\texttt{OUT}\in\type\,|\,D\in\IN\,|\,(I_1,\dotsc,I_D)\in\IDXs^D>}: \\[3px]
&\hspace*{100px}
\underbrace{T^\texttt{INP}\to T^\texttt{OUT}}_{f} \ \to \  \underbrace{T^\texttt{INP}[ \ I_1  \ , \, \dotsc \, , \ I_D \ \ ] \ \to \ T^\texttt{OUT}[ \ I_1  \ , \, \dotsc \, , \ I_D \  \ ]}_{\texttt{map}^{<T^\texttt{INP},T^\texttt{OUT}\,|\,D\,|\,(I_1,\dotsc,I_D)>}(f)}
\end{align*}
and it is computed as:
\[
\MDA \ \ \ \ \overset{\texttt{map}^{<T^\texttt{INP},T^\texttt{OUT}\,|\,D\,|\,(I_1,\dotsc,I_D)>}(f)}{\mapsto} \ \ \ \
\underset{i_1\in I_1}{\dplus_1}
\ \ \dotsc \ \
\underset{i_D\in I_D}{\dplus_D}
\vec{f}_\texttt{map}( \ \MDA|_{\{i_1\}\times\dotsc\times\{i_D\}} \ )
\]
Here, $\dplus_d:=\dplus^{<T^\texttt{OUT}\,|\,D\,|\,d>}$ denotes concatenation (Example~\ref{def:mda_concat}) in dimension $d$,
MDA $\MDA|_{\{i_1\}\times\dotsc\times\{i_D\}}$ is the restriction of $\MDA$ to the single scalar element accessed via indices $(i_1,\dotsc,i_D)$, and $\vec{f}_\texttt{map}$ denotes the adaption of function $f$ to operate on MDAs comprising a single value only:~it is of type
\[
\vec{f}_\texttt{map}^{<i_1,\dotsc,i_D\in\IN>}:T^\texttt{INP}[ \, \{i_1\},\dotsc,\{i_D\} \, ]\to T^\texttt{OUT}[ \ \{i_1\},\dotsc,\{i_D\} \ ]
\]
and defined as~%
\[
\vec{f}_\texttt{map}(x)[\,i_1,\dotsc,i_D\,]:=f(x[\,i_1,\dotsc,i_D\,])
\]

\vspace*{5px}

It is easy to see that $\texttt{map}^{<T^\texttt{INP},T^\texttt{OUT}\,|\,D>}(f)$ is an MDH of type $\MDH^{<T^\texttt{INP},T^\texttt{OUT}\,|\,D\,|\,id,\dotsc,id>}$ whose combine operators are concatenation $\dplus_d$%
$\in{\CO}^{<id\,|\,T^\texttt{OUT}\,|\,D\,|\,d>}$
in all dimensions $d\in[1,D]_\IN$.

We have chosen \texttt{map} function's order of stages~--~$T^\texttt{INP},T^\texttt{OUT}\in\type$ (stage 1), $D\in\IN$ (stage 2), and $(I_1,\dotsc,I_D)\IDXs^D$ (stage 3)~--~according to the recommendations in~\citet{haskell_parameter_order}, i.e., earlier stages (such as the scalar types  $T^\texttt{INP},T^\texttt{OUT}$) are expected to change less frequently than later stages (e.g., the MDAs' index sets $I_1,\dotsc,I_D$).

\paragraph*{Reduction}
Function $\texttt{red}^{<T\,|\,D\,|\,(I_1,\dotsc,I_D)>}(\oplus)$ combines all elements within an MDA that has scalar type $T\in\type$, dimensionality $D\in\IN$, and index sets $I:=(I_1,\dotsc,I_D)\in\IDXs^D$, using an associative and commutative binary function $\oplus:T\times T\to T$.
The function is of type
\[
\texttt{red}^{<T\in\type\,|\,D\in\IN\,|\,(I_1,\dotsc,I_D)\in\IDXs^D>} \ : \ \underbrace{T\times T\to T}_{\oplus} \ \to \ \underbrace{T[ \ I_1  \ , \, \dotsc \, , \ I_D \ \ ] \ \to \ T[ \ 1,\dotsc,1 \ ]}_{\texttt{red}^{<T\,|\,D\,|\,(I_1,\dotsc,I_D)>}(\oplus)}
\]
and it is computed as:
\[
\MDA \ \ \ \overset{\texttt{red}^{<T\,|\,D\,|\,(I_1,\dotsc,I_D)>}(\oplus)}{\mapsto} \ \ \
\underset{i_1\in I_1 \ \ }{\overrightarrow{\bullet_1}(\oplus)}
\ \ \dotsc \ \
\underset{i_D\in I_D \ \ }{\overrightarrow{\bullet_D}(\oplus)}
\ \vec{f}_\texttt{red}( \ \MDA|_{\{i_1\}\times\dotsc\times\{i_D\}} \ )
\]
Here, $\overrightarrow{\bullet_d}(\oplus):=\overrightarrow{\bullet}^{<T\,|\,D\,|\,d>}(\oplus)$ denotes point-wise combination (Example~\ref{def:mda_pw}) in dimension $d$,
MDA $\MDA|_{\{i_1\}\times\dotsc\times\{i_D\}}$ is
defined as above,
and $\vec{f}_\texttt{red}$ is the function of type
\[
\vec{f}_\texttt{red}^{<i_1,\dotsc,i_D\in\IN>}:T^\texttt{INP}[ \, \{i_1\},\dotsc,\{i_D\} \, ]\to T^\texttt{OUT}[ \ \{0\} \, ,\dotsc, \, \{0\} \ ]
\]
that is defined as~%
\[
\vec{f}_\texttt{red}(x)[\,0,\dotsc,0\,] := x[\,i_1,\dotsc,i_D\,]
\]

It is easy to see that $\texttt{red}^{<T\,|\,D>}(\oplus)$ is an MDH of type $\MDH^{<T,T\,|\,D\,|\,0_f,\dotsc,0_f>}$ whose combine operators are point-wise addition $\overrightarrow{\bullet_d}(\,\oplus\,)$%
$\in{\CO}^{<id\,|\,T\,|\,D\,|\,d>}$ in all dimensions $d\in[1,D]_\IN$.
The same as for function \texttt{map}, function \texttt{red}'s order of meta-parameter stages are chosen according to~\cite{haskell_parameter_order}.

\subsection{Design Decisions: \mdh}
\label{app_design_dec_mdh}

We list some design decisions for higher-order function $\mdh$ (Definition~\ref{def_md_hom}).

\begin{note}
For some MDHs (such as Mandelbrot), the scalar function $f$ (Definition~\ref{def_md_hom}) is dependent on the position in the input MDA, i.e., it takes as arguments, in addition to $\MDA[i_1,\dotsc,i_D]$, also the indices $i_1,\dotsc,i_D$.
Such MDHs can be easily expressed via $\mdh$ after a straightforward type adjustment:~type $\mathrm{SF}^{<T^\texttt{INP},T^\texttt{OUT}>}$ has to be defined as the set of functions $f:T^\texttt{INP}\times\IDXs^D\to T^\texttt{OUT}$ (rather than of functions $f:T^\texttt{INP}\to T^\texttt{OUT}$, as in Definition~\ref{def_md_hom}).

Since we do not aim at forcing scalar functions to always take MDA indices as input arguments~--~for expressing most computations, this is not required (Figure~\ref{fig_hl_examples}) and only causes additional complexity~--~we assume in the following two different definitions of pattern $\mdh$:~one variant exactly as in Definition~\ref{def_md_hom}, and one variant with the adjusted type for scalar functions and that passes automatically indices $i_1,\dotsc,i_D$ to $f$.
The two variants can be easily differentiated, via an additional, boolean meta-parameter \texttt{USE\_MDA\_INDICES}:~first variant iff $\texttt{USE\_MDA\_INDICES}=\texttt{false}$ and second variant iff $\texttt{USE\_MDA\_INDICES}=\texttt{true}$.

For simplicity, we focus in this paper on the first variant (as in Definition~\ref{def_md_hom}), because it is the more common variant, and because all insights presented in this work apply to both variants.
\end{note}

\subsection{Proof $\mdh$ Lemma~\ref{lemma_mdhom}}
\label{app_proof_mdhhom_lemma}

\begin{proof}
    Let
    $\MDA_1\in T[I^{\MDA_1}_1,\dotsc,I^{\MDA_2}_D]$
    and
    $\MDA_2\in T[I^{\MDA_2}_1,\dotsc,I^{\MDA_2}_D]$
    be two arbitrary MDAs
    that are concatenable in dimension $d$.

    According to Definition~\ref{def_md_hom}, we have to show that
    \begin{align*}
        &
        \texttt{md\_hom}( \ f \ , \ (\circledast_1,\dotsc, \circledast_D) \ ) (\ \MDA_1 \dplus_d \ \MDA_2 \ ) \ = \\[5pt]
        &\hspace*{80px}
        \texttt{md\_hom}( \ f \ , \ (\circledast_1,\dotsc, \circledast_D) \ ) ( \, \MDA_1 \, )
        \ \ \circledast_d \ \
        \texttt{md\_hom}( \ f \ , \ (\circledast_1,\dotsc, \circledast_D) \ ) ( \, \MDA_2 \, )
    \end{align*}
    For this, we first show for arbitrary $k\in[1,D)_\IN$ that
    \[
        \dotsc\underset{i_{k}\in I_{k}}{\circledast_{k}}\underset{i_{k+1}\in I_{k+1}}{\circledast_{k+1}}\dotsc \ \
        x|_{\dotsc,\{i_{k}\},\{i_{k+1}\},\dotsc}
        \ \ = \ \
        \dotsc\underset{i_{k+1}\in I_{k+1}}{\circledast_{k+1}}\underset{i_{k}\in I_{k}}{\circledast_{k}}\dotsc \ \
        x|_{\dotsc,\{i_{k}\},\{i_{k+1}\},\dotsc}
    \]
    from which follows
    \[
        \underset{i_1\in I_1}{\circledast_1}\dotsc \underset{i_D\in I_D}{\circledast_D}
        x|_{\{i_{1}\},\dotsc,\{i_{D}\}}
        \ \ = \ \
        \underset{i_\sigma(1)\in I_\sigma(1)}{\circledast_{\sigma(1)}}\dotsc \underset{i_\sigma(D)\in I_\sigma(D)}{\circledast_{\sigma(D)}}
        x|_{\{i_{1}\},\dotsc,\{i_{D}\}}
    \]
    for any permutation $\sigma:\{1,\dotsc,D\}\bij\{1,\dotsc,D\}$.
    Afterward, in our assumption above, we can assume w.l.o.g. that $d=1$.

\vspace*{10px}

    \begin{enumerate}
    \setlength\itemsep{1em}
        \item[] \underline{Case 1:} [$\co{k}=\co{k+1}$]
        Follows immediately from the commutativity of $\dplus$ or $\overrightarrow{\bullet}(\oplus)$
        for commutative $\oplus$, respectively. $\ \checkmark$

        \item[] \underline{Case 2:} [$\co{k}\neq\co{k+1}$]
        Trivial, as it is either $\co{k}=\dplus$ or $\co{k+1}=\dplus$, and
        \[
            (\underset{i_d\in I_d}{\dplus} \ x|_{\dotsc,\{i_d\},\dotsc})[i_1,\dotsc,i_D]
            =
            (x|_{\dotsc,\{i_d\},\dotsc})[i_1,\dotsc,i_D]
        \]
        according to the definition of MDA concatenation $\dplus$ (Example~\ref{def:mda_concat}).~~~$\checkmark$
    \end{enumerate}

    Let now be $d=1$ (see assumption above), it holds:
    \begin{align*}
        &\hspace*{-20px}
        \texttt{md\_hom}( \ f \ , \ (\circledast_1,\dotsc, \circledast_D) \ )(\ \MDA_1 \dplus_1 \ \MDA_2 \ )
        \\
        =& \
        \underset{i_1\in I_1}{\circledast_1}\dotsc \underset{i_D\in I_D}{\circledast_D} \ \vec{f}( \ (\MDA_1 \dplus_1 \ \MDA_2)|_{\{i_1\}\times\dotsc\times\{i_D\}} \ ) \\
        =& \
        \underset{i_1\in I^{\MDA_1}_1}{\circledast_1}\dotsc \underset{i_D\in I_D}{\circledast_D} \ \vec{f}( \ \MDA_1|_{\{i_1\}\times\dotsc\times\{i_D\}} \ )
        \ \ \ \co{1} \ \ \
        \underset{i_1\in I^{\MDA_2}_1}{\circledast_1}\dotsc \underset{i_D\in I_D}{\circledast_D} \ \vec{f}( \ \MDA_2|_{\{i_1\}\times\dotsc\times\{i_D\}} \ ) \\
        =& \ \ \
        \texttt{md\_hom}( \ f \ , \ (\circledast_1,\dotsc, \circledast_D) \ ) ( \, \MDA_1 \, )
        \ \ \ \, \co{1} \, \ \ \
        \texttt{md\_hom}( \ f \ , \ (\circledast_1,\dotsc, \circledast_D) \ ) ( \, \MDA_2 \, )
        \ \ \ \checkmark
    \end{align*}

\end{proof}

\subsection{Examples of Index Functions}
\label{app_examples_idx_funcs}

We present examples of index functions (Definition~\ref{def_idx_func}).

\begin{example}[Index Functions~--~Matrix-Vector Multiplication]
  The index functions we use for expressing Matrix-Vector Multiplication (\texttt{MatVec}) are:
  \begin{itemize}
    \setlength\itemsep{1em}

    \item \underline{Input Matrix:}
    \[\idx(i,k):=(i,k)\in\IDXFCT^{<D=2,D_b=2\,|\,\sbsize{}{MDA}{BUF}>}\]
    for
    $\bsize{}{MDA}{BUF}(I^\texttt{MDA}_1,I^\texttt{MDA}_2)\ := \ [0,\max(I^\texttt{MDA}_1)]_{\INz} \, , \, [0,\max(I^\texttt{MDA}_2)]_{\INz}$

    \item \underline{Input Vector:}
    \[\idx(i,k):=(k)\in\IDXFCT^{<D=2,D_b=1\,|\,\sbsize{}{MDA}{BUF}>}\]
    for
    $\bsize{}{MDA}{BUF}(I^\texttt{MDA}_1,I^\texttt{MDA}_2)\ := \ [0,\max(I^\texttt{MDA}_2)]_{\INz}$

    \item \underline{Output Vector:}
    \[\idx(i,k):=(i)\in\IDXFCT^{<D=2,D_b=1\,|\,\sbsize{}{MDA}{BUF}>}\]
    for
    $\bsize{}{MDA}{BUF}(I^\texttt{MDA}_1,I^\texttt{MDA}_2)\ := \ [0,\max(I^\texttt{MDA}_1)]_{\INz}$

  \end{itemize}
\end{example}

\vspace*{5px}

\begin{example}[Index Functions~--~Jacobi 1D]
  The index functions we use for expressing Jacobi 1D (\texttt{Jacobi1D}) are:
  \begin{itemize}
    \setlength\itemsep{1em}

    \item \underline{Input Buffer, 1. Access:}
    \[\idx(i):=(i+0)\in\IDXFCT^{<D=1,D_b=1\,|\,\sbsize{}{MDA}{BUF}>}\]
    for
    $\bsize{}{MDA}{BUF}(I^\texttt{MDA}_1) \ := \ [0,\max(I^\texttt{MDA}_1)+0]_{\INz}$

    \item \underline{Input Buffer, 2. Access:}
    \[\idx(i):=(i+1)\in\IDXFCT^{<D=1,D_b=1\,|\,\sbsize{}{MDA}{BUF}>}\]
    for
    $\bsize{}{MDA}{BUF}(I^\texttt{MDA}_1) \ := \ [0,\max(I^\texttt{MDA}_1)+1]_{\INz}$

    \item \underline{Input Buffer, 3. Access:}
    \[\idx(i):=(i+2)\in\IDXFCT^{<D=1,D_b=1\,|\,\sbsize{}{MDA}{BUF}>}\]
    for
    $\bsize{}{MDA}{BUF}(I^\texttt{MDA}_1) \ := \ [0,\max(I^\texttt{MDA}_1)+2]_{\INz}$

    \item \underline{Output Buffer:}
    \[\idx(i):=(i)\in\IDXFCT^{<D=1,D_b=1\,|\,\sbsize{}{MDA}{BUF}>}\]
    for
    $\bsize{}{MDA}{BUF}(I^\texttt{MDA}_1) \ := \ [0,\max(I^\texttt{MDA}_1)]_{\INz}$

  \end{itemize}
\end{example}

\vspace*{10px}

\subsection{Representation of Scalar Values}
\label{app_scalar_values}

Scalar values can be considered as $0$-dimensional BUFs (Definition~\ref{def_buffer}).
Consequently, in Definition~\ref{def_buffer}, the cartesian product $[0,N_1)_{\INz}\,\times\dotsc\times\,[0,N_D)_{\INz}$ is empty for $D=0$, and thus results in the neutral element of the cartesian product.
As any singleton set can be considered as neutral element of cartesian product (up to bijection), we define the set $\{\epsilon\}$ containing the dedicated symbol epsilon only, as the uniquely determined neutral element of cartesian product (inspired by the notation of the \emph{empty word}).

We often refrain from explicitly stating symbol $\epsilon$, e.g., by writing $\Buf$ instead of $\Buf[\epsilon]$ for accessing a BUF, or $(i_1,\dotsc,i_D)\to()$ instead of $(i_1,\dotsc,i_D)\to(\epsilon)$ for index functions.

Note that alternatively, scalar values can be considered as any multi-dimensional BUF containing a single element only.
For example, a scalar value $s$ can be represented as $1$-dimensional BUF $\Buf_\texttt{1D}[0]:=s$, or a $2$-dimensional BUF $\Buf_\texttt{2D}[0,0]:=s$, or a $3$-dimensional BUF $\Buf_\texttt{3D}[0,0,0]:=s$, etc.
However, this results in an ambiguous representation of scalar values, which we aim to avoid by considering scalars as $0$-dimensional BUFs, as described above.

\subsection{Runtime Complexity of Histograms}
\label{histo_complexity}

Our implementation of Histograms~(Subfigure~5 in Figure~\ref{fig_hl_examples})
has a \emph{work complexity} of $\mathcal{O}(E*B)$, where $E$ is the number of elements to check and $B$ the number of bins, i.e., our MDH Histogram implementation is not work efficient.
However, our Histograms' \emph{step complexity}~\cite{harris2007optimizing} is $\mathcal{O}( log(E) )$:~step complexity is often used for parallel algorithms and assumes an infinite number of cores, i.e., we can ignore in our implementation of Histogram the concatenation dimension~$B$~(which has a step complexity of $\mathcal{O}(1)$) and take into account its reduction dimension $E$ only, which has a step complexity of $log(E)$ (parallel reduction~\cite{harris2007optimizing}).
In contrast, related approaches~\cite{9355244} are often work efficient, by having a work complexity of $\mathcal{O}(E)$;~however, their high work efficiency is at the cost of their step complexity which is also $\mathcal{O}(E)$, rather than $\mathcal{O}(log(E))$ as for our implementation in Subfigure~5, thereby being asymptotically less efficient for parallel machines consisting of many cores.
Our future work will show that the work-efficient Histogram implementation introduced in~\citet{9355244} can also be expressed in our approach, by using for scalar function $f$ an optimized micro kernel for Histogram computation, similarly as done in the related work.

\newpage

\subsection{Combine Operator of Prefix-Sum Computations}
\label{app_co_scan}

We define \emph{prefix-sum} which is the combine operator of compute pattern \texttt{scan} and example \texttt{MBBS} in Section~\ref{ch_hl_examples}.

\begin{example}[Prefix-Sum]
\label{def:mda_scan}
    We define \emph{prefix-sum},
    according to a binary function $\oplus:T\times T\to T$ (e.g. addition),
    as function $\circledast_\texttt{prefix-sum}$ of type
    \begin{align*}
    &
    \circledast_\texttt{prefix-sum}^{<T\in\type\,|\,D\in\IN\,|\,d\in[1,D]_\IN\,|\,(I_1,\dotsc,I_{d-1},I_{d+1},\dotsc,I_D)\in\IDXs^{D-1},(P,Q)\in\IDXsxIDXs>}: \\[5pt]
    &\hspace*{0px}
    \underbrace{T\times T\to T
    }_{\oplus}
    \ \to \
    \underbrace{
    T[I_1,\dotsc,\underset{\underset{d}{\uparrow}}{\underbrace{id(P)}},\dotsc,I_D]
    \times
    T[I_1,\dotsc,\underset{\underset{d}{\uparrow}}{\underbrace{id(Q)}},\dotsc,I_D]
    \to T[I_1,\dotsc,\underset{\underset{d}{\uparrow}}{\underbrace{id(P\cupdot Q)}},\dotsc,I_D]
    }_{\substack{\texttt{prefix-sum (according to $\oplus$)} }}
    \end{align*}
    where $id:\IDXs\to\IDXs$ is the identity function on MDA index sets.
    The function is computed as:
    \begin{align*}
    &
    \circledast_\texttt{prefix-sum}^{<T\,|\,D\,|\,d\,|\,(I_1,\dotsc,I_{d-1},I_{d+1},\dotsc,I_D),(P,Q)>}(\,\oplus\,)(\, \MDA_1,\MDA_2 \,)[ \, i_1 \, ,\dotsc, \ \ i_d \ \ ,\dotsc, \, i_D \, ] \\[5pt]
    &\hspace*{35px}
    := \
    \begin{cases}
     \ \MDA_1[ \, i_1 \, ,\dotsc, \hspace*{23px}  i_d \hspace*{23px} ,\dotsc, \, i_D \, ] \ \oplus \ \MDA_2[ \, i_1 \, ,\dotsc, \ \ max_{\le i_d}(Q) \ \ ,\dotsc, \, i_D \, ] &, \ \ i_d\in P \\
     \ \MDA_1[ \, i_1 \, ,\dotsc, \ \ max_{\le i_d}(P) \ \ ,\dotsc, \, i_D \, ] \ \oplus \ \MDA_2[  \, i_1 \, ,\dotsc, \hspace*{24px}  i_d \hspace*{24px} ,\dotsc, \, i_D \, ] &, \ \ i_d\in Q \\
    \end{cases}
    \end{align*}
    Function $\circledast_\texttt{prefix-sum}^{<T\,|\,D\,|\,d>}(\oplus)$ (meaning:~$\circledast_\texttt{prefix-sum}$ is partially applied to ordinary function parameter $\oplus$;~formal details provided in the Appendix, Definition~\ref{app_partial_meta_func_app}) is a combine operator of type ${\CO}^{<id\,|\,T\,|\,D\,|\,d>}$ for any binary operator $\oplus:T\times T\to T$.
\end{example}

\section{Addendum Section~\ref{ch:low_level}}

\subsection{Constraints of Programming Models}
\label{app_constraints}

\setlength\dashlinedash{0.5pt}
\setlength\dashlinegap{3pt}
\setlength\arrayrulewidth{0.3pt}
\begin{table}[b!]
\center
\resizebox{\textwidth}{!}{
\begin{tabular}{c|lr}
No. & Constraint& \\
\hline & &\\
\texttt{0} & $\prod_{d\in[1,D]_\IN}\texttt{\#PRT}(\overset{\overset{\texttt{$\bullet$-MDH}}{\leftarrow}}{\texttt{CC}},\overset{\overset{\texttt{$\bullet$-MDH}}{\leftarrow}}{\texttt{d}})\le 1024$ &\texttt{ (Number of \texttt{CC}s limited)} \\
& \\
\texttt{R3} & $\texttt{\#PRT}(\overset{\overset{\texttt{$\uparrow$-MDH}}{\leftarrow}}{\texttt{BLK}},\overset{\overset{\texttt{$\uparrow$-MDH}}{\leftarrow}}{\texttt{d}})>1 \wedge \co{\overset{\overset{\texttt{$\uparrow$-MDH}}{\leftarrow}}{\texttt{d}}}\neq\dplus_{\overset{\overset{\texttt{$\uparrow$-MDH}}{\leftarrow}}{\texttt{d}}} \,\Rightarrow\, \texttt{$\uparrow$-mem}^\texttt{<ob>}(\overset{\overset{\texttt{$\uparrow$-MDH}}{\leftarrow}}{\texttt{BLK}},\overset{\overset{\texttt{$\uparrow$-MDH}}{\leftarrow}}{\texttt{d}})\in\{\texttt{DM}\}$ &\texttt{ (\texttt{SMX}s combine in \texttt{DM})} \\
   & $\texttt{\#PRT}(\overset{\overset{\texttt{$\uparrow$-MDH}}{\leftarrow}}{\texttt{CC}},\overset{\overset{\texttt{$\uparrow$-MDH}}{\leftarrow}}{\texttt{d}})>1 \wedge \co{\overset{\overset{\texttt{$\uparrow$-MDH}}{\leftarrow}}{\texttt{d}}}\neq\dplus_{\overset{\overset{\texttt{$\uparrow$-MDH}}{\leftarrow}}{\texttt{d}}} \,\Rightarrow\, \texttt{$\uparrow$-mem}^\texttt{<ob>}(\overset{\overset{\texttt{$\uparrow$-MDH}}{\leftarrow}}{\texttt{CC}},\overset{\overset{\texttt{$\uparrow$-MDH}}{\leftarrow}}{\texttt{d}})\in\{\texttt{DM,SM}\}$ &\texttt{ (\texttt{CC}s combine in \texttt{DM/SM})} \\
& \\
\end{tabular}
}
\caption{
CUDA model constraints on tuning parameters
}
\label{tab_model_constraints_cuda}
\end{table}

\setlength\dashlinedash{0.5pt}
\setlength\dashlinegap{3pt}
\setlength\arrayrulewidth{0.3pt}
\begin{table}[t!]
\center
\resizebox{\textwidth}{!}{
\begin{tabular}{c|lr}
No. & Constraint& \\
\hline & &\\
\texttt{0} & $\prod_{d\in[1,D]_\IN}\texttt{\#PRT}(\overset{\overset{\texttt{$\bullet$-MDH}}{\leftarrow}}{\texttt{CC}},\overset{\overset{\texttt{$\bullet$-MDH}}{\leftarrow}}{\texttt{d}})\le 1024$ &\texttt{ (Number of \texttt{CC}s limited)} \\
& \\
\texttt{R3} & $\texttt{\#PRT}(\overset{\overset{\texttt{$\uparrow$-MDH}}{\leftarrow}}{\texttt{BLK}},\overset{\overset{\texttt{$\uparrow$-MDH}}{\leftarrow}}{\texttt{d}})>1 \wedge \co{\overset{\overset{\texttt{$\uparrow$-MDH}}{\leftarrow}}{\texttt{d}}}\neq\dplus_{\overset{\overset{\texttt{$\uparrow$-MDH}}{\leftarrow}}{\texttt{d}}} \,\Rightarrow\, \texttt{$\uparrow$-mem}^\texttt{<ob>}(\overset{\overset{\texttt{$\uparrow$-MDH}}{\leftarrow}}{\texttt{BLK}},\overset{\overset{\texttt{$\uparrow$-MDH}}{\leftarrow}}{\texttt{d}})\in\{\texttt{DM}\}$ &\texttt{ (\texttt{SMX}s combine in \texttt{DM})} \\
   & $\texttt{\#PRT}(\overset{\overset{\texttt{$\uparrow$-MDH}}{\leftarrow}}{\texttt{WRP}},\overset{\overset{\texttt{$\uparrow$-MDH}}{\leftarrow}}{\texttt{d}})>1 \wedge \co{\overset{\overset{\texttt{$\uparrow$-MDH}}{\leftarrow}}{\texttt{d}}}\neq\dplus_{\overset{\overset{\texttt{$\uparrow$-MDH}}{\leftarrow}}{\texttt{d}}} \,\Rightarrow\, \texttt{$\uparrow$-mem}^\texttt{<ob>}(\overset{\overset{\texttt{$\uparrow$-MDH}}{\leftarrow}}{\texttt{WRP}},\overset{\overset{\texttt{$\uparrow$-MDH}}{\leftarrow}}{\texttt{d}})\in\{\texttt{DM,SM}\}$ &\texttt{ (\texttt{WRP}s combine in \texttt{DM/SM})} \\
& \\
\end{tabular}
}
\caption{
CUDA+WRP model constraints on tuning parameters
}
\label{tab_model_constraints_cuda_wrp}
\end{table}
\setlength\dashlinedash{0.5pt}
\setlength\dashlinegap{3pt}
\setlength\arrayrulewidth{0.3pt}
\begin{table}[t!]
\center
\resizebox{\textwidth}{!}{
\begin{tabular}{c|lr}
No. & Constraint& \\
\hline & &\\
\texttt{0} & $\prod_{d\in[1,D]_\IN}\texttt{\#PRT}(\overset{\overset{\texttt{$\bullet$-MDH}}{\leftarrow}}{\texttt{WI}},\overset{\overset{\texttt{$\bullet$-MDH}}{\leftarrow}}{\texttt{d}})\le C_\texttt{DEV}$ &\texttt{ (Number of \texttt{PE}s limited)} \\
& \\
\texttt{R3} & $\texttt{\#PRT}(\overset{\overset{\texttt{$\uparrow$-MDH}}{\leftarrow}}{\texttt{WG}},\overset{\overset{\texttt{$\uparrow$-MDH}}{\leftarrow}}{\texttt{d}})>1 \wedge \co{\overset{\overset{\texttt{$\uparrow$-MDH}}{\leftarrow}}{\texttt{d}}}\neq\dplus_{\overset{\overset{\texttt{$\uparrow$-MDH}}{\leftarrow}}{\texttt{d}}} \,\Rightarrow\, \texttt{$\uparrow$-mem}^\texttt{<ob>}(\overset{\overset{\texttt{$\uparrow$-MDH}}{\leftarrow}}{\texttt{WG}},\overset{\overset{\texttt{$\uparrow$-MDH}}{\leftarrow}}{\texttt{d}})\in\{\texttt{GM}\}$ &\texttt{ (\texttt{CU}s combine in \texttt{GM})} \\
   & $\texttt{\#PRT}(\overset{\overset{\texttt{$\uparrow$-MDH}}{\leftarrow}}{\texttt{WI}},\overset{\overset{\texttt{$\uparrow$-MDH}}{\leftarrow}}{\texttt{d}})>1 \wedge \co{\overset{\overset{\texttt{$\uparrow$-MDH}}{\leftarrow}}{\texttt{d}}}\neq\dplus_{\overset{\overset{\texttt{$\uparrow$-MDH}}{\leftarrow}}{\texttt{d}}} \,\Rightarrow\, \texttt{$\uparrow$-mem}^\texttt{<ob>}(\overset{\overset{\texttt{$\uparrow$-MDH}}{\leftarrow}}{\texttt{WI}},\overset{\overset{\texttt{$\uparrow$-MDH}}{\leftarrow}}{\texttt{d}})\in\{\texttt{GM},\texttt{LM}\}$ &\texttt{ (\texttt{PE}s combine in \texttt{GM}/\texttt{LM})} \\
& \\
\end{tabular}
}
\caption{
OpenCL model constraints on tuning parameters
}
\label{tab_model_constraints_ocl}
\end{table}

Constraints of programming models
can be expressed in our formalism;~we demonstrate this using the example models CUDA and OpenCL.
For this, we add to the general, model-unspecific constraints (described in Section~\ref{ch_gen_ll_exp}) the new, model-specific constraints for CUDA (in Table~\ref{tab_model_constraints_cuda} or Table~\ref{tab_model_constraints_cuda_wrp}) or for OpenCL (in Table~\ref{tab_model_constraints_ocl}), respectively.

For brevity, we use in the following:
\[
    (
    \overset{\overset{\texttt{$\bullet$-MDH}}{\leftarrow}}{l_\texttt{ASM}}
    ,
    \overset{\overset{\texttt{$\bullet$-MDH}}{\leftarrow}}{d_\texttt{ASM}}
    )
    \ := \
    \leftrightarrow^{-1}_\texttt{$\bullet$-ass}(l_\texttt{ASM},d_\texttt{ASM}) \ , \ \ \bullet\in\{\downarrow,f,\uparrow\}
\]

In Tables~\ref{tab_model_constraints_cuda} and~\ref{tab_model_constraints_cuda_wrp} for CUDA, the constraint No. \texttt{0} (which constrains tuning parameter No. \texttt{0} in Table~\ref{tab_tps}) limits the number of cuda cores (\texttt{CC}) to $1024$, according to the CUDA specification~\cite{cuda-specification}.
The constraints on tuning parameter~\texttt{R3} specify that the results of \texttt{SMX} can be combined in device memory (\texttt{DM}) only in CUDA, and the results of \texttt{CC}s/\texttt{WRP}s in only device memory~(\texttt{DM}) or shared memory~(\texttt{SM}).
Note that in the case of Table~\ref{tab_model_constraints_cuda_wrp}, \texttt{CC}s are not constrained in parameter~14, as \texttt{CC}s within a \texttt{WRP} have access to all CUDA memory regions:~\texttt{DM}, \texttt{SM}, as well as \texttt{RM} (via warp shuffles~\cite{cuda_shuffles}).

In Table~\ref{tab_model_constraints_ocl} for OpenCL, the constraints are similar to the CUDA's constraints in Tables~\ref{tab_model_constraints_cuda}:~they limit the number of \texttt{PE}s to $C_\texttt{DEV}$ (which is a device-specific constant in OpenCL), and the constraints specify the valid memory regions for combining the results of cores, according to the OpenCL specification~\cite{opencl-specification}.

Note that the tables present some important example constraints only and are not complete:~for example, CUDA and OpenCL devices are also constrained regarding their memory sizes (shared/private memory), which is not considered in the tables for brevity.

\newpage

\subsection{Inverse Concatenation}
\label{app_inverse_concat}

\begin{definition}[Inverse Concatenation]
    \label{def:inverse_mda_concat}

    The inverse of operator \emph{concatenation} (Example~\ref{def:mda_concat}) is function $\dplus^{-1}$ which is of type
    \begin{align*}
    &
    {\dplus^{-1}}^{<T\in\type\,|\,D\in\IN\,|\,d\in[1,D]_\IN\,|\,(I_1,\dotsc,I_{d-1},I_{d+1},\dotsc,I_D)\in\IDXs^{D-1},(P,Q)\in\IDXsxIDXs>}: \\
    &\hspace*{20px}
    T[ \, I_1 \, ,\dotsc, \, \underset{\underset{d}{\uparrow}}{\underbrace{id(P\cupdot Q)}} \, ,\dotsc, \, I_D \, ]
    \ \to \
    T[ \, I_1 \, ,\dotsc, \, \underset{\underset{d}{\uparrow}}{\underbrace{id(P)}} \, ,\dotsc, \, I_D \, ]\ \times \ T[ \,
    I_1 \, ,\dotsc, \, \underset{\underset{d}{\uparrow}}{\underbrace{id(Q)}} \, ,\dotsc, \, I_D \, ]
    \end{align*}
    where $id:\IDXs\to\IDXs$ is the identity function on MDA index sets.
    The function is computed as:
    \begin{align*}
    &
    {\dplus^{-1}}^{<T\,|\,D\,|\,d\,|\,(I_1,\dotsc,I_{d-1},I_{d+1},\dotsc,I_D),(P,Q)>}( \, \MDA \ ) \ := \ (\MDA_1,\MDA_2)
    \end{align*}
    for
    \[
        \MDA_1[ \, i_1 \, ,\dotsc, \ \ i_d \ \ ,\dotsc, \, i_D \, ] \ := \ \MDA[ \, i_1 \, ,\dotsc, \ \ i_d \ \ ,\dotsc, \, i_D \, ] \ , \ \ i_d\in P
    \]
    and
    \[
        \MDA_2[ \, i_1 \, ,\dotsc, \ \ i_d \ \ ,\dotsc, \, i_D \, ] \ := \ \MDA[ \, i_1 \, ,\dotsc, \ \ i_d \ \ ,\dotsc, \, i_D \, ] \ , \ \ i_d\in Q
    \]
    i.e., $\MDA_1$ and $\MDA_2$ behave exactly as MDA $\MDA$ on their restricted index sets $P$ or $Q$, respectively.

    We often write for $(\MDA_1,\MDA_2) := {\dplus^{-1}}^{<\dotsc>}(\MDA)$ (meta-parameters omitted via ellipsis) also
    \[
        \MDA =: \MDA_1 {\dplus}^{<\dotsc>} \MDA_2
    \]
    Our notation is justified by the fact that the inverse of MDA $\MDA$ is uniquely determined, as the two MDAs $\MDA_1$ and $\MDA_2$ which are equal to MDA $\MDA$ when concatenating them.
\end{definition}

\subsection{Example~\ref{fig_ll_example} in Verbose Math Notation}
\label{app_ll_example}

Figures~\ref{fig_app_decomp}-\ref{fig_app_comp} show our low-level representation from Example~\ref{fig_ll_example} in verbose math notation.
The symbols $\blacksquare_\bot,\dotsc,\blacksquare_f$ used in the figures are a textual abbreviation for:
\begin{align*}
  &\blacksquare_\bot         &:=& & &*     , *     \ &&|&& \ *     , *     \ &&|&& \ *     , *     \\
  &\blacksquare^\texttt{1}_1 &:=& & &*     , *     \ &&|&& \ *     , *     \ &&|&& \ *     , *     \\
  &\blacksquare^\texttt{1}_2 &:=& & &p^1_1 , *     \ &&|&& \ *     , *     \ &&|&& \ *     , *     \\
  &\blacksquare^\texttt{2}_1 &:=& & &p^1_1 , p^1_2 \ &&|&& \ *     , *     \ &&|&& \ *     , *     \\
  &\blacksquare^\texttt{2}_2 &:=& & &p^1_1 , p^1_2 \ &&|&& \ p^2_1 , *     \ &&|&& \ *     , *     \\
  &\blacksquare^\texttt{3}_1 &:=& & &p^1_1 , p^1_2 \ &&|&& \ p^2_1 , p^2_2 \ &&|&& \ *     , *     \\
  &\blacksquare^\texttt{3}_2 &:=& & &p^1_1 , p^1_2 \ &&|&& \ p^2_1 , p^2_2 \ &&|&& \ p^3_1 , *     \\
  &\blacksquare_f            &:=& & &p^1_1 , p^1_2 \ &&|&& \ p^2_1 , p^2_2 \ &&|&& \ p^3_1 , p^3_2
\end{align*}
where symbol $*$ indicates generalization in meta-parameters (Definition~\ref{app_gen_meta_func}).

\vspace*{5px}

In Example~\ref{fig_ll_example}, the arrow annotation of combine operators is formally an abbreviation.
For example, operator~$\dplus^{\texttt{(COR,y)}}_2$ in Figure~\ref{fig_ll_example}
is annotated with~\texttt{$\rightarrow$ M: HM[1,2],v: HM[1]} which abbreviates

\[
\dotsc \ \ \
{^\downarrow}{\MDA^2_2}^{
<\, p^1_1,p^1_2 \ | \ p^2_1,p^2_2\,:=\,*  \ | \ p^3_1\,:=\,*, p^3_2\,:=\,* \,>}
=:\underset{p^2_2\in[0,16)_{\INz}\phantom{COR}}{\dplus^\texttt{(COR,y)}_2}
\ \dotsc
\]

Here, ${^\downarrow}{\MDA^2_2}$ represents
the low-level MDA (Definition~\ref{def_ll_mda}) that is already
partitioned
for layer $1$ in dimensions $1$ and $2$, and for layer $2$ in dimension $1$ (because in Figure~\ref{fig_ll_example}, operators $\dplus^\texttt{(HM,x)}_1$, $\dplus^\texttt{(HM,y)}_2$, $\dplus^\texttt{(COR,x)}_1$ appear before operator $\dplus^\texttt{(COR,y)}_2$), but not yet for layer $2$ in dimension~$2$ as well as for layer~$3$ in both dimensions (indicated by symbol $*$ which is described formally in Definition~\ref{app_gen_meta_func} of our Appendix).
In our generated code (discussed in Section~\ref{app_sec_code_generation} of our Appendix), we store low-level MDAs, like ${^\downarrow}{\MDA^2_2}$, using their domain-specific data representation, as the domain-specific representation is usually more efficient:~in the case of \texttt{MatVec}, we physically store matrix $M$ and vector $v$ for the input MDA, and vector $w$ for the output MDA.
For example, low-level MDA
\[
    {^\downarrow}{\MDA^2_2}^{<\, p^1_1,p^1_2 \ | \ p^2_1,p^2_2\,:=\,*  \ | \ p^3_1\,:=\,*, p^3_2\,:=\,* \,>}
\]
can be transformed via view functions (Definitions~\ref{def_iv} and~\ref{def_ov}) to \emph{low-level BUFs} (Definition~\ref{ll_buffers})
\[
    {M^2_2}^{ <\texttt{HM} \:|\: id><\, p^1_1,p^1_2 \ | \ p^2_1,p^2_2\,:=\,*  \ | \ p^3_1\,:=\,*, p^3_2\,:=\,* \,>}
{ \ , \ \ }
    {v^2_2}^{<\texttt{HM} \:|\: id><\, p^1_1,p^1_2 \ | \ p^2_1,p^2_2\,:=\,*  \ | \ p^3_1\,:=\,*, p^3_2\,:=\,* \,>}
\]
and back (Lemma~\ref{theorem_views}).
Similarly as for data structures in low-level programming models (e.g., \emph{C arrays} as in OpenMP, CUDA, and OpenCL), low-level BUFs are defined to have an explicit notion of memory regions
and memory layouts.

\vspace*{5px}

In Figure~\ref{fig_app_decomp}, we de-compose the input MDA ${^\downarrow}{\MDA}$, step by step, for the MDH levels $(1,1),\dotsc,(3,2)$:
\[
{^\downarrow}\MDA
=:
{^\downarrow}\MDA_\bot^{<\blacksquare_\bot>}
\to
{^\downarrow}{\MDA^1_1}^{<\blacksquare^1_1>}
\to
{^\downarrow}{\MDA^1_2}^{<\blacksquare^1_2>}
\to
{^\downarrow}{\MDA^2_1}^{<\blacksquare^2_1>}
\to
{^\downarrow}{\MDA^2_2}^{<\blacksquare^2_2>}
\to
{^\downarrow}{\MDA^3_1}^{<\blacksquare^3_1>}
\to
{^\downarrow}{\MDA^3_2}^{<\blacksquare^3_2>}
\to
{^\downarrow}\MDA_f^{<\blacksquare_f>}
\]
The input MDAs $({^\downarrow}{\MDA^l_d})_{l\in[1,3]_\IN,d\in[1,2]_\IN}$, as well as ${^\downarrow}\MDA_\bot$ and ${^\downarrow}{\MDA_f}$, are all low-level MDA representations~(Definition~\ref{def_ll_mda}).
We use as partitioning schema $P$ (Definition~\ref{def_ll_mda})
\[
P:=
(\,
(P^1_1,P^1_2)
\,,\,\,
(P^2_1,P^2_2)
\,,\,\,
(P^3_1,P^3_2)
)
=
(\,
(2,4)
\,,\,
(8,16)
\,,\,
(32,64)
)
\]
and we use the index sets $I_d$ from Definition~\ref{app_mda_part} (which define a uniform index set partitioning):
\[
  ( \
    \fsize{d}{\dplus_{d}}{MDA}{MDA}(I_d^{<p^1_1,p^1_2\,|\,p^2_1,p^2_2\,|\,p^3_1,p^3_2>})^{<(p^1_1,p^1_2)\in P^1_1\times P^1_2\,|\,(p^2_1,p^2_2)\in P^2_1\times P^2_2\,|\,(p^3_1,p^3_2)\in P^3_1\times P^3_2>}
  \ )_{d\in[1,D]_\IN}
\]
Here, $\fsize{d}{\dplus_{d}}{MDA}{MDA}$ denotes the index set function of combine operator concatenation (Example~\ref{def:mda_concat}), which is the identity function
and explicitly stated for the sake of completeness only.
Note that in Figure~\ref{fig_app_decomp}, we access low-level MDAs ${^\downarrow}{\MDA^l_d}$ as generalized in some partition sizes, via $*$ (Definition~\ref{app_gen_meta_func}), according to the definitions of the $\blacksquare^l_d$.

Each MDA
${^\downarrow}{\MDA}$
can be transformed to its domain-specific data representation matrix~%
${^\downarrow}{M}$
and vector~%
${^\downarrow}{v}$
and vice versa, using the view functions, as discussed above.

Figure~\ref{fig_app_f} shows our scalar phase, which is formally trivial.

In Figure~\ref{fig_app_comp}, we re-compose the computed
data ${^\uparrow}{\MDA_f}^{<\blacksquare_f>}$, step by step, to the final result~${^\uparrow}{\MDA}$:
\[
{^\uparrow}{\MDA_f}^{<\blacksquare_f>}
\to
{^\uparrow}{\MDA^1_1}^{<\blacksquare^1_1>}
\to
{^\uparrow}{\MDA^1_2}^{<\blacksquare^1_2>}
\to
{^\uparrow}{\MDA^2_2}^{<\blacksquare^2_2>}
\to
{^\uparrow}{\MDA^3_1}^{<\blacksquare^3_1>}
\to
{^\uparrow}{\MDA^3_2}^{<\blacksquare^3_2>}
\to
{^\uparrow}\MDA_\bot^{<\blacksquare_\bot>}
=:
{^\uparrow}{\MDA}
\]
Analogously to the de-composition phase, each output MDA $({^\uparrow}{\MDA^l_d})_{l\in[1,3]_\IN,d\in[1,2]_\IN}$, as well as ${^\uparrow}{\MDA_f}$ and ${^\uparrow}{\MDA_\bot}$, are low-level MDA representations, for $P$ as defined above and index sets
\[
  ( \
    \fsize{d}{\co{d}}{MDA}{MDA}(I_d^{<p^1_1,p^1_2\,|\,p^2_1,p^2_2\,|\,p^3_1,p^3_2>})^{<(p^1_1,p^1_2)\in P^1_1\times P^1_2\,|\,(p^2_1,p^2_2)\in P^2_1\times P^2_2\,|\,(p^3_1,p^3_2)\in P^3_1\times P^3_2>}
  \ )_{d\in[1,D]_\IN}
\]
where $\fsize{d}{\co{d}}{MDA}{MDA}$ are the index set functions of the combine operators (Definition~\ref{def_combine_op}) used in the re-composition phase.
The same as in the de-composition phase, we access the output low-level MDAs as generalized in some partition sizes, according to our definitions of the $\blacksquare^l_d$, and
we identify each MDA with its domain-specific data representation (the output vector $w$).

\begin{figure}[!p]
  \includegraphics[width=0.85\textwidth]{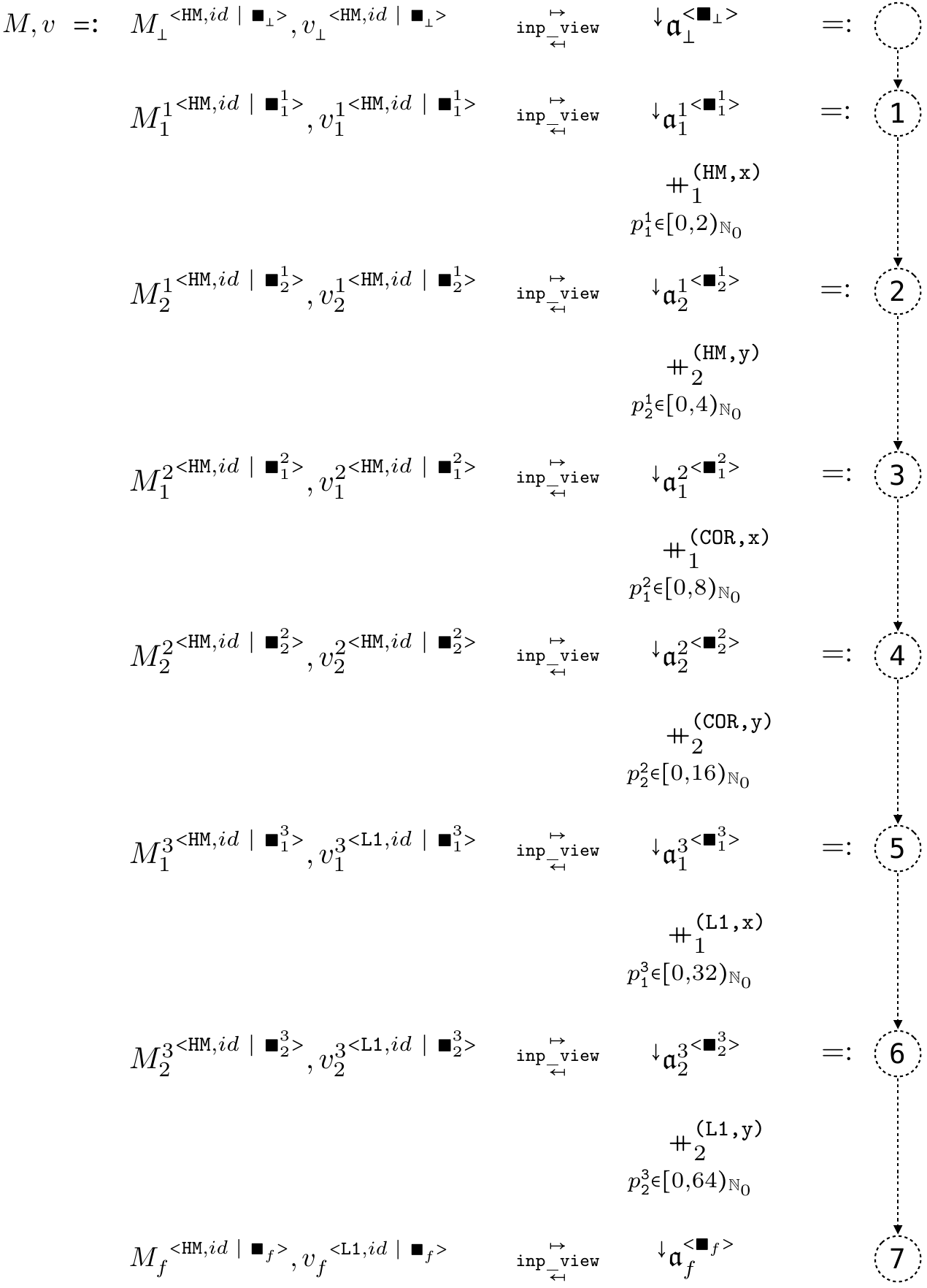}
\caption{
De-composition phase of Example~\ref{fig_ll_example} in verbose math notation.}
\label{fig_app_decomp}
\end{figure}
\clearpage

\begin{figure}[!p]
  \includegraphics[width=0.7\textwidth]{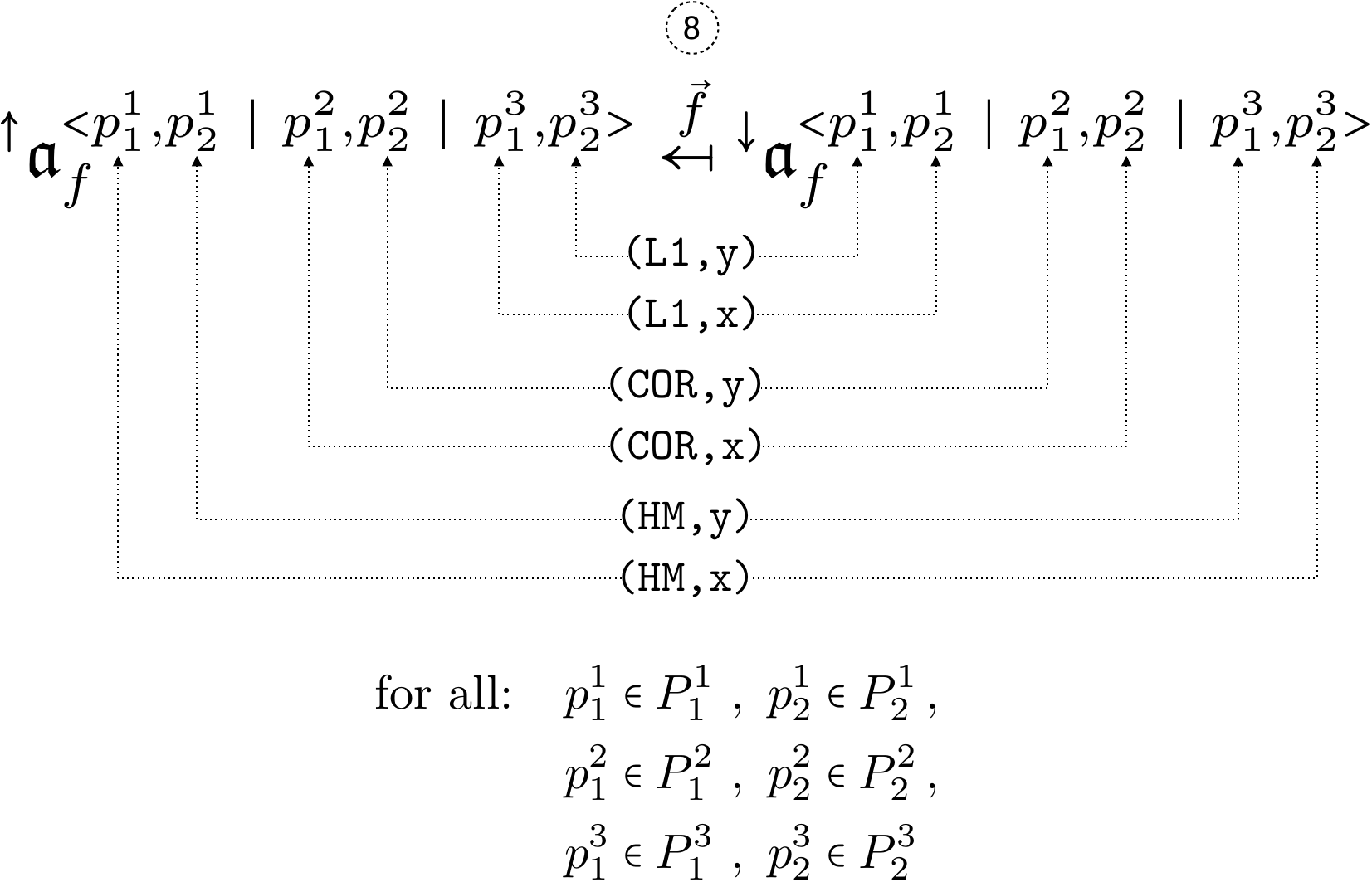}
\caption{Scalar phase of Example~\ref{fig_ll_example} in verbose math notation.}
\label{fig_app_f}
\end{figure}

\clearpage

\begin{figure}[!p]
  \includegraphics[width=0.68\textwidth]{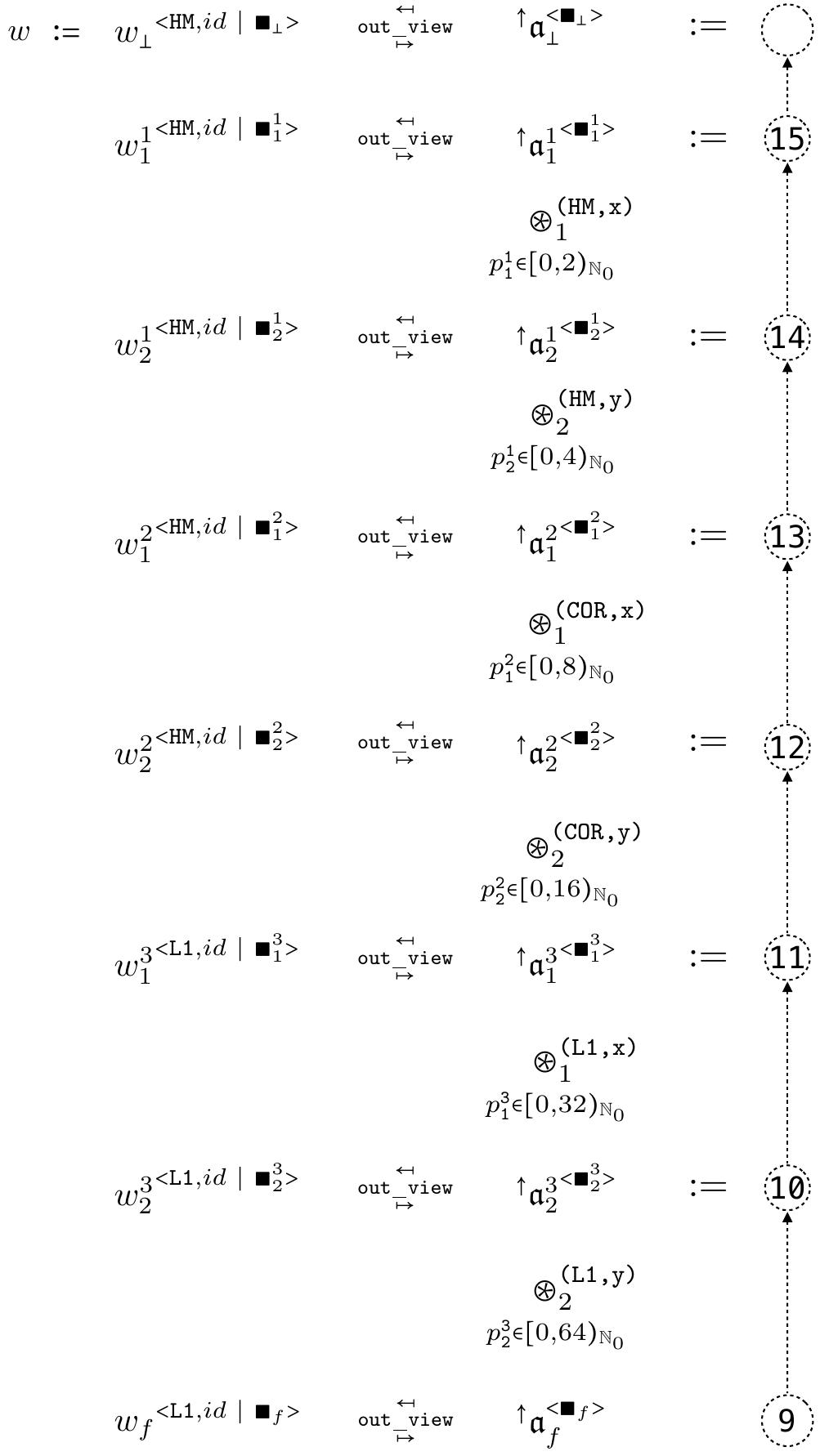}
\caption{Re-composition phase of Example~\ref{fig_ll_example} in verbose math notation.}
\label{fig_app_comp}
\end{figure}

\clearpage

\subsection{Multi-Dimensional ASM Arrangements}
\label{app_md_arrangement_mem_cor}

\begin{figure}[b]
    \includegraphics[width=0.75\textwidth]{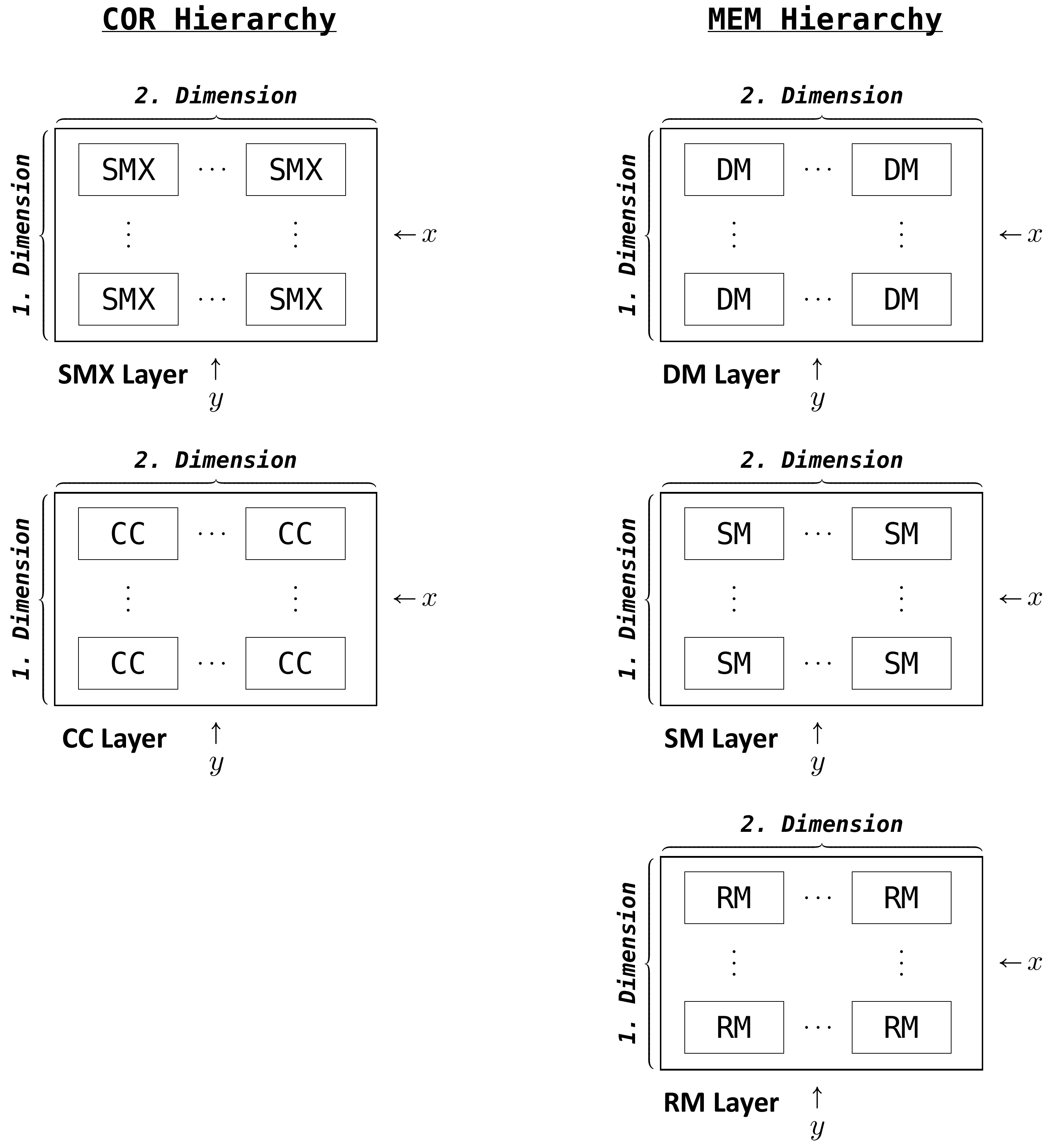}
    \caption{Multi-dimensional ASM arrangement illustrated using CUDA for the case $D=2$ (two dimensions)}
    \label{img_cuda_md_arrangement}
\end{figure}

We demonstrate how we arrange memory regions and cores of ASM-represented systems (Section~\ref{sec_amm}) in multiple dimensions using the example of CUDA.

\paragraph*{Cores (COR):}

In CUDA,
\texttt{SMX} cores are programmed via so-called \emph{CUDA Blocks}, and CUDA's \texttt{CC} cores are programmed via \emph{CUDA Threads}.
CUDA has native support for arranging its blocks and threads in up to three dimensions which are called \texttt{x}, \texttt{y}, and \texttt{z} in CUDA~\cite{cuda-programming-guide}.
Consequently, even though the original CUDA specification~\cite{cuda-specification} introduces \texttt{SMX} and \texttt{CC} without having an order,
the CUDA programmer benefits from imagining \texttt{SMX} and \texttt{CC} as three-dimensionally arranged.

Additional dimensions can be explicitly programmed in CUDA.
For example, to add a fourth dimension to CUDA, we can embed the additional dimension in the CUDA's \texttt{z} dimension, thereby splitting CUDA dimension \texttt{z} in the explicitly programmed dimensions \texttt{z\_1} (third dimension) and \texttt{z\_2} (fourth dimension), as follows:
\[
    \texttt{z\_1 := z \% Z\_1 and z\_2 := z / Z\_1}
\]
Here, \texttt{Z\_1} represents the number of threads in the additional dimension, and symbol \texttt{\%} denotes the modulo operator.

\paragraph*{Memory (MEM):}

In CUDA, memory is managed via \emph{C arrays} which may be multi-dimensional:~%
to arrange $(\texttt{DIM\_1}\times\dotsc\times\texttt{DIM\_D})$-many memory regions, each of size \texttt{N}, we use a CUDA array of the following type (pseudocode):
\[
  \texttt{array[ DIM\_1 ]$\dotsc$[ DIM\_D ][ N ]}
\]

Note that CUDA implicitly arranges its \emph{shared} and \emph{private} memory allocations in multiple dimensions, depending on the number of blocks and threads:~a shared memory array of type
\texttt{shared\_array[ DIM\_1 ]$\dotsc$[ DIM\_D ][ N ]}
is internally managed in CUDA as
\texttt{shared\_array[ blockIdx.x ][ blockIdx.y ][ blockIdx.z ][ DIM\_1 ]$\dotsc$[ DIM\_D ][ N ]}, i.e., each CUDA block has its own shared memory region.
Analogously, a private memory array
\texttt{private\_array[ DIM\_1 ]$\dotsc$[ DIM\_D ][ N ]}
is managed in CUDA as
\texttt{private\_array[ blockIdx.x ][ blockIdx.y ][ blockIdx.z ][ threadIdx.x ][ threadIdx.y ][ threadIdx.z ][ DIM\_1 ]$\dotsc$[ DIM\_D ][ N ]}, correspondingly.
Our arrangement methodology continues the CUDA's approach by explicitly programming the additional arrangement dimensions $\texttt{DIM\_1},\dotsc,\texttt{DIM\_D}$.

\vspace*{2px}

Figure~\ref{img_cuda_md_arrangement} illustrates our multi-dimensional core and memory arrangement using the example of CUDA, for $D=2$ (two-dimensional arrangement).

\vspace*{10px}

\subsection{ASM Levels}
\label{app_asm_levels}

ASM levels are pairs $(l_\texttt{ASM},d_\texttt{ASM})$ consisting of an ASM layer $l_\texttt{ASM}\in\IN$ and ASM dimension $d_\texttt{ASM}\in\IN$.

Figure~\ref{img_cuda_thread_hierarchy} illustrates ASM levels using the example of CUDA's thread hierarchy.
The figure shows that thread hierarchies can be considered as a tree in which each level is uniquely determined by a particular combination of a layer (\texttt{block} or \texttt{thread} in the case of CUDA) and dimension (\texttt{x}, \texttt{y}, or~\texttt{z}).
In the figure, we use \texttt{lvl} as an abbreviation for \emph{level}, \texttt{l} for \emph{layer}, and \texttt{d} for \emph{dimension}.

For ASM layers and dimensions, we usually use their domain-specific identifiers, e.g., \texttt{BLK}/\texttt{CC} and \texttt{x}/\texttt{y}/\texttt{z} as aliases for numerical values of layers and dimensions.

\vspace*{20px}

\begin{figure}[h!]
    \includegraphics[width=\textwidth]{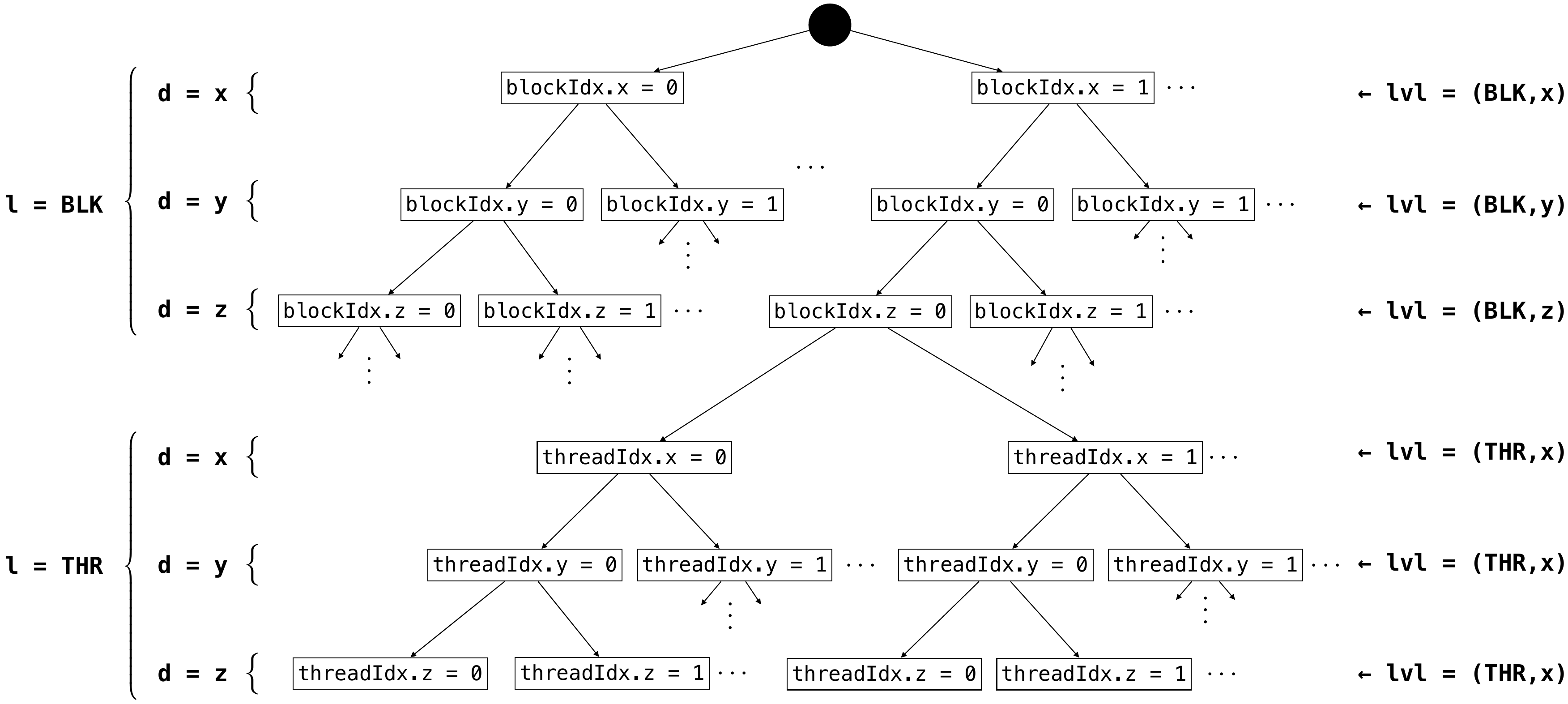}
    \caption{ASM levels illustrated using CUDA's thread hierarchy}
    \label{img_cuda_thread_hierarchy}
\end{figure}

\newpage

\subsection{MDH Levels}
\label{app_mdh_levels}

MDH levels are pairs $(l_\texttt{MDH},d_\texttt{MDH})$ consisting of an MDH layer $l_\texttt{MDH}\in\IN$ and MDH dimension $d_\texttt{MDH}\in\IN$.

\begin{figure}[h!]
    \includegraphics[width=0.9\textwidth]{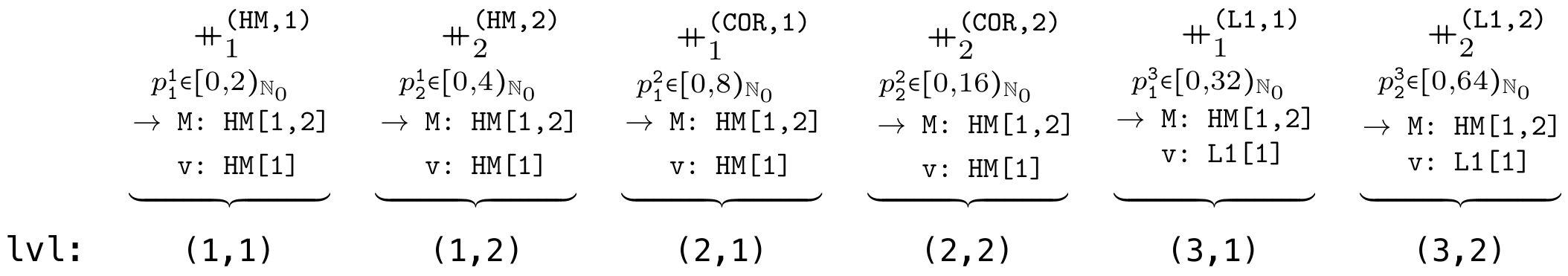}
    \caption{MDH levels illustrated using as example the de-composition phase in Figure~\ref{fig_ll_example}}
    \label{img_mdh_thread_hierarchy}
\end{figure}

Figure~\ref{img_mdh_thread_hierarchy} illustrates MDH levels using as example the de-composition phase in Figure~\ref{fig_ll_example}.
The levels $(l_\texttt{MDH},d_\texttt{MDH})$ can be derived from the super- and subscripts of combine operators' variables~$p^{l_\texttt{MDH}}_{d_\texttt{MDH}}$.

\subsection{MDA Partitioning}
\label{app_sec_mda_part}

We demonstrate how we partition MDAs into equally sized parts (a.k.a. \emph{uniform partitioning}).

\begin{definition}[MDA Partitioning]
\label{app_mda_part}
Let $\MDA\in T[I_1,\dotsc,I_D]$ be an arbitrary MDA that has scalar type $T\in\type$, dimensionality $D\in\IN$, index sets $I=(I_1,\dotsc,I_D)\in\IDXs^D$, and size $N=\{|I_1|,\dotsc,|I_D|\}\in\IN^D$.
We consider $I_d = \{ i^d_1,\dotsc,i^d_{N_d}\}$, $d\in[1,D]_\IN$, such that $i^d_1<\dotsc<i^d_{N_d}$ represents a sorted enumeration of the elements in $I_d$.
Let further
$
P=(\,
(P^1_1,\dotsc,P^1_D)
\ ,\,\dotsc\,, \
(P^L_1,\dotsc,P^L_D)\,)
$
be an arbitrary tuple of $L$-many $D$-tuples of positive natural numbers such that
$\prod_{l\in[1,L]_\IN} P^l_d$
divides~$N_d$ (the number of indices of MDA~$\MDA$ in dimension~$d$),
for each $d\in\{1,\dotsc,D\}$.

The \emph{$L$-layered, $D$-dimensional, $P$-partitioning of MDA $\MDA$} is
the $L$-layered, $D$-dimensional, $P$-partitioned low-level MDA $\MDA_\texttt{prt}$ (Definition~\ref{def_ll_mda}) that has scalar type $T$ and index sets
\begin{align*}
&
I_d^{<\, p^1_d,\dotsc,p^L_d  \,>} :=
\\
&\hspace*{20px}
\{ i_j\in I_d \ | \ j = OS + j' \, , \ \text{for} \ \ OS:=\sum_{l\in[1,L]_\IN}p^l_d*\frac{N_d}{\prod_{l'\in[1,l]_\IN} P^{l'}_d} \ \ \text{and} \ \ j'\in PS:=\frac{N_d}{\prod_{l'\in[1,L]_\IN} P^{l'}_d}
\}
\end{align*}
i.e., set
$I_d^{<\, p^1_d,\dotsc,p^L_d  \,>}$
denotes for each choice of parameters $p^1_d,\dotsc,p^L_d$ a part of the uniform partitioning of the ordered index set $I_d$ ($OS$ in the formula above represents the \underline{O}ff\underline{S}et to the part, and $PS$ the \underline{P}art's \underline{S}ize).
The partitioned MDA $\MDA_\texttt{prt}$ is defined as:
\[
  \MDA \ \ =: \ \ \ \
\underbrace{\concat{1}{p^1_1\in P^1_1}\dotsc\concat{D}{p^1_D\in P^1_D}}_\text{Layer $1$}
\ \ \ \ \dotsc \ \ \ \
\underbrace{\concat{1}{p^L_1\in P^L_1}\dotsc\concat{D}{p^L_D\in P^L_D}}_\text{Layer $L$}
\ \ \ \
\MDA_\texttt{prt}^{<
\, p^1_1,\dotsc,p^1_D \ | \ \dotsc \ | \ p^L_1,\dotsc,p^L_D \,
>}
\]
i.e., the parts $\MDA_\texttt{prt}^{<\, p^1_1,\dotsc,p^1_D \ | \ \dotsc \ | \ p^L_1,\dotsc,p^L_D \,>}$ are defined such that concatenating them results in the original MDA $\MDA$.
\end{definition}

\vspace*{5px}

\subsection{TVM Schedule for \texttt{MatMul}}
\label{app_tvm_schedule_matmul_resnet}

Listing~\ref{ansor_a100_resnet50_training} shows TVM's Ansor-generated schedule program for \texttt{MatMul} on input matrices of sizes $16\times2048$ and $2048\times1000$ taken from ResNet-50's training phase, discussed in Section~\ref{ll_examples_matmul_resnet}.
Code formatting, like names of variables and comments, have been shortened and adapted in the listing for brevity.

\begin{lstlisting}[
    language=Python,
    morekeywords={},
    breaklines=true,
    mathescape=true,
    numbers=left,
    xleftmargin=1.5em,
    float=h!,
    frame=lines,
    caption={TVM schedule for Matrix Multiplication on \texttt{NVIDIA\:Ampere\:GPU} (variable names shortened for brevity)},
    label={ansor_a100_resnet50_training}
    ]
# exploiting fast memory resources for computed results
matmul_local, = s.cache_write([matmul], "local")
matmul_1, matmul_2, matmul_3 = tuple(matmul_local.op.axis) + tuple(matmul_local.op.reduce_axis)
SHR_1, REG_1 = s[matmul_local].split(matmul_1, factor=1)
THR_1, SHR_1 = s[matmul_local].split(SHR_1, factor=1)
DEV_1, THR_1 = s[matmul_local].split(THR_1, factor=4)
BLK_1, DEV_1 = s[matmul_local].split(DEV_1, factor=2)
SHR_2, REG_2 = s[matmul_local].split(matmul_2, factor=1)
THR_2, SHR_2 = s[matmul_local].split(SHR_2, factor=1)
DEV_2, THR_2 = s[matmul_local].split(THR_2, factor=20)
BLK_2, DEV_2 = s[matmul_local].split(DEV_2, factor=1)
SHR_3, REG_3 = s[matmul_local].split(matmul_3, factor=2)
DEV_3, SHR_3 = s[matmul_local].split(SHR_3, factor=128)
s[matmul_local].reorder(BLK_1, BLK_2, DEV_1, DEV_2, THR_1, THR_2, DEV_3, SHR_3, SHR_1, SHR_2, REG_3, REG_1, REG_2)

# low-level optimizations:
s[matmul_local].pragma(BLK_1, "auto_unroll_max_step", 512)
s[matmul_local].pragma(BLK_1, "unroll_explicit", True)

# tiling
matmul_1, matmul_2, matmul_3 = tuple(matmul.op.axis) + tuple(matmul.op.reduce_axis)
THR_1, SHR_REG_1 = s[matmul].split(matmul_1, factor=1)
DEV_1, THR_1 = s[matmul].split(THR_1, factor=4)
BLK_1, DEV_1 = s[matmul].split(DEV_1, factor=2)
THR_2, SHR_REG_2 = s[matmul].split(matmul_2, factor=1)
DEV_2, THR_2 = s[matmul].split(THR_2, factor=20)
BLK_2, DEV_2 = s[matmul].split(DEV_2, factor=1)
s[matmul].reorder(BLK_1, BLK_2, DEV_1, DEV_2, THR_1, THR_2, SHR_REG_1, SHR_REG_2)
s[matmul_local].compute_at(s[matmul], THR_2)

# block/thread assignments:
BLK_fused = s[matmul].fuse(BLK_1, BLK_2)
s[matmul].bind(BLK_fused, te.thread_axis("blockIdx.x"))
DEV_fused = s[matmul].fuse(DEV_1, DEV_2)
s[matmul].bind(DEV_fused, te.thread_axis("vthread"))
THR_fused = s[matmul].fuse(THR_1, THR_2)
s[matmul].bind(THR_fused, te.thread_axis("threadIdx.x"))

# exploiting fast memory resources for first input matrix:
A_shared = s.cache_read(A, "shared", [matmul_local])$\label{ansor_a100_resnet50_training_1}$
A_shared_ax0, A_shared_ax1 = tuple(A_shared.op.axis)
A_shared_ax0_ax1_fused = s[A_shared].fuse(A_shared_ax0, A_shared_ax1)
A_shared_ax0_ax1_fused_o, A_shared_ax0_ax1_fused_i = s[A_shared].split(A_shared_ax0_ax1_fused, factor=1)
s[A_shared].vectorize(A_shared_ax0_ax1_fused_i)
A_shared_ax0_ax1_fused_o_o, A_shared_ax0_ax1_fused_o_i = s[A_shared].split(A_shared_ax0_ax1_fused_o, factor=80)
s[A_shared].bind(A_shared_ax0_ax1_fused_o_i, te.thread_axis("threadIdx.x"))
s[A_shared].compute_at(s[matmul_local], DEV_3)$\label{ansor_a100_resnet50_training_2}$

# exploiting fast memory resources for second input matrix:
# ... (analogous to lines $\ref{ansor_a100_resnet50_training_1}-\ref{ansor_a100_resnet50_training_2}$)
  \end{lstlisting}
\clearpage

\section{Addendum Section~\ref{ch:eval}}

\subsection{Data Characteristics used in Deep Neural Networks}
\label{app_data_characteristics}

Figure~\ref{fig_data_characteristics} shows the data characteristics used for the deep learning experiments in Figures~\ref{img_eval_dl_gpu} and~\ref{img_eval_dl_cpu} of Section~\ref{ch:eval}.
We use real-world characteristics taken from the neural networks~\texttt{ResNet-50},~\texttt{VGG-16}, and~\texttt{MobileNet}.
For each network, we consider computations \texttt{MCC} and \texttt{MatMul} (Table~\ref{fig_hl_examples}), because these are the networks' most time-intensive building blocks.
Each computation is called in each network on different data characteristics~--~we use for each combination of network and computation the two most time-intensive characteristics.
Note that the MobileNet network does not use \texttt{MatMul} in its implementation.

The capsule variants \texttt{MCC\_Capsule} in Figures~\ref{img_eval_dl_gpu} and~\ref{img_eval_dl_cpu} of Section~\ref{ch:eval} have the same characteristics as those listed for \texttt{MCC}s in Figure~\ref{fig_data_characteristics};~the only difference is that \texttt{MCC\_Capsule}, in addition to the dimensions $\texttt{N}, \texttt{H},\texttt{W},\texttt{K},\texttt{R},\texttt{S},\texttt{C}$, uses three additional dimensions $\texttt{MI}, \texttt{MJ}, \texttt{MK}$, each with a fixed size of $4$.
This is because \texttt{MCC\_Capsule} operates on $4\times4$ matrices, rather than scalars as \texttt{MCC} does.

\vspace{5px}

\begin{figure}[h!]
    \begin{subfigure}{\textwidth}
    \centering
        \resizebox{\columnwidth}{!}{%
        \begin{tabular}{|l|l|r|r|r|r|r|r|r|r|r|l|r|r|l|l|l|}
        \hline
        \multicolumn{1}{|c|}{\textbf{Network}} & \multicolumn{1}{c|}{\textbf{Phase}} & \multicolumn{1}{c|}{\textbf{N}} & \multicolumn{1}{c|}{\textbf{H}} & \multicolumn{1}{c|}{\textbf{W}} & \multicolumn{1}{c|}{\textbf{K}} & \multicolumn{1}{c|}{\textbf{R}} & \multicolumn{1}{c|}{\textbf{S}} & \multicolumn{1}{c|}{\textbf{C}} & \multicolumn{1}{c|}{\textbf{Stride H}} & \multicolumn{1}{c|}{\textbf{Stride W}} & \multicolumn{1}{c|}{\textbf{Padding}} & \multicolumn{1}{c|}{\textbf{P}} & \multicolumn{1}{c|}{\textbf{Q}} & \multicolumn{1}{c|}{\textbf{Image Format}} & \multicolumn{1}{c|}{\textbf{Filter Format}} & \multicolumn{1}{c|}{\textbf{Output Format}} \\ \hline
        \multirow{2}{*}{\texttt{ResNet-50}}             & Training                            & 16                              & 230                             & 230                             & 64                              & 7                               & 7                               & 3                               & 2                                      & 2                                      & VALID                                 & 112                             & 112                             & NHWC                                       & KRSC                                        & NPQK                                        \\ \cline{2-17}
                                               & Inference                           & 1                               & 230                             & 230                             & 64                              & 7                               & 7                               & 3                               & 2                                      & 2                                      & VALID                                 & 112                             & 112                             & NHWC                                       & KRSC                                        & NPQK                                        \\ \hline
        \multirow{2}{*}{\texttt{VGG-16}}                & Training                            & 16                              & 224                             & 224                             & 64                              & 3                               & 3                               & 3                               & 1                                      & 1                                      & VALID                                 & 224                             & 224                             & NHWC                                       & KRSC                                        & NPQK                                        \\ \cline{2-17}
                                               & Inference                           & 1                               & 224                             & 224                             & 64                              & 3                               & 3                               & 3                               & 1                                      & 1                                      & VALID                                 & 224                             & 224                             & NHWC                                       & KRSC                                        & NPQK                                        \\ \hline
        \multirow{2}{*}{\texttt{MobileNet}}             & Training                            & 16                              & 225                             & 225                             & 32                              & 3                               & 3                               & 3                               & 2                                      & 2                                      & VALID                                 & 112                             & 112                             & NHWC                                       & KRSC                                        & NPQK                                        \\ \cline{2-17}
                                               & Inference                           & 1                               & 225                             & 225                             & 32                              & 3                               & 3                               & 3                               & 2                                      & 2                                      & VALID                                 & 112                             & 112                             & NHWC                                       & KRSC                                        & NPQK                                        \\ \hline
        \end{tabular}
        }
        \caption{Data characteristics used for \texttt{MCC} experiments}
    \label{data_char_mcc}
    \end{subfigure}
    \newline
    \vspace{10px}
    \begin{subfigure}{\textwidth}
    \centering
        \resizebox{0.4\columnwidth}{!}{%
        \begin{tabular}{|l|l|r|r|r|l|}
        \hline
        \multicolumn{1}{|c|}{\textbf{Network}} & \multicolumn{1}{c|}{\textbf{Phase}} & \multicolumn{1}{c|}{\textbf{M}} & \multicolumn{1}{c|}{\textbf{N}} & \multicolumn{1}{c|}{\textbf{K}} & \multicolumn{1}{c|}{\textbf{Transposition}} \\ \hline
        \multirow{2}{*}{\texttt{ResNet-50}}             & Training                            & 16                              & 1000                            & 2048                            & NN                                          \\ \cline{2-6}
                                               & Inference                           & 1                               & 1000                            & 2048                            & NN                                          \\ \hline
        \multirow{2}{*}{\texttt{VGG-16}}                & Training                            & 16                              & 4096                            & 25088                           & NN                                          \\ \cline{2-6}
                                               & Inference                           & 1                               & 4096                            & 25088                           & NN                                          \\ \hline
        \end{tabular}
        }
    \caption{Data characteristics used for \texttt{MatMul} experiments}
    \label{data_char_matmul}
    \end{subfigure}
    \vspace{5px}
    \caption{Data characteristics used for experiments in Section~\ref{ch:eval}}
    \label{fig_data_characteristics}
\end{figure}

\vspace*{10px}

\subsection{Runtime and Accuracy of \texttt{cuBLASEx}}
\label{app_cuBLASEx_runtime_accuracy}

Listing~\ref{lst_runtime_cublasex_v100_matmul_1024} shows the runtime of \texttt{cuBLASEx} for its different \emph{algorithm} variants.
For demonstration, we use the example of matrix multiplication \texttt{MatMul} on \texttt{NVIDIA\:Volta\:GPU} for square input matrices of sizes $1024\times1024$.
For each algorithm variant, we list both:~%
1)~the runtime achieved by \texttt{cuBLASEx} (in nanoseconds~\texttt{ns}), as well as
2)~the maximum absolute deviation~($\texttt{delta}_\texttt{max}$ values) compared to a straightforward, sequential CPU computation.
For example, the $\texttt{delta}_\texttt{max}$ value of algorithm \texttt{CUBLAS\_GEMM\_DEFAULT} is \texttt{3.14713e-05}, i.e., at least one value $c^\texttt{GPU}_{i,j}$ in the GPU-computed output matrix deviates by \texttt{3.14713e-05} from its corresponding, sequentially computed value $c^\texttt{seq}_{i,j}$ such that
$|c^\texttt{GPU}_{i,j}|=|c^\texttt{seq}_{i,j}|+\texttt{3.14713e-05}$~(bar symbols $|\dotsc|$ denote absolute value).
All other GPU-computed values $c^\texttt{GPU}_{i',j'}$ deviate from their sequentially computed CPU-variant by \texttt{3.14713e-05} or less.

Note that \texttt{cuBLASEx} offers $42$ algorithm variants, but not all of them are supported for all potential characteristics of the input and output data (size, memory layout, $\dotsc$).
For our \texttt{MatMul} example, the list of unsupported variants includes: \texttt{CUBLAS\_GEMM\_ALGO1}, \texttt{CUBLAS\_GEMM\_ALGO12}, etc.

\newpage

\begin{lstlisting}[
    language=C,
    mathescape=true,
    xleftmargin=1.5em,
    numbers=none,
    frame=lines,
    label={lst_runtime_cublasex_v100_matmul_1024},
    caption={Runtime of \texttt{cuBLASEx} for its different \emph{algorithm} variants on \texttt{NVIDIA\:Volta\:GPU} when computing \texttt{MatMul} on square $1024\times1024$ input matrices}
    ]
CUBLAS_GEMM_DEFAULT: 188416ns (delta_max: 3.14713e-05)
CUBLAS_GEMM_ALGO2: 190464ns (delta_max: 6.86646e-05)
CUBLAS_GEMM_ALGO3: 186368ns (delta_max: 6.86646e-05)
CUBLAS_GEMM_ALGO4: 185344ns (delta_max: 6.86646e-05)
CUBLAS_GEMM_ALGO5: 181248ns (delta_max: 6.86646e-05)
CUBLAS_GEMM_ALGO6: 181248ns (delta_max: 6.86646e-05)
CUBLAS_GEMM_ALGO7: 178176ns (delta_max: 4.1008e-05)
CUBLAS_GEMM_ALGO8: 189440ns (delta_max: 4.1008e-05)
CUBLAS_GEMM_ALGO9: 171008ns (delta_max: 4.1008e-05)
CUBLAS_GEMM_ALGO10: 188416ns (delta_max: 4.1008e-05)
CUBLAS_GEMM_ALGO11: 191488ns (delta_max: 4.1008e-05)
CUBLAS_GEMM_ALGO18: 185344ns (delta_max: 2.67029e-05)
CUBLAS_GEMM_ALGO19: 172032ns (delta_max: 2.67029e-05)
CUBLAS_GEMM_ALGO20: 192512ns (delta_max: 2.67029e-05)
CUBLAS_GEMM_ALGO21: 201728ns (delta_max: 1.90735e-05)
CUBLAS_GEMM_ALGO22: 177152ns (delta_max: 1.90735e-05)
CUBLAS_GEMM_ALGO23: 194560ns (delta_max: 1.90735e-05)
CUBLAS_GEMM_DEFAULT_TENSOR_OP: 184320ns (delta_max: 3.14713e-05)
CUBLAS_GEMM_ALGO0_TENSOR_OP:  62464ns (delta_max: 0.0131454)
CUBLAS_GEMM_ALGO1_TENSOR_OP:  52224ns (delta_max: 0.0131454)
CUBLAS_GEMM_ALGO2_TENSOR_OP:  190464ns (delta_max: 3.14713e-05)
CUBLAS_GEMM_ALGO3_TENSOR_OP:  189440ns (delta_max: 3.14713e-05)
CUBLAS_GEMM_ALGO4_TENSOR_OP:  183296ns (delta_max: 3.14713e-05)
CUBLAS_GEMM_ALGO5_TENSOR_OP:  183296ns (delta_max: 3.14713e-05)
CUBLAS_GEMM_ALGO6_TENSOR_OP:  183296ns (delta_max: 3.14713e-05)
CUBLAS_GEMM_ALGO7_TENSOR_OP:  189440ns (delta_max: 3.14713e-05)
CUBLAS_GEMM_ALGO8_TENSOR_OP:  183296ns (delta_max: 3.14713e-05)
CUBLAS_GEMM_ALGO9_TENSOR_OP:  189440ns (delta_max: 3.14713e-05)
CUBLAS_GEMM_ALGO10_TENSOR_OP: 188416ns (delta_max: 3.14713e-05)
CUBLAS_GEMM_ALGO11_TENSOR_OP: 183296ns (delta_max: 3.14713e-05)
CUBLAS_GEMM_ALGO12_TENSOR_OP: 183296ns (delta_max: 3.14713e-05)
CUBLAS_GEMM_ALGO13_TENSOR_OP: 188416ns (delta_max: 3.14713e-05)
CUBLAS_GEMM_ALGO14_TENSOR_OP: 183296ns (delta_max: 3.14713e-05)
CUBLAS_GEMM_ALGO15_TENSOR_OP: 189440ns (delta_max: 3.14713e-05)
\end{lstlisting}

\newpage

\section{Code Generation}
\label{app_sec_code_generation}

This section outlines how imperative-style pseudocode is generated from our low-level program representation in Section~\ref{ch:low_level}.
Optimizations that operate below the abstraction level of our low-level representation (e.g., loop unrolling) are beyond the scope of this section and outlined in Section~\ref{ch_ll_opts}.
We aim to discuss and illustrate our code generation approach in detail in future work.

In the following, we highlight tuning parameters gray in our pseudocode, which are substituted by concrete, optimized values in our executable program code.
Static parameters, such as scalar types and the number of input/output buffers, are denoted in math font and also substituted by concrete values in our executable code.
We list meta-parameters in angle brackets \texttt{<...>}, and other static function annotations in double angle brackets \texttt{<<...>>}, e.g.,
\texttt{idx<<OUT>><<1,1>>} for denoting index function $\idx^\texttt{OUT}_{1,1}$ (used in Figure~\ref{fig_generic_hl}) in our pseudocode.

\subsection*{Overall Structure}

Listing~\ref{lst_app_code_gen_overview} shows the overall structure of our generated code.
We implement a particular expression in our low-level representation (Figure~\ref{fig_gen_ll}) as a compute kernel that is structured in the following phases:~%
0)~preparation (Section~\ref{app_code_gen_prep}),~%
1)~de-composition phase (Section~\ref{app_code_gen_decomp}),~%
2)~scalar phase (Section~\ref{app_code_gen_scalar}),~%
3)~re-composition phase (Section~\ref{app_code_gen_comp}).

\vspace*{10px}

\begin{lstlisting}[
  language=C,
  morekeywords={kernel},
  mathescape=true,
  numbers=left,
  xleftmargin=1.5em,
  frame=lines,
  caption={Overall structure of our generated code
  },
  label={lst_app_code_gen_overview}
  ]
kernel mdh(
  $T^\texttt{IB}_1$ trans_ll_IB<<$\bot$>><<$1$>><*,...,*>,...,$T^\texttt{IB}_{B^\texttt{IB}}$ trans_ll_IB<<$\bot$>><<$B^\texttt{IB}$>><*,...,*> ,
  $T^\texttt{OB}_1$ trans_ll_OB<<$\bot$>><<$1$>><*,...,*>,...,$T^\texttt{OB}_{B^\texttt{OB}}$ trans_ll_OB<<$\bot$>><<$B^\texttt{OB}$>><*,...,*> )
{
  // 0. preparation
  ...
  // 2. de-composition phase
  ...
  // 3. scalar phase
  ...
  // 4. re-composition phase
  ...
}
\end{lstlisting}

\setcounter{subsection}{-1}
\subsection{Preparation}
\label{app_code_gen_prep}

Listing~\ref{lst_app_code_gen_0} shows the preparation phase.
It prepares in five sub-phases the basic building blocks used in our low-level representation:~%
1)~$\mdh$ (Section~\ref{sec_app_code_gen_0_1}),~%
2)~$\iv$ (Section~\ref{sec_app_code_gen_0_2}),~%
3)~$\ov$ (Section~\ref{sec_app_code_gen_0_3}),~%
4)~BUFs (Section~\ref{sec_app_code_gen_0_4}),~%
5)~MDAs (Section~\ref{sec_app_code_gen_0_5}).

\begin{lstlisting}[
  caption={Preparation Phase},
  label={lst_app_code_gen_0}
  ]
// 0. preparation
  // 0.1. md_hom
  ...
  // 0.2. inp_view
  ...
  // 0.3. out_view
  ...
  // 0.4. BUFs
  ...
  // 0.5. MDAs
  ...
\end{lstlisting}

\subsubsection{\mdh}
\label{sec_app_code_gen_0_1}
\mbox{}

Listing~\ref{lst_app_code_gen_0_1} shows the user-defined scalar function and low-level combine operators (Definition~\ref{def_ll_comb}) which are both provided by the user via higher-order function $\mdh$ (Definition~\ref{def_md_hom}).

Listing~\ref{lst_app_code_gen_0_1_2_a} shows how we pre-implement for the user the two combine operators \emph{concatenation}~(Example~\ref{def:mda_concat}) and \emph{point-wise combination} (Example~\ref{def:mda_pw}).

Listing~\ref{lst_app_code_gen_0_1_2_b} shows how we pre-implement the \emph{inverse of concatenation} (Definition~\ref{app_inverse_concat}), which we use in the de-composition phase (via Definition~\ref{app_mda_part}).

\vspace*{15px}

\begin{lstlisting}[
  caption={Scalar Function \& Combine Operators},
  label={lst_app_code_gen_0_1}
  ]
// 0.1. md_hom

  // 0.1.1. scalar function
  f( $T^\texttt{INP}$ inp ) -> $T^\texttt{OUT}$ out
  {
    // ... (user defined)
  }

  // 0.1.2. combine operators
  $\forall d\in[1,D]_\IN$:
  co<<$d$>><$I_1,...,I_{{d}-1},I_{{d}+1},...,I_D\in\IDXs \, , \, (P,Q)\in\IDXsxIDXs$>(
    $ T^\texttt{OUT}[I_1,...,I_{{d}-1} , \hspace*{2px}\size{{d}}{MDA}{MDA}(P)\hspace*{3px} , I_{{d}+1},...,I_D]$ lhs ,
    $ T^\texttt{OUT}[I_1,...,I_{{d}-1} , \hspace*{2px}\size{{d}}{MDA}{MDA}(Q)\hspace*{2px} , I_{{d}+1},...,I_D]$ rhs ) -> $T^\texttt{OUT}[I_1,...,I_{{d}-1} , \size{{d}}{MDA}{MDA}(P\cupdot Q) , I_{{d}+1},...,I_D]$ res
  {
    // ... (user defined)
  }
\end{lstlisting}

\newpage

\begin{lstlisting}[
  caption={Pre-Implemented Combine Operators},
  label={lst_app_code_gen_0_1_2_a}
  ]
// 0.1.2. combine operators

  // pre-implemented combine operators

  // concatenation
  $\forall d\in\IN$:
  cc<<$d$>><$I_1,...,I_{{d}-1},I_{{d}+1},...,I_D\in\IDXs \, , \, (P,Q)\in\IDXsxIDXs$>(
    $ T^\texttt{OUT}[I_1,...,I_{{d}-1} , \hspace*{2px}id(P)\hspace*{3px} , I_{{d}+1},...,I_D]$ lhs ,
    $ T^\texttt{OUT}[I_1,...,I_{{d}-1} , \hspace*{2px}id(Q)\hspace*{2px} , I_{{d}+1},...,I_D]$ rhs ) -> $T^\texttt{OUT}[I_1,...,I_{{d}-1} , id(P\cupdot Q) , I_{{d}+1},...,I_D]$ res
  {
    int i_1 $\in \ I_1$
      $\ddots$
        int i_$\{d-1\}$ $\in$ I$_{d-1}$
          int i_$\{d+1\}$ $\in$ I$_{d+1}$
            $\ddots$
              int i_$D$ $\in \ I_D$
              {
                int i_$d$ $\in \ P$
                  res[ i_1,...,i_$d$,...,i_$D$] := lhs[ i_1,...,i_$d$,...,i_$D$];
                int i_$d$ $\in \ Q$
                  res[ i_1,...,i_$d$,...,i_$D$] := rhs[ i_1,...,i_$d$,...,i_$D$];
              }
  }

  // point-wise combination
  $\forall d\in\IN$:
  pw<<$d$>><$I_1,...,I_{{d}-1},I_{{d}+1},...,I_D\in\IDXs \, , \, (P,Q)\in\IDXsxIDXs$>(
    $\oplus:T^\texttt{OUT}\times T^\texttt{OUT}\to T^\texttt{OUT}$)( $\hspace*{1px}T^\texttt{OUT}[I_1,...,I_{{d}-1} , \hspace*{2px}0_f(P)\hspace*{3px} , I_{{d}+1},...,I_D]$ lhs ,
                  $\hspace*{2px}T^\texttt{OUT}[I_1,...,I_{{d}-1} , \hspace*{2px}0_f(Q)\hspace*{2px} , I_{{d}+1},...,I_D]$ rhs )
  -> $T^\texttt{OUT}[I_1,...,I_{{d}-1} , 0_f(P\cupdot Q) , I_{{d}+1},...,I_D]$ res
  {
    int i_1 $\in \ I_1$
      $\ddots$
        int i_$\{d-1\}$ $\in$ I$_{d-1}$
          int i_$\{d+1\}$ $\in$ I$_{d+1}$
            $\ddots$
              int i_$D$ $\in \ I_D$
              {
                res[ i_1,...,i_$\{d-1\}$ , 0 , i_$\{d+1\}$,...,i_$D$]
                  :=          $\hspace{-7px}$lhs[ i_1,...,i_$\{d-1\}$ , 0 , i_$\{d+1\}$,...,i_$D$]
                     $\texttt{atomic}(\oplus)\hspace*{1px}$ rhs[ i_1,...,i_$\{d-1\}$ , 0 , i_$\{d+1\}$,...,i_$D$];
              }
  }
\end{lstlisting}

\newpage

\begin{lstlisting}[
  caption={Pre-Implemented Combine Operators},
  label={lst_app_code_gen_0_1_2_b}
  ]
// 0.1.2. combine operators

  // pre-implemented combine operators

  // inverse concatenation
  $\forall d\in\IN$:
  cc_inv<<$d$>><$I_1,...,I_{{d}-1},I_{{d}+1},...,I_D\in\IDXs \, , \, (P,Q)\in\IDXsxIDXs$>(
    $T^\texttt{INP}[I_1,...,I_{{d}-1} , id(P\cupdot Q) , I_{{d}+1},...,I_D]$ res ) -> ( $T^\texttt{INP}[I_1,...,I_{{d}-1} , \hspace*{2px}id(P)\hspace*{3px} , I_{{d}+1},...,I_D]$ lhs ,
                                            $\hspace{10px}$$T^\texttt{INP}[I_1,...,I_{{d}-1} , \hspace*{2px}id(Q)\hspace*{2px} , I_{{d}+1},...,I_D]$ rhs )
  {
    int i_1 $\in \ I_1$
      $\ddots$
        int i_$\{d-1\}$ $\in$ I$_{d-1}$
          int i_$\{d+1\}$ $\in$ I$_{d+1}$
            $\ddots$
              int i_$D$ $\in \ I_D$
              {
                int i_$d$ $\in \ P$
                  res[ i_1,...,i_$d$,...,i_$D$] =: lhs[ i_1,...,i_$d$,...,i_$D$];
                int i_$d$ $\in \ Q$
                  res[ i_1,...,i_$d$,...,i_$D$] =: rhs[ i_1,...,i_$d$,...,i_$D$];
              }
  }
\end{lstlisting}

\vspace*{20px}

\subsubsection{\iv}
\label{sec_app_code_gen_0_2}
\mbox{}

Listing~\ref{lst_app_code_gen_0_2} shows the user-defined index functions provided by the user via higher-order function $\iv$ (Definition~\ref{def_iv}).
\begin{lstlisting}[
  caption={Index Functions (input)},
  label={lst_app_code_gen_0_2}
  ]
// 0.2. inp_view

  // index functions
  $\forall b\in[1,B^\texttt{IB}]_\IN$,$a\in[1,A^\texttt{IB}_b]_\IN$: $\forall d\in[1,D^\texttt{IB}_b]_\IN$:
  static
  idx<<INP>><<$b$,$a$>><<$d$>>( int i_MDA_1 , ... , i_MDA_$D$ ) -> int i_BUF_$d$
  {
      // ... (user defined)
  }
\end{lstlisting}

\newpage

\subsubsection{\ov}
\label{sec_app_code_gen_0_3}
\mbox{}

Listing~\ref{lst_app_code_gen_0_3} shows the user-defined index functions provided by the user via higher-order function $\ov$ (Definition~\ref{def_ov}).
\begin{lstlisting}[
  caption={Index Functions (output)},
  label={lst_app_code_gen_0_3}
  ]
// 0.3. out_view

  // index functions
  $\forall b\in[1,B^\texttt{OB}]_\IN$,$a\in[1,A^\texttt{OB}_b]_\IN$: $\forall d\in[1,D^\texttt{OB}_b]_\IN$:
  static
  idx<<OUT>><<$b$,$a$>><<$d$>>( int i_MDA_1 , ... , i_MDA_$D$ ) -> int i_BUF_$d$
  {
      // ... (user defined)
  }
\end{lstlisting}

%\vspace*{20px}

\subsubsection{BUFs}
\label{sec_app_code_gen_0_4}
\mbox{}

Listing~\ref{lst_app_code_gen_0_4} shows our implementation of low-level BUFs (Definition~\ref{ll_buffers}).
We compute BUFs' sizes using the ranges of their index functions (Definitions~\ref{def_iv} and~\ref{def_ov}).
Moreover, we partially evaluate BUFs' meta-parameters \texttt{MEM} (memory region) and $\sigma$ (memory layout) immediately, as the same values are re-used for them during program runtime.

The BUFs in lines~\ref{app_lst_ll_buffers_alloc_LVL_INP}~and ~\ref{app_lst_ll_buffers_alloc_f_INP} as well as in lines~\ref{app_lst_ll_buffers_alloc_LVL_OUT} and~\ref{app_lst_ll_buffers_alloc_f_OUT} represent the
BUFs' transposed function representation (Definition~\ref{ll_buffers}),
and the BUFs in lines~\ref{app_lst_ll_sig_1_INP},~\ref{app_lst_ll_sig_2_INP},~and~\ref{app_lst_ll_sig_3_INP} as well as in lines~\ref{app_lst_ll_sig_1_OUT},~\ref{app_lst_ll_sig_2_OUT},~and~\ref{app_lst_ll_sig_3_OUT} are
the transposed BUFs' ordinary low-level BUF representation.

\begin{lstlisting}[
  caption={Low-Level BUFs},
  label={lst_app_code_gen_0_4}
  ]
// 0.4. BUFs

  // 0.4.1. compute BUF sizes
  $\forall \texttt{IO}\in\{\texttt{IB},\texttt{OB}\}$: $\forall {b}\in[1,B^\texttt{IO}]_\IN$ $\forall {d}\in[1,D^\texttt{IO}_b]_\IN$:
  static N<<IO>><<$b$>><<$d$>>( mda_idx_set I_1 , ... , I_$D$ ) -> int N_$b$_$d$
  {
    N_$b$_$d$ := 0;

    i_1 $\in$ I_1
      $\ddots$
        i_$D$ $\in$ I_$D$
        {
          $\forall{a}\in[1,A^\texttt{IB}_b]_\IN$:
          N_$b$_$d$ :=$_\texttt{max}$ 1 + idx<<IO>><<$b$,$a$>><<$d$>>( i_1,...,i_$D$ );
        }
  }

  // 0.4.2. input BUFs

  // initial BUFs
  $\forall {b}\in[1,B^\texttt{IB}]_\IN$:
  static ll_IB<<$\bot$>><<$b$>><${\overset{(\bot)}{\blacktriangledown}}{^1_1}\in\tp{\texttt{\#PRT}(1,1)}$,...,${\overset{(\bot)}{\blacktriangledown}}{^L_D}\in\tp{\texttt{\#PRT}(L,D)}$>( int i_1,...,
    int i_$D^\texttt{IB}_b$ )$\label{app_lst_ll_sig_1_INP}$ -> $T^\texttt{IB}_b$ a
  {
    a := trans_ll_IB<<$\bot$>><<$b$>><${\overset{(\bot)}{\blacktriangledown}}{^1_1},\dotsc,{\overset{(\bot)}{\blacktriangledown}}{^L_D}$>[ i_1 , ... , i_$D^\texttt{IB}_b$ ];
  }

  // de-composition BUFs
  $\forall (l,d)\in\MDHLVL$: $\forall {b}\in[1,B^\texttt{IB}]_\IN$:
  auto trans_ll_IB<<$l$,$d$>><<$b$>><${\overset{(l,d)}{\blacktriangledown}}{^1_1}\in\tp{\texttt{\#PRT}(1,1)}$,...,${\overset{(l,d)}{\blacktriangledown}}{^L_D}\in\tp{\texttt{\#PRT}(L,D)}$>$\label{app_lst_ll_buffers_alloc_LVL_INP}$
    := $\tp{\downarrow\texttt{-mem}^\texttt{<b>}(l,d)}$ $T^\texttt{IB}_b$[ N<<IB>><<$b$>><<$\tp{\sigma_{\downarrow\texttt{-mem}}^\texttt{<b>}({l,d})(1)}\hspace*{9px}$>>(  $(\:\fsize{d}{\dplus}{MDA}{MDA}(N_d)\:)_{d\in[1,D]_\IN}$  ) ,$\label{app_lst_ll_buf_size_1}$
                          $\vdots$
                        $\hspace*{2px}$N<<IB>><<$b$>><<$\tp{\sigma_{\downarrow\texttt{-mem}}^\texttt{<b>}({l,d})(D^\texttt{IB}_{b})}$>>(   $(\:\fsize{d}{\dplus}{MDA}{MDA}(N_d)\:)_{d\in[1,D]_\IN}$   ) ];$\label{app_lst_ll_buf_size_2}$

  $\forall (l,d)\in\MDHLVL$: $\forall {b}\in[1,B^\texttt{IB}]_\IN$:
  static ll_IB<<$l$,$d$>><<$b$>><${\overset{(l,d)}{\blacktriangledown}}{^1_1}\in\tp{\texttt{\#PRT}(1,1)}$,...,${\overset{(l,d)}{\blacktriangledown}}{^L_D}\in\tp{\texttt{\#PRT}(L,D)}$>( int i_1,...,
    int i_$D^\texttt{IB}_b$ ) -> $T^\texttt{IB}_b$ a$\label{app_lst_ll_sig_2_INP}$
  {
    a := trans_ll_IB<<$l$,$d$>><<b><${\overset{(l,d)}{\blacktriangledown}}{^1_1},\dotsc,{\overset{(l,d)}{\blacktriangledown}}{^L_D}$>[ i_$\tp{\sigma_{\downarrow\texttt{-mem}}^\texttt{<b>}(l,d)(1)}$ , ... ,
      i_$\tp{\sigma_{\downarrow\texttt{-mem}}^\texttt{<b>}(l,d)(D^\texttt{IB}_{b})}$ ];
  }

  // scalar BUFs
  $\forall {b}\in[1,B^\texttt{IB}]_\IN$:
  auto trans_ll_IB<f>><<$b$>><${\overset{(f)}{\blacktriangledown}}{^1_1}\in\tp{\texttt{\#PRT}(1,1)}$,...,${\overset{(f)}{\blacktriangledown}}{^L_D}\in\tp{\texttt{\#PRT}(L,D)}$>$\label{app_lst_ll_buffers_alloc_f_INP}$
    := $\tp{\texttt{f}^\downarrow\texttt{-mem}^\texttt{<b>}}$ $T^\texttt{IB}_b$[ N<<IB>><<$b$>><<$\tp{\sigma_{\texttt{f}^\downarrow\texttt{-mem}}^\texttt{<b>}(1)}\hspace*{8px}$>>(  $(\:\fsize{d}{\dplus}{MDA}{MDA}(N_d)\:)_{d\in[1,D]_\IN}$  ) ,$\label{app_lst_ll_buf_size_3}$
                      $\vdots$
                    $\hspace*{3px}$N<<IB>><<$b$>><<$\tp{\sigma_{\texttt{f}^\downarrow\texttt{-mem}}^\texttt{<b>}(D^\texttt{IB}_{b})}$>>(  $(\:\fsize{d}{\dplus}{MDA}{MDA}(N_d)\:)_{d\in[1,D]_\IN}$  ) ];$\label{app_lst_ll_buf_size_4}$

  $\forall {b}\in[1,B^\texttt{IB}]_\IN$:
  static ll_IB<<f>><<$b$>><${\overset{(f)}{\blacktriangledown}}{^1_1}\in\tp{\texttt{\#PRT}(1,1)}$,...,${\overset{(f)}{\blacktriangledown}}{^L_D}\in\tp{\texttt{\#PRT}(L,D)}$>( int i_1,...,
    int i_$D^\texttt{IB}_b$ )  -> $T^\texttt{IB}_b$ a$\label{app_lst_ll_sig_3_INP}$
  {
    a := trans_ll_IB<<f>><<b><${\overset{(\texttt{f})}{\blacktriangledown}}{^1_1},\dotsc,{\overset{(\texttt{f})}{\blacktriangledown}}{^L_D}$>[ i_$\tp{\sigma_{f^\downarrow\texttt{-mem}}^\texttt{<b>}(1)}$ , ... , i_$\tp{\sigma_{f^\downarrow\texttt{-mem}}^\texttt{<b>}(D^\texttt{IB}_{b})}$ ];
  }

  // 0.4.3. output BUFs

  // initial BUFs
  $\forall {b}\in[1,B^\texttt{OB}]_\IN$:
  static ll_OB<<$\bot$>><<$b$>><${\overset{(\bot)}{\blacktriangle}}{^1_1}\in\tp{\texttt{\#PRT}(1,1)}$,...,${\overset{(\bot)}{\blacktriangle}}{^L_D}\in\tp{\texttt{\#PRT}(L,D)}$>( int i_1,...,
    int i_$D^\texttt{OB}_b$ )$\label{app_lst_ll_sig_1_OUT}$ -> $T^\texttt{OB}_b$ a
  {
    a := trans_ll_OB<<$\bot$>><<$b$>><${\overset{(\bot)}{\blacktriangle}}{^1_1},\dotsc,{\overset{(\bot)}{\blacktriangle}}{^L_D}$>[ i_1 , ... , i_$D^\texttt{OB}_b$ ];$\label{app_lst_ll_buf_size_0init_1}$
  }

  // re-composition BUFs
  $\forall (l,d)\in\MDHLVL$: $\forall {b}\in[1,B^\texttt{OB}]_\IN$:
  auto trans_ll_OB<<$l$,$d$>><<$b$>><${\overset{(l,d)}{\blacktriangle}}{^1_1}\in\tp{\texttt{\#PRT}(1,1)}$,...,${\overset{(l,d)}{\blacktriangle}}{^L_D}\in\tp{\texttt{\#PRT}(L,D)}$>$\label{app_lst_ll_buffers_alloc_LVL_OUT}$
    := $\tp{\uparrow\texttt{-mem}^\texttt{<b>}(l,d)}$ $T^\texttt{OB}_b$[ N<<OB>><<$b$>><<$\tp{\sigma_{\uparrow\texttt{-mem}}^\texttt{<b>}({l,d})(1)}\hspace*{9px}$>>(  $(\:\fsize{d}{\co{}}{MDA}{MDA}(N_d)\:)_{d\in[1,D]_\IN}$  ) ,$\label{app_lst_ll_buf_size_5}$
                          $\vdots$
                        $\hspace*{2px}$N<<OB>><<$b$>><<$\tp{\sigma_{\uparrow\texttt{-mem}}^\texttt{<b>}({l,d})(D^\texttt{OB}_{b})}$>>(  $(\:\fsize{d}{\co{}}{MDA}{MDA}(N_d)\:)_{d\in[1,D]_\IN}$  ) ];$\label{app_lst_ll_buf_size_6}$

  $\forall (l,d)\in\MDHLVL$: $\forall {b}\in[1,B^\texttt{OB}]_\IN$:
  static ll_OB<<$l$,$d$>><<$b$>><${\overset{(l,d)}{\blacktriangle}}{^1_1}\in\tp{\texttt{\#PRT}(1,1)}$,...,${\overset{(l,d)}{\blacktriangle}}{^L_D}\in\tp{\texttt{\#PRT}(L,D)}$>( int i_1,...,
    int i_$D^\texttt{OB}_b$ ) -> $T^\texttt{OB}_b$ a$\label{app_lst_ll_sig_2_OUT}$
  {
    a := trans_ll_OB<<$l$,$d$>><<b><${\overset{(l,d)}{\blacktriangle}}{^1_1},\dotsc,{\overset{(l,d)}{\blacktriangle}}{^L_D}$>[ i_$\tp{\sigma_{\uparrow\texttt{-mem}}^\texttt{<b>}(l,d)(1)}$ , ... ,
      i_$\tp{\sigma_{\uparrow\texttt{-mem}}^\texttt{<b>}(l,d)(D^\texttt{OB}_{b})}$ ];$\label{app_lst_ll_buf_size_0init_2}$
  }

  // scalar BUFs
  $\forall {b}\in[1,B^\texttt{OB}]_\IN$:
  auto trans_ll_OB<f>><<$b$>><${\overset{(f)}{\blacktriangle}}{^1_1}\in\tp{\texttt{\#PRT}(1,1)}$,...,${\overset{(f)}{\blacktriangle}}{^L_D}\in\tp{\texttt{\#PRT}(L,D)}$>$\label{app_lst_ll_buffers_alloc_f_OUT}$
    := $\tp{\texttt{f}^\uparrow\texttt{-mem}^\texttt{<b>}}$ $T^\texttt{OB}_b$[ N<<OB>><<$b$>><<$\tp{\sigma_{\texttt{f}^\uparrow\texttt{-mem}}^\texttt{<b>}(1)}\hspace*{8px}$>>(  $(\:\fsize{d}{\co{}}{MDA}{MDA}(N_d)\:)_{d\in[1,D]_\IN}$  ) ,$\label{app_lst_ll_buf_size_7}$
                      $\vdots$
                    $\hspace*{3px}$N<<OB>><<$b$>><<$\tp{\sigma_{\texttt{f}^\uparrow\texttt{-mem}}^\texttt{<b>}(D^\texttt{OB}_{b})}$>>(  $(\:\fsize{d}{\co{}}{MDA}{MDA}(N_d)\:)_{d\in[1,D]_\IN}$  ) ];$\label{app_lst_ll_buf_size_8}$

  $\forall {b}\in[1,B^\texttt{OB}]_\IN$:
  static ll_OB<<f>><<$b$>><${\overset{(f)}{\blacktriangle}}{^1_1}\in\tp{\texttt{\#PRT}(1,1)}$,...,${\overset{(f)}{\blacktriangle}}{^L_D}\in\tp{\texttt{\#PRT}(L,D)}$>( int i_1,...,
    int i_$D^\texttt{OB}_b$ )  -> $T^\texttt{OB}_b$ a$\label{app_lst_ll_sig_3_OUT}$
  {
    a := trans_ll_OB<<f>><<b><${\overset{(\texttt{f})}{\blacktriangle}}{^1_1},\dotsc,{\overset{(\texttt{f})}{\blacktriangle}}{^L_D}$>[ i_$\tp{\sigma_{f^\uparrow\texttt{-mem}}^\texttt{<b>}(1)}$ , ... , i_$\tp{\sigma_{f^\uparrow\texttt{-mem}}^\texttt{<b>}(D^\texttt{OB}_{b})}$ ];$\label{app_lst_ll_buf_size_0init_3}$
  }
\end{lstlisting}

\noindent
where ${\overset{(\bullet)}{\blacktriangledown}}{^\mathfrak{l}_\mathfrak{d}}$ and ${\overset{(\bullet)}{\blacktriangle}}{^\mathfrak{l}_\mathfrak{d}}$, for
$\bullet\in\{\bot\}\,\cup\,\MDHLVL\,\cup\,\{\texttt{f}\}$, are textually replaced by:
\begin{align*}
{\overset{(\bullet)}{\blacktriangledown}}{^\mathfrak{l}_\mathfrak{d}} \ = \
\begin{cases}
\ p^\mathfrak{l}_\mathfrak{d} &: \ \tp{\sigma_{\downarrow\texttt{-ord}}}(\mathfrak{l},\mathfrak{d}) \ < \ {\bullet} \\[10pt]
\ * &: \ \text{else}
\end{cases}
\end{align*}
\begin{align*}
{\overset{(\bullet)}{\blacktriangle}}{^\mathfrak{l}_\mathfrak{d}} \ = \
\begin{cases}
\ p^\mathfrak{l}_\mathfrak{d} &: \ \tp{\sigma_{\uparrow\texttt{-ord}}}(\mathfrak{l},\mathfrak{d}) \ < \ {\bullet} \\[10pt]
\ * &: \ \text{else}
\end{cases}
\end{align*}
(symbol $*$ is taken from Definition~\ref{app_gen_meta_func})
where $<$ is defined according to the lexicographical order on $\MDHLVL=[1,L]_\IN\times[1,D]_\IN$,
and:
\[
\forall (l,d)\in\MDHLVL: \ \bot \ :< \ (l,d) \ :< \ f
\]
Functions
\[
  \fsize{1}{\dplus}{MDA}{MDA} \ , \ \dotsc \ , \ \fsize{D}{\dplus}{MDA}{MDA}
\]
are the index set functions $id$ of combine operator concatenation $\dplus$ (Example~\ref{def:mda_concat}), and functions
\[
  \fsize{1}{\co{}}{MDA}{MDA} \ , \ \dotsc \ , \ \fsize{D}{\co{}}{MDA}{MDA}
\]
are the index set functions of combine operators $\co{1},\dotsc,\co{D}$.

Note that we use generous BUFs sizes (lines~\ref{app_lst_ll_buf_size_1}-\ref{app_lst_ll_buf_size_2},~\ref{app_lst_ll_buf_size_3}-\ref{app_lst_ll_buf_size_4},~\ref{app_lst_ll_buf_size_5}-\ref{app_lst_ll_buf_size_6},~\ref{app_lst_ll_buf_size_7}-\ref{app_lst_ll_buf_size_8}), as imperative-style programming models usually struggle with non-contiguous index ranges.
We discuss optimizations targeting BUF sizes in Section~\ref{ch_ll_opts}.

Note further that we do not need to initialize output buffers with neutral elements of combine operators in lines~\ref{app_lst_ll_buf_size_0init_1}, \ref{app_lst_ll_buf_size_0init_2}, and~\ref{app_lst_ll_buf_size_0init_3} of Listing~\ref{lst_app_code_gen_0_4}, as the buffers are initialized implicitly in the re-composition phase (Section~\ref{app_code_gen_comp}).

\subsubsection{MDAs}
\label{sec_app_code_gen_0_5}
\mbox{}

Listing~\ref{lst_app_code_gen_0_5} shows our implementation of low-level MDAs (Definitions~\ref{def_ll_mda} and~\ref{app_mda_part}).

Note that for a particular choice of meta-parameters, low-level BUFs (Definition~\ref{ll_buffers}) are ordinary BUFs (Definition~\ref{def_buffer}), as required by the types of functions $\iv$ and $\ov$ (Definitions~\ref{def_iv} and~\ref{def_ov}).

\begin{lstlisting}[
  caption={Low-Level MDAs},
  label={lst_app_code_gen_0_5}
  ]
// 0.5. MDAs

  // 0.5.1. partitioned index sets$\label{app_lst_mda_part_start}$
  $\forall d\in[1,D]_\IN:$
  static I<<$d$>><$p^1_d\in\tp{\texttt{\#PRT}(1,d)}$,...,$p^L_d\in\tp{\texttt{\#PRT}(L,d)}$>( int j' ) -> int i_j$\label{app_lst_mda_part}$
  {
    i_j := ( $p^1_d$ * ( $N_d$ / ($\tp{\texttt{\#PRT}(1,d)}$)                           $\hspace*{5px}$) +
              $\vdots$
             $p^L_d$ * ( $N_d$ / ($\tp{\texttt{\#PRT}(1,d)}$*...*$\tp{\texttt{\#PRT}(L,d)}$) ) + j' );
  }$\label{app_lst_mda_part_end}$

  // 0.5.2. input MDAs
  $\forall \bullet\in\{\bot\}\:\cup\:\MDHLVL\:\cup\:\{\texttt{f}\}$:
  static ll_inp_mda<<$\bullet$>><${\overset{(\bullet)}{\blacktriangledown}}{^1_1}\in\tp{\texttt{\#PRT}(1,1)}$,...,${\overset{(\bullet)}{\blacktriangledown}}{^L_D}\in\tp{\texttt{\#PRT}(L,D)}$>( int i_1,...,
    int i_$D$ ) -> $T^\texttt{IB}_b$ a
  {
    $\forall {b}\in[1,B^\texttt{IB}]_\IN$,${a}\in[1,A^\texttt{IB}_{b}]_\IN$:
    a := ll_IB<<$\bullet$>><<$b$>><${\overset{(\bullet)}{\blacktriangledown}}{^1_1}$,...,${\overset{(\bullet)}{\blacktriangledown}}{^L_D}$>( idx<<INP>><<$b$,$a$>><<$\ $1$\ $>>( i_1,...,i_$D$ ) ,$\label{app_lst_mda_acc_1}$
                                  $\hspace*{25px}$   $\vdots$
                                  $\hspace*{15px}$ $\hspace*{3px}$idx<<INP>><<$b$,$a$>><<$D^\texttt{IB}_{b}$>>( i_1,...,i_$D$ )   );$\label{app_lst_mda_acc_2}$
  }

  // 0.5.3. output MDAs
  $\forall \bullet\in\{\bot\}\:\cup\:\MDHLVL\:\cup\:\{\texttt{f}\}$:
  static ll_out_mda<<$\bullet$>><${\overset{(\bullet)}{\blacktriangle}}{^1_1}\in\tp{\texttt{\#PRT}(1,1)}$,...,${\overset{(\bullet)}{\blacktriangle}}{^L_D}\in\tp{\texttt{\#PRT}(L,D)}$>( int i_1,...,
    int i_$D$ ) -> $T^\texttt{OB}_b$ a
  {
    $\forall {b}\in[1,B^\texttt{OB}]_\IN$,${a}\in[1,A^\texttt{OB}_{b}]_\IN$:
    a := ll_OB<<$\bullet$>><<$b$>><${\overset{(\bullet)}{\blacktriangle}}{^1_1}$,...,${\overset{(\bullet)}{\blacktriangle}}{^L_D}$>( idx<<OUT>><<$b$,$a$>><<$\ $1$\ $>>( i_1,...,i_$D$ ) ,$\label{app_lst_mda_acc_3}$
                                  $\hspace*{25px}$   $\vdots$
                                  $\hspace*{15px}$ $\hspace*{3px}$idx<<OUT>><<$b$,$a$>><<$D^\texttt{OB}_{b}$>>( i_1,...,i_$D$ )   );$\label{app_lst_mda_acc_4}$
  }
\end{lstlisting}

For computing the partitioned index sets (lines~\ref{app_lst_mda_part_start}-\ref{app_lst_mda_part_end}), we exploit the following proposition.

\begin{proposition}
Let $\MDA\in T[N_1,\dotsc,N_D]$ be an arbitrary MDA that operates on contiguous index sets $[1,N_d)_\IN$, $d\in[1,D]_\IN$.
Let further be
\begin{align*}
&
\MDA^{<\,
(p^1_1,\dotsc,p^1_d)\in P^1_1\times\dotsc\times P^1_D
\ \ | \ \ \dotsc \ \ | \ \
(p^L_1,\dotsc,p^L_D)\in P^L_1\times\dotsc\times P^L_D
\,>}: \
I_1^{<p^1_1\,\dotsc\,p^L_1>}
\ \ \times \ \dotsc \  \times \ \
I_D^{<p^1_D\,\dotsc\,p^L_D>}
\ \to \ \
T
\end{align*}
an arbitrary $L$-layered, $D$-dimensional, $P$-partitioning of MDA $\MDA$.

It holds that $j$-th element within an MDA's part is accessed via index $j$:
\[
I_d^{<p^1_d\,\dotsc\,p^L_d>}
=
\{
\underbrace{\sum_{l\in[1,L]_\IN}p^l_d*\frac{N_d}{\prod_{l'\in[1,l]_\IN} P^{l'}_d}}_\text{OS} \ \ + \ \  0 \, , \ \
\underbrace{\sum_{l\in[1,L]_\IN}p^l_d*\frac{N_d}{\prod_{l'\in[1,l]_\IN} P^{l'}_d}}_\text{OS} \ \ + \ \  1 \, , \ \
\,\,\dotsc\
\}
\]
\end{proposition}
\begin{proof}
  Since MDA $\MDA$'s index sets are contiguous ranges of natural numbers, it holds the $i_j$~--~the index to access the $j$-th element within an MDA's part (Definition~\ref{app_mda_part})~--~is equal to $j$ itself.
\end{proof}

\vspace*{20px}

\subsection{De-Composition Phase}
\label{app_code_gen_decomp}

Listing~\ref{lst_app_code_gen_1} shows our implementation of the de-composition phase (Figure~\ref{fig_gen_ll}).

\begin{lstlisting}[
  caption={De-Composition Phase},
  label={lst_app_code_gen_1}
  ]
// 1. de-composition phase

  // 1.1. initialization
  ll_inp_mda<<$\bot$>> =: ll_inp_mda<<$\tp{\sigma_{\downarrow\texttt{-ord}}(1,1)}$>>

  // 1.2. main
  int p_$\tp{\sigma_{\downarrow\texttt{-ord}}(1,1)}$ $\in^{<\tp{\leftrightarrow_{\downarrow\texttt{-ass}}(1,1)}>}$ $\tp{\texttt{\#PRT}}$($\tp{\sigma_{\downarrow\texttt{-ord}}(1,1)}$)
  {
    ll_inp_mda<<$\tp{\sigma_{\downarrow\texttt{-ord}}(1,1)}$>> =:$_{\texttt{cc<}\tp{\scriptstyle\sigma_{\downarrow\texttt{-ord}}(1,1)}\texttt{>}}$ ll_inp_mda<<$\tp{\sigma_{\downarrow\texttt{-ord}}(1,2)}$>>;
    int p_$\tp{\sigma_{\downarrow\texttt{-ord}}(1,2)}$ $\in^{<\tp{\leftrightarrow_{\downarrow\texttt{-ass}}(1,2)}>}$ $\tp{\texttt{\#PRT}}$($\tp{\sigma_{\downarrow\texttt{-ord}}(1,2)}$)
    {
      ll_inp_mda<<$\tp{\sigma_{\downarrow\texttt{-ord}}(1,2)}$>> =:$_{\texttt{cc<}\tp{\scriptstyle\sigma_{\downarrow\texttt{-ord}}(1,2)}\texttt{>}}$ ll_inp_mda<<$\tp{\sigma_{\downarrow\texttt{-ord}}(1,3)}$>>;
        $\ddots$
          int p_$\tp{\sigma_{\downarrow\texttt{-ord}}(L,D)}$ $\in^{<\tp{\leftrightarrow_{\downarrow\texttt{-ass}}(L,D)}>}$ $\tp{\texttt{\#PRT}}$($\tp{\sigma_{\downarrow\texttt{-ord}}(L,D)}$)
          {
            ll_inp_mda<<$\tp{\sigma_{\downarrow\texttt{-ord}}(L,D)}$>> =:$_{\texttt{cc<}\tp{\scriptstyle\sigma_{\downarrow\texttt{-ord}}(L,D)}\texttt{>}}$ ll_inp_mda<<f>>;
          }
        $\udots$
    }
  }
\end{lstlisting}

\newpage

\noindent
where
\begin{lstlisting}[numbers=none,frame=none]
ll_inp_mda<<$l$,$d$>> =:$_{\texttt{cc<}l,d\texttt{>}}$ ll_inp_mda<<$l'$,$d'$>>
\end{lstlisting}
abbreviates
\begin{lstlisting}[numbers=none,frame=none]
ll_inp_mda<<$l'$,$d'$>><${\overset{(l',d')}{\blacktriangledown}}{^1_1}$ ,..., ${\overset{(l',d')}{\blacktriangledown}}{^L_D}$>, ll_inp_mda<<$l$,$d$>><${\overset{(l,d)}{\blacktriangledown}}{^1_1}$ ,..., ${\overset{(l,d)}{\blacktriangledown}}{^L_D}$>
:=  cc_inv<<d>><$\fsize{1}{\dplus}{MDA}{MDA}$( I<<$\hspace*{2px}1\hspace*{2px}$>><${\overset{(l,d)}{\blacksquare}}{^1_1}$,...,${\overset{(l,d)}{\blacksquare}}{^L_1}$>(0) ), // $I_1$
                 $\vdots$                                  // ...,$I_{d-1}$,$I_{d+1}$,...
                $\fsize{D}{\dplus}{MDA}{MDA}$( I<<$D$>><${\overset{(l,d)}{\blacksquare}}{^1_D}$,...,${\overset{(l,d)}{\blacksquare}}{^L_D}$>(0) ),  $\hspace{2px}$// $I_D$

                $\fsize{d}{\dplus}{MDA}{MDA}$( I<<$\hspace*{2px}d\hspace*{2px}$>><${\overset{(l,d)}{\boxplus}}{^1_d}$,...,${\overset{(l,d)}{\boxplus}}{^L_d}$>(0) ), // $P$
                $\fsize{d}{\dplus}{MDA}{MDA}$( I<<$\hspace*{2px}d\hspace*{2px}$>><${\overset{(l,d)}{\boxtimes}}{^1_d}$,...,${\overset{(l,d)}{\boxtimes}}{^L_d}$>(0) )   $\ $// $Q$
          >( ll_inp_mda<<l,d>><${\overset{(l,d)}{\blacktriangledown}}{^1_1}$ ,..., ${\overset{(l,d)}{\blacktriangledown}}{^L_D}$> )
\end{lstlisting}
Here, functions
$
\fsize{1}{\dplus}{MDA}{MDA},\dotsc,\fsize{D}{\dplus}{MDA}{MDA}
$
are the index set functions $id$ of combine operator concatenation~$\dplus_{1},\dotsc,\dplus_{D}$~(Example~\ref{def:mda_concat}), and ${\overset{(\bullet)}{\blacksquare}}{^\mathfrak{l}_\mathfrak{d}}$, ${\overset{(\bullet)}{\boxplus}}{^\mathfrak{l}_\mathfrak{d}}$, ${\overset{(\bullet)}{\boxtimes}}{^\mathfrak{l}_\mathfrak{d}}$, for $\bullet\in\MDHLVL$, are textually replaced by:
\begin{align*}
&&&
{\overset{(\bullet)}{\blacksquare}}{^\mathfrak{l}_\mathfrak{d}} \ := \
\begin{cases}
\ p\_(\mathfrak{l},\mathfrak{d}) &: \ \tp{\sigma_{\downarrow\texttt{-ord}}}(\mathfrak{l},\mathfrak{d}) \ < \ \bullet \\[10pt]
\ [p\_(\mathfrak{l},\mathfrak{d}),\tp{\texttt{\#PRT}(\mathfrak{l},\mathfrak{d})})_{\INz} &: \ (\mathfrak{l},\mathfrak{d}) \  \: = \hspace*{30.5px} \bullet \\[10pt]
\ [0,\tp{\texttt{\#PRT}(\mathfrak{l},\mathfrak{d})})_{\IN_0} &: \ \tp{\sigma_{\downarrow\texttt{-ord}}}(\mathfrak{l},\mathfrak{d}) \ > \ \bullet
\end{cases}
\\
&&&
\\
&&&
{\overset{(\bullet)}{\boxplus}}{^\mathfrak{l}_\mathfrak{d}} \ := \
\begin{cases}
\ p\_(\mathfrak{l},\mathfrak{d}) &: \ \tp{\sigma_{\downarrow\texttt{-ord}}}(\mathfrak{l},\mathfrak{d}) \ < \ \bullet \\[10pt]
\ p\_(\mathfrak{l},\mathfrak{d}) &: \ (\mathfrak{l},\mathfrak{d}) \  \: = \hspace*{30.5px} \bullet \\[10pt]
\ [0,\tp{\texttt{\#PRT}(\mathfrak{l},\mathfrak{d})})_{\IN_0} &: \ \tp{\sigma_{\downarrow\texttt{-ord}}}(\mathfrak{l},\mathfrak{d}) \ > \ \bullet
\end{cases}
\\
&&&
\\
&&&
{\overset{(\bullet)}{\boxtimes}}{^\mathfrak{l}_\mathfrak{d}} \ := \
\begin{cases}
\ p\_(\mathfrak{l},\mathfrak{d}) &: \ \tp{\sigma_{\downarrow\texttt{-ord}}}(\mathfrak{l},\mathfrak{d}) \ < \ \bullet \\[10pt]
\ (p\_(\mathfrak{l},\mathfrak{d}),\tp{\texttt{\#PRT}(\mathfrak{l},\mathfrak{d})})_{\INz} &: \ (\mathfrak{l},\mathfrak{d}) \  \: = \hspace*{30.5px} \bullet \\[10pt]
\ [0,\tp{\texttt{\#PRT}(\mathfrak{l},\mathfrak{d})})_{\IN_0} &: \ \tp{\sigma_{\downarrow\texttt{-ord}}}(\mathfrak{l},\mathfrak{d}) \ > \ \bullet
\end{cases}
\end{align*}

\noindent
where $<$ is defined as lexicographical order, according to Section~\ref{sec_app_code_gen_0_4}.

\noindent
Note that we re-use \texttt{ll\_inp\_mda<<$l$,$d$>>} for the intermediate results given by different iterations of variable \texttt{p\_({l},{d})}.
Correctness is ensured, as it holds:
\[
A \subseteq B \ \ \Rightarrow \ \ \fsize{d}{\dplus}{MDA}{MDA}(A) \ \subseteq \ \fsize{d}{\dplus}{MDA}{MDA}(B)
\]
MDA \texttt{ll\_inp\_mda<<$l$,$d$>>} has the following type when used for the intermediate result in a particular iteration of \texttt{p\_({l},{d})}:
\begin{align*}
\fsize{1}{\dplus}{MDA}{MDA}( \, I_1^{<{\overset{(l,d)}{\blacksquare}}{^1_1},...,{\overset{(l,d)}{\blacksquare}}{^1_D}\,|\,\dotsc\,|\,{\overset{(l,d)}{\blacksquare}}{^L_1},...,{\overset{(l,d)}{\blacksquare}}{^L_D}>} \, )
\ \times \ \dotsc \ \times \
\fsize{D}{\dplus}{MDA}{MDA}( \, I_D^{<{\overset{(l,d)}{\blacksquare}}{^1_1},...,{\overset{(l,d)}{\blacksquare}}{^1_D}\,|\,\dotsc\,|\,{\overset{(l,d)}{\blacksquare}}{^L_1},...,{\overset{(l,d)}{\blacksquare}}{^L_D}>} \, )
\ \to \ T^\texttt{INP}
\end{align*}
Here, for a set $P\subseteq[0,\tp{\texttt{\#PRT}({l},{d})})_{\IN_0}$, index set
$
  I_d^{<\dotsc\,|\,\dotsc\, P \,\dotsc\,|\,\dotsc>}
$
denotes
$
  \bigcup_{p^l_d\in P} I_d^{<\dotsc\,|\,\dotsc\, p^l_d \,\dotsc\,|\,\dotsc>}
$.

\vspace*{20px}

\subsection{Scalar Phase}
\label{app_code_gen_scalar}

Listing~\ref{lst_app_code_gen_2} shows our implementation of the scalar phase (Figure~\ref{fig_gen_ll}).

\begin{lstlisting}[
  caption={Scalar Phase},
  label={lst_app_code_gen_2}
  ]
// 2. scalar phase
int p_$\tp{\sigma_{f\texttt{-ord}}(1,1)}$ $\in^{<\tp{\leftrightarrow_{f\texttt{-ass}}(1,1)>}}$ $\tp{\texttt{\#PRT}}$($\tp{\sigma_{f\texttt{-ord}}(1,1)}$)
  $\ddots$
    int p_$\tp{\sigma_{f\texttt{-ord}}(L,D)}$ $\in^{<\tp{\leftrightarrow_{f\texttt{-ass}}(L,D)>}}$ $\tp{\texttt{\#PRT}}$($\tp{\sigma_{f\texttt{-ord}}(L,D)}$)
    {
      (
        ll_out_mda<<$f$>><<
            p_(1,1) ,..., p_(1,D),
                            ...
            p_(L,1) ,..., p_(L,D)>><<$b,a$>>(
                $\fsize{1}{\co{}}{MDA}{MDA}$( I<<$1$>><p_(1,1),...,p_(L,1)>(0) ) ,
                $\vdots$
                $\fsize{D}{\co{}}{MDA}{MDA}$( I<<$D$>><p_(1,D),...,p_(L,D)>(0) ) )
        )$_{b\in[1,B^\texttt{OB}]_\IN, a\in[1,A^\texttt{OB}_b]_\IN}$ := f( ( ll_inp_mda<<$f$>><< p_(1,1) ,..., p_(1,D),
                                                 |\hspace*{7px}|...
                                                 |\hspace*{7px}|p_(L,1) ,..., p_(L,D)>><<$b,a$>>(
                                    ${\fsize{d}{\dplus}{MDA}{MDA}}$( I<<$1$>><p_(1,1),...,p_(L,1)>(0) ) ,
                                        $\vdots$
                                    ${\fsize{d}{\dplus}{MDA}{MDA}}$( I<<$D$>><p_(1,D),...,p_(L,D)>(0) ) )
                               )$_{b\in[1,B^\texttt{IB}]_\IN, a\in[1,A^\texttt{IB}_b]_\IN}$ )
    }
\end{lstlisting}

\newpage

\subsection{Re-Composition Phase}
\label{app_code_gen_comp}

Listing~\ref{lst_app_code_gen_3} shows our implementation of the re-composition phase (Figure~\ref{fig_gen_ll}).

\begin{lstlisting}[
  caption={Re-Composition Phase},
  label={lst_app_code_gen_3}
  ]
// 3. re-composition phase

  // 3.1. main
  int p_$\tp{\sigma_{\uparrow\texttt{-ord}}(1,1)}$ $\in^{<\tp{\leftrightarrow_{\uparrow\texttt{-ass}}(1,1)}>}$ $\tp{\texttt{\#PRT}}$($\tp{\sigma_{\uparrow\texttt{-ord}}(1,1)}$)
  {
    int p_$\tp{\sigma_{\uparrow\texttt{-ord}}(1,2)}$ $\in^{<\tp{\leftrightarrow_{\uparrow\texttt{-ass}}(1,2)}>}$ $\tp{\texttt{\#PRT}}$($\tp{\sigma_{\uparrow\texttt{-ord}}(1,2)}$)
    {
      $\ddots$
        int p_$\tp{\sigma_{\uparrow\texttt{-ord}}(L,D)}$ $\in^{<\tp{\leftrightarrow_{\uparrow\texttt{-ass}}(L,D)}>}$ $\tp{\texttt{\#PRT}}$($\tp{\sigma_{\uparrow\texttt{-ord}}(L,D)}$)
        {
          ll_out_mda<<$\tp{\sigma_{\uparrow\texttt{-ord}}(L,D)}$>> :=$_{\texttt{co<}\tp{\scriptstyle\sigma_{\uparrow\texttt{-ord}}(L,D)}\texttt{>}}$ ll_out_mda<<f>>;$\label{lst_app_code_gen_3_1}$
        }
      $\udots$
      ll_out_mda<<$\tp{\sigma_{\uparrow\texttt{-ord}}(1,2)}$>> :=$_{\texttt{co<}\tp{\scriptstyle\sigma_{\uparrow\texttt{-ord}}(1,2)}\texttt{>}}$ ll_out_mda<<$\tp{\sigma_{\uparrow\texttt{-ord}}(1,3)}$>>;$\label{lst_app_code_gen_3_2}$
    }
    ll_out_mda<<$\tp{\sigma_{\uparrow\texttt{-ord}}(1,1)}$>> :=$_{\texttt{co<}\tp{\scriptstyle\sigma_{\uparrow\texttt{-ord}}(1,1)}\texttt{>}}$ ll_out_mda<<$\tp{\sigma_{\uparrow\texttt{-ord}}(1,2)}$>>;$\label{lst_app_code_gen_3_3}$
  }

  // 3.2. finalization
  ll_out_mda<<$\bot$>> := ll_out_mda<<$\tp{\sigma_{\uparrow\texttt{-ord}}(1,1)}$>>
\end{lstlisting}

\noindent
where
\begin{lstlisting}[numbers=none,frame=none]
ll_out_mda<<$l$,$d$>> :=$_{\texttt{co<}l,d\texttt{>}}$ ll_out_mda<<$l'$,$d'$>>
\end{lstlisting}
abbreviates
\begin{lstlisting}[numbers=none,frame=none]
ll_out_mda<<$l$,$d$>><${\overset{(l,d)}{\blacktriangle}}{^1_1}$ ,..., ${\overset{(l,d)}{\blacktriangle}}{^L_D}$>
:= co<<d>><$\fsize{1}{\co{}}{MDA}{MDA}$( I<<$\hspace*{2px}1\hspace*{2px}$>><${\overset{(l,d)}{\blacksquare}}{^1_1}$,...,${\overset{(l,d)}{\blacksquare}}{^L_1}\hspace*{0px}$>(0) ), // $I_1$
            $\vdots$                                  // ...,$I_{d-1}$,$I_{d+1}$,...
           $\fsize{D}{\co{}}{MDA}{MDA}$( I<<$D$>><${\overset{(l,d)}{\blacksquare}}{^1_D}$,...,${\overset{(l,d)}{\blacksquare}}{^L_D}$>(0) ), $\hspace{1px}$// $I_D$

           $\fsize{d}{\co{}}{MDA}{MDA}$( I<<$\hspace*{2px}d\hspace*{2px}$>><${\overset{(l,d)}{\boxplus}}{^1_d}$,...,${\overset{(l,d)}{\boxplus}}{^L_d}$>(0) ), // $P$
           $\fsize{d}{\co{}}{MDA}{MDA}$( I<<$\hspace*{2px}d\hspace*{2px}$>><${\overset{(l,d)}{\boxtimes}}{^1_d}$,...,${\overset{(l,d)}{\boxtimes}}{^L_d}$>(0) )   $\ $// $Q$
          >( ll_out_mda<<$l\,\,$,$d\,\,\,$>><${\overset{(l,d)}{\blacktriangle}}{^1_1} \ $ ,..., $ \ {\overset{(l,d)}{\blacktriangle}}{^L_D}$> ,
             ll_out_mda<<$l'$,$d'$>><${\overset{(l',d')}{\blacktriangle}}{^1_1}$ ,..., ${\overset{(l',d')}{\blacktriangle}}{^L_D}$> )
\end{lstlisting}
Here, functions
$
\fsize{1}{\co{}}{MDA}{MDA},\dotsc,\fsize{D}{\co{}}{MDA}{MDA}
$
are the index set function
of combine operators $\co{1},\dotsc,\co{D}$ (Definition~\ref{def_combine_op}), and ${\overset{(\bullet)}{\blacksquare}}{^\mathfrak{l}_\mathfrak{d}}$, ${\overset{(\bullet)}{\boxplus}}{^\mathfrak{l}_\mathfrak{d}}$, ${\overset{(\bullet)}{\boxtimes}}{^\mathfrak{l}_\mathfrak{d}}$, for $\bullet\in\MDHLVL$, are textually replaced by:
\begin{align*}
    &&&
    {\overset{(\bullet)}{\blacksquare}}{^\mathfrak{l}_\mathfrak{d}} \ := \
    \begin{cases}
    \ p\_(\mathfrak{l},\mathfrak{d}) &: \ \tp{\sigma_{\uparrow\texttt{-ord}}}(\mathfrak{l},\mathfrak{d}) \ < \ \bullet \\[10pt]
    \ [0,p\_(\mathfrak{l},\mathfrak{d})]_{\INz} &: \ (\mathfrak{l},\mathfrak{d}) \  \: = \hspace*{30.5px} \bullet \\[10pt]
    \ [0,\tp{\texttt{\#PRT}(\mathfrak{l},\mathfrak{d})})_{\IN_0} &: \ \tp{\sigma_{\uparrow\texttt{-ord}}}(\mathfrak{l},\mathfrak{d}) \ > \ \bullet
    \end{cases}
    \\
    &&&
    \\
    &&&
    {\overset{(\bullet)}{\boxplus}}{^\mathfrak{l}_\mathfrak{d}} \ := \
    \begin{cases}
    \ p\_(\mathfrak{l},\mathfrak{d}) &: \ \tp{\sigma_{\uparrow\texttt{-ord}}}(\mathfrak{l},\mathfrak{d}) \ < \ \bullet \\[10pt]
    \ [0,p\_(\mathfrak{l},\mathfrak{d}))_{\IN_0} &: \ (\mathfrak{l},\mathfrak{d}) \  \: = \hspace*{30.5px} \bullet \\[10pt]
    \ [0,\tp{\texttt{\#PRT}(\mathfrak{l},\mathfrak{d})})_{\IN_0} &: \ \tp{\sigma_{\uparrow\texttt{-ord}}}(\mathfrak{l},\mathfrak{d}) \ > \ \bullet
    \end{cases}
    \\
    &&&
    \\
    &&&
    {\overset{(\bullet)}{\boxtimes}}{^\mathfrak{l}_\mathfrak{d}} \ := \
    \begin{cases}
    \ p\_(\mathfrak{l},\mathfrak{d}) &: \ \tp{\sigma_{\uparrow\texttt{-ord}}}(\mathfrak{l},\mathfrak{d}) \ < \ \bullet \\[10pt]
    \ p\_(\mathfrak{l},\mathfrak{d}) &: \ (\mathfrak{l},\mathfrak{d}) \  \: = \hspace*{30.5px} \bullet \\[10pt]
    \ [0,\tp{\texttt{\#PRT}(\mathfrak{l},\mathfrak{d})})_{\IN_0} &: \ \tp{\sigma_{\uparrow\texttt{-ord}}}(\mathfrak{l},\mathfrak{d}) \ > \ \bullet
    \end{cases}
\end{align*}
where $<$ is defined as lexicographical order, according to Section~\ref{sec_app_code_gen_0_4}.

\noindent
Note that we assume for index set functions $\fsize{d}{\co{}}{MDA}{MDA}$ that
\[
A \subseteq B \ \ \Rightarrow \ \ \fsize{d}{\co{}}{MDA}{MDA}(A) \ \subseteq \ \fsize{d}{\co{}}{MDA}{MDA}(B)
\]
(which holds for all kinds of index set functions
used in this paper, e.g., in Examples~\ref{def:mda_concat} and~\ref{def:mda_pw}) so that we can re-use \texttt{ll\_out\_mda<<$l$,$d$>>} for the intermediate results given by different iterations of~\texttt{p\_(\texttt{l},\texttt{d})}.
MDA \texttt{ll\_out\_mda<<$l$,$d$>>} has the following type when used for the intermediate result in a particular iteration of variable \texttt{p\_({l},{d})}:
\begin{align*}
\fsize{1}{\co{}}{MDA}{MDA}( \, I_1^{<{\overset{(l,d)}{\blacksquare}}{^1_1},...,{\overset{(l,d)}{\blacksquare}}{^1_D}\,|\,\dotsc\,|\,{\overset{(l,d)}{\blacksquare}}{^L_1},...,{\overset{(l,d)}{\blacksquare}}{^L_D}>} \, )
\ \times \ \dotsc \ \times \
\fsize{D}{\co{}}{MDA}{MDA}( \, I_D^{<{\overset{(l,d)}{\blacksquare}}{^1_1},...,{\overset{(l,d)}{\blacksquare}}{^1_D}\,|\,\dotsc\,|\,{\overset{(l,d)}{\blacksquare}}{^L_1},...,{\overset{(l,d)}{\blacksquare}}{^L_D}>} \, )
\ \to \ T^\texttt{OUT}
\end{align*}

Note that in lines~\ref{lst_app_code_gen_3_1},~\ref{lst_app_code_gen_3_2},~\ref{lst_app_code_gen_3_3} of Listing~\ref{lst_app_code_gen_3}, we implicitly override the uninitialized value in \texttt{ll\_out\_mda<<l,d>>} (not explicitly stated in the listing for brevity), thereby avoiding initializing output buffers with neutral elements of combine operators.

\vspace*{15px}

\section{Code-Level Optimizations}
\label{ch_ll_opts}
We consider optimizations that operate below the abstraction level of our low-level representation~(Section~\ref{ch:low_level}) as \emph{code-level optimizations}.
For some code-level optimizations, e.g., loop fusion,
we do not want to rely on the underlying compiler (e.g., the OpenMP/CUDA/OpenCL compiler):~we exactly know the structure of our code presented in Section~\ref{app_sec_code_generation} and thus, we are able to implement code-level optimizations without requiring complex compiler analyses for optimizations.

We outline our conducted code-level optimizations, which our systems performs automatically for the underlying compiler (OpenMP, CUDA, OpenCL, etc).
Our future work will thoroughly discuss our code-level optimizations and how we apply them to our generated program code.
Some code-level optimizations, such as loop unrolling, are (currently) left to the underlying compiler, e.g., the OpenMP, CUDA, or OpenCL compiler.
In our future work, we aim to incorporate code-level optimizations, as loop unrolling, into our auto-tunable optimization process.

\subsubsection*{Loop Fusion}
In Listings~\ref{lst_app_code_gen_1},~\ref{lst_app_code_gen_2},~\ref{lst_app_code_gen_3}, the lines containing symbol "$\in$" are mapped to \texttt{for}-loops.
These loops can often be fused;~for example, when parameters \texttt{D1}, \texttt{S1}, \texttt{R1} as well as parameters \texttt{D2}, \texttt{S2}, \texttt{R2} in Table~\ref{tab_tps} coincide (as in Figure~\ref{fig_ll_example}).
Besides reducing the overhead caused by loop control structures, loop fusion in particular allows significantly reducing the memory footprint:~we can re-use the same memory region for each BUF partition (Definition~\ref{ll_buffers}), rather than allocating memory for all the partitions.

\subsubsection*{Buffer Elimination}
In Listing~\ref{lst_app_code_gen_0_4}, we allocate BUFs for each combination of a layer and dimension.
However, when memory regions and memory layouts of BUFs coincide, we can avoid a new BUF allocation, by re-using the BUF of the upper level, thereby again reducing memory footprint.

\subsubsection*{Buffer Size Reduction}
We can reduce the sizes of BUFs when specific classes of index functions are used for views (Definitions~\ref{def_iv} and~\ref{def_ov}).
For example, in the case of \emph{Dot Product (\texttt{DOT})} (Figure~\ref{fig_hl_examples}), when accessing its input in a strided fashion~--~via index function $k\mapsto(2*k)$, instead of function $k\mapsto(k)$ (as in Figure~\ref{fig_hl_examples})~--~we would have to allocate BUFs (Listing~\ref{lst_app_code_gen_0_4}, lines~\ref{app_lst_ll_buffers_alloc_LVL_INP} and~\ref{app_lst_ll_buffers_alloc_f_INP}) of size $2*K$ for an input size of $K\in\IN$;~in these BUFs, each second value (accessed via indices $2*k+1$) would be undefined.
We avoid this waste of memory, by using index function $k\mapsto(k)$ instead of $k\mapsto(2*k)$ for allocated BUFs (Listing~\ref{lst_app_code_gen_0_5}, lines~\ref{app_lst_mda_acc_1}-\ref{app_lst_mda_acc_2} and~\ref{app_lst_mda_acc_3}-\ref{app_lst_mda_acc_4}, case "$\bullet\neq\bot$"), which avoids index functions' leading factors and potential constant additions.
Thereby, we reduce the memory footprint from $2*K$ to $K$.
Furthermore, according to our partitioning strategy (Listing~\ref{lst_app_code_gen_0_5}, line~\ref{app_lst_mda_part}, and Listings~\ref{lst_app_code_gen_1},~\ref{lst_app_code_gen_2},~\ref{lst_app_code_gen_3}), we often access BUFs via offsets: $k\mapsto(2*k)$ for $k\in\{OS+k' \ | \ k'\in\IN\}$ and offset $OS\in\IN$.
We avoid such offset by using $k\mapsto(\,2*(k-OS)\,)$, \mbox{thereby further reducing memory footprint.}

\subsubsection*{Memory Operation Minimization}
In our code, we access BUFs uniformly via MDAs (Listing~\ref{lst_app_code_gen_0_5}), which may cause unnecessary memory operations.
For example, in the de-composition phase (Listing~\ref{lst_app_code_gen_1}) of, e.g., matrix multiplication (\texttt{MatMul}) (Figure~\ref{fig_hl_examples}), we iterate over all dimensions of the input (i.e., the $i,j,k$ dimensions) for de-composition (Listing~\ref{lst_app_code_gen_0_1_2_b}).
However, the $A$ input matrix of \texttt{MatMul} is accessed via MDA indices $i$ and $k$ only (Figure~\ref{fig_hl_examples}).
We avoid these unnecessary memory operations ($J$-many in the case of an input MDA of size $J$ in dimension $j$) by using index $0$ only in dimension $j$ for the de-composition of the $A$ matrix.
Analogously, we use index $0$ only in the $i$ dimension for the de-composition of \texttt{MatMul}'s $B$ matrix which is accessed via MDA indices $k$ and $j$ only.
Moreover, we exploit all available parallelism for memory copy operations.
For example, for \texttt{MatMul}, we use also the threads intended for the $j$ dimension when de-composing the $A$ matrix, and we use the threads in the $i$ dimension for the $B$ matrix.
For this, we flatten the thread ids over all dimensions $i,j,k$ and re-structure them only in dimensions $i,k$ (for the $A$ matrix) or $k,j$ (for the $B$ matrix).

\subsubsection*{Constant Substitution}

We use constants whenever possible.
For example, in CUDA, variable \texttt{threadIdx.x} returns the thread id in dimension $x$.
However, in our code, we use constant \texttt{0} instead of \texttt{threadIdx.x} when only one thread is started in dimension \texttt{x}, enabling the CUDA compiler to significantly simplify arithmetic expressions.

\end{appendix}

\end{document}